\documentclass[12pt,letterpaper]{article}
\def\blind{1} % 0 = anonymous; 1 = with authors
\def\documentpart{0} % 0 = full; 1 = main; 2 = supplement
\input{preamble-shared}
\ifnum\documentpart=0
  \newcommand{\suppref}[1]{\ref{#1}}
  
  \newcommand{\mainref}[1]{\ref{#1}}
  \newcommand{\maineqref}[1]{\eqref{#1}}
\else
  \usepackage{xr}
  \ifnum\documentpart=1
    \newcommand{\suppref}[1]{\ref*{external-#1}}
    
    \newcommand{\mainref}[1]{\ref{#1}}
    \newcommand{\maineqref}[1]{\eqref{#1}}
  \else
    \newcommand{\suppref}[1]{\ref{#1}}
    
    \newcommand{\mainref}[1]{\ref*{external-#1}}
    \newcommand{\maineqref}[1]{\textup{(\ref*{external-#1})}}
  \fi
\fi
\begin{document}
\ifnum\documentpart<2
  \def\spacingset#1{\renewcommand{\baselinestretch}%
{#1}\small\normalsize} \spacingset{1}

%%%%%%%%%%%%%%%%%%%%%%%%%%%%%%%%%%%%%%%%%%%%%%%%%%%%%%%%%%%%%%%%%%%%%%%%%%%%%%

\if1\blind
{
  \title{\bf \papertitle}
\author[1]{Daeyoung Ham\thanks{These authors contributed equally to this work.}}
\author[2]{Myeonghun Yu\textsuperscript{*}}
\author[3]{Tate Jacobson}

\affil[1]{Department of Statistics and Data Science, University of Texas at San Antonio}
\affil[2]{Department of Statistics, Ewha Womans University}
\affil[3]{Department of Statistics and Data Science, Oregon State University}

  \maketitle
} \fi

\if0\blind
{
  \bigskip
  \bigskip
  \bigskip
  \begin{center}
    {\LARGE\bf \papertitle}
\end{center}
  \medskip
} \fi

\bigskip

% NOTE: Should be no more than 35 pages, including the title page, abstract, figures, tables, and references. 
\begin{abstract}
With the growing prevalence of streaming data, renewable estimation has become an important tool for updating statistical analyses without retaining historical raw observations. We study estimation and linear hypothesis testing in high-dimensional online $M$-estimation. Existing renewable procedures based on quadratic approximations can incur first-order approximation errors because the gradient of the unpenalized loss generally does not vanish at a preceding penalized estimator, posing a particular challenge for valid inference. To address this issue, we propose a corrected renewable loss that preserves historical first-order information and combine it with partial folded-concave regularization and an iterative local linear approximation. We establish finite-sample $\ell_1$- and $\ell_2$-error bounds, together with contraction bounds and strong oracle properties for both the constrained and unconstrained iterative estimators. Building on these results, we derive oracle Bahadur representations and establish asymptotically valid Wald and score tests for general linear hypotheses, allowing both the testing-oracle dimension and the number of restrictions to diverge. We further characterize their asymptotic power under local alternatives. The proposed procedure requires only recursively updated summaries while recovering the first-order inferential behavior of its pooled-data oracle counterpart. Simulation studies and real-data applications demonstrate the finite-sample performance of the proposed estimation and testing procedures.

\end{abstract}

\noindent%
\vfill

\def\myspacing{1.7}
\spacingset{\myspacing}

\clearpage

  \section{Introduction}\label{sec:intro}

Many contemporary applications involve data that arrive sequentially over time, including financial transactions, electronic health records, sensor networks, and online platforms. The continual accumulation of data offers the prospect of increasingly precise estimation and inference, but also poses substantial computational and data-management challenges: repeatedly refitting a statistical model using the entire historical sample can become prohibitively expensive, while retaining all past observations may be impractical because of storage, communication, or privacy constraints. These challenges call for statistical procedures that can update existing estimates and inferential quantities as each new batch of data arrives, using a limited collection of summaries retained from the historical data.
Such coordinatewise guarantees do not directly establish the validity of joint tests when both the number of tested coefficients and the number of linear restrictions diverge. Our framework accommodates this setting by developing renewable Wald and score tests that allow both dimensions to increase with the cumulative sample size.

Such considerations motivate \emph{renewable} or \emph{online} statistical procedures. \citet{schifano2016online} study online updating for linear models and estimating equations, while \citet{luo2020renewable} introduce renewable estimation and incremental inference for generalized linear models (GLMs) by replacing historical data with suitably chosen summary statistics and combining these summaries with the current batch. This framework has subsequently been extended to other models, including quantile regression \citep{jiang2022renewable} and Huber regression \citep{jiang2024renewable}. In high-dimensional settings, where the number of covariates may exceed the available sample size, \citet{han2021online} and \citet{luo2023online} develop renewable Lasso procedures with debiased inference for linear and generalized linear models, respectively, with related developments for high-dimensional single-index models \citep{han2024inference}, smoothed quantile regression \citep{xie2024statistical}, and support vector machines \citep{zhang2026renewable}. A complementary line of work combines stochastic optimization with online debiasing to perform sequential inference on individual regression coefficients \citep{han2024online,han2026adaptive}. These debiasing-based approaches are primarily designed for coordinatewise inference and typically require auxiliary estimation of the relevant inverse information matrix, often through nodewise regression \citep{vdGeer2014,zhang2014confidence}.

For linear hypothesis testing on a low-dimensional subvector $\betab^*_{\mathcal M}$ of the true parameter $\betab^*$, partial penalization provides an alternative to debiasing when the full dataset is available. 
Instead of penalizing all coordinates and subsequently correcting the resulting regularization bias, partial penalization leaves the coordinates involved in the hypothesis unpenalized while regularizing the remaining components. 
\citet{shi2019linear} develop likelihood ratio, Wald, and score tests for linear hypotheses in high-dimensional GLMs using constrained and unconstrained partial penalization. 
Their approach employs folded-concave penalties on the nuisance coordinates to reduce the shrinkage bias associated with the Lasso and achieve consistent variable selection.
An important advantage of partial penalization is its behavior under local alternatives. Consider alternatives of the form $H_a:\betab^*_{\mc M}=\mathbf h_n$ with $\|\mathbf h_n\|_2\ll\lambda_n$, where $\lambda_n$ denotes the regularization level and $\|\cdot\|_2$ is the Euclidean norm. Under such alternatives, penalization may shrink the estimated coefficients in $\mathcal M$ toward zero, thereby attenuating the associated Wald statistic and reducing its ability to distinguish the alternative from the null. 
In contrast, by leaving the tested coordinates unpenalized, partial penalization avoids this attenuation and can retain nontrivial power against such local alternatives. \citet{shi2019linear} further show that, in the settings considered there, the resulting tests can exhibit greater local power than corresponding debiasing-based procedures; see Appendix~S1 of \citet{shi2019linear}.

Because folded-concave penalties lead to nonconvex optimization problems with potentially multiple local minima, computing the corresponding partial-penalized estimators poses substantial computational challenges. To bridge the resulting gap between theory and computation, \citet{jacobson2025strong} employs the local linear approximation (LLA) algorithm \citep{zou2008one,fan2014strong}, which reduces the nonconvex problem to a sequence of weighted $\ell_1$-penalized problems, and establishes strong oracle properties for the resulting LLA iterates. 
Related work extends partial-penalization-based inference from GLMs to other high-dimensional models and data settings, including high-dimensional Tobit models \citep{jacobsonzou2025testing},
nonparametric Box--Cox transformation models \citep{zhouzou2024testing},
high-dimensional quantile regression with distributed data \citep{yang2026communication}, and general $M$-estimation in distributed settings \citep{gu2025collaborative}. However, these methods are not directly applicable to streaming data, where historical observations are no longer available once processed and subsequent estimation and inference must rely on recursively updated summaries.
Moreover, the renewable loss of \citet{luo2023online} cannot be directly combined with partial penalization. Because the gradient of the unpenalized empirical loss generally does not vanish at a penalized estimator, a purely quadratic surrogate omits a nonzero first-order term. The resulting approximation error therefore depends on the preceding estimation error at first order and can propagate across batches. This creates a particular difficulty for inference: even if the oracle support is correctly recovered at the terminal batch, an oracle refit based on the conventional renewable loss need not be first-order equivalent to its pooled-data counterpart. Consequently, terminal oracle recovery alone is insufficient to obtain the Bahadur representation required for valid partial-penalized inference; see Remark~\ref{remark:ls no info lost} for further discussion. This motivates a corrected renewable loss that preserves the historical first-order information needed for valid inference.

In this paper, we develop a framework for estimation and linear hypothesis testing in high-dimensional online $M$-estimation. The proposed procedure combines a corrected renewable loss that preserves historical first-order information with partial folded-concave regularization and a contraction-and-tightening strategy motivated by iterative local adaptive majorize-minimization (I-LAMM) \citep{fan2018lamm}. By preserving the historical first-order term, the resulting approximation error depends on the preceding estimation error only at the second order, thereby preventing first-order errors from propagating across batches. A renewable $\ell_1$ estimator serves as an initializer, after which partially penalized LLA updates are used to construct the constrained and unconstrained estimators required for testing. The entire procedure is implemented using recursively updated summaries and therefore does not require access to historical raw observations once a batch has been processed. Our main contributions are as follows.

First, we establish finite-sample error bounds for the renewable estimator under suitable regularity conditions, based on the corrected loss. This estimator serves as the initializer for the subsequent partial-penalized I-LAMM procedure. Under sparsity, we derive $\ell_2$- and $\ell_1$-error bounds that match the convergence rates of their offline counterparts. A further consequence of the first-order correction is that the constants in these bounds do not deteriorate with the batch index. In contrast, the corresponding error bounds in \citet{luo2023online} involve factors that can grow exponentially with the number of processed batches. Thus, the corrected renewable loss prevents first-order approximation errors from accumulating across successive updates.

Second, for the partially penalized I-LAMM procedure, we analyze both the constrained and unconstrained LLA sequences and establish contraction bounds and strong oracle properties for the resulting estimators. The corresponding oracle estimators retain all coordinates involved in the hypothesis together with the truly active nuisance coordinates. We show that, after finitely many LLA iterations, the computed constrained and unconstrained estimators coincide with these oracle estimators with high probability. Importantly, the required minimum-signal condition is imposed only on the active nuisance coefficients; no such condition is needed for the tested coordinates. Hence, the oracle recovery result remains valid when the tested coefficients are zero under the null or approach zero under local alternatives, which is essential for subsequent power analysis.

Third, we establish asymptotically valid renewable Wald and score tests for general linear hypotheses, allowing both the number of tested coefficients and the number of linear restrictions to diverge with the cumulative sample size. We derive oracle Bahadur representations for the constrained and unconstrained estimators, whose leading terms coincide with those of their pooled-data oracle counterparts. Under some growth conditions, the first-order correction renders the renewable approximation error negligible on the inferential scale. Combining these representations with consistent and renewable covariance estimation and a Gaussian approximation that accommodates a diverging number of restrictions, we establish chi-square approximations under the null and local alternatives. Thus, despite using only recursively updated summaries, the proposed tests recover the first-order inferential behavior of their pooled-data oracle counterparts.

The remainder of the paper is organized as follows.
Section~\ref{sec:method} introduces the methodology, and
Section~\ref{sec:theory} establishes its statistical properties.
Section~\ref{sec: tuning parameter selection} details tuning parameter selection.
Sections~\ref{sec:numerical} and~\ref{sec:real-data} present the simulation studies and real data analysis, respectively. The Supplementary Materials contain additional simulation reports, all technical proofs, and model-specific verifications of the general assumptions and theory for GLMs, Tobit models, and the nonparametric Box--Cox transformation model.

\section{Model setup and methodology}\label{sec:method}

Consider a time point $B \ge 2$, at which a total of $N_B = \sum_{b=1}^B n_b$ observations have arrived in $B$ sequential data batches $\mathcal{D}_B^\star = \bigcup_{b=1}^B \mathcal{D}_b$. For each $1 \leq b \leq B$, the batch $\mathcal{D}_b$ contains $n_b = |\mathcal{D}_b|$ i.i.d. observations $\{(Y_i^{(b)}, X_i^{(b)})\}_{i=1}^{n_b}$ drawn from the joint distribution of $(Y,X)$, where $Y \in \RR$ and $X \in \RR^p$ over the sample space $\mathcal{Y} \times \mathcal{X}$.

For $\betab \in \RR^p$, define the population risk $\mathcal{L}(\betab) = \EE\{\ell(\betab;Y,X)\}$, where $\ell:\RR^p \times \mathcal{Y} \times \mathcal{X} \to \RR$ is twice continuously differentiable and convex in $\betab$. Let $\betab^* = \argmin_{\betab \in \RR^p} \mathcal{L}(\betab)$, which is assumed to be unique.
We study the estimation of $\betab^*$ in a streaming-data setting and inference on a low-dimensional subvector of $\betab^*$ through the linear hypothesis
\#
H_0: \Cb\betab^*_{\mathcal{M}} = \tb. \label{hypothesis}
\#
Here $\mathcal{M} \subset \{1,\ldots,p\}$ indexes the coordinates under test, with $m = |\mathcal{M}| \ll p$, $\boldsymbol{C} \in \RR^{r \times m}$, and $\tb \in \RR^r$. Without loss of generality, we relabel the coordinates so that $\mathcal{M} = \{1,\ldots,m\}$ and $\mathcal{M}^c = \{m+1,\ldots,p\}$. We assume $\mathrm{rank}(\boldsymbol{C}) = r$, so the null restrictions are nonredundant.
Let $\mathcal{A} = \mathrm{supp}(\betab^*)$ denote the support of $\betab^*$, and we assume that the true model is sparse, that is, $|\mathcal{A}| \ll p$. Define $\mathcal{S} = \mathcal{M}^c \cap \mathcal{A}$ and let $s = |\mathcal{S}|$. Write $\overline{\mc M} = \mathcal{M} \cup \mathcal{S}$, and, by reordering indices if necessary, assume $\mathcal{S} = \{m+1,\ldots,m+s\}$.

\subsection{Online estimation via a corrected surrogate loss}
\label{sec:online-estimation}

If all observations were available offline, one could compute the pooled Lasso estimator $\bar{\betab}^{(B)} = \argmin_{\betab \in \RR^p} \left\{(1/N_B) \sum_{i \in \mathcal{D}_B^\star} \ell(\betab;Y_i,X_i) + \lambda \|\betab\|_1\right\}$.
In a strictly online setting, however, repeated access to the historical raw data $\mathcal D_{B-1}^\star$ is not feasible. We therefore construct a surrogate objective that can be updated recursively using only the current batch and stored summary statistics.

Define the batch-wise empirical loss
$\widehat{\mathcal{L}}^{(b)}(\betab) = (1/n_b) \sum_{i \in \mathcal{D}_b} \ell(\betab;Y_i,X_i)$,
and let $\widehat{\betab}^{(b)}$ denote the estimator updated with the $b$th batch. We further define the aggregated curvature matrix
$\widehat{\Jb}^{(b)} = \sum_{j=1}^b n_j \nabla^2 \widehat{\mathcal{L}}^{(j)}(\widehat{\betab}^{(j)})$.
A key limitation of existing online surrogate losses is that they depend on the previous estimators via first-order terms, yielding a non-negligible first-order remainder. To address this issue, we recursively define a corrected surrogate objective. 
Set $\widehat{\mathcal{Q}}_{1}(\betab) = \widehat{\mathcal{L}}^{(1)}(\betab)$, and for $b \geq 2$ let
\#
\widehat{\mathcal{Q}}_{b}(\betab) & = \frac{1}{N_b} \bigg\{ n_b\widehat{\mathcal{L}}^{(b)}(\betab) + \frac{1}{2}(\betab - \widehat{\betab}^{(b-1)})^\T \widehat{\Jb}^{(b-1)}(\betab - \widehat{\betab}^{(b-1)}) + N_{b-1}\nabla \widehat{\mathcal{Q}}_{b-1}(\widehat{\betab}^{(b-1)})^\T \betab \bigg\}. \label{def:q.function}
\#

The three terms in \eqref{def:q.function} play distinct roles. 
The term $n_b\widehat{\mathcal{L}}^{(b)}(\betab)$ incorporates the contribution of the current batch. 
The quadratic term conveys curvature information from the historical data via $\widehat{\Jb}^{(b-1)}$ and is centered on the most recent estimator $\widehat{\betab}^{(b-1)}$. 
The linear term $N_{b-1}\nabla \widehat{\mathcal{Q}}_{b-1}(\widehat{\betab}^{(b-1)})^\T \betab$ corrects the first-order remainder inherited from the previous surrogate objective.
It is worth noting that, in the low-dimensional unpenalized regime with $p \ll N_{b-1}$, the first-order optimality condition implies $\nabla \widehat{\mathcal{Q}}_{b-1}(\widehat{\betab}^{(b-1)}) = \mathbf{0}$. 
In that case, the correction term disappears, and $\widehat{\mathcal{Q}}_b(\betab)$ reduces to the original renewable surrogate objective of \citet{luo2020renewable}. Our construction can therefore be viewed as a bias-corrected extension of the conventional renewable objective for the penalized high-dimensional setting, where the gradient of the previous-step surrogate objective at $\widehat{\betab}^{(b-1)}$ is generally nonzero.

For implementation, store the first-order correction $\widehat{\mathbf g}^{(b)} = \nabla\widehat{\mc Q}_b(\widehat{\betab}^{(b)})$.
Given $\widehat{\mathcal{Q}}_b$, we update the estimator at batch $b$ by solving
\#
\widehat{\betab}^{(b)} \in \argmin_{\betab \in \RR^p} \left\{ \widehat{\mathcal{Q}}_b(\betab) + \lambda_b \|\betab\|_1 \right\}. \label{def:bth.estimator}
\#
The estimator in \eqref{def:bth.estimator} depends only on the current batch $\mathcal{D}_b$ and the stored summaries $\{\widehat{\betab}^{(b-1)}, \widehat{\Jb}^{(b-1)}, \widehat{\mathbf g}^{(b-1)}\}$, and hence can be updated without revisiting the historical data.

\begin{remark} \label{remark:ls no info lost}
To illustrate the importance of the first-order correction, consider the least-squares loss $\ell(\betab;Y,X)=\frac{1}{2}(Y-X^\top\betab)^2$ under the linear model $Y_i=X_i^\top\betab^*+\epsilon_i$.
The conventional renewable loss of \citet{luo2023online} is
$\widetilde{\mc Q}_b(\betab)
=
\frac{1}{N_b}
\left\{
n_b\widehat{\mathcal L}^{(b)}(\betab)
+
\frac{1}{2}
(\betab-\widehat{\betab}^{(b-1)})^\top
\widehat{\Jb}^{(b-1)}
(\betab-\widehat{\betab}^{(b-1)})
\right\}$.
At the second batch, suppose that $\widehat\betab^{(1)}$ is the Lasso estimator based on $\mc D_1$ with $\lambda_1\asymp\sqrt{\log p/n_1}$. Whenever $\widehat\betab^{(1)}\neq\mathbf 0$, the KKT conditions imply $\|\sqrt{N_2}\nabla\widetilde{\mc Q}_2(\betab^*) + \frac{1}{\sqrt{N_2}} \sum_{i\in\mc D_2^*}X_i\epsilon_i \|_\infty \asymp \sqrt{(n_1\log p)/N_2}$.
Hence, for the two scores in the preceding display to be asymptotically equivalent on the $\sqrt{N_2}$-scaled $\ell_\infty$ scale, we require $(n_1\log p)/N_2\to0$, which amounts to $n_2\gg n_1\log p$. Repeated renewable updates can accumulate this first-order discrepancy, leading to increasingly stringent batch-size requirements involving successive powers of $\log p$. Such growth is particularly restrictive in high-dimensional streaming settings. 
In contrast, the corrected surrogate eliminates this discrepancy by preserving the historical first-order term. 
Indeed, for every $b\geq1$, $\widehat{\mathcal Q}_b(\betab) = \frac{1}{N_b}\sum_{i\in\mc D_b^*}\ell(\betab;Y_i,X_i)+C_b$, where $C_b$ does not depend on $\betab$. Thus, for least squares, the corrected surrogate coincides with the pooled empirical loss up to an irrelevant additive constant, and renewable updating does not lose any information relevant to estimation or inference.
See Appendix~\suppref{sec:ls no info lost proof} for details.
\end{remark}

\begin{remark}
    During the final stages of preparing this manuscript, we became aware of two contemporaneous works, \citet{liu2026high} and \citet{gao2026renewable}, whose renewable constructions also preserve historical first-order information and are closely related to our corrected loss. Their theoretical objectives, however, differ from ours. Our correction arises directly from deriving the oracle Bahadur representation: retaining the historical first-order term reduces the dependence of the renewable approximation error on preceding estimation errors, making the resulting remainder sufficiently small for inference; see Remark~\ref{remark:ls no info lost} and Appendix~\suppref{sec:appendix_lemmas}. As a byproduct, we also obtain estimation error bounds whose constants remain uniform over the batch index; see Section~\ref{sec:finite_sample_theory}, in contrast to the batch-dependent bounds of \citet{luo2023online}. Related estimation guarantees are obtained by \citet{liu2026high}, who derive batch-uniform error bounds, signal-adaptive improvements in estimation rates, and approximate support recovery, and by \citet{gao2026renewable}, who establish nonasymptotic error bounds for gradient-enhanced Lasso estimation and extend their method to distributed streaming data. Our main contribution is inferential: we establish strong oracle properties for both the constrained and unconstrained LLA estimators and, building on these results, derive oracle Bahadur representations and asymptotically valid Wald and score tests for general linear hypotheses.
\end{remark}

\subsection{Inference on linear hypotheses at the final batch}

After processing the final batch, we turn to inference for the hypothesis in \eqref{hypothesis}. To motivate the proposed procedure, we first introduce oracle estimators under the idealized assumption that the support set $\mathcal{A}$ is known. Because the inferential target is $\betab_{\mathcal{M}}^*$ and $\mathcal{M}$ need not be contained in $\mathcal{A}$, the appropriate oracle benchmark for testing \eqref{hypothesis} is not the conventional support-restricted oracle based only on $\mathcal{A}$. Instead, following \citet{jacobson2025strong}, we define testing oracles on $\overline{\mathcal{M}} = \mathcal{M} \cup \mathcal{A}$, so that the coordinates in $\mathcal{M}$ are retained regardless of whether the corresponding components of $\betab^*$ are zero.

Specifically, we define the reduced renewable-oracle estimator and the full renewable-oracle estimator by
\begin{align}
&\widehat{\betab}_0^{{\rm ora}} = \argmin_{\betab_{\overline{\mathcal{M}}^c} = \mathbf{0}} \widehat{\mathcal{Q}}_B(\betab)
\quad \text{subject to} \quad \Cb\betab_{\mathcal{M}} = \tb, \label{def:constrained_oracle_estimator}\\
&\widehat{\betab}_a^{{\rm ora}} = \argmin_{\betab_{\overline{\mathcal{M}}^c} = \mathbf{0}} \widehat{\mathcal{Q}}_B(\betab), \label{def:unconstrained_oracle_estimator}
\end{align}
respectively. These oracle estimators are defined in terms of the corrected surrogate objective $\widehat{\mathcal{Q}}_B(\betab)$ rather than the full-sample loss. Moreover, they keep the coordinates in $\mathcal{M}$ regardless of whether the corresponding components of $\betab^*$ are zero, making them appropriate targets for inference under \eqref{hypothesis}.

In practice, however, the support set $\mathcal{A}$ is unknown, and the online estimator $\widehat{\betab}^{(B)}$ obtained from the $\ell_1$-regularized problem generally exhibits shrinkage bias, making it unsuitable for direct inference. 
To mitigate this bias and construct estimators that asymptotically mimic $\widehat{\betab}_0^{\rm ora}$ and $\widehat{\betab}_a^{\rm ora}$, we replace the $\ell_1$ penalty with a folded-concave penalty $p_\lambda(\cdot)$ and adopt a partial penalization scheme that leaves the coordinates in $\mathcal{M}$ unpenalized \citep{shi2019linear,jacobson2025strong}.

\begin{assumption}[Folded-concave penalties]\label{assm:penalty_function}
The penalty has the form $p_\lambda(t)=\lambda^2p(t/\lambda)$ for $t\ge0$, where
$p:[0,\infty)\to[0,\infty)$ is nondecreasing and concave with $p(0)=0$.  Its right derivative
$p'_+(t)$ exists for every $t\ge0$, is nonincreasing, and satisfies $p'_+(0)=1$.  In addition,
$p'_+(c_1)=0$ for some constant $c_1>0$.  We write $p'_{\lambda,+}(t)=\lambda p'_+(t/\lambda)$, $t\geq 0$.
\end{assumption}

Examples satisfying Assumption \ref{assm:penalty_function} include the smoothly clipped absolute deviation penalty \citep{fan2001variable} and the minimax concave penalty \citep{Zhang2010MCP}. The $\ell_1$ penalty satisfies all preceding shape conditions except the flat-tail condition.

Based on the final surrogate loss $\widehat {\mc Q}_B$, we define the reduced and full partially penalized estimators by
\#
\widehat{\betab}_0 \in \argmin_{\betab \in \RR^p} \left\{ \widehat{\mathcal{Q}}_B(\betab) + \sum_{j \notin \mathcal{M}} p_\lambda(|\beta_j|) \right\}
\quad \text{subject to} \quad \boldsymbol{C}\betab_{\mathcal{M}} = \tb,
\#
and
\#
\widehat{\betab}_a \in \argmin_{\betab \in \RR^p} \left\{ \widehat{\mathcal{Q}}_B(\betab) + \sum_{j \notin \mathcal{M}} p_\lambda(|\beta_j|) \right\},
\#
respectively. The coefficients in $\mathcal{M}$ are left unpenalized because they enter the null hypothesis directly, whereas the coordinates in $\mathcal{M}^c$ remain penalized. This partial penalization structure reduces bias in the parameters of interest without requiring a minimal-signal condition on $\betab^*_{\mathcal{M}}$, thereby helping preserve power against local alternatives \citep{shi2019linear,jacobson2025strong}.

Because $p_\lambda(\cdot)$ is nonconvex, the reduced and full optimization problems may admit multiple local minima. We therefore use local linear approximation (LLA) \citep{zou2008one,fan2014strong} to obtain computable estimators. Given an initial value $\tilde{\betab}$, we approximate the folded-concave penalty by
$\sum_{j \notin \mathcal{M}} p_\lambda(|\beta_j|) \approx \sum_{j \notin \mathcal{M}} \left\{ p_\lambda(|\tilde{\beta}_j|) + p'_{\lambda,+}(|\tilde{\beta}_j|)\big(|\beta_j| - |\tilde{\beta}_j|\big) \right\}$,
which yields a convex weighted $\ell_1$ surrogate. We initialize the procedure at $\widehat{\betab}^{(B)}$ and set $\widehat{\betab}_a^{({\rm LLA},0)}=\widehat{\betab}_0^{({\rm LLA},0)}=\widehat{\betab}^{(B)}$. For $k \geq 1$, define
\#
\widehat{\betab}_0^{({\rm LLA},k)} \in \argmin_{\betab \in \RR^p} \left\{ \widehat{\mathcal{Q}}_B(\betab) + \sum_{j \notin \mathcal{M}} \widehat{w}_{0,j}^{(k-1)} |\beta_j| \right\}
\quad \text{subject to} \quad \Cb \betab_{\mathcal{M}} = \tb, \label{def:LLA.0}
\#
and
\#
\widehat{\betab}_a^{({\rm LLA},k)} \in \argmin_{\betab \in \RR^p} \left\{ \widehat{\mathcal{Q}}_B(\betab) + \sum_{j \notin \mathcal{M}} \widehat{w}_{a,j}^{(k-1)} |\beta_j| \right\}, \label{def:LLA.a}
\#
where $\widehat{w}_{0,j}^{(k-1)} = p'_{\lambda,+}(|\widehat{\beta}_{0,j}^{({\rm LLA},k-1)}|)$ and $\widehat{w}_{a,j}^{(k-1)} = p'_{\lambda,+}(|\widehat{\beta}_{a,j}^{({\rm LLA},k-1)}|)$ for $j \notin \mathcal{M}$.

The LLA iterates replace the nonconvex folded-concave penalty with a sequence of convex weighted $\ell_1$ problems while preserving the partial-penalization structure required for testing. The resulting reduced and full estimators form the basis for the inferential procedures developed in Section~\ref{sec:theory}. We note that the constrained update in \eqref{def:LLA.0} is formulated in a renewable setting with linear restrictions, which differs substantially from the unconstrained renewable framework studied in \citet{luo2023online}.

The overall procedure proceeds in two stages. The first is a contraction stage, in which the initial estimator $\widehat{\betab}^{(B)}$ is explicitly computed in the renewable setting via the $\ell_1$-regularized corrected surrogate loss. 
The second is a tightening stage, which refines this initial estimator through the LLA iterates $\{\widehat{\betab}_0^{({\rm LLA},k)}, \widehat{\betab}_a^{({\rm LLA},k)}\}_{k \geq 1}$. 
As shown in Section~\ref{sec:theory}, the contraction stage yields a consistent initial estimator, while the tightening stage further improves the rate and recovers oracle-type behavior.

After computing the reduced- and full-model estimators, we construct partially penalized Wald and score test statistics for the hypothesis in \eqref{hypothesis}. Let $\Jb_0 = \nabla^2_{\overline{\mc M}} {\mc L}(\betab^*)$ and 
$\Kb_0
=
\EE\left\{
\nabla_{\overline{\mc M}}\ell(\betab^*;Y_i, X_i)
\nabla_{\overline{\mc M}}\ell(\betab^*;Y_i, X_i)^\T
\right\}$
denote, respectively, the oracle Hessian matrix and oracle score covariance matrix on the oracle index set $\overline{\mc M}$. 
% Similarly, let $\Jb_a = \nabla^2_{\mc M} {\mc L}(\betab^*)$ and 
% $\Kb_a
% =
% \EE\left\{
% \nabla_{\mc M}\ell(\betab^*;Y_i, X_i)
% \nabla_{\mc M}\ell(\betab^*;Y_i, X_i)^\T
% \right\}$
% denote the corresponding oracle quantities associated with the full model indexed by $\mc M$.
Define $\Vb_0 = \Cb_0 \Jb_0^{-1}\Kb_0 \Jb_0^{-1}\Cb_0^\T$, where $\Cb_0 = (\Cb, \mathbf{0}_{r \times s}) \in \RR^{r \times (m + s)}$.
Under the standard regularity conditions for maximum likelihood estimation, the second Bartlett identity implies that $\Jb_0 = \Kb_0$ and $\Jb_a = \Kb_a$ when $\ell(\cdot)$ is the correctly specified negative log-likelihood function \citep{mccullagh1989generalized}. Consequently, the sandwich covariance matrix reduces to $\Vb_0 = \Cb_0 \Jb_0^{-1}\Cb_0^\T$.

To construct test statistics for the hypothesis in \eqref{hypothesis}, define the second-moment summary for batch $b$ by
$\widehat{\Kb}^{(b)}
=(1/N_b)
\sum_{j=1}^b\sum_{i\in\mc D_j}
\nabla \ell(\widehat{\betab}^{(j)};Y_i^{(j)},X_i^{(j)})
\nabla \ell(\widehat{\betab}^{(j)};Y_i^{(j)},X_i^{(j)})^\T $.
Note that, like the surrogate Hessian summary matrix, this summary matrix can be updated online by $\widehat{\Kb}^{(b)}
=(N_{b-1}/N_b)\widehat{\Kb}^{(b-1)}
+(1/N_b)\sum_{i\in\mc D_b}
\nabla \ell(\widehat{\betab}^{(b)};Y_i^{(b)},X_i^{(b)})
\nabla \ell(\widehat{\betab}^{(b)};Y_i^{(b)},X_i^{(b)})^\T$, $b\geq 2$.
At the final batch $b = B$, for $e \in \{0,a\}$, let $\widehat{\mc S}^{(k)}_e = \operatorname{supp}\{\widehat{\betab}_{e,\mc M^c}^{({\rm LLA},k)}\}$, $\widehat{\overline{\mc M}}^{(k)}_e=\mc M\cup\widehat{\mc S}^{(k)}_e$.
Define $\widehat{\Cb}^{(k)}_e=(\Cb,\mathbf 0_{r\times |\widehat{\mc S}^{(k)}_e|})$ with columns ordered as $(\mc M,\widehat{\mc S}^{(k)}_e)$.
We denote $\widehat{\Jb}^{(k)}_e
=
\nabla^2_{\widehat{\overline{\mc M}}^{(k)}_e\widehat{\overline{\mc M}}^{(k)}_e}
\widehat{\mc Q}_B(\widehat{\betab}^{({\rm LLA},k)}_e)$, 
$\widehat{\Kb}^{(k)}_e
=
\widehat{\Kb}^{(B)}_{\widehat{\overline{\mc M}}^{(k)}_e\widehat{\overline{\mc M}}^{(k)}_e}$,
and let
$\widehat{\Vb}^{(k)}_e
=
\widehat{\Cb}^{(k)}_e
(\widehat{\Jb}^{(k)}_e)^{-1}
\widehat{\Kb}^{(k)}_e
(\widehat{\Jb}^{(k)}_e)^{-1}
(\widehat{\Cb}^{(k)}_e)^\T$,
for $e \in \{0,a\}$.
Then, the partially penalized Wald and score statistics are defined as
\begin{equation*}
\begin{split}
&T_W^{(k)}
=
N_B
(\Cb\widehat{\betab}^{({\rm LLA},k)}_{a,\mc M}-\tb)^\T
(\widehat{\Vb}^{(k)}_a)^{-1}
(\Cb\widehat{\betab}^{({\rm LLA},k)}_{a,\mc M}-\tb),\\
&T_S^{(k)}
=
N_B
\left\{
\nabla_{\widehat{\overline{\mc M}}^{(k)}_0}
\widehat{\mc Q}_B(\widehat{\betab}^{({\rm LLA},k)}_0)
\right\}^{\T}
\!(\widehat{\Jb}^{(k)}_0)^{-1}
\!(\widehat{\Cb}^{(k)}_0)^\T
\!(\widehat{\Vb}^{(k)}_0)^{-1}
\widehat{\Cb}^{(k)}_0
\!(\widehat{\Jb}^{(k)}_0)^{-1}
\left\{
\nabla_{\widehat{\overline{\mc M}}^{(k)}_0}
\widehat{\mc Q}_B(\widehat{\betab}^{({\rm LLA},k)}_0)
\right\}.
\end{split}
\end{equation*}
For comparison, we define the oracle Wald and score statistics, respectively, as
\#
T_W^{\rm ora} = N_B \bigl(\Cb_0 \widehat \betab_{a, {\overline{\mc M}}}^{\rm ora} - \tb \bigr)^\T \Vb_0^{-1} \bigl(\Cb_0 \widehat \betab_{a, \overline{\mc M}}^{\rm ora} - \tb \bigr),  \label{def:oracle.wald}
\#
and
\#
T_S^{\rm ora} = N_B \bigg\{\nabla_{\overline{\mc M}} \widehat{\mc Q}_B(\widehat{\betab}_0^{\rm ora})\bigg\}^\T \Jb_0^{-1}\Cb_0^\T \Vb_0^{-1}\Cb_0 \Jb_0^{-1}\bigg\{\nabla_{\overline{\mc M}} \widehat{\mc Q}_B(\widehat{\betab}_0^{\rm ora})\bigg\}. \label{def:oracle.score}
\#

Our theoretical results show that, under mild regularity conditions, the proposed partially penalized Wald and score statistics, $T_W^{(k)}$ and $T_S^{(k)}$, exhibit an oracle property: for all sufficiently large iteration indices $k$, their limiting distributions coincide with those of their oracle counterparts, $T_W^{\rm ora}$ and $T_S^{\rm ora}$, respectively. 
Consequently, either statistic may be used to test the linear hypothesis in \eqref{hypothesis}. 
Specifically, for a prescribed significance level $\alpha$, we reject $H_0$ whenever $T > \chi_\alpha^2(r)$, where $T$ denotes either $T_W^{(k)}$ or $T_S^{(k)}$, and $\chi_\alpha^2(r)$ is the upper $\alpha$ quantile of the $\chi^2$ distribution with $r$ degrees of freedom. These results are established rigorously in Section~\ref{sec:inference}.

\begin{remark}
We note that LLA may also be applied at intermediate batches, as examined in our simulations. Under the corresponding conditions of Theorem~\ref{thm:batch_b_tightening}, such tightening can reduce shrinkage bias and improve estimation rates, at additional computational cost. 
Our simulations show broadly similar inferential performance with and without intermediate LLA in the correctly specified likelihood examples (GLMs and Tobit), where \(\Jb_0=\Kb_0\), whereas batchwise LLA improves calibration in the composite-likelihood example (nonparametric Box--Cox), where \(\Jb_0\neq \Kb_0\). 
We therefore favor terminal-only LLA for computational considerations in the former settings and batchwise LLA in the latter when computation permits.

Under either scheme, all \(p\) coordinates remain eligible at subsequent batches; intermediate support estimates are not used to restrict later optimization. As \(N_b\) increases and the regularization level decreases, the beta-min requirement in Theorem~\ref{thm:oracle_recovery} becomes less restrictive, while Theorems~\ref{thm:oracle.test} and \ref{thm:test} accommodate growing numbers of tested coefficients and jointly tested linear restrictions under their stated conditions. Accordingly, the support used for inference is determined at the terminal batch.
\end{remark}

\section{Statistical analysis}\label{sec:theory}

\subsection{Finite-sample theory}\label{sec:finite_sample_theory}

For $R,l>0$, define $\BB(R)=\{u\in\RR^p:\|u\|_2\le R\}$, and 
$\CC(l)=\{u\in\RR^p:\|u\|_1\le l\|u\|_2\}$.
For matrices, $\|\cdot\|_2$ and $\|\cdot\|_{\max}$ denote the operator norm and the largest absolute entry, respectively.
Our theoretical analysis relies crucially on the following ``good'' events, which encode the local restricted strong convexity (RSC) and local restricted Lipschitz Hessian (RLH) properties of the empirical loss function for batch $b$, denoted by $\widehat{\mc L}^{(b)}$. We refer readers to \citet{negahban2012unified} and \citet{loh2015regularized} for detailed discussions of local RSC.
The local RLH condition, in turn, is required to control the higher-order Taylor remainder, thereby quantifying the discrepancy between the proposed surrogate loss and the corresponding offline loss based on all observations.

\begin{definition}[Local restricted strong convexity]\label{def:local_rsc}
For $R,l,\kappa>0$, let ${\mc G}_{\rm rsc}^{(b)}(R,l,\kappa)$ be the event on which
$\left\langle
\nabla\widehat{\mc L}^{(b)}(\betab_1)-
\nabla\widehat{\mc L}^{(b)}(\betab_2),
\betab_1-\betab_2
\right\rangle
\ge
\kappa\|\betab_1-\betab_2\|_2^2$,
for every pair $(\betab_1,\betab_2)$ satisfying $\betab_1-\betab_2\in\BB(R)\cap\CC(l)$, and $\betab_2-\betab^*\in\BB(R)\cap\CC(l)$.
\end{definition}

\begin{definition}[Local restricted Lipschitz Hessian]\label{def:local_lipschitz}
For $R,l>0$ and $M\ge0$, let ${\mc E}_{\rm rlh}^{(b)}(R,l,M)$ be the event on which
$\left\|
\left\{
\nabla^2\widehat{\mc L}^{(b)}(\betab^*+\vb)-
\nabla^2\widehat{\mc L}^{(b)}(\betab^*)
\right\}\wb
\right\|_\infty
\le
M\|\vb\|_2\|\wb\|_2$,
for every $\vb,\wb\in\BB(R)\cap\CC(l)$.
\end{definition}

For ease of presentation, for any $R,\kappa,l>0$ and $M\geq 0$, define the batchwise regularity event
\begin{equation*}
{\mc R}_b(R,l,\kappa,M)
=
\left\{
\bigcap_{j=1}^{b}
{\mc G}_{\rm rsc}^{(j)}(R,l,\kappa)
\right\}
\cap
\left\{
\bigcap_{j=1}^{b-1}
{\mc E}_{\rm rlh}^{(j)}(R,l,M)
\right\},
\qquad 1\leq b\leq B.
\end{equation*}
For $1\le b\le B$, define $\omegab^{(b)} = (1/N_b)\sum_{j=1}^b n_j \nabla\widehat{\mc L}^{(j)}(\betab^*)$ and $\phi^{(b)} = (1/N_b)\sum_{j=1}^{b-1}n_j \|\widehat{\betab}^{(j)}-\betab^*\|_2^2$.
Thus, $\omegab^{(b)}$ is the pooled empirical score at the truth, and $\phi^{(b)}$ is the accumulated squared error from the historical batches.

\begin{theorem}[Initializer bound]\label{thm:batch_b_init}
Fix $b\in\{1,\ldots,B\}$ and let $a=|\mc A|$.  Suppose $R_0>1.5 \kappa^{-1} a^{1/2}\lambda_b$,
and, for every $j<b$,
\begin{align}\label{event: historical localization}
\widehat{\betab}^{(j)}-\betab^* \in \BB(R_0)\cap\CC(4a^{1/2}).
\end{align}
On the event
\begin{equation}
\begin{split}
&\left\{
\lambda_b\ge
2\left(\|\omegab^{(b)}\|_\infty+\frac32M\phi^{(b)}\right)
\right\}
\cap
{\mc R}_b(R_0, 4a^{1/2}, \kappa, M),
\end{split}
\label{batch_b_event}
\end{equation}
any minimizer in \eqref{def:bth.estimator} obeys
\[
\|\widehat{\betab}^{(b)}-\betab^*\|_2 \le \frac{3}{2\kappa}a^{1/2}\lambda_b, ~~\mbox{ and }~~ \|\widehat{\betab}^{(b)}-\betab^*\|_1 \le \frac6\kappa a\lambda_b.
\]
%Moreover, $\widehat{\betab}^{(b)}-\betab^*\in\CC_{\mc A}(3) \subset\CC(4a^{1/2})$.
\end{theorem}

Theorem~\ref{thm:batch_b_init} establishes deterministic bounds conditional on the event in \eqref{batch_b_event}. 
Its key feature is that, because the corrected surrogate objective in \eqref{def:q.function} is used, the historical approximation error enters through $\phi^{(b)}$, which is quadratic in the estimation errors from the preceding batches.
To obtain an explicit bound from Theorem~\ref{thm:batch_b_init}, we impose the following sub-Gaussian condition on the score vector.

\begin{assumption}[Sub-Gaussian score]\label{assm:subgaussian}
The score $\nabla\ell(\betab^*;Y,X)$ is centered and $\sigma_0$-sub-Gaussian: $\EE\exp\{v^\T\nabla\ell(\betab^*;Y,X)\}
\le
\exp(\sigma_0^2\|v\|_2^2/2)$, $v\in\RR^p$.
\end{assumption}

\begin{corollary}\label{cor:batch_b_init}
Suppose Assumption~\ref{assm:subgaussian} holds and let $\lambda_b=C\sqrt{\log(p)/N_b}$, where $C>8\sigma_0$.  
Choose a common radius $ R_0 > \max_{1\le b\le B} (1.5\kappa^{-1}|\mc A|^{1/2}\lambda_b)$, and assume, for every $b=2,\ldots,B$, that
\begin{equation}
\sqrt{\frac{N_b}{|\mc A|^2\log p}}
\ge
\frac{27CM}{2\kappa^2}
\left\{1+\log\left(\frac{N_{b-1}}{n_1}\right)\right\}.
\label{sample.size.requirement}
\end{equation}
Then there is an event $\mc S_B$ satisfying $\PP(\mc S_B)\ge1-2B/p$ such that, on
$\mc R_B(R_0, 4|\mc A|^{1/2}, \kappa, M) \cap\mc S_B$, simultaneously for $b=1,\ldots,B$,
\[
\|\widehat{\betab}^{(b)}-\betab^*\|_2
\le
\frac{3C}{2\kappa}
\sqrt{\frac{|\mc A|\log p}{N_b}},
\qquad
\|\widehat{\betab}^{(b)}-\betab^*\|_1
\le
\frac{6C}{\kappa}|\mc A|
\sqrt{\frac{\log p}{N_b}}.
\]
% Consequently,
% \[
% \PP\!\left(\mc R_B\cap
% \left\{\text{all preceding bounds hold}\right\}\right)
% \ge
% \PP(\mc R_B)-\frac{2B}{p}.
% \]
% In particular, if $\mc R_B$ is deterministic under a fixed design, or holds almost surely,
% the simultaneous bounds have probability at least $1-2B/p$.
\end{corollary}

A sufficient condition for the scaling requirement in \eqref{sample.size.requirement} is
$N_b\gtrsim |\mc A|^2\log p$, up to a squared logarithmic factor in the number of batches. In particular, there exists a constant $C_\star$, depending only on $(C,M,\kappa)$, such that, if
$n_1 \geq C_\star |\mc A|^2\log p$, then \eqref{sample.size.requirement} holds for every $b\geq 2$. 
By contrast, ensuring that the event ${\mc R}_B$ holds with high probability generally requires
$n_b\gtrsim |\mc A|\log^2 p$; see Appendix~\suppref{sec:examples} for concrete examples.

% More precisely,
% set
% \[
% \alpha
% =
% \frac{729C^2M^2}{4\kappa^4}|\mc A|^2\log p.
% \]
% If $C_\star>1$ is the larger solution of
% $c=(1+\log c)^2$ and $n_1\ge C_\star\alpha$, then
% \eqref{sample.size.requirement} holds for every $b\ge2$.  Thus the sufficient initial-batch
% order is $|\mc A|^2\log p$, rather than $|\mc A|\sqrt{\log p}$.

Having established the convergence rates of $\{\widehat \betab^{(b)}\}_{b = 1}^B$, we next analyze the iteratively reweighted renewable estimators $\widehat \betab_e^{({\rm LLA}, k)}$, for $e \in \{0,a\}$.
Because our ultimate goal is to use these estimators to test linear hypotheses, we develop the convergence rate in a local-alternative framework.
For notational simplicity, we write $\omegab^* = \omegab^{(B)}$, $\phi^* = \phi^{(B)}$, $\lambda=\lambda_B$, and $q=m+s=|\overline{\mc M}|$.
All beta-min conditions below are vacuous when $\mc S=\varnothing$.

\begin{assumption}[Local alternatives]
\label{assm:local_alternative}
The true parameter $\betab^*$ satisfies $\Cb \betab^*_{\mc M} - \tb = \hb_B$ for some $\hb_B \in \RR^r$, and $\lambda_{\max}\{(\Cb\Cb^\T)^{-1}\}\le C_{\max}^2$.
The null corresponds to $\hb_B=0$. 
\end{assumption}

Assumption~\ref{assm:local_alternative} characterizes the local alternative through $\hb_B = \Cb \betab^*_{\mc M} - \tb$. In particular, $\hb_B = \mathbf 0$ corresponds to the null hypothesis. The assumption also imposes regularity conditions on the constraint matrix $\Cb$. Representative examples include $\Cb = (\Ib_{r \times r}, \Ob_{r \times (m-r)})$, which corresponds to testing $r$ individual components, and the first-order difference matrix, whose $(i,i)$th entry is $1$, whose $(i,i+1)$th entry is $-1$, and whose remaining entries are $0$, which corresponds to testing successive differences of components.

Let $\Zb\in\RR^{m\times(m-r)}$ have orthonormal columns spanning $\ker(\Cb)$.
Define $\betab^\dagger_{\mc M}
=
\betab^*_{\mc M}-\Cb^\T(\Cb\Cb^\T)^{-1}\hb_B$, $\betab^\dagger_{\mc M^c}=\betab^*_{\mc M^c}$,
so that $\Cb\betab^\dagger_{\mc M}=\tb$, and $\|\betab^\dagger-\betab^*\|_2 \leq C_{\max}\|\hb_B\|_2$.
Denote $d_B=C_{\max}\|\hb_B\|_2$, and $\Jb=\nabla^2\mc L(\betab^*)$.
% Then
% \[
% \|\betab^\dagger-\betab^*\|_2
% =
% \{\hb_B^\T(\Cb\Cb^\T)^{-1}\hb_B\}^{1/2}
% \le d_B.
% \]

\begin{definition}[Hessian concentration]\label{def:hessian_concentration}
For $\eta_\infty>0$, define the event
\[
{\mc H}(\eta_\infty)
=
\left\{
\left\|
\frac1{N_B}\sum_{b=1}^B n_b
\nabla^2\widehat{\mc L}^{(b)}(\betab^*)-\Jb
\right\|_{\max}
\le\eta_\infty
\right\}.
\]
\end{definition}

\begin{assumption}[Population Hessian regularity]\label{assm:hessian}
There are constants $0<J_l\le J_u<\infty$ such that
$J_l\le\lambda_{\min}(\Jb)\le\lambda_{\max}(\Jb)\le J_u$.
\end{assumption}

For ease of presentation, define the local-alternative envelope by
$\Gamma_B := \Gamma_B(M,\eta_\infty)
=
2MC_{\max}^2\|\hb_B\|_2^2
+
C_{\max}(J_u+\sqrt m\,\eta_\infty)\|\hb_B\|_2$.
Under the null hypothesis, $\hb_B=\mathbf{0}$, and hence $\Gamma_B=0$.

\begin{theorem}[Tightening-stage contraction]\label{thm:batch_b_tightening}
Suppose Assumption~\ref{assm:penalty_function} holds.  Let $\gamma_0>0$ satisfy $p_0:=p'_+(\gamma_0)>0$, $\kappa\gamma_0>\sqrt{5}/2$, and $\delta=\sqrt{5}/{(2\kappa\gamma_0)}<1$.
Let $c_{\rm aug}>0$ be the unique positive constant  satisfying
\#
2+p_0\sqrt{c_{\rm aug}^2+1} = c_{\rm aug}\kappa\gamma_0,
\label{def:c.aug}
\#
and set $l = (2+2/p_0)\sqrt{(m+s)(c_{\rm aug}^2 + 1)}$.
Assume that $R_0>c_{\rm aug}\gamma_0(m+s)^{1/2}\lambda$ and that the localization conditions, $\widehat{\betab}^{(j)}-\betab^*
\in
\BB(R_0)\cap\CC(4|\mc A|^{1/2})$, hold for $ j=1,\ldots,B-1$.

\noindent\textup{(i) Full model.}
On the event
\#
\left\{
\lambda\ge\frac2{p_0}
\left(\|\omegab^*\|_\infty+2M\phi^*\right)
\right\}
\cap
{\mc R}_B(R_0, l, \kappa, M),
\label{event:theorem2-full}
\#
for every $k\ge1$,
\#
\|\widehat{\betab}^{({\rm LLA},k)}_a-\betab^*\|_2
\le
\delta
\|\widehat{\betab}^{({\rm LLA},k-1)}_a-\betab^*\|_2
+\kappa^{-1}\big\|p'_{\lambda,+}
\bigl((|\betab^*_{\mc S}|-\gamma_0\lambda)_+\bigr)\big\|_2 + r_a^{\rm ora},
\label{eq:theorem2-full-contraction}
\#
where $r_a^{\rm ora}
:=
\kappa^{-1}
\left\{
\|\omegab^*_{\overline{\mc M}}\|_2
+2M \sqrt{s + m}\phi^*
\right\}$.

\noindent\textup{(ii) Reduced model.}
Suppose, in addition, that Assumptions~\ref{assm:local_alternative} and \ref{assm:hessian} hold and $R_0 > C_{\max}\|\hb_B\|_2$.
Moreover, in addition to the event in \eqref{event:theorem2-full}, assume the following event:
\begin{equation}
\left\{
\lambda\ge\frac2{p_0}
\left(
\|\omegab^*\|_\infty+2M\phi^*+\Gamma_B
\right) + C_{\max}\kappa \|\hb_B\|_2
\right\}
\cap
{\mc H}(\eta_\infty) \cap
{\mc E}_{\rm rlh}^{(B)}(R_0,l,M).
\label{event:theorem2-reduced}
\end{equation}
Then, for every $k\ge1$,
\$
\|\widehat{\betab}^{({\rm LLA},k)}_0-\betab^*\|_2
\le
\delta
\|\widehat{\betab}^{({\rm LLA},k-1)}_0-\betab^*\|_2
+\kappa^{-1} \big\|p'_{\lambda,+} \bigl((|\betab^*_{\mc S}|-\gamma_0\lambda)_+\bigr)\big\|_2 + r_{0}^{\rm ora},
\$
where $r_0^{\rm ora} := \kappa^{-1} \big[\|\Zb^\T\omegab^*_{\mc M}\|_2+\|\omegab^*_{\mc S}\|_2  +2\sqrt{s + m}\{2M\phi^*+\Gamma_B\}
\big] +C_{\max}\|\hb_B\|_2$.
\end{theorem}

Theorem~\ref{thm:batch_b_tightening} shows that, for both the constrained and unconstrained estimators, each tightening step contracts the estimation error inherited from the preceding iteration by a factor of $\delta$. As an immediate consequence, for each $e\in\{0,a\}$,
\[
\|\widehat{\betab}_e^{({\rm LLA},k)}-\betab^*\|_2
\leq
\delta^k\|\widehat{\betab}^{(B)}-\betab^*\|_2
+
\frac{1}{1-\delta}
\left\{
\kappa^{-1}
\left\|
p'_{\lambda,+}\left((|\betab^*_{\mc S}|-\gamma_0\lambda)_+\right)
\right\|_2
+
r_e^{\rm ora}
\right\}.
\]
Thus, for sufficiently large $k$, the contribution of the initial estimator becomes negligible, and the estimation error is governed by the second term on the right-hand side.

The quantity $\kappa^{-1}\|p'_{\lambda,+}((|\betab^*_{\mc S}|-\gamma_0\lambda)_+)\|_2$ represents the shrinkage bias induced by the penalty function; see \citet{fan2018lamm}. 
Under the $\ell_1$ penalty, $p'_{\lambda,+}(t)=\lambda$ for $t\geq0$, so the derivative does not vanish even for arbitrarily strong signals. 
Consequently, the shrinkage bias can be of order $s^{1/2}\lambda$, illustrating a fundamental limitation of the $\ell_1$ penalty. 
By contrast, consider a concave penalty satisfying $p'_{\lambda,+}(\gamma_1\lambda)=0$ for some $\gamma_1>0$. 
Under the beta-min condition $\min_{j\in\mc S}|\beta_j^*|\geq(\gamma_0+\gamma_1)\lambda$, the shrinkage-bias term vanishes. 
In this case, for $e \in \{0,a\}$, the final convergence rate is governed by $r_e^{\rm ora}$ and therefore matches that of the corresponding oracle estimator $\widehat{\betab}^{\rm ora}_e$ defined in \eqref{def:constrained_oracle_estimator} and \eqref{def:unconstrained_oracle_estimator}, respectively.

In particular, up to multiplicative constants, $r_0^{\rm ora}$ consists of the following three components:
\[
\underbrace{
\|\omegab^*_{\mc S}\|_2
+
\|\Zb^\T\omegab^*_{\mc M}\|_2
}_{\mbox{\scriptsize offline oracle rate}},
\qquad
\underbrace{
(s+m)^{1/2}\phi^*
}_{\mbox{\scriptsize surrogate-loss bias}},
\qquad
\underbrace{
(s+m)^{1/2}\Gamma_B+\|\hb_B\|_2
}_{\mbox{\scriptsize local-alternative error}}.
\]
The first component is the stochastic error of the corresponding offline oracle estimator. The second quantifies the approximation error arising from the use of the renewable surrogate loss in place of the full-data offline loss. The third captures the additional error induced by the local alternative in Assumption~\ref{assm:local_alternative}.
Consequently, under the beta-min condition and for sufficiently large $k$, $\widehat{\betab}_0^{({\rm LLA},k)}$ attains the same convergence rate as the oracle estimator $\widehat{\betab}^{\rm ora}_0$, which is constructed with prior knowledge of the true support $\mc S$. We refer to this rate equivalence as the \textit{weak oracle property}. The same conclusion holds for the unconstrained estimator $\widehat{\betab}_a^{({\rm LLA},k)}$.

A similar surrogate-loss bias term appears in Theorem~\ref{thm:batch_b_init}; however, since no penalty is imposed on the indices in $\mc M$ in the definition of $\widehat{\betab}_e^{({\rm LLA}, k)}$, the scaling becomes $(s+m)^{1/2}$ rather than $|\mc A|^{1/2}$.
Furthermore, the unconstrained oracle rate $r^{\rm ora}_a$ involves the term $\|\omegab^*_{\overline{\mc M}}\|_2$, whereas the constrained counterpart $r^{\rm ora}_0$ involves $\|\omegab^*_{\mc S}\|_2 + \|\Zb^\T \omegab^*_{\mc M}\|_2$. 
As shown in Appendix~\suppref{sec:examples}, under standard regularity conditions such as Assumption~\ref{assm:subgaussian}, $\|\omegab^*_{\overline{\mc M}}\|_2$ scales as $\sqrt{(s+m)/N_B}$, while $\|\omegab^*_{\mc S}\|_2 + \|\Zb^\T \omegab^*_{\mc M}\|_2$ scales as $\sqrt{(s+m-r)/N_B}$. This distinction is particularly pronounced in regimes where the number of constraints $r$ diverges with $N_B$, highlighting the advantage of the constrained estimator.
This gain comes at a cost: $r^{\rm ora}_0$ includes the additional term $(s+m)^{1/2}\Gamma_B + \|\hb_B\|_2$, which captures the misspecification induced by local alternatives to the linear constraint.

The preceding results give sufficient conditions under which the LLA sequence $\{\widehat{\betab}_e^{({\rm LLA},k)}\}_{k\geq 1}$, for $e\in\{0,a\}$, satisfies a contraction property. 
We next establish a stronger result: in the renewable-data setting, finitely many LLA iterations suffice for the proposed estimator to recover the corresponding oracle estimator exactly.
This exact recovery property is commonly referred to as the \textit{strong oracle property} \citep{fan2014strong,fan2018lamm}.
% For $e\in\{0,a\}$, define the oracle score
% \[
% \mathfrak{w}^{\mathrm{ora}}_e
% =
% \nabla \widehat{\mc Q}_B(\widehat{\betab}^{\mathrm{ora}}_e),
% \qquad e\in\{0,a\}.
% \]

\begin{theorem}[Strong oracle recovery]\label{thm:oracle_recovery}
Suppose the assumptions and conditions of
Theorem~\ref{thm:batch_b_tightening} hold.  
% Suppose, moreover, that
% \[
% \bigcap_{b=1}^B
% {\mc G}_{\rm rsc}^{(b)}(R_0,l,\kappa)
% \]
% holds and that \suppeqref{eq:crude-lla-localization} is valid for every $k\ge0$ and both
% $e=0$ and $e=a$. 
For the same $\gamma_0$ and $p_0=p'_+(\gamma_0)$ as in
Theorem~\ref{thm:batch_b_tightening}, assume that there is a $\gamma_1>\gamma_0$ such that $p_0\in(0,1/2]$, $p'_+(\gamma_1)=0$, and
$\min_{j\in\mc S}|\beta_j^*|
\ge(\gamma_0+\gamma_1)\lambda$.
Assume $\kappa\gamma_0>5/4$ and choose $\bar\delta\in (5/(4\kappa\gamma_0),1)$.
%When $\mc S=\varnothing$, the beta-min condition above and the beta-min condition inherited from Theorem~\ref{thm:batch_b_tightening} are understood to be vacuous.
For each $e\in\{0,a\}$, suppose
\#
\big\|\nabla_{\overline {\mc M}^{\,c}} \widehat{\mc Q}_B(\widehat{\betab}^{\mathrm{ora}}_e)\big\|_\infty \le \frac{p_0\lambda}{2}, ~~\mbox{ and }~~ \|\widehat{\betab}^{\rm ora}_e-\betab^*\|_2 \le R_0\wedge\frac{\lambda}{5\bar\delta\kappa}. \label{eq:oracle-recovery-score-radius}
\#
Let $\widehat s = |\operatorname{supp}(\widehat{\betab}^{(B)})|+s$, and define $K_0=\lceil\log(\widehat s^{1/2}/\bar\delta)/\log(1/\bar\delta)\rceil$, if $\widehat s \geq 1$, otherwise $K_0=0$.
Then
\[
\widehat{\betab}^{({\rm LLA},k)}_e
=
\widehat{\betab}^{\rm ora}_e,
\qquad
k\ge K_0+1,
\qquad e\in\{0,a\}.
\]
\end{theorem}

\subsection{Theoretical guarantees for inference}\label{sec:inference}

In this section, we derive theoretical guarantees for the proposed inference procedure.
We first study the properties of the oracle test statistics $T_W^{\rm ora}$ and $T_S^{\rm ora}$ defined in \eqref{def:oracle.wald} and \eqref{def:oracle.score}, respectively.
The regularity conditions are collected below.

\begin{assumption} [Score covariance] \label{assm:score}
    There exist constants $0 < K_l \leq K_u \lesssim 1$ such that $K_l \leq \lambda_{\min}(\Kb_0) \leq \lambda_{\max}(\Kb_0) \leq K_u$.
\end{assumption}

Define the standardized oracle score map as $\Omega_0 = \Vb_0^{-1/2} \Cb_0 \Jb_0^{-1}$, and noncentrality parameter $e_B = N_B \hb_B^\T \Vb_0^{-1} \hb_B$.
For the Gaussian approximation, we assume the following Lyapunov-type condition.
\begin{assumption}[Lyapunov-type condition]
\label{assm:lyapunov}
$(r^{1/4}/\sqrt{N_B})\EE\left\|
\Omega_0\nabla_{\overline{\mc M}}\ell(\betab^*;Y,X)
\right\|_2^3
\rightarrow 0$, as $N_B\to\infty$. 
\end{assumption}

Assumption~\ref{assm:lyapunov} provides a sufficient Lyapunov-type condition for the multivariate Gaussian approximation used below; see \citet{bentkus2005lyapunov}. 
Since $\EE\left\|\Omega_0\nabla_{\overline{\mc M}}\ell(\betab^*;Y,X)\right\|_2^2 =  \operatorname{tr}\left(\Omega_0\Kb_0\Omega_0^\T\right) = r$,
Jensen's inequality implies $\EE\| \Omega_0\nabla_{\overline{\mc M}}\ell(\betab^*;Y,X)\|_2^3 \geq r^{3/2}$.
Consequently, when $r$ diverges, Assumption~\ref{assm:lyapunov} necessarily requires $r^{7/2} = o(N_B)$.
For fixed $r$, the assumption holds whenever the standardized third moment is uniformly bounded.

We now establish the limiting distributions of the oracle test statistics $T_W^{\rm ora}$ and $T_S^{\rm ora}$. 
Importantly, the resulting distributional approximation remains valid even when both the number of constraints $r$ and the effective model dimension $m+s$ diverge with the cumulative sample size.

\begin{theorem} \label{thm:oracle.test}
Suppose that Assumptions~\ref{assm:local_alternative}, \ref{assm:hessian}, \ref{assm:score}, and \ref{assm:lyapunov} hold, with $\|\hb_B\|_2 = O(\sqrt{r/N_B})$. 
Suppose, further, that with $l \geq (m+s)^{1/2}$, the event
$\bigcap_{j=1}^{B-1} \big\{ \widehat{\betab}^{(j)}-\betab^*\in\BB(R_0)\cap\CC(l) \big\} \cap
{\mc R}_B(R_0,l,\kappa,M) \cap {\mc E}_{\rm rlh}^{(B)}(R_0,l,M) \cap {\mc H}(\eta_\infty)$
holds with probability tending to one.
Moreover, $R_0$ is sufficiently large so that $R_0 \geq 2 \max(r_0^{\rm ora}, r_a^{\rm ora})$ with probability tending to one.
Finally, we assume that, for each $e \in \{0,a\}$,
\#
& \sqrt{\frac{m+s}{N_B}}
\sum_{b=1}^{B-1} n_b
\left\{
\|\widehat\betab^{(b)}-\betab^*\|_2^2
+
\|\widehat\betab^{(b)}-\betab^*\|_2
\|\widehat\betab_e^{\rm ora}-\betab^*\|_2
\right\} \nn
\\
&\qquad
+
\|\widehat\betab_e^{\rm ora}-\betab^*\|_2
\left\{
\sqrt{\frac{n_B^2(m+s)}{N_B}}
\|\widehat\betab_e^{\rm ora}-\betab^*\|_2
+
\sqrt{N_B}(m+s)\eta_\infty
\right\}
\overset{p}{\to} 0. \label{eq:inference-remainder-condition}
\#
Then, for $T = T_W^{\rm ora}$ or $T_S^{\rm ora}$,
\[
\sup_{x \geq 0}
\left|
\PP(T \leq x)
-
\PP\{\chi^2(r,e_B) \leq x\}
\right|
\longrightarrow 0,
\]
where $\chi^2(r,e_B)$ denotes a noncentral chi-square random variable with $r$ degrees of freedom and noncentrality parameter $e_B$.
\end{theorem}

Theorem~\ref{thm:oracle.test} establishes that the oracle test statistics admit a valid noncentral chi-square approximation in the renewable setting, despite the unavailability of raw observations from previous batches. The result also accommodates a high-dimensional regime in which both $r$ and $m+s$ may diverge. Thus, the corrected surrogate loss function in \eqref{def:q.function} preserves the first-order inferential properties of the corresponding full-data offline procedure. In particular, for generalized linear models, the resulting limiting distribution coincides with that obtained from the offline loss based on all observations, as derived by \citet{shi2019linear}.

We next show that the proposed partially penalized test statistics satisfy an oracle property: despite the unknown support set, their limiting distributions coincide with those of the corresponding oracle statistics in Theorem~\ref{thm:oracle.test}.

\begin{theorem} \label{thm:test}
In addition to the conditions of Theorem~\ref{thm:oracle.test}, suppose that the conditions of Theorem~\ref{thm:oracle_recovery} hold with probability tending to one. Assume moreover that $\lambda_{\max}(\Cb\Cb^\T)=O(1)$ and that
\begin{equation}
\sqrt r\left[
\max_{e\in\{0,a\}}
\left\|
\nabla^2_{\overline{\mc M}\,\overline{\mc M}}
\widehat{\mc Q}_B(\widehat{\betab}^{\rm ora}_e)-\Jb_0
\right\|_2
+
\left\|
\widehat{\Kb}^{(B)}_{\overline{\mc M}\,\overline{\mc M}}-\Kb_0
\right\|_2
\right]
=o_p(1).
\label{eq:feasible-plugin}
\end{equation}
Then, for every $k\ge K_0+1$ and for $T = T_W^{(k)}$ or $T_S^{(k)}$,
\[
\sup_{x \geq 0}
\left|
\PP(T \leq x)
-
\PP\{\chi^2(r,e_B) \leq x\}
\right|
\to 0.
\]
\end{theorem}
\noindent It follows immediately from Theorem~\ref{thm:test} that
$\sup_{x \geq 0}
\left|
\PP\{T_W^{(k)} \leq x\}
-
\PP\{T_S^{(k)} \leq x\}
\right|
\to 0$,
thereby establishing the asymptotic equivalence of the partially penalized Wald and score statistics. Recall that the testing procedure for \eqref{hypothesis} is defined as follows: for a prescribed significance level $\alpha$, we reject the null hypothesis whenever $T > \chi_\alpha^2(r)$, where $T$ denotes either $T_W^{(k)}$ or $T_S^{(k)}$. The following corollary establishes the asymptotic validity of this procedure.

\begin{corollary} \label{cor:test}
Suppose that the conditions of Theorem~\ref{thm:test} hold. Under the null hypothesis, that is, when $\hb_B = \mathbf{0}$, for any $0<\alpha<1$ and $T = T_W^{(k)}$ or $T_S^{(k)}$,
\[
\lim_{N_B \to \infty}
\PP\left\{T > \chi_\alpha^2(r)\right\}
=
\alpha.
\]
Under local alternatives satisfying $\|\hb_B\|_2 = O(
\sqrt{r/N_B})$, for any $0<\alpha<1$ and $T = T_W^{(k)}$ or $T_S^{(k)}$,
\[
\lim_{N_B \to \infty}
\left|
\PP\left\{T > \chi_\alpha^2(r)\right\}
-
\PP\left\{\chi^2(r,e_B) > \chi_\alpha^2(r)\right\}
\right| = 0.
\]
\end{corollary}
\noindent Corollary~\ref{cor:test} shows that the asymptotic power function of the proposed test statistics is $\PP\{\chi^2(r,e_B) > \chi_\alpha^2(r)\}$.
When $\Cb$ is the identity matrix, the Wald test based on the debiased Lasso estimator of \citet{luo2023online} can also be applied to the corresponding hypothesis. 
The proposed procedure may nevertheless attain greater power because the debiasing step can inflate the variance of the resulting estimator \citep{shi2019linear}.

\section{Tuning parameter selection}
\label{sec: tuning parameter selection}

A pooled $K$-fold cross-validation requires repeated access to historical observations and is therefore unavailable in the strictly online setting considered in this paper. 
A cross-validation procedure within the current batch remains possible, however, since the stored historical summaries can be held fixed across its training and validation folds. 
We consider two tuning strategies: accumulated one-step-ahead validation with a one-standard-error
(1-SE) rule, and current-batch cross-validation, followed by a terminal generalized information criterion (GIC) when the objective is inference.

\subsection{Accumulated one-step-ahead validation}\label{sec: 1se rule}
Upon receiving batch $\mc D_b$, we
first assess predictions made by estimators stored after batch $b-1$;
we then use these predictive losses to choose the penalty for the update
based on $\mc D_b$. 
Note that $\mc D_2$ is the first one-step-ahead validation batch based on the initial batch. 
This construction follows the prequential and rolling-origin principle of evaluating predictions before updating the fitted model, as considered in \citet{rissanen1986stochastic,tashman2000out,han2021online}.

Let $\Lambda_1\subset(0,\infty)$ be a finite candidate set, constructed
using the first batch if necessary and fixed thereafter. For each
$c\in\Lambda_1$, define $\lambda_b(c):=c\sqrt{n_1/N_b}$, for $b=1,\ldots,B$.
Here $c$ identifies a sequence of penalties, whereas $\lambda_b(c)$ is
the penalty actually used at batch $b$. Given the historical summaries
available before processing $\mc D_b$, let
$\widehat{\betab}^{(b)}(c)
\in\operatorname*{argmin}_{\betab\in\RR^p}
\left\{\widehat{\mc Q}_b(\betab)
       +\lambda_b(c)\|\betab\|_1\right\}$.
All candidates at the same batch use the same $\widehat{\mc Q}_b$.
Its historical terms are constructed from the previously \emph{selected}
estimators, rather than from separate histories indexed by $c$.

\paragraph{Initialization at batch 1.}
For $n_1\ge2$, partition $\mc D_1$ into $K_1$ folds, and write $k(i)$ for the fold
containing observation $i$. 
For each $c\in\Lambda_1$, fit the first-batch Lasso on the other folds using penalty $c$, and denote the result by
$\widehat{\betab}_{1,-k(i)}(c)$. Define the individual held-out losses,
their mean, and the empirical standard error of that mean by $ d_{1i}^{\mathrm{CV}}(c)
 =\ell\!\left(\widehat{\betab}_{1,-k(i)}(c);Y_i,X_i\right)$, $ R_1^{\mathrm{CV}}(c)=(1/n_1)
       \sum_{i\in\mc D_1}d_{1i}^{\mathrm{CV}}(c)$, 
and $ s_1^{\mathrm{CV}}(c)
 =\left\{(\sum_{i\in\mc D_1}
       [d_{1i}^{\mathrm{CV}}(c)-R_1^{\mathrm{CV}}(c)]^2)/
       (n_1(n_1-1))\right\}^{1/2}.$
With ties in the minimizing penalty resolved in favor of the largest
candidate, select
\[
 \widetilde c_1=\max\operatorname*{argmin}_{c\in\Lambda_1}
                   R_1^{\mathrm{CV}}(c),
 \qquad
 \widehat c_1=\max\left\{c\in\Lambda_1:
 R_1^{\mathrm{CV}}(c)\le R_1^{\mathrm{CV}}(\widetilde c_1)
                       +s_1^{\mathrm{CV}}(\widetilde c_1)\right\}.
\]
We then fit every candidate on the entire first batch, set
$\widehat\lambda_1=\widehat c_1$ and
$\widehat{\betab}^{(1)}=\widehat{\betab}^{(1)}(\widehat c_1)$,
and construct the historical summaries from this selected estimate.
The full-batch candidate estimates
$\{\widehat{\betab}^{(1)}(c):c\in\Lambda_1\}$ are stored to predict $\mc D_2$. The first-batch cross-validation losses are used only for
initialization and are not included in the subsequent accumulated
one-step-ahead criterion.

\paragraph{Updating at batch $b\ge2$.}
Before using $\mc D_b$ to fit new candidate estimates, evaluate $ d_{bi}(c)=\ell\!\left(\widehat{\betab}^{(b-1)}(c);Y_i,X_i\right)$, $i\in\mc D_b,\quad c\in\Lambda_1$.
Note that each predictor in this expression was fitted before $\mc D_b$ arrived.
Let $m_b=N_b-n_1$ be the number of observations in the validation batches
$\mc D_2,\ldots,\mc D_b$. Their weighted mean predictive
loss and empirical standard error are $R_b(c)=\frac{1}{m_b}\sum_{j=2}^b\sum_{i\in\mc D_j}d_{ji}(c)$, 
$s_b(c)=\left\{(\sum_{j=2}^b\sum_{i\in\mc D_j}[d_{ji}(c)-R_b(c)]^2)/(m_b(m_b-1))\right\}^{1/2},$
where $m_b\ge2$. We choose
\[
 \widetilde c_b=\max\operatorname*{argmin}_{c\in\Lambda_1}R_b(c),
 \qquad
 \widehat c_b=\max\left\{c\in\Lambda_1:
 R_b(c)\le R_b(\widetilde c_b)+s_b(\widetilde c_b)\right\},
\]
and set $ \widehat\lambda_b=\widehat c_b\sqrt{n_1/N_b}$, $\widehat{\betab}^{(b)}=\widehat{\betab}^{(b)}(\widehat c_b)$.
The maximum in the 1-SE rule selects the strongest regularization among the admissible candidates.

We update the common historical summaries using only
$\widehat{\betab}^{(b)}$, while keeping all current candidate estimates over $c\in\Lambda_1$ for predicting $\mc D_{b+1}$.
The criterion for a fixed $c$ aggregates predictions from the successive
estimators $\widehat{\betab}^{(1)}(c),\ldots,
\widehat{\betab}^{(b-1)}(c)$.
In detail, initialize $S_1(c)=T_1(c)=0$ and update $ S_b(c)=S_{b-1}(c)+\sum_{i\in\mc D_b}d_{bi}(c)$, 
$T_b(c)=T_{b-1}(c)+\sum_{i\in\mc D_b}d_{bi}(c)^2$.
Then one has $R_b(c)=S_b(c)/m_b$ and
$ s_b(c)=\left\{(T_b(c)-S_b(c)^2/m_b)/(m_b(m_b-1))\right\}^{1/2}$.
The loss accumulators require $O(|\Lambda_1|)$ storage, and storing the candidate coefficient vectors additionally requires
$O(p|\Lambda_1|)$ storage.

The motivation is to assess predictive performance using observations
that were not used to fit the candidate estimators, following the same
principle as held-out evaluation in cross-validation. The sequential batch arrivals provide a natural separation between fitting and validation:
when $\mc D_b$ arrives, it serves as a validation sample for the candidate
estimators fitted through batch $b-1$.

\subsection{Batchwise CV and terminal GIC}\label{sec: CV/GIc rule}
We give more details on the cross-validation procedure in this section.
Let $\Lambda_b$ denote a finite penalty grid for batch $b$.
Partition $\mc D_b$ into folds
$\mc I_{b1},\ldots,\mc I_{bK_{\mathrm{CV}}}$. 
Write the unnormalized historical contribution as
$ H_{b-1}(\betab)
 =(1/2)(\betab-\widehat{\betab}^{(b-1)})^\T
       \widehat{\Jb}^{(b-1)}(\betab-\widehat{\betab}^{(b-1)})
 +N_{b-1}\{\widehat{\mathbf g}^{(b-1)}\}^{\T}\betab$,
with $H_0=0$ and $N_0=0$. The fold-specific surrogate is
$ \widehat{\mc Q}_{b,-k}(\betab)
 =(\sum_{i\in\mc D_b\setminus\mc I_{bk}}
                 \ell(\betab;Y_i,X_i)+H_{b-1}(\betab))/
        (N_{b-1}+n_b-|\mc I_{bk}|)$.
For each $\lambda\in\Lambda_b$, minimize
$\widehat{\mc Q}_{b,-k}(\betab)+\lambda\|\betab\|_1$ to obtain
$\widehat{\betab}_{b,-k}(\lambda)$, and select
\[
 \widehat\lambda_b^{\mathrm{CV}}
 =\max\operatorname*{argmin}_{\lambda\in\Lambda_b}
 \left\{\frac{1}{n_b}\sum_{k=1}^{K_{\mathrm{CV}}}\sum_{i\in\mc I_{bk}}
 \ell\!\left(\widehat{\betab}_{b,-k}(\lambda);Y_i,X_i\right)\right\}.
\]
The selected penalty is then used to fit the full current batch and update the summaries.

In the estimation experiments, this minimum-CV rule is applied at every batch. 
For inference tasks, we take the CV/GIC procedure, which uses batchwise CV for $b<B$ and the final GIC for the inferential fits. 
For a specified hypothesis, let $\widehat{\betab}_{a,B}(\lambda)$ and
$\widehat{\betab}_{0,B}(\lambda)$ denote the unrestricted and restricted
partial-penalized fits, respectively, computed at a candidate penalty
$\lambda$ by the LLA procedure described above. For each fitted path,
the terminal criterion is
\[
 \operatorname{GIC}_{B,e}(\lambda)
 =N_B\widehat{\mc Q}_B\!\left(\widehat{\betab}_{e,B}(\lambda)\right)
   +a_{n_B,p}\,\big\|\widehat{\betab}_{e,B}(\lambda)\big\|_0,
 \qquad e\in\{a,0\},
\]
where $a_{n_B,p}=\max\{\log n_B,\log(\log n_B)\log p\}$ in our
simulations. Here $n_B$ is the size of the final batch, whereas $N_B$ is
the cumulative sample size. 
We note that GIC is minimized separately on
the unrestricted and restricted candidate paths for the Wald and score procedures, respectively.
By contrast, the accumulated 1-SE procedure in Section~\ref{sec: 1se rule} keeps its selected
$\widehat\lambda_B$ for the terminal inferential fits.
In Section~\ref{sec: simulation} below, we compare these tuning strategies in detail with comprehensive simulation studies.

\section{Numerical studies}\label{sec: simulation}\label{sec:numerical}

In this section, we conduct simulation studies to evaluate the finite-sample performance of the proposed estimators and the associated hypothesis tests. 
We consider the following three models:
\begin{enumerate}
  \item linear model, $Y_i=X_i^\T\betab^*+\epsilon_i$ with $\epsilon_i\sim N(0,1)$;
  \item logistic model,
        ${\rm logit}\,\PP(Y_i=1\mid X_i)=X_i^\T\betab^*$, where
        ${\rm logit}(u)=\log\{u/(1-u)\}$;
  \item Poisson model, $Y_i\mid X_i\sim{\rm Poisson}(\mu_i)$ with
        $\log\mu_i=X_i^\T\betab^*$.
\end{enumerate}
For the first two models, we generate $X_i \sim N(\mathbf{0}_p,\Sigma)$, where $\Sigma$ is a Toeplitz covariance matrix with entries $\Sigma_{jk}=\rho^{|j-k|}$ and $\rho=0.5$. For the Poisson model, we impose a bounded design by setting $X_{ij} =\sqrt{3}\{2\Phi(Z_{ij})-1\}$, where $Z_i\sim N(\mathbf{0}_p,\Sigma)$ and $\Phi(\cdot)$ denotes the cumulative distribution function of the standard normal distribution.

In this section, we focus on hypothesis testing. A comparison of the
estimation performance of estimators based on the proposed first-order
corrected loss and the renewable loss of \citet{luo2023online} is
presented in Section~\suppref{sec: estimation error simulation}.

We evaluate the empirical Type I error rates and power of the proposed renewable Wald and score tests in logistic and Poisson regression. For linear regression, the first-order correction makes the renewable squared-error loss coincide with the pooled loss up to an additive constant (see Remark~\ref{remark:ls no info lost} above). Thus, for a given penalty and tuning parameter, the terminal partially penalized objectives coincide with their offline counterparts studied by \citet{shi2019linear} and \citet{jacobson2025strong}. 
We present the linear-model results in Section~\suppref{sec:linear_additional_simulation_I}. 
Additional experiments for the Tobit and nonparametric Box--Cox models are presented in Sections~\suppref{sec:tobit_additional_simulation_II}
and~\suppref{sec:BC_additional_simulation_III}, respectively.

Throughout the inferential experiments, we set $n_b=100$ and $B=8$, yielding $N_B=800$, and consider
$p\in\{400,1200\}$.
For each setting, we use 500 independent replications and consider the following five linear hypotheses:
\begin{align*}
H_0^{(1)} &: \ \beta_1^*=0, &
H_0^{(2)} &: \ \beta_4^*+\beta_5^*=0,\\
H_0^{(3)} &: \ (\beta_1^*,\beta_8^*,\beta_9^*)^\T=\mathbf{0}_3, &
H_0^{(4)} &: \ \beta_1^*=0,\ \beta_4^*+\beta_5^*=0,\\
H_0^{(5)} &: \ (\beta_1^*,\beta_8^*,\beta_9^*)^\T=\mathbf{0}_3,
\ \beta_4^*+\beta_5^*=0, &&
\end{align*}
where the corresponding degrees of freedom are $1,1,3,2$, and $4$.
All five null hypotheses hold when $h=0$, while positive values of $h$
produce departures through $\beta_1^*$ and $\beta_4^*+\beta_5^*$.
We set the nominal significance level to $\alpha=0.05$.

We compare two different estimation strategies, Lasso $\to$ LLA and LLA $\to$ LLA, each combined with either tuning procedure in Section~\ref{sec: tuning parameter selection}.
Under Lasso $\to$ LLA, the historical summaries at batches $1,\ldots,B-1$ use corrected renewable Lasso estimates. Under LLA $\to$ LLA, each historical estimate is obtained by four SCAD reweighted iterations following a Lasso initialization. 
The one-step-ahead predictive-loss procedure (1-SE) in Section~\ref{sec: 1se rule} uses five-fold cross-validation for initialization and exploits the
penalty selected at batch $B$ for final inference. 
The procedure (CV/GIC) in Section~\ref{sec: CV/GIc rule} uses ten-fold batchwise cross-validation for the first $B-1$ batches and separate terminal
selection on the unrestricted and restricted paths.
The pooled partial-penalized Wald benchmark ({\tt offline)} of \citet{shi2019linear}
uses four SCAD reweightings and the full-sample GIC for tuning parameter selection as in \citet{jacobson2025strong}.
For $H_0^{(1)}$, we additionally report the debiased online Wald test
of \citet{luo2023online}, denoted by {\tt Luo}.

\paragraph{Logistic regression.}
For logistic regression, we use the regression coefficient vector $\betab^*=(h,0,0,0.7+h,-0.7,\mathbf{0}_{p-5}^\T)^\T$, with 
$h\in\{0,0.025,0.05,0.1,0.2,0.4\}$.
Tables~\suppref{tab:logistic-p400} and~\suppref{tab:logistic-p1200} in Section~\suppref{sec: additional tables GLM} report the
empirical rejection probabilities. Under the null, rejection probabilities
are generally close to the nominal level in both dimensions. Lasso
$\to$ LLA with 1-SE tuning has mildly elevated rejection probabilities
for the first hypothesis, whereas the corresponding LLA $\to$ LLA
results are closer to the nominal level. For $H_0^{(3)}$, CV/GIC tuning
tends to be mildly conservative. Under the alternatives, power generally
increases with the departure and broadly follows the pooled benchmark.
The left panels of Figure~\ref{fig:glm-power-h1} display the rejection
curves for $H_0^{(1)}$ in logistic regression at both dimensions.
For $H_0^{(1)}$, the renewable procedures show higher power
than the procedure of \citet{luo2023online}. Wald and score tests are
particularly similar under CV/GIC tuning. With 1-SE tuning, the score
test can have lower power at intermediate departures, especially with
Lasso $\to$ LLA estimation. Related comparisons for the linear, Tobit,
and nonparametric Box--Cox models appear in
Sections~\suppref{sec:linear_additional_simulation_I},
\suppref{sec:tobit_additional_simulation_II}, and
\suppref{sec:BC_additional_simulation_III}, respectively.

\paragraph{Poisson regression.}
For the Poisson model, the responses satisfy
$Y_i\mid X_i\sim\operatorname{Poisson}\{\exp(X_i^\top\boldsymbol\beta^*)\}$,
with
$\boldsymbol\beta^*=(h,0,0,0.45+h,-0.45,\mathbf{0}_{p-5}^\top)^\top$.
We use $h\in\{0,0.04,0.07,0.1,0.2\}$.
Tables~\suppref{tab:poisson-p400} and~\suppref{tab:poisson-p1200}
in Section~\suppref{sec: additional tables GLM} report the empirical
rejection probabilities. Under $H_0^{(2)}$, $H_0^{(4)}$, and
$H_0^{(5)}$, null rejection probabilities remain close to the nominal level in both dimensions. Under $H_0^{(3)}$, the
pooled Wald test and Lasso $\to$ LLA with 1-SE Wald testing exhibit
higher null rejection probabilities in the larger dimension, whereas
the score tests and the other Wald procedures remain closer to the nominal level.

Under the alternatives, our renewable Wald tests and CV/GIC score tests generally track the pooled
benchmark. The 1-SE score tests can have lower power at intermediate departures, particularly for $H_0^{(2)}$.
The right panels of Figure~\ref{fig:glm-power-h1} display the rejection
curves for $H_0^{(1)}$ in Poisson regression at both dimensions.
For $H_0^{(1)}$, our procedures have higher power than {\tt Luo} at the positive departures, as expected. These patterns are similar across dimensions.

\begin{figure}[t]
\centering
\includegraphics[width=\linewidth]{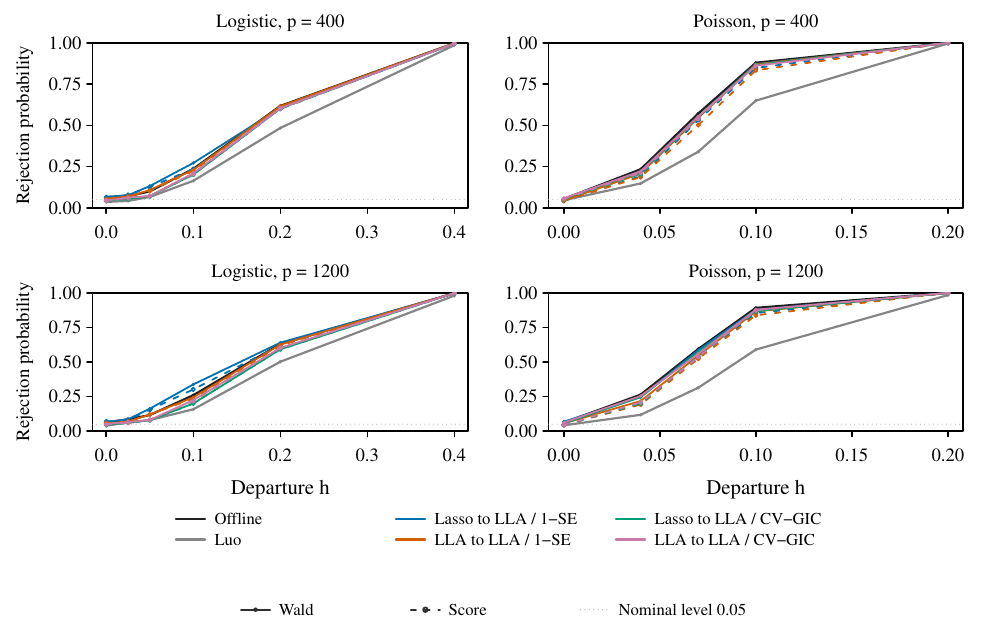}
\caption{Empirical rejection probabilities for $H_0^{(1)}$ in logistic
regression (left) and Poisson regression (right). The top and bottom rows
correspond to $p=400$ and $p=1200$, respectively. Each point uses 500
replications with $n_b=100$ and $B=8$. Solid and dashed colored curves
denote Wald and score tests; black and gray curves denote {\tt Offline}
and {\tt Luo}, respectively.}
\label{fig:glm-power-h1}
\end{figure}

\section{Real data analysis}\label{sec:real-data}
We apply our online estimation and testing procedures to analyze bike rental data from Capital Bikeshare \citep{capitalbikeshare}. We examine data from January and February 2024 and focus on bike rides ending in a generous rectangle around the National Mall, spanning from just north of the Pentagon (38$^{\circ}$52'40.8''N, 77$^{\circ}$03'18.0''W) to Union Station (38$^{\circ}$53'56.4''N, 77$^{\circ}$00'18.0''W). We joined these data with Local Climatological Data from the NCEI Reagan National Airport weather station \citep{kantor2023lcdv2}.

We take whether the rider is a Capital Bikeshare member or a guest user as our response and model it using logistic regression. (A secondary analysis using Poisson regression can be found in Section~\suppref{sec:real data poisson}.)
Our training data consist of 25,260 rides from January 2024 and our test data consist of 35,564 rides from February 2024.
The data contain 646 predictors, including characteristics of the ride -- indicators for which rental stations the ride started and ended at, whether the bicycle was electric or manual, the duration of the ride, and the end time of the ride -- and of the day -- which day of the week it was, whether it was the weekend, the minimum and maximum dry bulb temperature, total precipitation, and the length of the day. 
We categorize the ride end times as follows: overnight (12am - 6am), morning rush (6am - 9am), midday (9am - 4pm), afternoon rush (4pm - 7pm), and evening (7pm - 12am).

We break the training data into three week-long batches and a ten-day final batch. We employ the one-step-ahead predictive-loss tuning procedure (1-SE) in Section~\ref{sec: 1se rule} with $10$-fold CV tuning in the first batch.
After fitting to each batch, we use the updated model to make predictions on the test data and compute the classification error and empirical loss.
Figure \ref{fig:bikeshare_binom_pred} plots the batch number against the classification error and loss on the test data for the corrected online Lasso estimator fit to that batch.
We see that, by both metrics, the predictions steadily improve with each additional batch of training.

\begin{figure}[h]
    \centering
    \includegraphics[width=\linewidth]{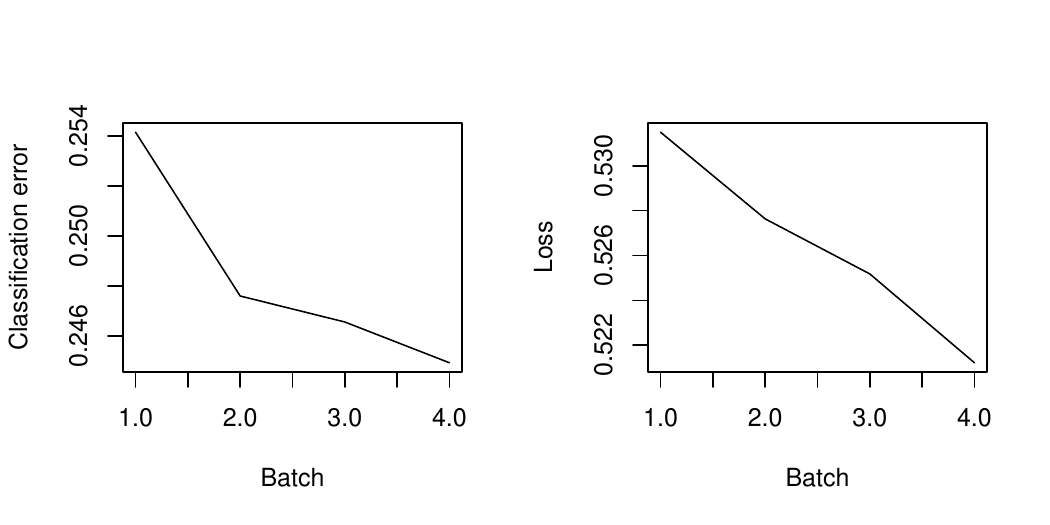}
    \caption{Classification error and loss on test data for predictions made after training on each batch in Capital Bikeshare logistic regression analysis.}
    \label{fig:bikeshare_binom_pred}
\end{figure}

We test whether having the ride end during the morning or afternoon rush hours is a significant predictor of whether the rider is a member -- that is, we test $H_0: (\beta_{\rm MR}, \beta_{\rm AR})^\T = (0,0)^\T$ where $\beta_{\rm MR}$ and $\beta_{\rm AR}$ are the coefficients for the morning and afternoon rush hour indicators, respectively.
We fit 5-step LLA estimators with the SCAD penalty with shape parameter $a = 3.7$ and tune the penalty parameter with GIC.  
The online partial penalized Wald and score statistics are $T_W = 202.47$ and $T_S = 186.63$, respectively, giving us $p$-values $< 2.2\times10^{-16}$ and providing overwhelming evidence that having a ride end during rush hour is a significant predictor of whether someone is a Capital Bikeshare member.

A complementary analysis of the same Capital Bikeshare data, modeling
hourly arrival counts at individual rental stations using Poisson
regression, is presented in Section~\suppref{sec:real data poisson}.

\section*{Data availability statement}
The Capital Bikeshare data used in  the real data analysis (Section \ref{sec:real-data}) are publicly available at \url{https://capitalbikeshare.com/system-data} \citep{capitalbikeshare}. These data were joined with Local Climatological Data (V2) from the NOAA National Centers for Environmental Information which are publicly available at \url{https://doi.org/10.25921/96dw-mb77} \citep{kantor2023lcdv2}.
  \if1\blind
    \InputIfFileExists{content-acks.tex}{}{}
  \fi
  \InputIfFileExists{content-data.tex}{}{}
\fi
\ifnum\documentpart=1
  \def\myspacing{1.55}\spacingset{\myspacing}
\else
  \clearpage
  % Reset supplement counters consistently in full and standalone builds.
  \setcounter{table}{0}\setcounter{figure}{0}
  % Distinct hyperlink destinations after resetting counters.
  \renewcommand{\theHsection}{supp.\arabic{section}}
  \renewcommand{\theHequation}{supp.\arabic{equation}}
  \renewcommand{\theHtable}{supp.\arabic{table}}
  \renewcommand{\theHfigure}{supp.\arabic{figure}}
  \renewcommand{\theHtheorem}{supp.\arabic{theorem}}
  \renewcommand{\theHlemma}{supp.\arabic{lemma}}
  \renewcommand{\theHproposition}{supp.\arabic{proposition}}
  \renewcommand{\theHcorollary}{supp.\arabic{corollary}}
  \renewcommand{\theHremark}{supp.\arabic{remark}}
  \def\spacingset#1{\renewcommand{\baselinestretch}%
{#1}\small\normalsize} \spacingset{1.2}

\if1\blind
{
  \bigskip
  \bigskip
  \bigskip
  \begin{center}
    {\LARGE\bf Supplementary Materials for ``\papertitle''}\\[1.2em]
    {\large Daeyoung Ham, Myeonghun Yu}\\
    {\large Department of Statistics and Data Science, University of Texas at San Antonio}\\
    {\large Department of Statistics, Ewha Womans University}\\[0.75em]
    {\large Tate Jacobson}\\
    {\large Department of Statistics and Data Science, Oregon State University}
  \end{center}
  \medskip
} \fi

\if0\blind
{
  \bigskip
  \bigskip
  \bigskip
  \begin{center}
    {\LARGE\bf Supplementary Materials for ``\papertitle''}
\end{center}
  \medskip
} \fi

% Continue the supplement on the title page.
\def\myspacing{1.7}
\spacingset{\myspacing}

\renewcommand{\thesection}{S\arabic{section}}
\renewcommand{\theequation}{S.\arabic{equation}}
\renewcommand{\thetheorem}{S.\arabic{theorem}}
\renewcommand{\thelemma}{S.\arabic{lemma}}
\renewcommand{\theproposition}{S.\arabic{proposition}}
\renewcommand{\thecorollary}{S.\arabic{corollary}}
\renewcommand{\theremark}{S.\arabic{remark}}

\setcounter{section}{0}
\setcounter{equation}{0}
\setcounter{theorem}{0}
\setcounter{lemma}{0}
\setcounter{proposition}{0}
\setcounter{corollary}{0}
\setcounter{remark}{0}

  \section{Auxiliary lemmas}\label{sec:appendix_lemmas}

This section collects all auxiliary lemmas used in the proofs of Theorems~\mainref{thm:batch_b_init}--\mainref{thm:test}.  
Their proofs are deferred to Appendix~\ref{sec:appendix_lemma_proofs}.

\begin{lemma}[Lemma 3 of \citet{han2021online}]
\label{lem:batch.size}
For $1\le b\le B$,
\[
\sum_{j=1}^b\frac{n_j}{N_j}
\le
1+\log\left(\frac{N_b}{n_1}\right).
\]
\end{lemma}

For a fixed vector $\theta\in\RR^p$, define
\[
R_j(\theta)
=
\nabla\widehat{\mc L}^{(j)}(\widehat{\betab}^{(j)})
-
\nabla\widehat{\mc L}^{(j)}(\theta)
+
\nabla^2\widehat{\mc L}^{(j)}(\widehat{\betab}^{(j)})
(\theta-\widehat{\betab}^{(j)}).
\]
Let $\widehat\delta_j=\widehat{\betab}^{(j)}-\betab^*$.
The following lemma provides an exact decomposition of the renewable-surrogate gradient into the gradient of the pooled loss and an accumulated linearization remainder, and bounds the latter in terms of the squared localization radii.

\begin{lemma}
\label{lem:decomposition}
For every $b\ge1$ and every $\theta\in\RR^p$,
\begin{equation}
\begin{split}
\nabla\widehat{\mc Q}_b(\theta)
=
\frac1{N_b}
\bigg[
&n_b\nabla\widehat{\mc L}^{(b)}(\theta)+
\sum_{j=1}^{b-1}n_j
\left\{
\nabla\widehat{\mc L}^{(j)}(\widehat{\betab}^{(j)})
+
\nabla^2\widehat{\mc L}^{(j)}(\widehat{\betab}^{(j)})
(\theta-\widehat{\betab}^{(j)})
\right\}
\bigg].
\end{split}
\label{eq:general-surrogate-gradient}
\end{equation}
Consequently,
\begin{equation}
\nabla\widehat{\mc Q}_b(\theta)
=
\frac1{N_b}\sum_{j=1}^b n_j
\nabla\widehat{\mc L}^{(j)}(\theta)
+
\frac1{N_b}\sum_{j=1}^{b-1}n_jR_j(\theta).
\label{eq:surrogate-gradient-remainder}
\end{equation}
In particular,
\[
\nabla\widehat{\mc Q}_b(\widehat{\betab}^{(b)})
=
\frac1{N_b}\sum_{j=1}^b n_j
\left\{
\nabla\widehat{\mc L}^{(j)}(\widehat{\betab}^{(j)})
-
\nabla^2\widehat{\mc L}^{(j)}(\widehat{\betab}^{(j)})
(\widehat{\betab}^{(j)}-\widehat{\betab}^{(b)})
\right\}.
\]
Moreover, if
$\widehat\delta_j\in\BB(R_0)\cap\CC(l)$ and
${\mc E}_{\rm rlh}^{(j)}(R_0,l,M)$ holds, then
\begin{equation}
\|R_j(\betab^*)\|_\infty
\le
\frac32M\|\widehat\delta_j\|_2^2.
\label{eq:R-beta-star-bound}
\end{equation}
If, in addition, $d=\theta-\betab^*\in\BB(R_0)\cap\CC(l)$, then
\begin{equation}
\|R_j(\theta)\|_\infty
\le
2M\|\widehat\delta_j\|_2^2+M\|d\|_2^2.
\label{eq:R-general-center-bound}
\end{equation}
\end{lemma}

Building on this decomposition, the next lemma shows that the restricted strong convexity of the per-batch losses is inherited by the renewable surrogate, uniformly over the localized region.

\begin{lemma}[Curvature inherited by the renewable surrogate]
\label{lem:surrogate-curvature-inheritance}
Fix $1\le b\le B$, $R,\kappa>0$, and $l\ge0$.  Suppose
\begin{equation}
\widehat{\betab}^{(j)}-\betab^*
\in\BB(R)\cap\CC(l), j=1,\ldots,b-1,
\qquad\text{and}\qquad
\bigcap_{j=1}^{b}{\mc G}_{\rm rsc}^{(j)}(R,l,\kappa)
\label{eq:surrogate-curvature-sufficient-conditions}
\end{equation}
holds.  Then, for every pair $(\betab_1,\betab_2)$ satisfying
\[
\betab_1-\betab_2\in\BB(R)\cap\CC(l),
\qquad
\betab_2-\betab^*\in\BB(R)\cap\CC(l),
\]
the renewable surrogate satisfies
\begin{equation}
\left\langle
\nabla\widehat{\mc Q}_b(\betab_1)-
\nabla\widehat{\mc Q}_b(\betab_2),
\betab_1-\betab_2
\right\rangle
\ge
\kappa\|\betab_1-\betab_2\|_2^2.
\label{eq:inherited-surrogate-curvature}
\end{equation}
\end{lemma}

Given a nonnegative weight vector
$\lambdab=(\lambda_1,\ldots,\lambda_p)^\T$ with
$\lambdab_{\mc M}=0$, consider
\begin{equation}
\min_{\betab\in\RR^p}
\left\{
\widehat{\mc Q}_B(\betab)+\|\lambdab\circ\betab\|_1
\right\},
\label{appendix.unconstrained.program}
\end{equation}
and its constrained analogue
\begin{equation}
\min_{\betab\in\RR^p}
\left\{
\widehat{\mc Q}_B(\betab)+\|\lambdab\circ\betab\|_1
\right\}
\quad\text{subject to}\quad
\Cb\betab_{\mc M}=\tb.
\label{appendix.constrained.program}
\end{equation}

\begin{lemma}
\label{lem:intermediate.rate}
Let $\theta\in\RR^p$ satisfy
$\theta_{\overline{\mc M}^{\,c}}=0$. For the reduced problem, further assume that
$\Cb\theta_{\mc M}=\tb$. Let ${\mc E}\supset\overline{\mc M}$,
$d=|\mc E|$, and suppose that
\[
\|\lambdab\|_\infty\le\lambda,
\qquad
\|\lambdab_{{\mc E}^c}\|_{\min}\ge\rho\lambda,
\qquad
0<\rho\le1.
\]
Define $l=(2+2/\rho)d^{1/2}$, and assume that
\[
\big\|\nabla \widehat{\mc Q}_B(\theta)\|_\infty\le\frac12\rho\lambda,
\qquad
\theta-\betab^*\in\BB(R_0)\cap\CC(l).
\]

\begin{enumerate}[label=(\roman*)]
\item Let $\widetilde{\betab}_a$ be any solution to
\eqref{appendix.unconstrained.program}. Then
$\widetilde{\betab}_a-\theta\in\CC(l)$. 
Moreover, if \eqref{eq:inherited-surrogate-curvature} holds with $(b,R)=(B,R_0)$ and
\[
\kappa^{-1}
\left\{
\|\lambdab_{\mc S}\|_2+\big\|\nabla_{\mc E} \widehat{\mc Q}_B(\theta)\big\|_2
\right\}<R_0,
\]
then
\begin{equation}
\|\widetilde{\betab}_a-\theta\|_2
\le
\kappa^{-1}
\left\{
\|\lambdab_{\mc S}\|_2+\big\|\nabla_{\mc E} \widehat{\mc Q}_B(\theta)\big\|_2
\right\}.
\label{eq:weighted-full-bound}
\end{equation}

\item Let $\widetilde{\betab}_0$ be any solution to \eqref{appendix.constrained.program}. 
Then $\widetilde{\betab}_0-\theta\in\CC(l)$. 
Under the same inherited-curvature condition, if
\[
\kappa^{-1}
\left\{
\|\lambdab_{\mc S}\|_2
+
\|\Zb^\T\nabla_{\mc M}\widehat{\mc Q}_B(\theta)\|_2+\|\nabla_{{\mc E}\setminus\mc M}\widehat{\mc Q}_B(\theta)\|_2
\right\}<R_0,
\]
then
\begin{equation}
\|\widetilde{\betab}_0-\theta\|_2 \le \kappa^{-1} \bigg\{\|\lambdab_{\mc S}\|_2 +\|\Zb^\T\nabla_{\mc M}\widehat{\mc Q}_B(\theta)\|_2+\|\nabla_{{\mc E}\setminus\mc M}\widehat{\mc Q}_B(\theta)\|_2 \bigg\}.
\label{eq:weighted-reduced-bound}
\end{equation}
\end{enumerate}
\end{lemma}

Recall that $\betab^\dagger$ is defined as
\$
\betab^\dagger_{\mc M} = \betab^*_{\mc M} - \Cb^\T(\Cb \Cb^\T)^{-1}\hb_B, ~~\mbox{ and }~~ \betab^\dagger_{{\mc M}^c} = \betab^*_{{\mc M}^c},
\$
so that $\Cb\betab^\dagger_{\mc M} = \tb$ and $\|\betab^\dagger - \betab^*\|_2 \leq C_{\max}\|\hb_B\|_2$.
Moreover, recall that $\omegab^* = N_B^{-1}\sum_{j = 1}^B n_j \nabla \widehat{\mc L}^{(j)}(\betab^*)$.
For the following auxiliary arguments, write
\$
\gb^* = \nabla \widehat{\mc Q}_B(\betab^*), \quad \mbox{ and } \quad \gb^\dagger = \nabla \widehat{\mc Q}_B(\betab^\dagger).
\$
Moreover, recall that
\$
\Gamma_B = \Gamma_B(M, \eta_\infty) = 2MC_{\max}^2 \|\hb_B\|_2^2 + C_{\max}(J_u + \sqrt{m}\eta_\infty)\|\hb_B\|_2.
\$

\begin{lemma}
\label{lem:omega_discrepancy}
Suppose, for $j=1,\ldots,B-1$, that
\[
\widehat{\betab}^{(j)}-\betab^*
\in
\BB(R_0)\cap\CC(l).
\]
On the event
\[
\bigcap_{b=1}^{B-1}{\mc E}_{\rm rlh}^{(b)}(R_0,l,M),
\]
we have
\begin{equation}
\|\gb^* -\omegab^*\|_\infty
\le
\frac32M\phi^*.
\label{eq:gstar-score}
\end{equation}

Suppose
$\betab^\dagger - \betab^* \in\BB(R_0)\cap\CC(l)$.  
Under Assumption~\mainref{assm:hessian}, if, in addition,
\[
{\mc H}(\eta_\infty)
\cap
\bigcap_{b=1}^B{\mc E}_{\rm rlh}^{(b)}(R_0,l,M)
\]
holds, then
\begin{equation}
\left\|
\frac1{N_B}\sum_{b=1}^Bn_b
\left\{
\nabla\widehat{\mc L}^{(b)}(\betab^\dagger)
-
\nabla\widehat{\mc L}^{(b)}(\betab^*)
\right\}
\right\|_\infty
\le
\frac12MC_{\max}^2\|\hb_B\|_2^2
+
C_{\max}(J_u+\sqrt m\,\eta_\infty)\|\hb_B\|_2.
\label{eq:omega-transport}
\end{equation}
Moreover,
\begin{equation}
\|\gb^\dagger-\omegab^*\|_\infty
\le
2M\phi^*+\Gamma_B(M,\eta_\infty).
\label{eq:gdagger-score}
\end{equation}
Consequently,
\begin{equation}
\|\Zb^\T(\gb^\dagger-\omegab^*)_{\mc M}\|_2 + \|(\gb^\dagger-\omegab^*)_{\mc S}\|_2 \le
2(s+m)^{1/2}\{2M\phi^*+\Gamma_B(M,\eta_\infty)\}.
\label{eq:gdagger-block-score}
\end{equation}
\end{lemma}

\begin{lemma}
\label{lem:intermediate.rate2}
Let $e\in\{0,a\}$ and define
\[
\mathfrak w_e^{\rm ora}
=
\nabla\widehat{\mc Q}_B(\widehat{\betab}^{\rm ora}_e).
\]
Let ${\mc E}\supset\overline{\mc M}$, set $d=|{\mc E}|$, and suppose that
\[
\|\lambdab\|_\infty\leq\lambda,
\qquad
\|\lambdab_{{\mc E}^c}\|_{\min}\geq\rho\lambda,
\qquad
\big\|\mathfrak w_{e,\overline{\mc M}^{\,c}}^{\rm ora}\big\|_\infty
\leq\frac{1}{2}\rho\lambda,
\qquad
\lambda>0,
\qquad
0<\rho\leq1.
\]
Define $l=(2+2/\rho)d^{1/2}$. Assume that
\[
\widehat{\betab}^{\rm ora}_e-\betab^*
\in
\BB(R_0)\cap\CC(l),
\]
and, for each $j=1,\ldots,B-1$, that
\[
\widehat{\betab}^{(j)}-\betab^*
\in
\BB(R_0)\cap\CC(l).
\]
Then any solution $\widetilde{\betab}_e$ to the corresponding weighted program---\eqref{appendix.unconstrained.program} when $e=a$ and \eqref{appendix.constrained.program} when $e=0$---satisfies
\[
\widetilde{\betab}_e-\widehat{\betab}^{\rm ora}_e
\in
\CC(l).
\]
If the curvature inequality \eqref{eq:inherited-surrogate-curvature} holds with $(b,R)=(B,R_0)$ and
\[
\kappa^{-1}
\left\{
\|\lambdab_{\mc S}\|_2
+
\|\mathfrak w^{\rm ora}_{e,{\mc E}\setminus\overline{\mc M}}\|_2
\right\}
<R_0,
\]
then
\begin{equation}
\|\widetilde{\betab}_e-\widehat{\betab}^{\rm ora}_e\|_2
\le
\kappa^{-1}
\left\{
\|\lambdab_{\mc S}\|_2
+
\|\mathfrak w^{\rm ora}_{e,{\mc E}\setminus\overline{\mc M}}\|_2
\right\}.
\label{eq:oracle-centered-weighted-bound}
\end{equation}
For the reduced oracle,
$\Zb^\T\mathfrak w^{\rm ora}_{0,\mc M}=0$ and
$\mathfrak w^{\rm ora}_{0,\mc S}=0$; for the full oracle,
$\mathfrak w^{\rm ora}_{a,\overline{\mc M}}=0$.
\end{lemma}

\begin{lemma}
\label{lem:oracle_rate}
Assume $l\ge (m+s)^{1/2}$, 
% and, for $j=1,\ldots,B-1$, that
% \[
% \widehat{\betab}^{(j)}-\betab^*
% \in
% \BB(R_0)\cap\CC(l),
% \]
and the curvature inequality \eqref{eq:inherited-surrogate-curvature} holds with $(b,R) = (B, R_0)$.

For the full oracle, suppose
\[
\kappa^{-1}\big\|\gb^*_{\overline{\mc M}}\big\|_2<R_0.
\]
Then
\begin{equation}
\|\widehat{\betab}^{\rm ora}_a-\betab^*\|_2
\le
\kappa^{-1}\|\gb^*_{\overline{\mc M}}\|_2.
\label{eq:full-oracle-rate-first}
\end{equation}
If, in addition,
\[
\bigcap_{b=1}^{B-1}{\mc E}_{\rm rlh}^{(b)}(R_0,l,M)
\]
holds, then
\begin{equation}
\|\widehat{\betab}^{\rm ora}_a-\betab^*\|_2
\le
\kappa^{-1}
\left\{
\|\omegab^*_{\overline{\mc M}}\|_2
+2M(s+m)^{1/2}\phi^*
\right\},
\label{eq:full-oracle-rate}
\end{equation}

For the reduced oracle, suppose $\betab^\dagger-\betab^*\in\BB(R_0)\cap\CC(l)$ and
\[
\kappa^{-1}
\left\{
\|\Zb^\T\gb^\dagger_{\mc M}\|_2
+
\|\gb^\dagger_{\mc S}\|_2
\right\}<R_0.
\]
Then
\begin{equation}
\begin{split}
\|\widehat{\betab}^{\rm ora}_0-\betab^*\|_2
\le
&\kappa^{-1}
\left\{
\|\Zb^\T\gb^\dagger_{\mc M}\|_2
+
\|\gb^\dagger_{\mc S}\|_2
\right\}
+\|\betab^\dagger - \betab^*\|_2.
\end{split}
\label{eq:reduced-oracle-rate-first}
\end{equation}
Under Assumption~\mainref{assm:hessian}, if, in addition,
\[
{\mc H}(\eta_\infty)
\cap
\bigcap_{b=1}^B{\mc E}_{\rm rlh}^{(b)}(R_0,l,M)
\]
holds, then
\begin{equation}
\|\widehat{\betab}^{\rm ora}_0-\betab^*\|_2
\le \kappa^{-1}
\big\{
\|\Zb^\T\omegab^*_{\mc M}\|_2
+
\|\omegab^*_{\mc S}\|_2
+2(s+m)^{1/2}
\left[2M\phi^*+\Gamma_B\right]
\big\}
+\|\betab^\dagger - \betab^*\|_2.
\label{eq:reduced-oracle-rate}
\end{equation}
\end{lemma}

Recall from Section~\mainref{sec:method} that $ \Jb_0=\Jb_{\overline{\mc M}\overline{\mc M}}$, and $\Cb_0=(\Cb,\mathbf0_{r\times s})$, and define
\[
\Psi_0=\Cb_0\Jb_0^{-1}\Cb_0^\T, ~~\mbox{ and }~~
\Pb_0
=
\Jb_0^{-1/2}\Cb_0^\T\Psi_0^{-1}\Cb_0    \Jb_0^{-1/2}.
\]

\begin{lemma}[Kiefer--Bahadur representations for the testing oracles]
\label{lem:oracle_bahadur}
Assume that $\lambda_{\min}(\Jb_0)\ge J_l$, that the localization conditions
\[
\widehat{\betab}^{(b)}-\betab^*\in\BB(R_0)\cap\CC(l),
\quad b=1,\ldots,B-1,
\qquad
\widehat{\betab}^{\rm ora}_e-\betab^*\in\BB(R_0)\cap\CC(l),
\quad e\in\{0,a\},
\]
hold, and that the event
\[
\bigcap_{b=1}^B{\mc E}_{\rm rlh}^{(b)}(R_0,l,M)
\cap
{\mc H}(\eta_\infty)
\]
holds.
Then
\begin{equation}
\widehat{\betab}^{\rm ora}_{0,\overline{\mc M}}
-
\betab^*_{\overline{\mc M}}
=
\Jb_0^{-1/2}(\Pb_0-\Ib)\Jb_0^{-1/2}
\left\{
\frac1{N_B}\sum_{b=1}^Bn_b
\nabla_{\overline{\mc M}}\widehat{\mc L}^{(b)}(\betab^*)
\right\} -
\Jb_0^{-1}\Cb_0^\T\Psi_0^{-1}\hb_B+
\mc R_0,
\label{eq:reduced-oracle-bahadur}
\end{equation}
where
\begin{equation}
\begin{split}
\|\mc R_0\|_2
\lesssim
&\frac{Mq^{1/2}}{N_B}
\bigg[
\sum_{b=1}^{B-1}n_b
\left\{
\|\widehat{\betab}^{(b)}-\betab^*\|_2^2
+
\|\widehat{\betab}^{(b)}-\betab^*\|_2
\|\widehat{\betab}^{\rm ora}_0-\betab^*\|_2
\right\}
\\
&\hspace{32mm}
+n_B\|\widehat{\betab}^{\rm ora}_0-\betab^*\|_2^2
\bigg]
+
(s+m)\eta_\infty
\|\widehat{\betab}^{\rm ora}_0-\betab^*\|_2.
\end{split}
\label{eq:reduced-bahadur-rem}
\end{equation}
Moreover,
\begin{equation}
\widehat{\betab}^{\rm ora}_{a,\overline{\mc M}}
-
\betab^*_{\overline{\mc M}}
=
-
\Jb_0^{-1}
\left\{
\frac1{N_B}\sum_{b=1}^Bn_b
\nabla_{\overline{\mc M}}\widehat{\mc L}^{(b)}(\betab^*)
\right\}
+
\mc R_a,
\label{eq:full-oracle-bahadur}
\end{equation}
with
\begin{equation}
\begin{split}
\|\mc R_a\|_2
\lesssim
&\frac{M(m+s)^{1/2}}{N_B}
\bigg[
\sum_{b=1}^{B-1}n_b
\left\{
\|\widehat{\betab}^{(b)}-\betab^*\|_2^2
+
\|\widehat{\betab}^{(b)}-\betab^*\|_2
\|\widehat{\betab}^{\rm ora}_a-\betab^*\|_2
\right\}
\\
&\hspace{32mm}
+n_B\|\widehat{\betab}^{\rm ora}_a-\betab^*\|_2^2
\bigg]
+
(m+s)\eta_\infty
\|\widehat{\betab}^{\rm ora}_a-\betab^*\|_2.
\end{split}
\label{eq:full-bahadur-rem}
\end{equation}
%The hidden constants depend only on $M$, $J_l^{-1}$, and fixed numerical constants.
\end{lemma}

The following lemma provides the probabilistic ingredient for the anti-concentration argument used in the proof of Theorem~\mainref{thm:oracle.test}. 
Upon a careful examination of the proof of Lemma~S.7 in \citet{shi2019linear}, we found that a minor gap appears to arise in the derivation. 
We therefore provide a corrected anti-concentration inequality for the noncentral chi-square distribution. 
Specifically, the lemma uniformly bounds the probability mass assigned to an interval of width proportional to $\sqrt{r}$, over the degrees of freedom, the noncentrality parameter, and the center of the interval.

\begin{lemma}[Noncentral chi-square anti-concentration]
\label{lem:chisq-anticoncentration}
There exists an absolute constant $C>0$ such that, for every $0<\epsilon\le1$,
\begin{equation}
\sup_{r\ge1,\ e\ge0,\ x\in\RR}
\left[
\PP\{\chi^2(r,e)\le x+\epsilon\sqrt r\}
-
\PP\{\chi^2(r,e)\le x-\epsilon\sqrt r\}
\right]
\le C\sqrt\epsilon.
\label{eq:chisq-anticoncentration}
\end{equation}
The supremum over $r$ is taken over the positive integers.
\end{lemma}

Now, recall the definition of $\Vb_0 = \Cb_0 \Jb_0^{-1} \Kb_0 \Jb_0^{-1}\Cb_0^\T$ and $\Omega_0 = \Vb_0^{-1/2} \Cb_0 \Jb_0^{-1}$.
The following lemma establishes the properties of $\Vb_0$ and $\Omega_0$.

\begin{lemma}\label{lem:V0-conditioning}
Under Assumptions~\mainref{assm:local_alternative}, \mainref{assm:hessian}, and \mainref{assm:score}, the matrix $\Vb_0$ is nonsingular with
\[
\lambda_{\min}(\Vb_0)
\ge
\frac{K_l}{J_u^2C_{\max}^2},
\qquad
\left\|
\Omega_0
\right\|_2
\le
K_l^{-1/2}.
\]
\end{lemma}

The final auxiliary lemma is a deterministic perturbation bound for the sandwich covariance.  
It converts operator-norm consistency of the plug-in matrices $\widehat{\Jb}^{(k)}_e$ and $\widehat{\Kb}^{(k)}_e$ into closeness of the feasible and oracle test statistics.  

\begin{lemma}[Plug-in stability]
\label{lem:plugin_stability}
Suppose Assumptions~\mainref{assm:local_alternative}, \mainref{assm:hessian}, and \mainref{assm:score} hold.  
Let $\widehat{\Jb}$ and $\widehat{\Kb}$ be symmetric matrices with
\[
\|\widehat{\Jb}-\Jb_0\|_2\le\tau_J\le J_l/2,
\qquad
\|\widehat{\Kb}-\Kb_0\|_2\le\tau_K,
\]
and define
$\widehat{\Vb}=\Cb_0\widehat{\Jb}^{-1}\widehat{\Kb}\widehat{\Jb}^{-1}\Cb_0^\T$.  
Then $\widehat{\Jb}$ is nonsingular and, with $C_V=6J_l^{-3}(K_u+J_l)$,
\begin{equation}
\|\widehat{\Vb}-\Vb_0\|_2
\le
C_V\|\Cb_0\|_2^2(\tau_J+\tau_K),
\label{eq:plugin-V-error}
\end{equation}
where $\|\Cb_0\|_2^2=\lambda_{\max}(\Cb\Cb^\T)$.  
If, in addition,
$C_V\|\Cb_0\|_2^2(\tau_J+\tau_K)\le\lambda_{\min}(\Vb_0)/2$, then $\widehat{\Vb}$ is
nonsingular with $\lambda_{\min}(\widehat{\Vb})\ge\lambda_{\min}(\Vb_0)/2$, and
\begin{equation}
\big\|
\Vb_0^{1/2}\widehat{\Vb}^{-1}\Vb_0^{1/2}-\Ib_r
\big\|_2
\le
2\lambda_{\min}(\Vb_0)^{-1}C_V\|\Cb_0\|_2^2(\tau_J+\tau_K).
\label{eq:plugin-whitened-error}
\end{equation}
In particular, for every $N>0$ and every $\ub\in\RR^r$,
\begin{equation}
\left|
N\|\widehat{\Vb}^{-1/2}\ub\|_2^2-N\|\Vb_0^{-1/2}\ub\|_2^2
\right|
\le
\big\|
\Vb_0^{1/2}\widehat{\Vb}^{-1}\Vb_0^{1/2}-\Ib_r
\big\|_2
\,\cdot
N\|\Vb_0^{-1/2}\ub\|_2^2.
\label{eq:plugin-quadratic-error}
\end{equation}
\end{lemma}

\section{Proofs of results in Section~\mainref{sec:finite_sample_theory}}\label{sec:theorem_proofs}

\subsection{Proof of Theorem~\mainref{thm:batch_b_init}}
\label{sec:proof_batch_b_init}

Let
\[
\widehat{\deltab}
=
\widehat{\betab}^{(b)}-\betab^*,
\qquad
D_b
=
\left\langle
\nabla\widehat{\mc Q}_b(\widehat{\betab}^{(b)})
-
\nabla\widehat{\mc Q}_b(\betab^*),
\widehat{\deltab}
\right\rangle.
\]
Convexity of $\widehat{\mc Q}_b$ gives $D_b\ge0$.  By the KKT condition, there is
$\widehat{\vb}\in\partial\|\widehat{\betab}^{(b)}\|_1$ such that
\[
\nabla\widehat{\mc Q}_b(\widehat{\betab}^{(b)})
+
\lambda_b\widehat{\vb}=0.
\]
Recall that $\mathcal A$ is the support of $\betab^*$.
Therefore, since $\widehat\vb \in \partial\|\widehat{\betab}^{(b)}\|_1$, we have
\begin{equation}\label{eq:thm1-subgradient}
\begin{split}
\langle \widehat \vb, \betab^* - \widehat{\betab}^{(b)} \rangle & \leq \|\betab^*\|_1 - \|\widehat{\betab}^{(b)}\|_1  \\
& = \|\betab^*\|_1 - \|\betab^* + \widehat{\deltab}\|_1  \\
& = \|\betab^*_{\mc A}\|_1 - \|\widehat{\deltab}_{{\mc A}^c}\|_1 - \|\widehat{\deltab}_{\mc A} + \betab^*_{\mc A}\|_1 \leq \|\widehat{\deltab}_{\mc A}\|_1 - \|\widehat{\deltab}_{{\mc A}^c}\|_1.
\end{split}
\end{equation}
Recall that $\omegab^{(b)} = N_b^{-1} \sum_{j = 1}^b n_j \nabla \widehat{\mc L}^{(j)}(\betab^*)$.
By Lemma~\ref{lem:decomposition},
\[
\nabla\widehat{\mc Q}_b(\betab^*)
=
\omegab^{(b)}
+
\frac1{N_b}\sum_{j=1}^{b-1}n_jR_j(\betab^*).
\]
The event in \maineqref{event: historical localization} and the RLH event in \maineqref{batch_b_event} with \eqref{eq:R-beta-star-bound} yield
\[
\left\|
\frac1{N_b}\sum_{j=1}^{b-1}n_jR_j(\betab^*)
\right\|_\infty
\le
\frac32M\phi^{(b)}.
\]
Combining this bound, \eqref{eq:thm1-subgradient}, and the KKT equation gives
\begin{align}
D_b & = \lambda_b \langle \widehat{\vb}, \betab^* - \widehat{\betab}^{(b)} \rangle + \langle \omegab^{(b)}, \betab^* - \widehat{\betab}^{(b)}\rangle + \bigg\langle \frac{1}{N_b} \sum_{j = 1}^{b-1} n_j R_j(\betab^*), \betab^* - \widehat{\betab}^{(b)} \bigg\rangle  \nn \\
&\le
\lambda_b
\left(
\|\widehat{\deltab}_{\mc A}\|_1
-
\|\widehat{\deltab}_{\mc A^c}\|_1
\right)
+
\left(
\|\omegab^{(b)}\|_\infty+\frac32M\phi^{(b)}
\right)
\|\widehat{\deltab}\|_1
\notag\\
&\le
\frac{\lambda_b}{2}
\left(
3\|\widehat{\deltab}_{\mc A}\|_1
-
\|\widehat{\deltab}_{\mc A^c}\|_1
\right),
\label{eq:thm1-upper-bregman}
\end{align}
where the last inequality follows by $\lambda_b \geq 2\{\|\omegab^{(b)}\|_\infty + 1.5 M \phi^{(b)}\}$.
Since $D_b\ge0$, this proves
\begin{equation}
\|\widehat{\deltab}_{\mc A^c}\|_1
\le
3\|\widehat{\deltab}_{\mc A}\|_1.
\label{eq:thm1-cone}
\end{equation}
Consequently,
\begin{equation}
\|\widehat{\deltab}\|_1
\le
4\|\widehat{\deltab}_{\mc A}\|_1
\le
4|\mc A|^{1/2}\|\widehat{\deltab}\|_2.
\label{eq:thm1-global-cone}
\end{equation}
Also, \eqref{eq:thm1-upper-bregman} implies
\begin{equation}
D_b
\le
\frac32|\mc A|^{1/2}\lambda_b
\|\widehat{\deltab}\|_2.
\label{eq:thm1-upper-l2}
\end{equation}

We now turn to the lower bound of $D_b$.
Define the intermediate vector $\widetilde\betab = \betab^* + \eta \widehat{\deltab}$ and $\widetilde \deltab = \widetilde \betab - \betab^* = \eta \widehat \deltab$, where
\$
\eta = \sup\left\{
 t\in[0,1]:t\widehat{\deltab}\in\BB(R_0)
\right\}.
\$
Note that if $\|\widehat{\deltab}\|_2 > R_0$, then $\eta \in (0,1)$ and $\|\widetilde \deltab\|_2 = R_0$.
If instead $\|\widehat{\deltab}\|_2 \leq R_0$, then $\eta = 1$, and hence $\widetilde \betab = \widehat\betab^{(b)}$.
Moreover, \eqref{eq:thm1-global-cone} also yields $\widetilde\deltab \in \BB(R_0) \cap \CC(4|\mc A|^{1/2})$.
By the convexity of $\betab \to \widehat{\mc Q}_b(\betab)$ and Lemma F.2 in the supplementary material of \citet{fan2018lamm}, we obtain
\#
0 \leq \langle \nabla \widehat{\mc Q}_b(\widetilde \betab) - \nabla \widehat{\mc Q}_b(\betab^*), \widetilde \betab - \betab^* \rangle \leq \eta \langle \nabla \widehat{\mc Q}_b(\widehat{\betab}^{(b)}) - \nabla \widehat{\mc Q}_b(\betab^*), \widehat{\betab}^{(b)} - \betab^* \rangle. \label{initial.rate.proof1}
\#
Now, since $\widetilde\betab - \betab^* \in \BB(R_0) \cap \CC(4|\mc A|^{1/2})$, applying Lemma~\ref{lem:surrogate-curvature-inheritance} gives
\$
\langle \nabla \widehat{\mc Q}_b(\widetilde \betab) - \nabla \widehat{\mc Q}_b(\betab^*), \widetilde \betab - \betab^* \rangle \geq \kappa \|\widetilde\betab - \betab^*\|_2^2.
\$
Combining this with \eqref{initial.rate.proof1} and \eqref{eq:thm1-upper-l2}, we have
\$
\kappa\|\widetilde\betab - \betab^*\|_2^2 \leq \frac{3}{2}|\mc A|^{1/2} \lambda_b \eta \|\widehat{\betab}^{(b)} - \betab^*\|_2 = \frac{3}{2}|\mc A|^{1/2} \lambda_b\|\widetilde\betab - \betab^*\|_2.
\$
Canceling $\|\widetilde\betab - \betab^*\|_2$, we obtain
\$
\|\widetilde\betab - \betab^*\|_2 \leq \frac{3}{2\kappa}|\mc A|^{1/2} \lambda_b < R_0,
\$
implying that $\widetilde\betab = \widehat{\betab}^{(b)}$.
This proves the $\ell_2$ bound.  
The $\ell_1$ bound follows from\eqref{eq:thm1-global-cone}.  
\qed

\subsection{Proof of Theorem~\mainref{thm:batch_b_tightening}}
\label{sec:batch_b_tightening_proof}

Throughout the proof, $p_0=p'_+(\gamma_0)$.  
Concavity and monotonicity of $p$ imply $0<p_0\le1$ and $0\le p'_{\lambda,+}(t)\le\lambda$.
The function
\[
h(c)=c\kappa\gamma_0-p_0\sqrt{c^2+1}-2
\]
is strictly increasing because, for every $c\ge0$,
\[
h'(c)
=
\kappa\gamma_0-\frac{p_0c}{\sqrt{c^2+1}}
>
\kappa\gamma_0-p_0
>
0,
\]
while $h(0)<0$ and $h(c)\to\infty$.  
Hence \maineqref{def:c.aug} has a unique positive solution.

For $e\in\{0,a\}$ and $k \geq 1$, define
\[
{\mc E}_{e,k}
=
\overline{\mc M}
\cup
\left\{
j\in\overline{\mc M}^{\,c}:
\widehat w_{e,j}^{(k-1)}<p_0\lambda
\right\},
\]
and let
\[
\lambda_{e,j}^{(k-1)}
=
p'_{\lambda,+}
\left(|\widehat\beta_{e,j}^{({\rm LLA},k-1)}|\right),
\qquad j\in\mc M^c,
\]
and set $\lambda_{e,j}^{(k-1)}=0$ on $\mc M$.
By construction, $|\lambdab_{e, {\mc E}_{e,k}^c}^{(k-1)}|_{\min} \geq p_+ ^\prime(\gamma_0)\lambda$.
If $j\in{\mc E}_{e,k}\setminus\overline{\mc M}$, then monotonicity of $p'_+$ and the fact that $p^\prime_+(\gamma_0) = p_0$ give
\#
|\widehat\beta_{e,j}^{({\rm LLA},k-1)}|>\gamma_0\lambda. \label{tighten.proof1}
\#

We first show the crude induction bound
\begin{equation}
\|\widehat{\betab}^{({\rm LLA},k)}_e-\betab^*\|_2
\le
c_{\rm aug}\gamma_0(s+m)^{1/2}\lambda,
\qquad
|{\mc E}_{e,k+1}|\le (c_{\rm aug}^2 + 1)(s+m).
\label{eq:crude-lla-localization}
\end{equation}
We first consider $k = 0$, and note that $\widehat{\betab}_e^{({\rm LLA}, k)} = \widehat{\betab}^{(B)}$.
Since $p_0\le1$, $\sqrt{c_{\rm aug}^2 + 1}\ge1$, and $|\mc A|\le s + m$,
\[
l=(2+2/p_0)\sqrt{(s+m)(c_{\rm aug}^2 + 1)}
\ge
4(s + m)^{1/2}
\ge
4|\mc A|^{1/2}.
\]
Consequently, the RSC and historical RLH events with cone parameter $l$ imply the corresponding events with cone parameter $4|\mc A|^{1/2}$.
Moreover, $2M\phi^*\ge(3/2)M\phi^*$, thus either event in the theorem implies the initializer tuning condition in Theorem~\mainref{thm:batch_b_init}, yielding
\#
\|\widehat{\betab}^{(B)}-\betab^*\|_2
\le
\frac{3}{2\kappa}|\mc A|^{1/2}\lambda
<
c_{\rm aug}\gamma_0(s+m)^{1/2}\lambda, \label{tighten.proof2}
\#
where the last inequality follows from $c_{\rm aug}\kappa\gamma_0=2+p_0\sqrt{c_{\rm aug}^2 + 1}>3/2$ and $|\mc A|\le s+m$.  
Moreover, if $j \in {\mc E}_{e,1} \setminus \overline{\mc M}$, then \eqref{tighten.proof1} implies $|\widehat{\beta}_{e,j}^{({\rm LLA}, 0)}| > \gamma_0 \lambda$.
Since $\beta_j^* = 0$ for $j \in {\mc E}_{e,1}\setminus\overline{\mc M}$, we get
\#
|{\mc E}_{e,1}\setminus\overline{\mc M}|^{1/2}
\le
\frac{
\|(\widehat{\betab}_e^{({\rm LLA},0)}-\betab^*)_
{{\mc E}_{e,1}\setminus\overline{\mc M}}\|_2
}{\gamma_0\lambda}
\le
\frac{
\|\widehat{\betab}_e^{({\rm LLA},0)}-\betab^*\|_2
}{\gamma_0\lambda} \leq c_{\rm aug}(s + m)^{1/2},
\label{eq:theorem2-cardinality-basic}
\#
where the last inequality follows by the crude bound \eqref{tighten.proof2}.
Therefore,
\$
|{\mc E}_{e,1}| \leq (s + m) + c_{\rm aug}^2(s+m) = (c_{\rm aug}^2 + 1)(s+m).
\$
This establishes the bounds in \eqref{eq:crude-lla-localization} for $k = 0$.

Now, assume \eqref{eq:crude-lla-localization} holds up to step $k-1$.
The same argument as in deriving \eqref{eq:theorem2-cardinality-basic} gives $|{\mc E}_{a,k}| \leq (c_{\rm aug}^2 + 1) (s + m)$.
We next consider the convergence rate bound in \eqref{eq:crude-lla-localization}.
For the full model, Lemma~\ref{lem:omega_discrepancy} and \maineqref{event:theorem2-full} imply
\$
\|\gb^*\|_\infty \leq \frac{1}{2}p_0 \lambda.
\$
Applying Lemma~\ref{lem:intermediate.rate} with $\theta=\betab^*$, $\rho=p_+^\prime(\gamma_0)$, and $\mc E={\mc E}_{a,k}$, where
\[
|{\mc E}_{a,k}|\leq(c_{\rm aug}^2+1)(s+m),
\]
and invoking the curvature conclusion of Lemma~\ref{lem:surrogate-curvature-inheritance}, which holds under the stated assumptions, we obtain
\$
\|\widehat{\betab}^{({\rm LLA},k)}_a-\betab^*\|_2
&\le
\kappa^{-1}
\left\{
\|\lambdab^{(k-1)}_{a,\mc S}\|_2
+
\|\gb^*_{{\mc E}_{a,k}}\|_2
\right\}
\\
&\le
\kappa^{-1}
\left(1+\frac12p_0\sqrt{c_{\rm aug}^2 + 1}\right)(s+m)^{1/2}\lambda
\\
&<
\kappa^{-1}(2+p_0\sqrt{c_{\rm aug}^2 + 1})(s+m)^{1/2}\lambda
=
c_{\rm aug}\gamma_0(s+m)^{1/2}\lambda,
\$
where the final equality follows from the definition of $c_{\rm aug}$.
Therefore, \eqref{eq:crude-lla-localization} holds for $e = a$.

For the reduced model, Lemma~\ref{lem:omega_discrepancy} and \maineqref{event:theorem2-reduced} give $\|\gb^\dagger\|_\infty\le p_0\lambda/2$.
Moreover, $\betab^\dagger-\betab^*$ is supported on $\mc M$, so
\[
\|\betab^\dagger-\betab^*\|_1 \le m^{1/2}\|\betab^\dagger - \betab^*\|_2 \leq l\|\betab^\dagger - \betab^*\|_2.
\]
Since $C_{\rm max}\|\hb_B\|_2<R_0$, this proves $\betab^\dagger-\betab^*\in\BB(R_0)\cap\CC(l)$.
Applying Lemma~\ref{lem:intermediate.rate} with center $\theta = \betab^\dagger$ gives
\begin{align*}
\|\widehat{\betab}^{({\rm LLA},k)}_0-\betab^\dagger\|_2
&\le
\kappa^{-1}
\bigl\{
(s+m)^{1/2}\lambda
+
\|\Zb^\T\gb^\dagger_{\mc M}\|_2
+
\|\gb^\dagger_{{\mc E}_{0,k}\setminus\mc M}\|_2
\bigr\}
\\
&\le
\kappa^{-1}
\left(1+\frac{p_0\sqrt{c_{\rm aug}^2 + 1}}{\sqrt2}\right)(s+m)^{1/2}\lambda,
\end{align*}
where we used $\sqrt m+\sqrt{|{\mc E}_{0,k}|-m}\le\sqrt{2|{\mc E}_{0,k}|}$.
Since $C_{\max} \|\hb_B\|_2 \le\lambda/\kappa\le (s+m)^{1/2}\lambda/\kappa$,
\begin{align*}
\|\widehat{\betab}^{({\rm LLA},k)}_0-\betab^*\|_2
&\le
\kappa^{-1}
\left(2+\frac{p_0\sqrt{c_{\rm aug}^2 + 1}}{\sqrt2}\right)(s+m)^{1/2}\lambda
\\
&<
\kappa^{-1}(2+p_0\sqrt{c_{\rm aug}^2 + 1})(s+m)^{1/2}\lambda
=
c_{\rm aug}\gamma_0(s+m)^{1/2}\lambda.
\end{align*}
This completes the crude induction for \eqref{eq:crude-lla-localization} for both models.

We now prove the refined contractions.  
For every $j\in\mc S$, if $|\widehat\beta_{e,j}^{({\rm LLA},k-1)}-\beta_j^*| \geq \gamma_0 \lambda$, then $\lambda_{e,j}^{(k-1)} \leq \lambda \leq \gamma_0^{-1}|\widehat\beta_{e,j}^{({\rm LLA},k-1)}-\beta_j^*|$.
Otherwise, if $|\widehat\beta_{e,j}^{({\rm LLA},k-1)}-\beta_j^*| \leq \gamma_0 \lambda$, then $\lambda_{e,j}^{(k-1)} \leq p_{\lambda, +}^\prime((|\beta_j^*|-\gamma_0\lambda)_+)$ due to the monotonicity of $p_{\lambda, +}^\prime$.
Putting together the pieces, we have for every $j \in \mc S$,
\$
\lambda_{e,j}^{(k-1)}
\le
p'_{\lambda,+}
\bigl((|\beta_j^*|-\gamma_0\lambda)_+\bigr)
+
\gamma_0^{-1}
|\widehat\beta_{e,j}^{({\rm LLA},k-1)}-\beta_j^*|.
\$
Taking Euclidean norms gives
\begin{equation}
\|\lambdab_{e,\mc S}^{(k-1)}\|_2
\le
\left\|p'_{\lambda,+}
\bigl((|\betab^*_{\mc S}|-\gamma_0\lambda)_+\bigr)\right\|_2
+
\gamma_0^{-1}
\|(\widehat{\betab}^{({\rm LLA},k-1)}_e-\betab^*)_{\mc S}\|_2.
\label{eq:penalty-weight-vector}
\end{equation}

For the full model, Lemma~\ref{lem:intermediate.rate} and the disjoint decomposition
${\mc E}_{a,k}=\overline{\mc M}\cup
({\mc E}_{a,k}\setminus\overline{\mc M})$ give
\begin{align*}
\|\widehat{\betab}^{({\rm LLA},k)}_a-\betab^*\|_2
\le
\kappa^{-1}\bigg\{
&\|\lambdab_{a,\mc S}^{(k-1)}\|_2
+
\|\gb^*_{\overline{\mc M}}\|_2
+
\|\gb^*_{{\mc E}_{a,k}\setminus\overline{\mc M}}\|_2
\bigg\}.
\end{align*}
By \eqref{eq:theorem2-cardinality-basic} applied on the indicated block,
\[
\|\gb^*_{{\mc E}_{a,k}\setminus\overline{\mc M}}\|_2
\le
\frac{p_0}{2\gamma_0}
\|(\widehat{\betab}^{({\rm LLA},k-1)}_a-\betab^*)_{{\mc E}_{a,k}\setminus\overline{\mc M}}\|_2.
\]
Combining this display with \eqref{eq:penalty-weight-vector}, and using that $\mc S$ and ${\mc E}_{a,k}\setminus\overline{\mc M}$ are disjoint, yields
\$
\|\widehat{\betab}^{({\rm LLA},k)}_a-\betab^*\|_2
& \le
\kappa^{-1}\big\{\big\|p'_{\lambda,+}
\bigl((|\betab^*_{\mc S}|-\gamma_0\lambda)_+\bigr)\big\|_2 + \|\gb_{\overline{\mc M}}^*\|_2\big\} \\
&\qquad + \frac1{\kappa\gamma_0}
\left(1+\frac{p_0^2}{4}\right)^{1/2}
\|\widehat{\betab}^{({\rm LLA},k-1)}_a-\betab^*\|_2.
\$
Since $p_0\le1$, the coefficient is at most $\sqrt5/(2\kappa\gamma_0)=\delta$.  
Furthermore, Lemma~\ref{lem:omega_discrepancy} gives
\$
\|\gb^*_{\overline{\mc M}}\|_2 \leq  \|\omegab^*_{\overline{\mc M}}\|_2 + \| (\omegab^* - \gb^*)_{\overline{\mc M}}\|_2 \leq \|\omegab^*_{\overline{\mc M}}\|_2 + 2M\sqrt{s+m}\phi^*.
\$
This proves \maineqref{eq:theorem2-full-contraction}.  
% The upper bound for $r_a^{\sharp}$ follows from
% \eqref{eq:gstar-score} and
% $\|(\gb^*-\omegab^*)_{\overline{\mc M}}\|_2
% \le q^{1/2}\|\gb^*-\omegab^*\|_\infty$.

For the reduced model, apply Lemma~\ref{lem:intermediate.rate} around $\betab^\dagger$ and split the score coordinates as $\mc S$ and ${\mc E}_{0,k}\setminus\overline{\mc M}$.  This gives
\begin{align*}
\|\widehat{\betab}^{({\rm LLA},k)}_0-\betab^\dagger\|_2
\le
\kappa^{-1}\bigg\{
\|\lambdab_{0,\mc S}^{(k-1)}\|_2
+
\|\Zb^\T\gb^\dagger_{\mc M}\|_2
+
\|\gb^\dagger_{\mc S}\|_2 +
\|\gb^\dagger_{{\mc E}_{0,k}\setminus\overline{\mc M}}\|_2
\bigg\}.
\end{align*}
As before,
\[
\|\gb^\dagger_{{\mc E}_{0,k}\setminus\overline{\mc M}}\|_2
\le
\frac{p_0}{2\gamma_0}
\|(\widehat{\betab}^{({\rm LLA},k-1)}_0-\betab^*)_{{\mc E}_{0,k}\setminus\overline{\mc M}}\|_2.
\]
Using \eqref{eq:penalty-weight-vector}, the disjointness of the two error blocks, and
$\|\widehat{\betab}^{({\rm LLA},k)}_0-\betab^*\|_2
\le
\|\widehat{\betab}^{({\rm LLA},k)}_0-\betab^\dagger\|_2+C_{\rm max}\|\hb_B\|_2$, we obtain
\$
\|\widehat{\betab}^{({\rm LLA},k)}_0-\betab^*\|_2
& \le
\!\!\kappa^{-1}\big\{\big\|p'_{\lambda,+}
\bigl((|\betab^*_{\mc S}|-\gamma_0\lambda)_+\bigr)\big\|_2 + \|\Zb^\T\gb^\dagger_{\mc M}\|_2
+
\|\gb^\dagger_{\mc S}\|_2\big\}
+
\delta
\|\widehat{\betab}^{({\rm LLA},k-1)}_0-\betab^*\|_2 \\
& \qquad+ C_{\rm \max} \|\hb_B\|_2.
\$
Finally, \eqref{eq:gdagger-block-score} gives that the previous inequality is further bounded as
\$
\|\widehat{\betab}^{({\rm LLA},k)}_0-\betab^*\|_2 \leq \delta
\|\widehat{\betab}^{({\rm LLA},k-1)}_0-\betab^*\|_2
+\kappa^{-1}
\left\|p'_{\lambda,+}
\bigl((|\betab^*_{\mc S}|-\gamma_0\lambda)_+\bigr)\right\|_2
+r_{0}^{\rm ora}.
\$
This completes the proof.  \qed

\subsection{Proof of Theorem~\mainref{thm:oracle_recovery}}
\label{sec:oracle_recoveery_proof}

Fix $e\in\{0,a\}$ and write
$\widehat{\betab}_e^k=\widehat{\betab}^{({\rm LLA},k)}_e$.
In the proof of Theorem~\mainref{thm:batch_b_tightening}, we have shown the crude bounds \eqref{eq:crude-lla-localization}.
Therefore,
\[
|{\mc E}_{e,k}|\le (c_{\rm aug}^2 + 1)(s+m),
\qquad k\ge1.
\]
Define the bad-coordinate set
\[
{\mc B}_{e,k}
=
\left\{
 j\in\mc M^c:
 |\widehat\beta_{e,j}^k-\beta_j^*|
 \ge\gamma_0\lambda
\right\}.
\]
If $j\in\mc S\setminus{\mc B}_{e,k-1}$, then the beta-min condition gives
\[
|\widehat\beta_{e,j}^{k-1}|
>
\gamma_1\lambda,
\]
so $p'_{\lambda,+}(|\widehat\beta_{e,j}^{k-1}|)=0$.  On the remaining signal coordinates
the weight is at most $\lambda$.  Therefore,
\begin{equation}
\|\lambdab_{e,\mc S}^{(k-1)}\|_2
\le
\lambda|\mc S\cap{\mc B}_{e,k-1}|^{1/2}.
\label{eq:strong-oracle-signal-weight}
\end{equation}
Moreover, if
$j\in{\mc E}_{e,k}\setminus\overline{\mc M}$, then
$|\widehat\beta_{e,j}^{k-1}|>\gamma_0\lambda$ and $\beta_j^*=0$, whence
\begin{equation}
{\mc E}_{e,k}\setminus\overline{\mc M}
\subset
{\mc B}_{e,k-1}\setminus\mc S.
\label{eq:strong-oracle-set-inclusion}
\end{equation}
At iteration $k$, apply Lemma~\ref{lem:intermediate.rate2} with $\lambdab=\lambdab_e^{(k-1)}$, $\rho=p_0$, and ${\mc E}={\mc E}_{e,k}$.
Assumption~\mainref{assm:penalty_function} gives $\|\lambdab_e^{(k-1)}\|_\infty\le\lambda$, while the definition of ${\mc E}_{e,k}$ gives
$\|\lambdab^{(k-1)}_{e,{\mc E}_{e,k}^c}\|_{\min}\ge p_0\lambda$.
Because ${\mc E}_{e,k}\supset\overline{\mc M}$ and
$|{\mc E}_{e,k}|\le (c_{\rm aug}^2 + 1)(s+m)$,
\[
4|\mc A|^{1/2}
\le
(2+2/p_0)|{\mc E}_{e,k}|^{1/2}
\le
l.
\]
The historical displacements therefore belong to the cone required by
Lemma~\ref{lem:intermediate.rate2}.  The oracle displacement is supported on
$\overline{\mc M}$ and hence belongs to the same cone, while its ball condition is
assumed in the theorem.  The RSC event with cone parameter $l$ implies the RSC event
with cone parameter $(2+2/p_0)|{\mc E}_{e,k}|^{1/2}$.

The score condition in \maineqref{eq:oracle-recovery-score-radius}, together with
\eqref{eq:strong-oracle-signal-weight} and
\eqref{eq:strong-oracle-set-inclusion}, bounds the right-hand side of
\eqref{eq:oracle-centered-weighted-bound} by
\begin{align}
\frac{\lambda}{\kappa}
\left\{
|\mc S\cap{\mc B}_{e,k-1}|^{1/2}
+
\frac{p_0}{2}|{\mc B}_{e,k-1}\setminus\mc S|^{1/2}
\right\} \leq
\frac{\sqrt{17}}{4}\frac{\lambda}{\kappa}
|{\mc B}_{e,k-1}|^{1/2},
\label{eq:strong-oracle-one-step-rhs}
\end{align}
where we used $p_0\leq1/2$.  
Moreover, the crude localization gives
\[
|{\mc B}_{e,k-1}|^{1/2}
\le
\frac{\|\widehat{\betab}_e^{k-1}-\betab^*\|_2}{\gamma_0\lambda}
\le
c_{\rm aug}(s+m)^{1/2}.
\]
Since $\sqrt{17}/4 < 5/4 < \kappa \gamma_0$, the right-hand side of \eqref{eq:strong-oracle-one-step-rhs} is strictly smaller than $c_{\rm aug}\gamma_0(s+m)^{1/2}\lambda < R_0$.
All conditions of Lemma~\ref{lem:intermediate.rate2} are therefore satisfied, and the lemma yields
\begin{equation}
\|\widehat{\betab}_e^k-\widehat{\betab}_e^{\rm ora}\|_2
\le
\frac{\sqrt{17}}{4}\frac{\lambda}{\kappa}
|{\mc B}_{e,k-1}|^{1/2}.
\label{eq:strong-oracle-upper}
\end{equation}

We next contract the bad-coordinate set. 
For any $j\in{\mc B}_{e,k}$,
\$
|\widehat\beta_{e,j}^k-\widehat\beta_{e,j}^{\rm ora}|
\ge |\widehat\beta_{e,j}^k-\beta_j^*| - |\widehat\beta_{e,j}^{\rm ora}-\beta_j^*| \ge \gamma_0\lambda-\frac{\lambda}{5\bar\delta\kappa} > \frac{21}{20}\frac{\lambda}{\bar\delta\kappa}.
\$
Therefore, we obtain
\begin{equation}
\frac{21}{20}\frac{\lambda}{\bar\delta\kappa}
|{\mc B}_{e,k}|^{1/2}
\le
\|\widehat{\betab}_e^k-\widehat{\betab}_e^{\rm ora}\|_2.
\label{eq:strong-oracle-lower}
\end{equation}
Combining \eqref{eq:strong-oracle-upper} and
\eqref{eq:strong-oracle-lower} yields, for every $k\geq1$,
\begin{equation}
|{\mc B}_{e,k}|^{1/2}
\leq
\frac{5\sqrt{17}}{21}\bar\delta
|{\mc B}_{e,k-1}|^{1/2}
\le
\bar\delta|{\mc B}_{e,k-1}|^{1/2}.
\label{eq:bad-set-contraction}
\end{equation}

At initialization, $|{\mc B}_{e,0}|\leq\widehat s$: every bad null coordinate belongs to $\operatorname{supp} (\widehat{\betab}^{(B)})$, and at most $s$ signal coordinates can be bad.  
If $\widehat s=0$, then \eqref{eq:bad-set-contraction} gives
${\mc B}_{e,k}=\varnothing$ for every $k\ge0$.
If $\widehat s\ge1$, then \eqref{eq:bad-set-contraction} gives
\[
|{\mc B}_{e,k}|^{1/2}
\le
\bar\delta^{k}|{\mc B}_{e,0}|^{1/2}
\leq
\bar\delta^{k}\widehat s^{1/2}
\leq
\bar\delta
<1
\]
for all $k \geq K_0$.
Because $|{\mc B}_{e,k}|$ is a nonnegative integer,
${\mc B}_{e,k}=\varnothing$. 
Consequently, by \eqref{eq:strong-oracle-upper}, we have $\widehat{\betab}_e^{k} = \widehat{\betab}_e^{\rm ora}$ for all $k \geq K_0 + 1$.

\qed

\subsection{Proof of Corollary~\mainref{cor:batch_b_init}}

For every fixed $b$ and coordinate $k$, Assumption~\mainref{assm:subgaussian} gives
\[
\PP\left(
|\omega_k^{(b)}|>\frac{C}{4}\sqrt{\frac{\log p}{N_b}}
\right)
\le
2\exp\left(-\frac{C^2\log p}{32\sigma_0^2}\right)
\le
2p^{-2}.
\]
A union bound over $k$ and then over $b$ gives, with probability at least $1-2B/p$,
\begin{equation}
\|\omegab^{(b)}\|_\infty\le\lambda_b/4,
\qquad b=1,\ldots,B.
\label{eq:cor-score-event}
\end{equation}
Work on this event.  At $b=1$, $\phi^{(1)}=0$, and
Theorem~\mainref{thm:batch_b_init} applies.  Suppose inductively that the asserted bounds hold
through batch $b-1$.  Then Lemma~\ref{lem:batch.size} gives
\begin{align*}
\phi^{(b)}
&\le
\frac{9C^2|\mc A|\log p}{4\kappa^2N_b}
\sum_{j=1}^{b-1}\frac{n_j}{N_j}
\\
&\le
\frac{9C^2|\mc A|\log p}{4\kappa^2N_b}
\left\{1+\log\left(\frac{N_{b-1}}{n_1}\right)\right\}.
\end{align*}
Condition \maineqref{sample.size.requirement} implies
$M\phi^{(b)}\le\lambda_b/6$, and hence
$(3/2)M\phi^{(b)}\le\lambda_b/4$.  Together with
\eqref{eq:cor-score-event}, this verifies the tuning condition in
Theorem~\mainref{thm:batch_b_init}.  The induction is complete.  \qed

\section{Proofs of results in Section~\mainref{sec:inference}}

\subsection{Proof of Theorem~\mainref{thm:oracle.test}}
To begin with, recall the notation
\[
\omegab^*_{\overline{\mc M}}=\frac{1}{N_B}\sum_{b=1}^B n_b\nabla_{\overline{\mc M}}\widehat{\mc L}^{(b)}(\betab^*),
\]
and define
\[
T_0=N_B\left\|\Vb_0^{-1/2}\underbrace{\left(\Cb_0\Jb_0^{-1}\omegab^*_{\overline{\mc M}}-\hb_B\right)}_{=:\widehat{\mub}}\right\|_2^2.
\]
By the definitions of $\Kb_0$ and $\Vb_0$,
\[
\EE\left\{N_B\left\|\Vb_0^{-1/2}\Cb_0\Jb_0^{-1}\omegab^*_{\overline{\mc M}}\right\|_2^2\right\}=r.
\]
Markov's inequality therefore gives
\[
N_B^{1/2}\left\|\Vb_0^{-1/2}\Cb_0\Jb_0^{-1}\omegab^*_{\overline{\mc M}}\right\|_2=O_p(r^{1/2}).
\]
Moreover, since $\|\hb_B\|_2=O(\sqrt{r/N_B})$, Lemma~\ref{lem:V0-conditioning} implies
\[
N_B^{1/2}\|\Vb_0^{-1/2}\hb_B\|_2=O(r^{1/2}).
\]
Combining the preceding two displays with the triangle inequality yields
\begin{equation}
N_B^{1/2}\|\Vb_0^{-1/2}\widehat{\mub}\|_2=O_p(r^{1/2}),
\qquad
T_0=O_p(r).
\label{eq:T0-leading-order}
\end{equation}

The remainder of the proof proceeds in three steps. First, we establish a Gaussian approximation for $T_0$. Second, we show that both oracle test statistics equal $T_0$ up to $o_p(\sqrt{r})$. 
Finally, we transfer the noncentral chi-square approximation from $T_0$ to the two oracle statistics using Lemma~\ref{lem:chisq-anticoncentration}.

\medskip
\noindent\textsc{Step 1: Gaussian approximation for $T_0$.}
For the $i$th observation in batch $b$, define
\[
\xib_{b,i}=N_B^{-1/2}\Omega_0\nabla_{\overline{\mc M}}\ell(\betab^*;Y_i^{(b)},X_i^{(b)}),
\qquad
\Omega_0=\Vb_0^{-1/2}\Cb_0\Jb_0^{-1}.
\]
Since $\betab^*$ is an interior minimizer of $\mc L(\betab)=\EE\{\ell(\betab;Y,X)\}$, we have $\EE(\xib_{b,i})=\mathbf 0_r$. 
Furthermore,
\[
\sum_{b=1}^B\sum_{i=1}^{n_b}\EE(\xib_{b,i}\xib_{b,i}^\T)
=\Omega_0\Kb_0\Omega_0^\T
=\Vb_0^{-1/2}\Vb_0\Vb_0^{-1/2}
=\Ib_r.
\]
Assumption~\mainref{assm:lyapunov} gives
\[
r^{1/4}\sum_{b=1}^B\sum_{i=1}^{n_b}\EE\|\xib_{b,i}\|_2^3
=r^{1/4}N_B^{-1/2}\EE\left\|\Omega_0\nabla_{\overline{\mc M}}\ell(\betab^*;Y,X)\right\|_2^3
\to 0.
\]
Therefore, the multivariate Lyapunov-type bound of \citet{bentkus2005lyapunov} implies that, for some absolute constant $C_0>0$,
\[
\sup_{\mc B}\left|\PP\left(\sum_{b=1}^B\sum_{i=1}^{n_b}\xib_{b,i}\in\mc B\right)-\PP(\Zb_r\in\mc B)\right|
\leq C_0r^{1/4}\sum_{b=1}^B\sum_{i=1}^{n_b}\EE\|\xib_{b,i}\|_2^3
\longrightarrow 0,
\]
where $\Zb_r\sim N(\mathbf 0_r,\Ib_r)$ and the supremum is taken over all convex Borel sets in $\RR^r$.

Because
\[
\sum_{b=1}^B\sum_{i=1}^{n_b}\xib_{b,i}
=\sqrt{N_B}\Omega_0\omegab^*_{\overline{\mc M}}
=\sqrt{N_B}\Vb_0^{-1/2}\Cb_0\Jb_0^{-1}\omegab^*_{\overline{\mc M}},
\]
and
\[
T_0=\left\|\sum_{b=1}^B\sum_{i=1}^{n_b}\xib_{b,i}-\sqrt{N_B}\Vb_0^{-1/2}\hb_B\right\|_2^2,
\]
we may apply the preceding approximation to Euclidean balls centered at $\sqrt{N_B}\Vb_0^{-1/2}\hb_B$. It follows that
\[
\sup_{x\geq 0}\left|\PP(T_0\leq x)-\PP\left\{\left\|\Zb_r-\sqrt{N_B}\Vb_0^{-1/2}\hb_B\right\|_2^2\leq x\right\}\right|\longrightarrow 0.
\]
Remark that the squared norm of $\Zb_r-\sqrt{N_B}\Vb_0^{-1/2}\hb_B$ follows a noncentral chi-square distribution with $r$ degrees of freedom and noncentrality parameter
\[
\left\|\sqrt{N_B}\Vb_0^{-1/2}\hb_B\right\|_2^2
=N_B\hb_B^\T\Vb_0^{-1}\hb_B
=e_B.
\]
Consequently,
\#
\sup_{x\geq 0}\left|\PP(T_0\leq x)-\PP\{\chi^2(r,e_B)\leq x\}\right| \to 0.
\label{eq:T0-noncentral-approximation}
\#

\medskip
\noindent\textsc{Step 2: Approximation of the oracle statistics to $T_0$.}

Denote by ${\mc E}_{\rm ora}$ the intersection of the event displayed in the statement of Theorem~\mainref{thm:oracle.test} with the radius event $\{R_0\ge2\max(r^{\rm ora}_0,r^{\rm ora}_a)\}$; by hypothesis both have probability tending to one, so $\PP({\mc E}_{\rm ora})\to1$. 
Recall that $\gb^*=\nabla\widehat{\mc Q}_B(\betab^*)$ and $\gb^\dagger=\nabla\widehat{\mc Q}_B(\betab^\dagger)$. 
We first verify that, on ${\mc E}_{\rm ora}$, the three hypotheses of Lemma~\ref{lem:oracle_bahadur} are satisfied.

First, $\Jb_0=\Jb_{\overline{\mc M}\,\overline{\mc M}}$ is a principal submatrix of $\Jb$, so Assumption~\mainref{assm:hessian} gives $\lambda_{\min}(\Jb_0)\ge\lambda_{\min}(\Jb)\ge J_l$.

Second, we establish the localization of the oracle estimators.  
Recall that $\gb^*=\nabla\widehat{\mc Q}_B(\betab^*)$ and $\gb^\dagger=\nabla\widehat{\mc Q}_B(\betab^\dagger)$.  
Since
\[
\|\betab^\dagger-\betab^*\|_2
\leq C_{\max}\|\hb_B\|_2
\leq r_0^{\rm ora}
\leq R_0/2
\]
and $\betab^\dagger-\betab^*$ is supported on $\mc M$, the condition $l\geq(m+s)^{1/2}$ implies $\betab^\dagger-\betab^*\in\BB(R_0)\cap\CC(l)$.  
We may therefore apply Lemma~\ref{lem:omega_discrepancy}.  
Equations~\eqref{eq:gstar-score} and \eqref{eq:gdagger-block-score}, together with the definitions of $r_a^{\rm ora}$ and $r_0^{\rm ora}$, yield
\$
& \kappa^{-1}\|\gb^*_{\overline{\mc M}}\|_2
 \leq r_a^{\rm ora}\leq R_0/2,  \\
& \kappa^{-1}\left\{
\|\Zb^\T\gb^\dagger_{\mc M}\|_2+\|\gb^\dagger_{\mc S}\|_2
\right\}
+\|\betab^\dagger-\betab^*\|_2
\leq r_0^{\rm ora}\leq R_0/2.
\$
Thus the strict score-radius conditions in Lemma~\ref{lem:oracle_rate} are satisfied.  
Applying that lemma gives, for $e\in\{0,a\}$,
\[
\|\widehat\betab_e^{\rm ora}-\betab^*\|_2
\leq r_e^{\rm ora}\leq R_0/2.
\]
Each oracle error is supported on $\overline{\mc M}$, whose cardinality is
$m+s\leq l^2$.  Consequently,
\[
\|\widehat\betab_e^{\rm ora}-\betab^*\|_1
\leq(m+s)^{1/2}\|\widehat\betab_e^{\rm ora}-\betab^*\|_2
\leq l\|\widehat\betab_e^{\rm ora}-\betab^*\|_2,
\]
and hence
$\widehat\betab_e^{\rm ora}-\betab^*\in\BB(R_0)\cap\CC(l)$ for both
$e=0$ and $e=a$.

Finally, the event $\bigcap_{b=1}^B{\mc E}_{\rm rlh}^{(b)}(R_0,l,M)\cap{\mc H}(\eta_\infty)$ holds on ${\mc E}_{\rm ora}$.
Consequently, the Kiefer--Bahadur representations in Lemma~\ref{lem:oracle_bahadur} hold.
Based on the Kiefer--Bahadur representation, we derive the approximation of the oracle Wald and score statistics to $T_0$, respectively.
In the remainder of the proof, all $o_p$ and $O_p$ statements derived from that lemma are understood after intersection with ${\mc E}_{\rm ora}$; the complement contributes at most $\PP({\mc E}_{\rm ora}^c)=o(1)$ to each Kolmogorov distance, which is absorbed in Step 3.

\smallskip
\noindent\textsc{Step 2-1: Wald statistic.}

The full-oracle representation in Lemma~\ref{lem:oracle_bahadur} gives
\[
\widehat\betab^{\rm ora}_{a,\overline{\mc M}}-\betab^*_{\overline{\mc M}}
=-\Jb_0^{-1}\omegab^*_{\overline{\mc M}}+{\mc R}_a.
\]
Since $\Cb_0\betab^*_{\overline{\mc M}}=\tb+\hb_B$,
\[
\Cb\widehat\betab^{\rm ora}_{a,\mc M}-\tb
=\hb_B-\Cb_0\Jb_0^{-1}\omegab^*_{\overline{\mc M}}+\Cb_0{\mc R}_a
=-\widehat{\mub}+\Cb_0{\mc R}_a.
\]
The remainder bound \eqref{eq:full-bahadur-rem}, together with \maineqref{eq:inference-remainder-condition} for $e=a$, gives
\[
\sqrt{N_B}\|\mc R_a \|_2=o_p(1).
\]
Lemma~\ref{lem:V0-conditioning} also implies
\[
\|\Vb_0^{-1/2}\Cb_0\|_2\leq J_uK_l^{-1/2},
\]
and therefore
\[
\sqrt{N_B}\|\Vb_0^{-1/2}\Cb_0\mc R_a\|_2=o_p(1).
\]
Remark that
\[
T_W^{\rm ora}
=\left\{-\sqrt{N_B}\Vb_0^{-1/2}\widehat{\mub}+\sqrt{N_B}\Vb_0^{-1/2}\Cb_0\mc R_a\right\}^\T
\left\{-\sqrt{N_B}\Vb_0^{-1/2}\widehat{\mub}+\sqrt{N_B}\Vb_0^{-1/2}\Cb_0\mc R_a\right\}.
\]
Using
\[
\left|\|\ub+\vb\|_2^2-\|\ub\|_2^2\right|
\leq 2\|\ub\|_2\|\vb\|_2+\|\vb\|_2^2,
\]
together with \eqref{eq:T0-leading-order}, yields
\begin{equation}
T_W^{\rm ora}=T_0+o_p(r^{1/2}).
\label{eq:oracle-wald-T0-reduction}
\end{equation}

\smallskip
\noindent\textsc{Step 2-2: Score statistic.}
The reduced-oracle representation in Lemma~\ref{lem:oracle_bahadur} is
\#
\widehat\betab^{\rm ora}_{0,\overline{\mc M}}-\betab^*_{\overline{\mc M}}
=\Jb_0^{-1/2}(\Pb_0-\Ib)\Jb_0^{-1/2}\omegab^*_{\overline{\mc M}}
-\Jb_0^{-1}\Cb_0^\T\Psi_0^{-1}\hb_B+{\mc R}_0. \label{score.oracle.proof1}
\#
Multiplying by $\Cb_0$ in the previous display and using
\[
\Cb_0\Jb_0^{-1/2}(\Pb_0-\Ib)=\mathbf 0,
\qquad
\Psi_0=\Cb_0\Jb_0^{-1}\Cb_0^\T,
\]
together with the reduced-oracle constraint
\[
\Cb_0\left(\widehat\betab^{\rm ora}_{0,\overline{\mc M}}-\betab^*_{\overline{\mc M}}\right)=-\hb_B,
\]
gives $\Cb_0{\mc R}_0=\mathbf 0$.
Recall that $\Pb_0 = \Jb_0^{-1/2}\Cb_0^\T\Psi_0^{-1}\Cb_0\Jb_0^{-1/2}$.
Multiplying \eqref{score.oracle.proof1} by $\Jb_0$ and adding $\omegab^*_{\overline{\mc M}}$ yields
\#
\omegab^*_{\overline{\mc M}}
+\Jb_0\left(\widehat\betab^{\rm ora}_{0,\overline{\mc M}}-\betab^*_{\overline{\mc M}}\right)
=\Cb_0^\T\Psi_0^{-1}\widehat{\mub}+\Jb_0{\mc R}_0.
\label{eq:score-normal-equation-T0}
\#

Lemma~\ref{lem:decomposition} further gives the exact identity
\[
\nabla_{\overline{\mc M}}\widehat{\mc Q}_B(\widehat\betab^{\rm ora}_0)
=\omegab^*_{\overline{\mc M}}
+\Jb_0\left(\widehat\betab^{\rm ora}_{0,\overline{\mc M}}-\betab^*_{\overline{\mc M}}\right)
+\widetilde{\mc R}_0,
\]
where
\[
\widetilde{\mc R}_0
=\frac{1}{N_B}\sum_{b=1}^{B-1}n_bR_b(\betab^*)_{\overline{\mc M}}
+\nabla_{\overline{\mc M}}\widehat{\mc Q}_B(\widehat\betab^{\rm ora}_0)
-\nabla_{\overline{\mc M}}\widehat{\mc Q}_B(\betab^*)
-\Jb_0\left(\widehat\betab^{\rm ora}_{0,\overline{\mc M}}-\betab^*_{\overline{\mc M}}\right).
\]
Substituting \eqref{eq:score-normal-equation-T0} and multiplying by $\Omega_0$ give
\[
\Omega_0\nabla_{\overline{\mc M}}\widehat{\mc Q}_B(\widehat\betab^{\rm ora}_0)
=\Omega_0\Cb_0^\T\Psi_0^{-1}\widehat{\mub}
+\Omega_0\Jb_0{\mc R}_0
+\Omega_0\widetilde{\mc R}_0.
\]
Because
\[
\Omega_0\Cb_0^\T\Psi_0^{-1}=\Vb_0^{-1/2},
\qquad
\Omega_0\Jb_0{\mc R}_0=\Vb_0^{-1/2}\Cb_0{\mc R}_0=\mathbf 0,
\]
we obtain
\#
\Omega_0\nabla_{\overline{\mc M}}\widehat{\mc Q}_B(\widehat\betab^{\rm ora}_0)
=\Vb_0^{-1/2}\widehat{\mub}+\Omega_0\widetilde{\mc R}_0.
\label{eq:oracle-score-T0-linearization}
\#

The bound established in the proof of Lemma~\ref{lem:oracle_bahadur} gives
\[
\begin{aligned}
\|\widetilde{\mc R}_0\|_2\lesssim{}&
\frac{\sqrt{m+s}}{N_B}\Bigg[
\sum_{b=1}^{B-1}n_b\left\{
\|\widehat\betab^{(b)}-\betab^*\|_2^2
+\|\widehat\betab^{(b)}-\betab^*\|_2\|\widehat\betab^{\rm ora}_0-\betab^*\|_2
\right\} \\
&\hspace{35mm}
+n_B\|\widehat\betab^{\rm ora}_0-\betab^*\|_2^2
\Bigg]
+(m+s)\eta_\infty\|\widehat\betab^{\rm ora}_0-\betab^*\|_2.
\end{aligned}
\]
Condition \maineqref{eq:inference-remainder-condition} for $e=0$ implies
\[
\|\widetilde{\mc R}_0\|_2=o_p(N_B^{-1/2}).
\]
Since $\|\Omega_0\|_2\leq K_l^{-1/2}$ by Lemma~\ref{lem:V0-conditioning},
\[
\sqrt{N_B}\|\Omega_0\widetilde{\mc R}_0\|_2=o_p(1).
\]
Using the definition of the oracle score statistic and \eqref{eq:oracle-score-T0-linearization},
\[
T_S^{\rm ora}
=N_B\left\|\Vb_0^{-1/2}\widehat{\mub}+\Omega_0\widetilde{\mc R}_0\right\|_2^2.
\]
Applying the same elementary inequality used for the oracle Wald statistic, together with \eqref{eq:T0-leading-order}, yields
\#
T_S^{\rm ora}=T_0+o_p(r^{1/2}).
\label{eq:oracle-score-T0-reduction}
\#

\medskip
\noindent\textsc{Step 3: Anti-concentration transfer.}
Let $Y_B\sim\chi^2(r,e_B)$, and let $T$ denote either $T_W^{\rm ora}$ or $T_S^{\rm ora}$. By \eqref{eq:oracle-wald-T0-reduction} and \eqref{eq:oracle-score-T0-reduction},
\[
T=T_0+\Delta_B,
\qquad
\Delta_B=o_p(\sqrt{r}).
\]
Fix $\epsilon\in(0,1)$. 
Uniformly over $x\geq 0$,
\[
\PP(T_0\leq x-\epsilon\sqrt{r})-\PP(|\Delta_B|>\epsilon\sqrt{r})
\leq \PP(T\leq x)
\leq \PP(T_0\leq x+\epsilon\sqrt{r})+\PP(|\Delta_B|>\epsilon\sqrt{r}).
\]
The probability $\PP(|\Delta_B|>\epsilon\sqrt{r})$ is $o(1)$. 
Combining the preceding inequalities with \eqref{eq:T0-noncentral-approximation} gives
\[
\sup_{x\geq 0}\left|\PP(T\leq x)-\PP(Y_B\leq x)\right|
\leq
\sup_{x\geq 0}\left[\PP(Y_B\leq x+\epsilon\sqrt{r})-\PP(Y_B\leq x-\epsilon\sqrt{r})\right]+o(1).
\]
By Lemma~\ref{lem:chisq-anticoncentration}, the first term on the right-hand side is bounded by $C\sqrt{\epsilon}$ uniformly over $r$, $e_B$, and $x$. Therefore,
\[
\limsup_{N_B\to\infty}\sup_{x\geq 0}\left|\PP(T\leq x)-\PP(Y_B\leq x)\right|
\leq C\sqrt{\epsilon}.
\]
Letting $\epsilon\downarrow 0$ yields
\[
\sup_{x\geq 0}\left|\PP(T_W^{\rm ora}\leq x)-\PP\{\chi^2(r,e_B)\leq x\}\right|\longrightarrow 0
\]
and
\[
\sup_{x\geq 0}\left|\PP(T_S^{\rm ora}\leq x)-\PP\{\chi^2(r,e_B)\leq x\}\right|\longrightarrow 0.
\]
This completes the proof.
\qed

\subsection{Proof of Theorem~\mainref{thm:test}}

Fix $k\geq K_0+1$ and define 
\[ 
\tau_J=\max_{e\in\{0,a\}}\left\|\nabla^2_{\overline{\mc M}\,\overline{\mc M}}\widehat{\mc Q}_B(\widehat{\betab}^{\rm ora}_e)-\Jb_0\right\|_2, \qquad \tau_K=\left\|\widehat{\Kb}^{(B)}_{\overline{\mc M}\,\overline{\mc M}}-\Kb_0\right\|_2. 
\] 
Then condition \maineqref{eq:feasible-plugin} can be written as 
\[ 
\sqrt{r}\,(\tau_J+\tau_K)=o_p(1), 
\] 
which, in particular, implies $\tau_J+\tau_K=o_p(1)$. 
The proof proceeds in four parts. 
We first show that, on an event with probability tending to 1, the feasible estimators and selected supports coincide exactly with their oracle counterparts. 
Thus, the remaining discrepancies arise only from replacing the population standardization matrices by their feasible plug-in estimators. 
Next, we establish 
\[ T_W^{(k)}=T_W^{\rm ora}+o_p(\sqrt{r}) \qquad\mbox{and}\qquad T_S^{(k)}=T_S^{\rm ora}+o_p(\sqrt{r}), 
\] 
respectively. 
Finally, we show that any statistic that is $o_p(\sqrt{r})$-close to its oracle counterpart inherits the uniform noncentral chi-square approximation established in Theorem~\mainref{thm:oracle.test}. 

\medskip \noindent\textsc{Step 1: Exact oracle replacement.} Let $\mc E_{\rm rec}$ denote the event on which the conditions of Theorem~\mainref{thm:oracle_recovery} hold. 
By assumption, $\PP(\mc E_{\rm rec})\to1$. 
On $\mc E_{\rm rec}$, Theorem~\mainref{thm:oracle_recovery} gives, jointly for $e\in\{0,a\}$ and $k \geq K_0 + 1$,  
\# 
\widehat{\betab}^{({\rm LLA},k)}_e=\widehat{\betab}^{\rm ora}_e. \label{eq:feasible-estimator-equality} 
\# 
We first show that \eqref{eq:feasible-estimator-equality} propagates to the estimated supports. 
Suppose that $\mc S\neq\varnothing$ and let $j\in\mc S$. 
The beta-min condition in Theorem~\mainref{thm:oracle_recovery}, together with the oracle-radius bound in \maineqref{eq:oracle-recovery-score-radius}, gives 
\$ 
|\widehat\beta^{\rm ora}_{e,j}| \geq |\beta_j^*|-\|\widehat{\betab}^{\rm ora}_e-\betab^*\|_\infty \geq(\gamma_0+\gamma_1)\lambda-\|\widehat{\betab}^{\rm ora}_e-\betab^*\|_2 \geq(\gamma_0+\gamma_1)\lambda-\frac{\lambda}{5\bar\delta\kappa}>0. 
\$ 
Indeed, $\bar\delta>5/(4\kappa\gamma_0)$ implies $\gamma_0>1/(5\bar\delta\kappa)$. 
Since $\mc S=\mc M^c\cap\mc A$ and both oracle estimators are supported on $\overline{\mc M}=\mc M\cup\mc S$ by their definition \maineqref{def:constrained_oracle_estimator}--\maineqref{def:unconstrained_oracle_estimator}, the nonzero coordinates of $\widehat{\betab}^{\rm ora}_{e,\mc M^c}$ are exactly those indexed by $\mc S$. 
Combining this observation with \eqref{eq:feasible-estimator-equality} yields 
\# 
\widehat{\mc S}^{(k)}_e=\mc S, \qquad \widehat{\overline{\mc M}}{}^{(k)}_e=\overline{\mc M}, \qquad \widehat{\Cb}^{(k)}_e=(\Cb,\mathbf 0_{r\times s})=\Cb_0, ~\mbox{ for } e \in \{0,a\}. \label{eq:feasible-support-identification} 
\#  
When $\mc S=\varnothing$, the margin argument is vacuous, and \eqref{eq:feasible-support-identification} follows directly from the oracle support containment. 
Consequently, on the event $\mc E_{\rm rec}$, the definitions of the plug-in matrices give 
\# 
\widehat{\Jb}^{(k)}_e =\nabla^2_{\overline{\mc M}\,\overline{\mc M}}\widehat{\mc Q}_B(\widehat{\betab}^{\rm ora}_e), \qquad \widehat{\Kb}^{(k)}_e =\widehat{\Kb}^{(B)}_{\overline{\mc M}\,\overline{\mc M}}, \qquad \widehat{\Vb}^{(k)}_e =\Cb_0(\widehat{\Jb}^{(k)}_e)^{-1}\widehat{\Kb}^{(k)}_e(\widehat{\Jb}^{(k)}_e)^{-1}\Cb_0^\T. \label{eq:feasible-oracle-plugin-identification} 
\#
Thus, the condition in Theorem~\mainref{thm:test} yields 
\[ 
\|\widehat{\Jb}^{(k)}_e-\Jb_0\|_2\leq\tau_J, \qquad \|\widehat{\Kb}^{(k)}_e-\Kb_0\|_2\leq\tau_K. 
\] 
By Lemma~\ref{lem:V0-conditioning} and the assumption $\lambda_{\max}(\Cb\Cb^\T) = O(1)$, 
\begin{equation}
\lambda_{\min}(\Vb_0)\ge\frac{K_l}{J_u^2C_{\max}^2},
\qquad
\|\Vb_0\|_2\le\|\Cb_0\|_2^2J_l^{-2}K_u,
\qquad
\|\Cb_0\|_2=O(1),
\label{eq:feasible-V0-bounds}
\end{equation}
where the upper bound follows from $\|\Vb_0\|_2\le\|\Cb_0\|_2^2\|\Jb_0^{-1}\|_2^2\|\Kb_0\|_2$ together with $\lambda_{\min}(\Jb_0)\ge J_l$ from Assumption~\mainref{assm:hessian} and $\lambda_{\max}(\Kb_0)\le K_u$ from Assumption~\mainref{assm:score}, and $\|\Cb_0\|_2=\|\Cb\|_2$.

Define the working event
\begin{equation}
\mc E_N = \mc E_{\rm rec} \cap\Big\{ \sup_{e \in \{0,a\}}\|\widehat{\Jb}^{(k)}_e-\Jb_0\|_2\leq\tau_J\leq J_l/2\Big\} \cap\left\{C_V\|\Cb_0\|_2^2(\tau_J+\tau_K)\leq\lambda_{\min}(\Vb_0)/2\right\}, \label{eq:feasible-working-event}
\end{equation} 
where $C_V=6J_l^{-3}(K_u+J_l)$ is the constant in Lemma~\ref{lem:plugin_stability}. 
By \eqref{eq:feasible-V0-bounds}, the two deterministic thresholds $J_l$ and $\lambda_{\min}(\Vb_0)$ in \eqref{eq:feasible-working-event} are bounded away from zero. 
Since $\tau_J+\tau_K=o_p(1)$, $\PP(\mc E_{\rm rec})\to1$, and $\PP(\mc E_B^{\rm ora})\to1$, it follows that 
\# 
\PP(\mc E_N)\to1.  \label{eq:feasible-working-event-probability}
\# 
On $\mc E_N$, Lemma~\ref{lem:plugin_stability} ensures that $\widehat{\Jb}^{(k)}_e$ and $\widehat{\Vb}^{(k)}_e$ are nonsingular, and hence the feasible statistics are well defined.
Outside $\mc E_N$, where the matrices might not be singular, we set the statistics equal to zero. 
This convention does not affect any asymptotic conclusion because $\PP(\mc E_N^c)=o(1)$, which is absorbed into the approximation error later. 

\medskip
\noindent\textsc{Step 2: Approximation of the feasible statistics to the oracle statistics.}

We first establish several stability bounds for the plug-in matrices. On the working event ${\mc E}_N$, Weyl's inequality and the resolvent identity yield
\#
\|(\widehat{\Jb}^{(k)}_e)^{-1}\|_2\leq\frac{2}{J_l},~
\|\widehat{\Kb}^{(k)}_e\|_2\leq K_u+\tau_K=O_p(1),~
\|(\widehat{\Jb}^{(k)}_e)^{-1}-\Jb_0^{-1}\|_2\leq\frac{2\tau_J}{J_l^2}=o_p(r^{-1/2}). \label{eq:hat-J-inverse-bounds}
\#
All conditions of Lemma~\ref{lem:plugin_stability} are now verified; applying it for $e=0$ and $e=a$ gives
\begin{equation}
\|\widehat{\Vb}^{(k)}_e-\Vb_0\|_2\leq C_V\|\Cb_0\|_2^2(\tau_J+\tau_K)=o_p(r^{-1/2}).
\end{equation}
With the definition of $\mc E_N$ and the first bound in \eqref{eq:feasible-V0-bounds},
\[
\lambda_{\min}(\widehat{\Vb}^{(k)}_e)\geq\frac{K_l}{2J_u^2C_{\max}^2},
\qquad
\|(\widehat{\Vb}^{(k)}_e)^{-1}\|_2\leq\frac{2J_u^2C_{\max}^2}{K_l},
\]
and a second application of the resolvent identity gives
\$
\|(\widehat{\Vb}^{(k)}_e)^{-1}-\Vb_0^{-1}\|_2
&\leq\|(\widehat{\Vb}^{(k)}_e)^{-1}\|_2\|\widehat{\Vb}^{(k)}_e-\Vb_0\|_2\|\Vb_0^{-1}\|_2\\
&\leq2\left(\frac{J_u^2C_{\max}^2}{K_l}\right)^2C_V\|\Cb_0\|_2^2(\tau_J+\tau_K)=o_p(r^{-1/2}).
\$
Lemma~\ref{lem:plugin_stability} also gives the relative inverse bound
\begin{equation}
\left\|\Vb_0^{1/2}(\widehat{\Vb}^{(k)}_e)^{-1}\Vb_0^{1/2}-\Ib_r\right\|_2\leq2\frac{J_u^2C_{\max}^2}{K_l}C_V\|\Cb_0\|_2^2(\tau_J+\tau_K)=o_p(r^{-1/2}).
\label{eq:feasible-whitened-inverse-error}
\end{equation}
For the score statistic, define
\[
\widehat{\Ab}_0=(\widehat{\Jb}^{(k)}_0)^{-1}\Cb_0^\T(\widehat{\Vb}^{(k)}_0)^{-1}\Cb_0(\widehat{\Jb}^{(k)}_0)^{-1},
\qquad
\Ab_0=\Jb_0^{-1}\Cb_0^\T\Vb_0^{-1}\Cb_0\Jb_0^{-1}.
\]
The telescoping identity
\[
\begin{aligned}
\widehat{\Ab}_0-\Ab_0
={}&\{(\widehat{\Jb}^{(k)}_0)^{-1}-\Jb_0^{-1}\}\Cb_0^\T(\widehat{\Vb}^{(k)}_0)^{-1}\Cb_0(\widehat{\Jb}^{(k)}_0)^{-1}\\
&+\Jb_0^{-1}\Cb_0^\T\{(\widehat{\Vb}^{(k)}_0)^{-1}-\Vb_0^{-1}\}\Cb_0(\widehat{\Jb}^{(k)}_0)^{-1}\\
&+\Jb_0^{-1}\Cb_0^\T\Vb_0^{-1}\Cb_0\{(\widehat{\Jb}^{(k)}_0)^{-1}-\Jb_0^{-1}\},
\end{aligned}
\]
and the fact that every factor other than a displayed matrix difference is $O_p(1)$ by \eqref{eq:feasible-V0-bounds} and \eqref{eq:hat-J-inverse-bounds} give
\begin{equation}
\|\widehat{\Ab}_0-\Ab_0\|_2=O_p(\tau_J+\tau_K)=o_p(r^{-1/2}).
\label{eq:feasible-score-matrix-error}
\end{equation}

Finally, by \eqref{eq:feasible-V0-bounds} and $\|\hb_B\|_2=O(\sqrt{r/N_B})$,
\[
e_B=N_B\hb_B^\T\Vb_0^{-1}\hb_B=O(r).
\]
Writing $Y_B\sim\chi^2(r,e_B)$ and $d_N=\sup_{x\geq0}|\PP(T^{\rm ora}\leq x)-\PP(Y_B\leq x)|$ for $T^{\rm ora}\in\{T_W^{\rm ora},T_S^{\rm ora}\}$, Theorem~\mainref{thm:oracle.test} gives $d_N\to0$, so Markov's inequality yields, for every $\Lambda>0$,
\[
\PP(T^{\rm ora}>\Lambda r)\leq\PP(Y_B>\Lambda r)+d_N\leq\frac{\EE(Y_B)}{\Lambda r}+d_N \leq \frac{r+e_B}{\Lambda r}+d_N.
\]
Therefore, $T_W^{\rm ora}=O_p(r)$ and $T_S^{\rm ora}=O_p(r)$.

\medskip
\noindent\textsc{Step 2-1: The feasible Wald statistic.}
By \eqref{eq:feasible-estimator-equality} and \eqref{eq:feasible-support-identification}, we have the following:
\[
\Cb\widehat{\betab}^{({\rm LLA},k)}_{a,\mc M}-\tb=\Cb_0\widehat{\betab}^{({\rm LLA},k)}_{a,\overline{\mc M}}-\tb=\Cb_0\widehat{\betab}^{\rm ora}_{a,\overline{\mc M}}-\tb=:\ub.
\]
Hence
\[
T_W^{(k)}=N_B\|\{\widehat{\Vb}^{(k)}_a\}^{-1/2}\ub\|_2^2,
\qquad
T_W^{\rm ora}=N_B\|\Vb_0^{-1/2}\ub\|_2^2,
\]
and \eqref{eq:feasible-whitened-inverse-error} gives
\[
\frac{|T_W^{(k)}-T_W^{\rm ora}|}{\sqrt r}\leq\sqrt r\left\|\Vb_0^{1/2}(\widehat{\Vb}^{(k)}_a)^{-1}\Vb_0^{1/2}-\Ib_r\right\|_2\frac{T_W^{\rm ora}}{r}=o_p(1),
\]
using $\sqrt r(\tau_J+\tau_K)=o_p(1)$ and $T_W^{\rm ora}/r=O_p(1)$. 
Therefore,
\begin{equation}
T_W^{(k)}=T_W^{\rm ora}+o_p(\sqrt r).
\label{eq:feasible-wald-perturbation}
\end{equation}

\medskip
\noindent\textsc{Step 2-2: The feasible score statistic.}
On $\mc E_N$, \eqref{eq:feasible-estimator-equality} and \eqref{eq:feasible-support-identification} give the exact identity
\begin{equation}
\nabla_{\widehat{\overline{\mc M}}^{(k)}_0}\widehat{\mc Q}_B(\widehat{\betab}^{({\rm LLA},k)}_0)=\nabla_{\overline{\mc M}}\widehat{\mc Q}_B(\widehat{\betab}^{\rm ora}_0)= (\mathfrak w_0^{\rm ora})_{\overline{\mc M}}.
\label{eq:feasible-score-numerator}
\end{equation}
By the definition of the constrained estimator, there exists $\widehat{\vb}\in\RR^r$ satisfying
\$
(\mathfrak w_0^{\rm ora})_{\overline{\mc M}}+\Cb_0^\T\widehat{\vb}=\mathbf 0_{m+s},
\qquad
\Cb_0\widehat{\betab}^{\rm ora}_{0,\overline{\mc M}}=\tb.
\$
With $\Psi_0=\Cb_0\Jb_0^{-1}\Cb_0^\T$ and $\lambda_{\max}(\Jb_0)\leq J_u$,
\[
\lambda_{\min}(\Psi_0)\geq J_u^{-1}\lambda_{\min}(\Cb_0\Cb_0^\T)\geq\frac{1}{J_uC_{\max}^2}.
\]
Moreover, since $\Omega_0=\Vb_0^{-1/2}\Cb_0\Jb_0^{-1}$, the identity $(\mathfrak w_0^{\rm ora})_{\overline{\mc M}}=-\Cb_0^\T\widehat{\vb}$ gives $\Omega_0(\mathfrak w_0^{\rm ora})_{\overline{\mc M}}=-\Vb_0^{-1/2}\Psi_0\widehat{\vb}$ and hence
\[
\|\widehat{\vb}\|_2\leq\lambda_{\min}(\Psi_0)^{-1}\|\Vb_0\|_2^{1/2}\|\Omega_0(\mathfrak w_0^{\rm ora})_{\overline{\mc M}}\|_2.
\]
Combining this with $(\mathfrak w_0^{\rm ora})_{\overline{\mc M}}=-\Cb_0^\T\widehat{\vb}$, \eqref{eq:feasible-V0-bounds}, and $T_S^{\rm ora}=N_B\|\Omega_0(\mathfrak w_0^{\rm ora})_{\overline{\mc M}}\|_2^2 = O_p(r)$ yields
\begin{equation}
N_B\|(\mathfrak w_0^{\rm ora})_{\overline{\mc M}}\|_2^2\leq\|\Cb_0\|_2^2\lambda_{\min}(\Psi_0)^{-2}\|\Vb_0\|_2T_S^{\rm ora}=O_p(r).
\label{eq:feasible-score-order}
\end{equation}
By \eqref{eq:feasible-score-numerator}, the two score statistics are quadratic forms in $(\mathfrak w_0^{\rm ora})_{\overline{\mc M}}$,
\[
T_S^{(k)}=N_B(\mathfrak w_0^{\rm ora})_{\overline{\mc M}}^\T\widehat{\Ab}_0(\mathfrak w_0^{\rm ora})_{\overline{\mc M}},
\qquad
T_S^{\rm ora}=N_B(\mathfrak w_0^{\rm ora})_{\overline{\mc M}}^\T\Ab_0(\mathfrak w_0^{\rm ora})_{\overline{\mc M}},
\]
so \eqref{eq:feasible-score-matrix-error} and \eqref{eq:feasible-score-order} give
\[
\frac{|T_S^{(k)}-T_S^{\rm ora}|}{\sqrt r}\leq\frac{N_B\|(\mathfrak w_0^{\rm ora})_{\overline{\mc M}}\|_2^2}{r}\left\{\sqrt r\,\|\widehat{\Ab}_0-\Ab_0\|_2\right\}=o_p(1).
\]
Therefore,
\begin{equation}
T_S^{(k)}=T_S^{\rm ora}+o_p(\sqrt r).
\label{eq:feasible-score-perturbation}
\end{equation}

\medskip
\noindent\textsc{Step 3: Anti-concentration transfer.}
Let $T$ denote either feasible statistic, $T^{\rm ora}$ its matching oracle statistic, $Y_B\sim\chi^2(r,e_B)$, and $d_N=\sup_{x\geq0}|\PP(T^{\rm ora}\leq x)-\PP(Y_B\leq x)|$; by Theorem~\mainref{thm:oracle.test}, $d_N\to0$. 
Equations \eqref{eq:feasible-wald-perturbation} and \eqref{eq:feasible-score-perturbation}, with \eqref{eq:feasible-working-event-probability}, give, for every fixed $\epsilon\in(0,1]$,
\[
\PP\left\{\mc E_N\cap(|T-T^{\rm ora}|>\epsilon\sqrt r)\right\}+\PP(\mc E_N^c)=o(1),
\]
so that, uniformly over $x\geq0$,
\[
\PP(T^{\rm ora}\leq x-\epsilon\sqrt r)-o(1)\leq\PP(T\leq x)\leq\PP(T^{\rm ora}\leq x+\epsilon\sqrt r)+o(1).
\]
So the Kolmogorov bound $d_N$ applies at both shifted arguments, and
\[
\begin{aligned}
\sup_{x\geq0}\left|\PP(T\leq x)-\PP(Y_B\leq x)\right|
\leq{}&\sup_{x\geq0}\left[\PP(Y_B\leq x+\epsilon\sqrt r)-\PP(Y_B\leq x-\epsilon\sqrt r)\right]\\
&+d_N+o(1).
\end{aligned}
\]
By Lemma~\ref{lem:chisq-anticoncentration}, the first term on the right is bounded by $C\sqrt\epsilon$ uniformly over every positive integer $r$, every $e_B\geq0$, and every center $x$; hence
\[
\limsup_{N_B\to\infty}\sup_{x\geq0}\left|\PP(T\leq x)-\PP(Y_B\leq x)\right|\leq C\sqrt\epsilon.
\]
Letting $\epsilon\downarrow0$ and applying the conclusion to the Wald and score pairs separately gives
\[
\sup_{x\geq0}\left|\PP(T_W^{(k)}\leq x)-\PP\{\chi^2(r,e_B)\leq x\}\right|\longrightarrow0
\]
and
\[
\sup_{x\geq0}\left|\PP(T_S^{(k)}\leq x)-\PP\{\chi^2(r,e_B)\leq x\}\right|\longrightarrow0.
\]
This completes the proof.
\qed

\section{Proofs of the auxiliary lemmas} \label{sec:appendix_lemma_proofs}

% \subsection{Proof of Lemma~\ref{lem:batch.size}}

% For $j\ge2$,
% \[
% \frac{n_j}{N_j}
% =
% \frac{N_j-N_{j-1}}{N_j}
% \le
% \int_{N_{j-1}}^{N_j}\frac{dx}{x}
% =
% \log\left(\frac{N_j}{N_{j-1}}\right).
% \]
% Since $n_1/N_1=1$, summation gives
% \[
% \sum_{j=1}^b\frac{n_j}{N_j}
% \le
% 1+\log\left(\frac{N_b}{N_1}\right)
% =
% 1+\log\left(\frac{N_b}{n_1}\right).
% \]
% \qed

\subsection{Proof of Lemma~\ref{lem:decomposition}}

To begin with, we first derive the equation \eqref{eq:general-surrogate-gradient}. 
The result is immediate for $b=1$. 
Suppose that it holds for $b-1$. 
Differentiating the recursive definition of $\widehat{\mc Q}_b$ yields
\[
N_b\nabla\widehat{\mc Q}_b(\theta)=n_b\nabla\widehat{\mc L}^{(b)}(\theta)+\widehat{\Jb}^{(b-1)}(\theta-\widehat{\betab}^{(b-1)})+N_{b-1}\nabla\widehat{\mc Q}_{b-1}(\widehat{\betab}^{(b-1)}).
\]
By the induction hypothesis, evaluated at $\widehat{\betab}^{(b-1)}$,
\[
N_{b-1}\nabla\widehat{\mc Q}_{b-1}(\widehat{\betab}^{(b-1)})=\sum_{j=1}^{b-1}n_j\left\{\nabla\widehat{\mc L}^{(j)}(\widehat{\betab}^{(j)})+\nabla^2\widehat{\mc L}^{(j)}(\widehat{\betab}^{(j)})(\widehat{\betab}^{(b-1)}-\widehat{\betab}^{(j)})\right\}.
\]
Recall that
\[
R_j(\theta)=\nabla\widehat{\mc L}^{(j)}(\widehat{\betab}^{(j)})-\nabla\widehat{\mc L}^{(j)}(\theta)+\nabla^2\widehat{\mc L}^{(j)}(\widehat{\betab}^{(j)})(\theta-\widehat{\betab}^{(j)}).
\]
Combining the preceding displays and using
\[
\widehat{\Jb}^{(b-1)}=\sum_{j = 1}^{b-1} n_j\nabla^2\widehat{\mc L}^{(j)}(\widehat{\betab}^{(j)})
\]
establishes \eqref{eq:general-surrogate-gradient}. Adding and subtracting $\nabla\widehat{\mc L}^{(j)}(\theta)$ in each historical summand, for $1\leq j\leq b-1$, gives \eqref{eq:surrogate-gradient-remainder}. The corresponding identity at $\widehat{\betab}^{(b)}$ follows by direct substitution.

We next establish the remainder bounds. Let $a_j=\widehat{\betab}^{(j)}-\betab^*$. By the fundamental theorem of calculus,
\[
R_j(\betab^*)=\int_0^1\left\{\nabla^2\widehat{\mc L}^{(j)}(\betab^*+t a_j)-\nabla^2\widehat{\mc L}^{(j)}(\widehat{\betab}^{(j)})\right\}a_j\,{\rm d}t.
\]
Adding and subtracting $\nabla^2\widehat{\mc L}^{(j)}(\betab^*)$ inside the integrand and applying the RLH condition yield
\[
\begin{aligned}
\|R_j(\betab^*)\|_\infty
&\leq \int_0^1 tM\|a_j\|_2^2\,{\rm d}t+\int_0^1 M\|a_j\|_2^2\,{\rm d}t \\
&=\frac{3}{2}M\|a_j\|_2^2.
\end{aligned}
\]
Now let $d=\theta-\betab^*$. A straightforward algebra shows that $R_j(\theta)-R_j(\betab^*)$ equals
\[
\begin{aligned}
\left\{\nabla^2\widehat{\mc L}^{(j)}(\widehat{\betab}^{(j)})-\nabla^2\widehat{\mc L}^{(j)}(\betab^*)\right\}d -\int_0^1\left\{\nabla^2\widehat{\mc L}^{(j)}(\betab^*+td)-\nabla^2\widehat{\mc L}^{(j)}(\betab^*)\right\}d\,{\rm d}t.
\end{aligned}
\]
Therefore, the RLH condition implies
\[
\|R_j(\theta)-R_j(\betab^*)\|_\infty\leq M\|a_j\|_2\|d\|_2+\frac{1}{2}M\|d\|_2^2.
\]
Combining this inequality with the preceding bound and using
\[
\|a_j\|_2\|d\|_2\leq (1/2)\left(\|a_j\|_2^2+\|d\|_2^2\right)
\]
yields
\[
\|R_j(\theta)\|_\infty\leq 2M\|a_j\|_2^2+M\|d\|_2^2.
\]
\qed

\subsection{Proof of Lemma~\ref{lem:surrogate-curvature-inheritance}}
Let $\vb=\betab_1-\betab_2$, where $(\betab_1,\betab_2)$ satisfies
\begin{equation}
\betab_1-\betab_2\in\BB(R)\cap\CC(l),
\qquad
\betab_2-\betab^*\in\BB(R)\cap\CC(l).
\label{surrogate-curvature-proof1}
\end{equation}
The current-batch RSC event implies
\[
\left\langle \nabla\widehat{\mc L}^{(b)}(\betab_1)-\nabla\widehat{\mc L}^{(b)}(\betab_2),\vb \right\rangle \geq \kappa\|\vb\|_2^2.
\]
For each historical batch $j<b$ and any $t\in(0,1]$, the pair $(\widehat{\betab}^{(j)}+t\vb,\widehat{\betab}^{(j)})$
also satisfies the conditions in \eqref{surrogate-curvature-proof1}: the cone condition is preserved by homogeneity, $t\vb\in\BB(R)$, and the second argument satisfies the localization condition in \eqref{eq:surrogate-curvature-sufficient-conditions}. Therefore,
\[
\left\langle \nabla\widehat{\mc L}^{(j)}(\widehat{\betab}^{(j)}+t\vb)-\nabla\widehat{\mc L}^{(j)}(\widehat{\betab}^{(j)}),t\vb \right\rangle \geq \kappa t^2\|\vb\|_2^2.
\]
Dividing by $t^2$ and letting $t\downarrow0$, twice differentiability yields
\begin{equation}
\vb^\T\nabla^2\widehat{\mc L}^{(j)}(\widehat{\betab}^{(j)})\vb \geq \kappa\|\vb\|_2^2,\qquad j<b.
\label{eq:historical-hessian-inherits-rsc}
\end{equation}
Finally, subtracting \eqref{eq:general-surrogate-gradient} evaluated at $\betab_2$ from the same identity evaluated at $\betab_1$ gives
\[
\begin{split}
&\left\langle \nabla\widehat{\mc Q}_b(\betab_1)-\nabla\widehat{\mc Q}_b(\betab_2),\vb \right\rangle\\
&\quad=\frac{1}{N_b}\left[n_b\left\langle \nabla\widehat{\mc L}^{(b)}(\betab_1)-\nabla\widehat{\mc L}^{(b)}(\betab_2),\vb \right\rangle+\sum_{j=1}^{b-1}n_j\vb^\T\nabla^2\widehat{\mc L}^{(j)}(\widehat{\betab}^{(j)})\vb\right].
\end{split}
\]
Applying the current-batch RSC bound and \eqref{eq:historical-hessian-inherits-rsc}, together with $\sum_{j=1}^b n_j=N_b$, establishes \eqref{eq:inherited-surrogate-curvature}.
\qed

\subsection{Proof of Lemma~\ref{lem:intermediate.rate}}

We prove the result for the reduced-model estimator. The corresponding result for the full-model estimator follows by omitting the Lagrange multiplier and replacing the projected score on $\mc M$ with the ordinary score. 

Let $\Delta=\widetilde{\betab}_0-\theta$. 
By the KKT conditions, there exist a subgradient $\xib$ and a Lagrange multiplier $\vb\in\RR^r$ such that 
\#
\nabla\widehat{\mc Q}_B(\widetilde{\betab}_0)+\lambdab\circ\xib+((\Cb^\T\vb)^\T,\mathbf 0_{p-m}^\T)^\T=0. \label{intermediate.rate.KKT} 
\#
Because both $\widetilde{\betab}_0$ and $\theta$ are feasible, $\Cb\Delta_{\mc M}=0$, and hence $\langle ((\Cb^\T\vb)^\T,\mathbf 0_{p-m}^\T)^\T,\Delta\rangle=0$. 
Moreover, the convexity of $\widehat{\mc Q}_B$ implies 
\[ 
\left\langle \nabla\widehat{\mc Q}_B(\widetilde{\betab}_0)-\nabla\widehat{\mc Q}_B(\theta),\Delta\right\rangle\geq 0. 
\] 
Combining this inequality with \eqref{intermediate.rate.KKT} gives 
\#
0\leq \left\langle \lambdab\circ\xib,\theta-\widetilde{\betab}_0\right\rangle+\left\langle \nabla\widehat{\mc Q}_B(\theta),\theta-\widetilde{\betab}_0\right\rangle. \label{eq:weighted-center-basic} 
\#
Because $\theta_{\overline{\mc M}^{\,c}}=0$ and $\lambdab_{\mc M}=0$, the subgradient property yields 
\$
\left\langle \lambdab\circ\xib,\theta-\widetilde{\betab}_0\right\rangle\leq \lambda\|\Delta_{\mc E}\|_1-\rho\lambda\|\Delta_{{\mc E}^c}\|_1. 
\$
By the assumed score bound, 
\$
\left\langle \nabla\widehat{\mc Q}_B(\theta),\theta-\widetilde{\betab}_0\right\rangle\leq \|\nabla\widehat{\mc Q}_B(\theta)\|_\infty\|\Delta\|_1\leq \frac{\rho\lambda}{2}\|\Delta\|_1. 
\$
Substituting the preceding two inequalities into \eqref{eq:weighted-center-basic} gives 
\$
\|\Delta_{{\mc E}^c}\|_1\leq \left(1+\frac{2}{\rho}\right)\|\Delta_{\mc E}\|_1. 
\$
Consequently, 
\$
\|\Delta\|_1\leq \left(2+\frac{2}{\rho}\right)\|\Delta_{\mc E}\|_1\leq (2+2/\rho)d^{1/2}\|\Delta\|_2=l\|\Delta\|_2, 
\$
and hence $\Delta\in\CC(l)$. 

We next derive the $\ell_2$ bound. Define 
\$
\eta=\sup\{t\in[0,1]:t\Delta\in\BB(R_0)\},\qquad \Delta_\eta=\eta\Delta,\qquad \betab_\eta=\theta+\Delta_\eta. 
\$
Equivalently, $\eta=1$ when $\|\Delta\|_2\leq R_0$, whereas $\eta=R_0/\|\Delta\|_2\in(0,1)$ when $\|\Delta\|_2>R_0$. 
In particular, if $\|\Delta_\eta\|_2 < R_0$, then $\eta = 1$ and $\|\Delta\|_2 < R_0$. 

Now, since $\widehat{\mc Q}_B$ is a convex function, Lemma~F.2 in \cite{fan2018lamm} implies that 
\#
\left\langle \nabla\widehat{\mc Q}_B(\betab_\eta)-\nabla\widehat{\mc Q}_B(\theta),\Delta_\eta\right\rangle\leq \eta\left\langle \nabla\widehat{\mc Q}_B(\widetilde{\betab}_0)-\nabla\widehat{\mc Q}_B(\theta),\Delta\right\rangle. \label{intermediate_rate_proof1} 
\#
We first find a lower bound of the left-hand side of \eqref{intermediate_rate_proof1}.
Since $\CC(l)$ is a cone, $\Delta_\eta\in\BB(R_0)\cap\CC(l)$, while $\theta-\betab^*\in\BB(R_0)\cap\CC(l)$ by assumption. 
Therefore, the inherited-curvature inequality gives 
\#
\left\langle \nabla\widehat{\mc Q}_B(\betab_\eta)-\nabla\widehat{\mc Q}_B(\theta),\Delta_\eta\right\rangle\geq \kappa\|\Delta_\eta\|_2^2. \label{eq:weighted-center-lower}
\#

Turning to bounding the right-hand side, since $\Cb\Delta_{\mc M}=0$ and $\operatorname{ker}(\Cb)=\operatorname{range}(\Zb)$, we have $\Delta_{\mc M}=\Zb\Zb^\T\Delta_{\mc M}$, and 
$\|\Zb^\T\Delta_{\mc M}\|_2=\|\Delta_{\mc M}\|_2$,
where the latter identity follows from $\Zb^\T\Zb=\Ib_{m-r}$. 
Using \eqref{intermediate.rate.KKT} and the orthogonality of the multiplier term, we obtain 
\$
\left\langle \nabla\widehat{\mc Q}_B(\widetilde{\betab}_0)-\nabla\widehat{\mc Q}_B(\theta),\Delta\right\rangle=-\left\langle \lambdab\circ\xib,\Delta\right\rangle-\left\langle \nabla\widehat{\mc Q}_B(\theta),\Delta\right\rangle, 
\$ 
For $j\in{\mc E}^c$, we have $\theta_j=0$ and hence $\xib_j\Delta_j=|\Delta_j|$. Therefore, 
\$ 
-\sum_{j\in{\mc E}^c}\lambda_j\xib_j\Delta_j-\left\langle \nabla_{{\mc E}^c} \widehat{\mc Q}_B(\theta),\Delta_{{\mc E}^c}\right\rangle\leq -\sum_{j\in{\mc E}^c}\big\{\lambda_j-\big|\nabla_j \widehat{\mc Q}_B(\theta)\big|\big\}|\Delta_j|\leq -\frac{\rho\lambda}{2}\|\Delta_{{\mc E}^c}\|_1\leq 0. 
\$
Furthermore, $\lambda_j = 0$ for $j \in \mc M$.
Therefore, the remaining terms satisfy 
\$
& \left\langle \nabla\widehat{\mc Q}_B(\widetilde{\betab}_0)-\nabla\widehat{\mc Q}_B(\theta),\Delta\right\rangle \\
& \leq \|\lambdab_{\mc S}\|_2\|\Delta_{\mc S}\|_2 +\big\|\Zb^\T\nabla_{\mc M}\widehat{\mc Q}_B(\theta)\big\|_2\|\Delta_{\mc M}\|_2  +\|\nabla_{{\mc E}\setminus\mc M} \widehat{\mc Q}_B(\theta)\|_2\|\Delta_{{\mc E}\setminus\mc M}\|_2. 
\$
Indeed, 
\$ 
\left|\left\langle \nabla_{\mc M}\widehat{\mc Q}_B(\theta),\Delta_{\mc M}\right\rangle\right|=\left|\left\langle \Zb^\T\nabla_{\mc M}\widehat{\mc Q}_B(\theta),\Zb^\T\Delta_{\mc M}\right\rangle\right|\leq \|\Zb^\T\nabla_{\mc M}\widehat{\mc Q}_B(\theta)\|_2\|\Delta_{\mc M}\|_2. 
\$
Combining this bound with \eqref{intermediate_rate_proof1}, \eqref{eq:weighted-center-lower}, and the fact that $\Delta_\eta=\eta\Delta$ yields 
\$
\kappa\|\Delta_\eta\|_2^2\leq \left\{\|\lambdab_{\mc S}\|_2+\big\|\Zb^\T\nabla_{\mc M}\widehat{\mc Q}_B(\theta)\big\|_2+\|\nabla_{{\mc E}\setminus\mc M}\widehat{\mc Q}_B(\theta)\|_2\right\}\|\Delta_\eta\|_2. 
\$
Thus, \[ \|\Delta_\eta\|_2\leq \kappa^{-1}\left\{\|\lambdab_{\mc S}\|_2+\|\Zb^\T\nabla_{\mc M}\widehat{\mc Q}_B(\theta)\|_2+\|\nabla_{{\mc E}\setminus\mc M}\widehat{\mc Q}_B(\theta)\|_2\right\}. \] By assumption, the right-hand side is strictly smaller than $R_0$. If $\eta<1$, however, then $\|\Delta_\eta\|_2=R_0$, which is a contradiction. 
Hence $\eta=1$, so $\Delta_\eta=\Delta$, and the claimed reduced-model bound follows. 

For the full-model estimator, the same argument applies without the constraint and the associated Lagrange multiplier. 
The score contributions on $\mc M$ and ${\mc E}\setminus\mc M$ are then combined into $\|\nabla_{\mc E}\widehat{\mc Q}_B(\theta)\|_2$, yielding 
\$
\|\widetilde{\betab}_a-\theta\|_2\leq \kappa^{-1}\left\{\|\lambdab_{\mc S}\|_2+\|\nabla_{\mc E}\widehat{\mc Q}_B(\theta)\|_2\right\}. 
\$
\qed

\subsection{Proof of Lemma~\ref{lem:omega_discrepancy}}

By \eqref{eq:surrogate-gradient-remainder},
\[
\gb^*-\omegab^*
=
\frac1{N_B}\sum_{b=1}^{B-1}n_bR_b(\betab^*).
\]
Hence \eqref{eq:R-beta-star-bound} gives
\[
\|\gb^*-\omegab^*\|_\infty
\le
\frac{3M}{2N_B}
\sum_{b=1}^{B-1}n_b
\|\widehat{\betab}^{(b)}-\betab^*\|_2^2
=
\frac32M\phi^*,
\]
which proves \eqref{eq:gstar-score} using only the historical RLH events.

Let
\[
\overline{\Hb}_B(\betab)
=
\frac1{N_B}\sum_{b=1}^Bn_b
\nabla^2\widehat{\mc L}^{(b)}(\betab),
\qquad
 d=\betab^\dagger-\betab^*.
\]
Taylor's expansion gives
\[
\frac1{N_B}\sum_{b=1}^Bn_b
\left\{
\nabla\widehat{\mc L}^{(b)}(\betab^\dagger)
-
\nabla\widehat{\mc L}^{(b)}(\betab^*)
\right\}
=
\int_0^1\overline{\Hb}_B(\betab^*+td)d\,{\rm d}t.
\]
Add and subtract $\overline{\Hb}_B(\betab^*)$ and $\Jb$.  The all-batch RLH event gives
\[
\left\|
\int_0^1
\left\{
\overline{\Hb}_B(\betab^*+td)
-
\overline{\Hb}_B(\betab^*)
\right\}d\,{\rm d}t
\right\|_\infty
\le
\frac12M\|d\|_2^2.
\]
Since $d$ is supported on $\mc M$,
\[
\|\{\overline{\Hb}_B(\betab^*)-\Jb\}d\|_\infty
\le
\eta_\infty\|d\|_1
\le
\sqrt m\,\eta_\infty\|d\|_2.
\]
Finally,
$\|\Jb d\|_\infty\le\|\Jb d\|_2\le J_u\|d\|_2$.
Combining the preceding three bounds proves \eqref{eq:omega-transport}.

Again by \eqref{eq:surrogate-gradient-remainder},
\[
\gb^\dagger-\omegab^*
=
\frac1{N_B}\sum_{b=1}^Bn_b
\left\{
\nabla\widehat{\mc L}^{(b)}(\betab^\dagger)
-
\nabla\widehat{\mc L}^{(b)}(\betab^*)
\right\}+
\frac1{N_B}\sum_{b=1}^{B-1}n_bR_b(\betab^\dagger).
\]
Using \eqref{eq:R-general-center-bound} and
\eqref{eq:omega-transport},
\begin{align*}
\|\gb^\dagger-\omegab^*\|_\infty
&\le
2M\phi^*+M\|d\|_2^2
+
\frac12M\|d\|_2^2
+
(J_u+\sqrt m\eta_\infty)\|d\|_2
\\
&\le
2M\phi^*+
2Md_B^2+(J_u+\sqrt m\eta_\infty)d_B.
\end{align*}
This is \eqref{eq:gdagger-score}.  The block bound follows from
$\|\Zb\|_2\le1$ and
\[
\|\Zb^\T v_{\mc M}\|_2+\|v_{\mc S}\|_2
\le
\|v_{\mc M}\|_2+\|v_{\mc S}\|_2
\le
(\sqrt m+\sqrt s)\|v\|_\infty
\le
2\sqrt q\|v\|_\infty.
\]
\qed

\subsection{Proof of Lemma~\ref{lem:intermediate.rate2}}

We first report the oracle KKT relations.  For the full oracle, we have $\mathfrak w^{\rm ora}_{a,\overline{\mc M}}=0$, $\widehat{\betab}^{\rm ora}_{a,\overline{\mc M}^{\,c}}=0$.
For the reduced oracle there is $\nu\in\RR^r$ such that $\mathfrak w^{\rm ora}_{0,\mc S}=0$, $\mathfrak w^{\rm ora}_{0,\mc M}+\Cb^\T\nu=0$, $\widehat{\betab}^{\rm ora}_{0,\overline{\mc M}^{\,c}}=0$.
Thus, we have $\Zb^\T\mathfrak w^{\rm ora}_{0,\mc M}=0$.

Let $\Delta=\widetilde{\betab}_e-\widehat{\betab}^{\rm ora}_e$.
If $\Delta=0$, both the cone conclusion and
\eqref{eq:oracle-centered-weighted-bound} are immediate.  Hence suppose
$\Delta\ne0$.
The weighted-program KKT conditions provide a subgradient $\xib$ with
$\|\xib\|_\infty\le1$.  In the reduced problem, the multiplier term has zero inner
product with $\Delta$ because both the weighted-program solution and the reduced oracle
are feasible.
The weighted-program KKT equation, convexity, and feasibility in the reduced case give
\[
0
\le
-\langle\lambdab\circ\xib,\Delta\rangle
-
\langle\mathfrak w_e^{\rm ora},\Delta\rangle.
\]
For the full oracle, the oracle-gradient term on $\overline{\mc M}$ vanishes.  For the
reduced oracle, its term on $\mc S$ vanishes, while
$\left\langle
\mathfrak w^{\rm ora}_{0,\mc M},
\Delta_{\mc M}
\right\rangle
=
-\langle\nu,\Cb\Delta_{\mc M}\rangle
=0$.
because both the weighted-program solution and the reduced oracle are feasible.  On
$\overline{\mc M}^{\,c}$, the oracle is zero, and hence the subgradient relation gives
$\xi_j\Delta_j=|\Delta_j|$.  Using
$\|\lambdab\|_\infty\le\lambda$ on $\mc S$, the oracle-score bound on
${\mc E}\setminus\overline{\mc M}$, and both the weight and oracle-score bounds on
${\mc E}^c$, the preceding KKT inequality yields
\[
0
\le
\lambda\|\Delta_{\mc S}\|_1
+
\frac12\rho\lambda
\|\Delta_{{\mc E}\setminus\overline{\mc M}}\|_1
-
\frac12\rho\lambda
\|\Delta_{{\mc E}^c}\|_1.
\]
Since $\mc S\subset\mc E$ and
${\mc E}\setminus\overline{\mc M}\subset\mc E$, when
${\mc E}^c\ne\varnothing$ the last display implies
\[
\|\Delta_{{\mc E}^c}\|_1
\le
\left(1+\frac2\rho\right)\|\Delta_{\mc E}\|_1,
\]
and the same inequality is immediate when ${\mc E}^c=\varnothing$.  Consequently,
\[
\|\Delta\|_1
\le
\left(2+\frac2\rho\right)\|\Delta_{\mc E}\|_1
\le
\left(2+\frac2\rho\right)d^{1/2}\|\Delta\|_2
=
l\|\Delta\|_2.
\]
Let
\[
\eta
=
\sup\{t\in[0,1]:t\Delta\in\BB(R_0)\},
\qquad
\Delta_\eta=\eta\Delta,
\qquad
\betab_\eta=\widehat{\betab}^{\rm ora}_e+\Delta_\eta.
\]
By construction and the cone bound just proved, $\Delta_\eta\in\BB(R_0)\cap\CC(l)$.  
By Lemma~\ref{lem:surrogate-curvature-inheritance}, the inherited curvature bound applies directly to the pair \((\betab_\eta,\widehat{\betab}^{\rm ora}_e)\), yielding
\[
\left\langle
\nabla\widehat{\mc Q}_B(\betab_\eta)
-
\nabla\widehat{\mc Q}_B(\widehat{\betab}^{\rm ora}_e),
\Delta_\eta
\right\rangle
\ge
\kappa\|\Delta_\eta\|_2^2.
\]
Convexity along the segment from $\widehat{\betab}^{\rm ora}_e$ to
$\widetilde{\betab}_e$ bounds the left-hand side by
\[
\eta
\left\langle
\nabla\widehat{\mc Q}_B(\widetilde{\betab}_e)
-
\nabla\widehat{\mc Q}_B(\widehat{\betab}^{\rm ora}_e),
\Delta
\right\rangle.
\]
The weighted-program KKT equation and the same oracle KKT cancellations used above show
that the last display is at most
$-\langle\lambdab\circ\xib,\Delta_\eta\rangle
-
\langle\mathfrak w_e^{\rm ora},\Delta_\eta\rangle$.
The terms on ${\mc E}^c$ are nonpositive.  On the remaining coordinates, the only
possibly positive terms are the penalty on $\mc S$ and the oracle score on
${\mc E}\setminus\overline{\mc M}$.  Hence
\[
\kappa\|\Delta_\eta\|_2^2
\le
\left\{
\|\lambdab_{\mc S}\|_2
+
\|\mathfrak w^{\rm ora}_{e,{\mc E}\setminus\overline{\mc M}}\|_2
\right\}\|\Delta_\eta\|_2.
\]
Because $\Delta\ne0$ and $R_0>0$, the definition of $\eta$ gives
$\Delta_\eta\ne0$.  Cancellation and the strict localization assumption therefore give
\[
\|\Delta_\eta\|_2
\le
\kappa^{-1}
\left\{
\|\lambdab_{\mc S}\|_2
+
\|\mathfrak w^{\rm ora}_{e,{\mc E}\setminus\overline{\mc M}}\|_2
\right\}
<R_0.
\]
If $\eta<1$, the definition of $\eta$ and $\Delta\ne0$ would imply
$\|\Delta_\eta\|_2=R_0$, contradicting the last display.  Hence $\eta=1$, and the
preceding bound becomes \eqref{eq:oracle-centered-weighted-bound}.  \qed

\subsection{Proof of Lemma~\ref{lem:oracle_rate}}

Let $\Delta_a=\widehat{\betab}^{\rm ora}_a-\betab^*$.
If $\Delta_a=0$, \eqref{eq:full-oracle-rate-first} is immediate.  Otherwise, put
\[
\eta
=
\sup\{t\in[0,1]:t\Delta_a\in\BB(R_0)\},
\qquad
\Delta_{a,\eta}=\eta\Delta_a.
\]
Because $\Delta_a$ is supported on $\overline{\mc M}$,
\[
\|\Delta_{a,\eta}\|_1
\le
(s+m)^{1/2}\|\Delta_{a,\eta}\|_2
\le
l\|\Delta_{a,\eta}\|_2.
\]
The current-batch RSC inequality, the historical Hessian inequalities, and
\eqref{eq:general-surrogate-gradient} therefore give
\[
\left\langle
\nabla\widehat{\mc Q}_B(\betab^*+\Delta_{a,\eta})
-
\nabla\widehat{\mc Q}_B(\betab^*),
\Delta_{a,\eta}
\right\rangle
\ge
\kappa\|\Delta_{a,\eta}\|_2^2.
\]
Convexity along the segment from $\betab^*$ to
$\widehat{\betab}^{\rm ora}_a$ bounds the left-hand side by
\[
\eta
\left\langle
\nabla\widehat{\mc Q}_B(\widehat{\betab}^{\rm ora}_a)
-
\nabla\widehat{\mc Q}_B(\betab^*),
\Delta_a
\right\rangle.
\]
The full-oracle KKT equation gives
$\nabla_{\overline{\mc M}}\widehat{\mc Q}_B
(\widehat{\betab}^{\rm ora}_a)=0$.  Consequently,
\begin{align*}
\kappa\|\Delta_{a,\eta}\|_2^2
&\le
-\left\langle
\gb^*_{\overline{\mc M}},
\Delta_{a,\eta,\overline{\mc M}}
\right\rangle
\\
&\le
\|\gb^*_{\overline{\mc M}}\|_2
\|\Delta_{a,\eta}\|_2.
\end{align*}
It follows that
\[
\|\Delta_{a,\eta}\|_2
\le
\kappa^{-1}\|\gb^*_{\overline{\mc M}}\|_2
<R_0.
\]
The localization point therefore cannot lie on the boundary of $\BB(R_0)$, so
$\eta=1$.  This proves \eqref{eq:full-oracle-rate-first}.  On the historical RLH event,
\eqref{eq:gstar-score} gives
\[
\|\gb^*_{\overline{\mc M}}\|_2
\le
\|\omegab^*_{\overline{\mc M}}\|_2
+
\frac32M(s+m)^{1/2}\phi^*
\le
\|\omegab^*_{\overline{\mc M}}\|_2
+
2M(s+m)^{1/2}\phi^*,
\]
which proves \eqref{eq:full-oracle-rate}.

For the reduced oracle, let $\Delta_0=\widehat{\betab}^{\rm ora}_0-\betab^\dagger$.
Both vectors are supported on $\overline{\mc M}$ and satisfy the same constraint, so
\[
\Cb\Delta_{0,\mc M}=0,
\qquad
\Delta_{0,\mc M}=\Zb u,
\qquad
\|u\|_2=\|\Delta_{0,\mc M}\|_2.
\]
If $\Delta_0=0$, the bound about $\betab^\dagger$ is immediate.  Otherwise, put
\[
\eta
=
\sup\{t\in[0,1]:t\Delta_0\in\BB(R_0)\},
\qquad
\Delta_{0,\eta}=\eta\Delta_0.
\]
Because $\betab^\dagger-\betab^*\in\BB(R_0)\cap\CC(l)$, the same aggregate RSC
argument and convexity along the feasible line segment give
\[
\kappa\|\Delta_{0,\eta}\|_2^2
\le
\eta
\left\langle
\nabla\widehat{\mc Q}_B(\widehat{\betab}^{\rm ora}_0)
-
\nabla\widehat{\mc Q}_B(\betab^\dagger),
\Delta_0
\right\rangle.
\]
For some multiplier $\nu\in\RR^r$, the reduced-oracle KKT equations are
\[
\nabla_{\mc S}\widehat{\mc Q}_B
(\widehat{\betab}^{\rm ora}_0)=0,
\qquad
\nabla_{\mc M}\widehat{\mc Q}_B
(\widehat{\betab}^{\rm ora}_0)+\Cb^\T\nu=0.
\]
Since $\Cb\Delta_{0,\mc M}=0$, the multiplier term vanishes.  Therefore,
\begin{align*}
\kappa\|\Delta_{0,\eta}\|_2^2
&\le
-\left\langle
\gb^\dagger_{\mc M},
\Delta_{0,\eta,\mc M}
\right\rangle
-
\left\langle
\gb^\dagger_{\mc S},
\Delta_{0,\eta,\mc S}
\right\rangle
\\
&\le
\left\{
\|\Zb^\T\gb^\dagger_{\mc M}\|_2
+
\|\gb^\dagger_{\mc S}\|_2
\right\}
\|\Delta_{0,\eta}\|_2.
\end{align*}
It follows that
\[
\|\Delta_{0,\eta}\|_2
\le
\kappa^{-1}
\left\{
\|\Zb^\T\gb^\dagger_{\mc M}\|_2
+
\|\gb^\dagger_{\mc S}\|_2
\right\}
<R_0.
\]
Thus $\eta=1$.  Adding
$\|\betab^\dagger-\betab^*\|_2\le d_B$ proves
\eqref{eq:reduced-oracle-rate-first}.  Under the additional event,
\eqref{eq:gdagger-block-score} proves
\eqref{eq:reduced-oracle-rate}.  \qed

\subsection{Proof of Lemma~\ref{lem:oracle_bahadur}}

Let
\[
\scoreb_B
=
\frac1{N_B}\sum_{b=1}^Bn_b
\nabla_{\overline{\mc M}}\widehat{\mc L}^{(b)}(\betab^*),
\qquad
\Delta_e
=
\widehat{\betab}^{\rm ora}_{e,\overline{\mc M}}
-
\betab^*_{\overline{\mc M}}.
\]
By Lemma~\ref{lem:decomposition},
\[
\nabla_{\overline{\mc M}}\widehat{\mc Q}_B(\betab^*)
=
\scoreb_B
+
\frac1{N_B}\sum_{b=1}^{B-1}n_b
R_b(\betab^*)_{\overline{\mc M}}.
\]
Define
\begin{align*}
\remb_e
=
\frac1{N_B}\sum_{b=1}^{B-1}n_b
R_b(\betab^*)_{\overline{\mc M}}
+
\nabla_{\overline{\mc M}}\widehat{\mc Q}_B
(\widehat{\betab}^{\rm ora}_e)
-
\nabla_{\overline{\mc M}}\widehat{\mc Q}_B(\betab^*)
-
\Jb_0\Delta_e.
\end{align*}

For the reduced oracle, the KKT conditions and the constraint give
\[
\Jb_0\Delta_0
=
-\scoreb_B-\remb_0-\Cb_0^\T\nu,
\qquad
\Cb_0\Delta_0=-\hb_B.
\]
Solving this linear system gives
\[
\Delta_0
=
\Jb_0^{-1/2}(\Pb_0-\Ib)\Jb_0^{-1/2}\scoreb_B
-
\Jb_0^{-1}\Cb_0^\T\Psi_0^{-1}\hb_B
+
\Jb_0^{-1/2}(\Pb_0-\Ib)\Jb_0^{-1/2}\remb_0.
\]
This is \eqref{eq:reduced-oracle-bahadur} with the last term defined as $\mc R_0$.
For the full oracle, the KKT equation is
$\Jb_0\Delta_a=-\scoreb_B-\remb_a$, which gives
\eqref{eq:full-oracle-bahadur}.

It remains to bound $\remb_e$.  By \eqref{eq:R-beta-star-bound},
\begin{equation}
\left\|
\frac1{N_B}\sum_{b=1}^{B-1}n_b
R_b(\betab^*)_{\overline{\mc M}}
\right\|_2
\le
\frac{3M(s+m)^{1/2}}{2N_B}
\sum_{b=1}^{B-1}n_b
\|\widehat{\betab}^{(b)}-\betab^*\|_2^2.
\label{eq:bahadur-historical-R}
\end{equation}
The centered gradient difference in the definition of $\remb_e$, $\nabla_{\overline{\mc M}}\widehat{\mc Q}_B
(\widehat{\betab}^{\rm ora}_e)
-
\nabla_{\overline{\mc M}}\widehat{\mc Q}_B(\betab^*) - \Jb_0\Delta_e$, equals
\begin{align*}
&\frac{n_B}{N_B}
\int_0^1
\left\{
\nabla^2_{\overline{\mc M}\overline{\mc M}}
\widehat{\mc L}^{(B)}(\betab^*+t\widetilde\Delta_e)
-
\nabla^2_{\overline{\mc M}\overline{\mc M}}
\widehat{\mc L}^{(B)}(\betab^*)
\right\}\Delta_e\,{\rm d}t
\\
&+
\frac1{N_B}\sum_{b=1}^{B-1}n_b
\left\{
\nabla^2_{\overline{\mc M}\overline{\mc M}}
\widehat{\mc L}^{(b)}(\widehat{\betab}^{(b)})
-
\nabla^2_{\overline{\mc M}\overline{\mc M}}
\widehat{\mc L}^{(b)}(\betab^*)
\right\}\Delta_e
\\
&+
\left\{
\frac1{N_B}\sum_{b=1}^Bn_b
\nabla^2_{\overline{\mc M}\overline{\mc M}}
\widehat{\mc L}^{(b)}(\betab^*)
-
\Jb_0
\right\}\Delta_e,
\end{align*}
where \(\widetilde{\Delta}_e\in\RR^p\) denotes the zero-padded extension of \(\Delta_e\), with \((\widetilde{\Delta}_e)_{\overline{\mc M}}=\Delta_e\) and \((\widetilde{\Delta}_e)_{\overline{\mc M}^{\,c}}=\mathbf 0\).
The RLH event bounds the first two lines in Euclidean norm by
\[
\frac{M(s+m)^{1/2}}{2N_B}n_B\|\Delta_e\|_2^2
+
\frac{M(s+m)^{1/2}}{N_B}
\sum_{b=1}^{B-1}n_b
\|\widehat{\betab}^{(b)}-\betab^*\|_2\|\Delta_e\|_2.
\]
For the last line, the event $\mc H(\eta_\infty)$ implies
\$
\left\|
\left\{
\frac1{N_B}\sum_{b=1}^Bn_b
\nabla^2_{\overline{\mc M}\overline{\mc M}}
\widehat{\mc L}^{(b)}(\betab^*)
-
\Jb_0
\right\}\Delta_e
\right\|_2
\le
q\eta_\infty\|\Delta_e\|_2.
\$
Indeed, each coordinate is bounded by
$\eta_\infty\|\Delta_e\|_1\le\sqrt q\eta_\infty\|\Delta_e\|_2$, and there are $q$
coordinates.  Combining these estimates with \eqref{eq:bahadur-historical-R}, and using
bounded operator norms of
$\Jb_0^{-1}$ and
$\Jb_0^{-1/2}(\Pb_0-\Ib)\Jb_0^{-1/2}$, proves
\eqref{eq:reduced-bahadur-rem} and \eqref{eq:full-bahadur-rem}.  \qed

\subsection{Proof of Lemma~\ref{lem:chisq-anticoncentration}}
Since the map $t \to \PP\{\chi^2(r,e)\leq t\}$
is nondecreasing, the difference on the left-hand side of \eqref{eq:chisq-anticoncentration} is equal to
\[
\PP\left\{
x-\epsilon\sqrt{r}
<
\chi^2(r,e)
\leq
x+\epsilon\sqrt{r}
\right\}.
\]
Therefore, it suffices to bound, uniformly over $r\geq 1$, $e\geq 0$, and $x\in\RR$, the probability that $\chi^2(r,e)$ lies in an interval of length $2\epsilon\sqrt{r}$. 
We treat the cases $r\geq 2$ and $r=1$ separately as in \citet{shi2019linear}.

First, suppose that $r\geq 2$. 
Let $f_\nu$ denote the density of a central chi-square random variable with $\nu$ degrees of freedom:
\[
f_\nu(y)=\frac{1}{2^{\nu/2}\Gamma(\nu/2)}y^{\nu/2-1}e^{-y/2},\qquad y>0,
\]
where $\Gamma(\cdot)$ denotes the gamma function.
We first establish that
\begin{equation}
\sup_{y\geq 0}f_\nu(y)
\leq
\nu^{-1/2},
\qquad
\nu\in\{2,3,4,\ldots\}.
\label{eq:central-chisq-density-bound}
\end{equation}
When $\nu=2$, the density $f_2(y)=\tfrac{1}{2}e^{-y/2}$ is decreasing, and hence
\[
\sup_{y\geq 0}f_2(y)
=
f_2(0)
=
\frac{1}{2}
\leq
2^{-1/2}.
\]
For $\nu\geq 3$, the density attains its maximum at $y=\nu-2$. Therefore,
\[
\sup_{y\geq 0}f_\nu(y)
=
f_\nu(\nu-2)
=
\frac{(\nu/2-1)^{\nu/2-1}}{2\Gamma(\nu/2)}
e^{-\nu/2+1}.
\]
By Stirling's lower bound,
\[
\Gamma(\nu/2)
=
\Gamma\bigl((\nu/2-1)+1\bigr)
\geq
\sqrt{2\pi}
(\nu/2-1)^{\nu/2-1/2}
e^{-\nu/2+1}.
\]
Consequently,
\[
\begin{aligned}
f_\nu(\nu-2)
&\leq
\frac{
(\nu/2-1)^{\nu/2-1}e^{-\nu/2+1}
}{
2\sqrt{2\pi}
(\nu/2-1)^{\nu/2-1/2}
e^{-\nu/2+1}
},
\end{aligned}
\]
which is equal to $1/\{2\sqrt{\pi(\nu-2)}\}$.
For every integer $\nu\geq 3$, we have $\nu\leq 3(\nu-2)$, and hence
\[
(\nu-2)^{-1/2}
\leq
\sqrt{3}\,\nu^{-1/2}.
\]
It follows that
\[
f_\nu(\nu-2)
\leq
\frac{\sqrt{3}}{2\sqrt{\pi}}\nu^{-1/2}
\leq
\nu^{-1/2},
\]
due to $\sqrt{3}/(2\sqrt{\pi})<1$. This proves
\eqref{eq:central-chisq-density-bound}.

The density of a noncentral chi-square random variable admits the
following representation
\[
f_{r,e}(y)
=
\sum_{k=0}^{\infty}
e^{-e/2}\frac{(e/2)^k}{k!}
f_{r+2k}(y),
\qquad y>0.
\]
For $r\geq 2$, each index satisfies $r+2k\geq 2$. Therefore, by
\eqref{eq:central-chisq-density-bound},
\[
\sup_{y\geq 0}f_{r+2k}(y)
\leq
(r+2k)^{-1/2}
\leq
r^{-1/2}.
\]
Because the weights are nonnegative and sum to one, we obtain
\$
\sup_{e\geq 0}\sup_{y\geq 0}f_{r,e}(y)
\leq
r^{-1/2},
\qquad r\geq 2.
\$
It follows that, for every $x\in\RR$ and $e\geq 0$, the probability that
$\chi^2(r,e)$ lies in an interval of length $2\epsilon\sqrt{r}$ is bounded by
\[
2\epsilon\sqrt{r}
\sup_{y\geq 0}f_{r,e}(y)
\leq
2\epsilon.
\]

It remains to consider $r=1$, in which case the interval has length
$2\epsilon$. Using the representation
\[
\chi^2(1,e)
\overset{d}{=}
(Z+\sqrt{e})^2,
\qquad
Z\sim N(0,1),
\]
consider an arbitrary interval $I\subseteq[0,\infty)$ of length at most
$2\epsilon$. Write $I=(a,b]$, where $0\leq a\leq b$ and
$b-a\leq 2\epsilon$. Its preimage under the map $u\mapsto u^2$ is contained in $[-\sqrt{b},-\sqrt{a})
\cup
(\sqrt{a},\sqrt{b}]$,
whose total Lebesgue measure is
$2(\sqrt{b}-\sqrt{a})
\leq
2\sqrt{b-a}
\leq
2\sqrt{2\epsilon}.$
Since the density of $Z+\sqrt{e}$ is uniformly bounded by
$(2\pi)^{-1/2}$, it follows that
\[
\PP\{(Z+\sqrt{e})^2\in I\}
\leq
\frac{1}{\sqrt{2\pi}}
\,2\sqrt{2\epsilon}
=
2\sqrt{\frac{\epsilon}{\pi}},
\]
uniformly over $e\geq 0$ and over all intervals $I\subseteq[0,\infty)$ of
length at most $2\epsilon$. The same bound applies to an arbitrary interval
in $\RR$, after intersecting it with the support $[0,\infty)$.

Combining the two cases yields
\[
\begin{aligned}
&\sup_{r\geq 1,\ e\geq 0,\ x\in\RR}
\left[
\PP\{\chi^2(r,e)\leq x+\epsilon\sqrt{r}\}
-
\PP\{\chi^2(r,e)\leq x-\epsilon\sqrt{r}\}
\right]
\leq
\max\left\{
2\epsilon,\,
2\sqrt{\frac{\epsilon}{\pi}}
\right\},
\end{aligned}
\]
which is bounded by $2\sqrt{\epsilon}$. This proves
\eqref{eq:chisq-anticoncentration}.
\qed

\subsection{Proof of Lemma~\ref{lem:V0-conditioning}}

Fix a unit vector \(u\in\RR^r\).  Since \(\Cb_0\Cb_0^\T=\Cb\Cb^\T\) and
Assumption~\mainref{assm:local_alternative} gives
\(\lambda_{\min}(\Cb\Cb^\T)\ge C_{\max}^{-2}\),
\[
u^\T\Vb_0u
=
(\Jb_0^{-1}\Cb_0^\T u)^\T
\Kb_0
(\Jb_0^{-1}\Cb_0^\T u)
\ge
K_l\|\Jb_0^{-1}\Cb_0^\T u\|_2^2
\ge
\frac{K_l}{J_u^2}\|\Cb_0^\T u\|_2^2
=
\frac{K_l}{J_u^2}u^\T\Cb\Cb^\T u
\ge
\frac{K_l}{J_u^2C_{\max}^2}.
\]
Here the first inequality uses \(\lambda_{\min}(\Kb_0)\ge K_l\) from
Assumptions~\mainref{assm:hessian} and \mainref{assm:score}, and the second uses
\(\|\Cb_0^\T u\|_2=\|\Jb_0\Jb_0^{-1}\Cb_0^\T u\|_2
\le\|\Jb_0\|_2\|\Jb_0^{-1}\Cb_0^\T u\|_2\)
together with \(\lambda_{\max}(\Jb_0)\le J_u\), which follows from
Assumption~\mainref{assm:hessian}.  
This proves the eigenvalue bound; in particular, \(\Vb_0\) is nonsingular.

Now, we prove the second statement: for any full-row-rank matrix
\(\Ab\) and any positive definite matrix \(\Kb\),
\[
\Ab^\T(\Ab\Kb\Ab^\T)^{-1}\Ab\preceq \Kb^{-1}.
\]
Indeed, writing \(\Bb=\Kb^{1/2}\Ab^\T\), the left-hand side equals $\Kb^{-1/2}\Bb(\Bb^\T\Bb)^{-1}\Bb^\T\Kb^{-1/2}$,
and
\(\Bb(\Bb^\T\Bb)^{-1}\Bb^\T\) is an orthogonal projection.
Applying this inequality with \(\Ab=\Cb_0\Jb_0^{-1}\) and \(\Kb=\Kb_0\), and noting that
\(\Ab\Kb_0\Ab^\T=\Vb_0\), gives
\[
\left\|
\Omega_0
\right\|_2^2
=
\lambda_{\max}\left\{
\Ab^\T\Vb_0^{-1}\Ab
\right\}
\le
\lambda_{\max}(\Kb_0^{-1})
\le
K_l^{-1}.
\] 

\qed

\subsection{Proof of Lemma~\ref{lem:plugin_stability}}

Note that $\lambda_{\min}(\mathbf J_0)\geq J_l$, and $\|\mathbf J_0^{-1}\|_2\leq J_l^{-1}$.
Weyl's inequality gives
\[
\lambda_{\min}(\widehat{\mathbf J})
\geq
\lambda_{\min}(\mathbf J_0)
-\|\widehat{\mathbf J}-\mathbf J_0\|_2
\geq J_l-\tau_J
\geq \frac{J_l}{2}>0.
\]
Thus, $\widehat{\mathbf J}$ is positive definite and
$\|\widehat{\mathbf J}^{-1}\|_2\leq 2/J_l$.
The resolvent identity gives
$\widehat{\mathbf J}^{-1}-\mathbf J_0^{-1}
=
\widehat{\mathbf J}^{-1}
(\mathbf J_0-\widehat{\mathbf J})
\mathbf J_0^{-1}$,
and hence
\[
\|\widehat{\mathbf J}^{-1}-\mathbf J_0^{-1}\|_2
\leq
\|\widehat{\mathbf J}^{-1}\|_2
\|\widehat{\mathbf J}-\mathbf J_0\|_2
\|\mathbf J_0^{-1}\|_2
\leq
\frac{2\tau_J}{J_l^2}.
\]
Using the exact identity
\[
\begin{split}
\widehat{\mathbf J}^{-1}\widehat{\mathbf K}
 \widehat{\mathbf J}^{-1}
-\mathbf J_0^{-1}\mathbf K_0\mathbf J_0^{-1}
&=
\widehat{\mathbf J}^{-1}
(\widehat{\mathbf K}-\mathbf K_0)
\widehat{\mathbf J}^{-1}\\
&\quad+
(\widehat{\mathbf J}^{-1}-\mathbf J_0^{-1})
\mathbf K_0\widehat{\mathbf J}^{-1}
+
\mathbf J_0^{-1}\mathbf K_0
(\widehat{\mathbf J}^{-1}-\mathbf J_0^{-1}),
\end{split}
\]
Assumption~\mainref{assm:score} and submultiplicativity yield
\[
\begin{split}
\bigl\|
\widehat{\mathbf J}^{-1}\widehat{\mathbf K}
\widehat{\mathbf J}^{-1}
-\mathbf J_0^{-1}\mathbf K_0\mathbf J_0^{-1}
\bigr\|_2&\leq
\|\widehat{\mathbf J}^{-1}\|_2^2
\|\widehat{\mathbf K}-\mathbf K_0\|_2+
\|\widehat{\mathbf J}^{-1}-\mathbf J_0^{-1}\|_2
\|\mathbf K_0\|_2
\|\widehat{\mathbf J}^{-1}\|_2\\
&\qquad\quad+
\|\mathbf J_0^{-1}\|_2
\|\mathbf K_0\|_2
\|\widehat{\mathbf J}^{-1}-\mathbf J_0^{-1}\|_2\\
&\quad\leq
\frac{4\tau_K}{J_l^2}
+\frac{4K_u\tau_J}{J_l^3}
+\frac{2K_u\tau_J}{J_l^3}\\
&\quad=
\frac{6K_u}{J_l^3}\tau_J
+\frac{4}{J_l^2}\tau_K\\
&\quad\leq
6J_l^{-3}(K_u+J_l)(\tau_J+\tau_K).
\end{split}
\]
Therefore,
\[
\begin{split}
\|\widehat{\mathbf V}-\mathbf V_0\|_2
\leq
\|\mathbf C_0\|_2^2
\bigl\|
\widehat{\mathbf J}^{-1}\widehat{\mathbf K}
\widehat{\mathbf J}^{-1}
-\mathbf J_0^{-1}\mathbf K_0\mathbf J_0^{-1}
\bigr\|_2\leq
C_V\|\mathbf C_0\|_2^2(\tau_J+\tau_K).
\end{split}
\]
Moreover,
\[
\|\mathbf C_0\|_2^2
=
\lambda_{\max}(\mathbf C_0\mathbf C_0^\top)
=
\lambda_{\max}(\mathbf C\mathbf C^\top),
\]
which proves~\eqref{eq:plugin-V-error}.

By Lemma~\ref{lem:V0-conditioning}, $\mathbf V_0$ is positive definite. Since
$\widehat{\mathbf J}^{-1}$ and $\widehat{\mathbf K}$ are symmetric,
$\widehat{\mathbf V}$ is symmetric. Under the additional condition in
the statement, Weyl's inequality and \eqref{eq:plugin-V-error} give
\[
\begin{split}
\lambda_{\min}(\widehat{\mathbf V})
&\geq
\lambda_{\min}(\mathbf V_0)
-\|\widehat{\mathbf V}-\mathbf V_0\|_2\\
&\geq
\lambda_{\min}(\mathbf V_0)
-C_V\|\mathbf C_0\|_2^2(\tau_J+\tau_K)\geq
\frac{\lambda_{\min}(\mathbf V_0)}{2}>0.
\end{split}
\]
Thus, $\widehat{\mathbf V}$ is positive definite. Furthermore,
\[
\begin{split}
\bigl\|
\mathbf V_0^{-1/2}\widehat{\mathbf V}\mathbf V_0^{-1/2}
-\mathbf I_r
\bigr\|_2
&=
\bigl\|
\mathbf V_0^{-1/2}
(\widehat{\mathbf V}-\mathbf V_0)
\mathbf V_0^{-1/2}
\bigr\|_2\\
&\leq
\lambda_{\min}(\mathbf V_0)^{-1}
\|\widehat{\mathbf V}-\mathbf V_0\|_2\\
&\leq
\lambda_{\min}(\mathbf V_0)^{-1}
C_V\|\mathbf C_0\|_2^2(\tau_J+\tau_K)
\leq \frac12.
\end{split}
\]
Consequently,
\[
\left\|
\left\{
\mathbf V_0^{-1/2}\widehat{\mathbf V}\mathbf V_0^{-1/2}
\right\}^{-1}
\right\|_2
\leq 2.
\]
Since
$\mathbf V_0^{1/2}\widehat{\mathbf V}^{-1}\mathbf V_0^{1/2}
=
\left\{
\mathbf V_0^{-1/2}\widehat{\mathbf V}\mathbf V_0^{-1/2}
\right\}^{-1}$,
we obtain
\[
\begin{split}
\bigl\|
\mathbf V_0^{1/2}\widehat{\mathbf V}^{-1}
\mathbf V_0^{1/2}
-\mathbf I_r
\bigr\|_2
&=
\left\|
\left\{
\mathbf V_0^{-1/2}\widehat{\mathbf V}\mathbf V_0^{-1/2}
\right\}^{-1}
\left\{
\mathbf I_r-
\mathbf V_0^{-1/2}\widehat{\mathbf V}\mathbf V_0^{-1/2}
\right\}
\right\|_2\\
&\leq
2
\bigl\|
\mathbf V_0^{-1/2}
(\widehat{\mathbf V}-\mathbf V_0)
\mathbf V_0^{-1/2}
\bigr\|_2\\
&\leq
2\lambda_{\min}(\mathbf V_0)^{-1}
\|\widehat{\mathbf V}-\mathbf V_0\|_2\\
&\leq
2\lambda_{\min}(\mathbf V_0)^{-1}
C_V\|\mathbf C_0\|_2^2(\tau_J+\tau_K),
\end{split}
\]
which proves~\eqref{eq:plugin-whitened-error}.

Finally, for every $N>0$ and every $\mathbf u\in\mathbb R^r$,
\[
\begin{split}
\left|
N\|\widehat{\mathbf V}^{-1/2}\mathbf u\|_2^2
-
N\|\mathbf V_0^{-1/2}\mathbf u\|_2^2
\right|&=
N\left|
(\mathbf V_0^{-1/2}\mathbf u)^\top
\left\{
\mathbf V_0^{1/2}\widehat{\mathbf V}^{-1}
\mathbf V_0^{1/2}
-\mathbf I_r
\right\}
(\mathbf V_0^{-1/2}\mathbf u)
\right|\\
&\leq
\bigl\|
\mathbf V_0^{1/2}\widehat{\mathbf V}^{-1}
\mathbf V_0^{1/2}
-\mathbf I_r
\bigr\|_2
\,N\|\mathbf V_0^{-1/2}\mathbf u\|_2^2,
\end{split}
\]
which proves~\eqref{eq:plugin-quadratic-error}.
\qed

\subsection{Sample-size calculation for Corollary~\mainref{cor:batch_b_init}}
\label{sec:sample_size_discussion}

Squaring \maineqref{sample.size.requirement} gives
\[
N_b
\ge
\alpha
\left\{1+\log\left(\frac{N_{b-1}}{n_1}\right)\right\}^2,
\qquad
\alpha
=
\frac{729C^2M^2}{4\kappa^4}|\mc A|^2\log p.
\]
Let $C_\star>1$ be the larger solution of
$c=(1+\log c)^2$.  Then
$x\ge(1+\log x)^2$ for every $x\ge C_\star$.
Assume $n_1\ge C_\star\alpha$.  Since $n_1\ge\alpha$ and
$N_{b-1}\ge n_1\ge C_\star\alpha$,
\begin{align*}
N_b
\ge N_{b-1}
\ge
\alpha
\left\{1+\log\left(\frac{N_{b-1}}{\alpha}\right)\right\}^2
\ge
\alpha
\left\{1+\log\left(\frac{N_{b-1}}{n_1}\right)\right\}^2.
\end{align*}
Thus the condition holds for every $b\ge2$.  

\qed

\section{Examples} \label{sec:examples}

\subsection{Generalized linear models} 

Fix $b\in\{1,\ldots,B\}$.  For $j=1,\ldots,b$, write
$X^{(j)}=(X_1^{(j)},\ldots,X_{n_j}^{(j)})^\T\in\RR^{n_j\times p}$,
and let $X_{\cdot k}^{(j)}$ be its $k$th column.  Conditionally on
$X^{(1)},\ldots,X^{(b)}$, the responses are independent with canonical
parameter $(X_i^{(j)})^\T\betab^*$ and
\begin{equation}
\widehat{\mc L}^{(j)}(\betab)
=
\frac1{n_j}\sum_{i=1}^{n_j}
\left[
b\{(X_i^{(j)})^\T\betab\}
-Y_i^{(j)}(X_i^{(j)})^\T\betab
\right].
\label{eq:glm-loss}
\end{equation}
Thus
\begin{equation}
\nabla\widehat{\mc L}^{(j)}(\betab^*)
=-\frac1{n_j}\sum_{i=1}^{n_j}\varepsilon_i^{(j)}X_i^{(j)},
\qquad
\varepsilon_i^{(j)}
=Y_i^{(j)}-b'\{(X_i^{(j)})^\T\betab^*\}.
\label{eq:glm-score}
\end{equation}
Denote
\$
c_b=
\max_{1\le j\le b}\max_{1\le i\le n_j}
\|X_i^{(j)}\|_\infty.
\$

\subsubsection{Assumptions for generalized linear models}

Throughout, we assume $p\ge2$ and $|\mc A|\ge1$.
\begin{enumerate}[label=(GLM-\arabic*),start=1,leftmargin=3\parindent]
\item\label{GLMassm: correct specification}
For every $j\le b$ and $i\le n_j$,
\[
\EE(\varepsilon_i^{(j)}\mid X^{(1)},\ldots,X^{(b)})=0.
\]

\item\label{GLMassm: covariate norm conditions}
There are constants $B_1,B_2>0$ such that, for every $j\le b$,
\begin{equation}
\max_{1\le k\le p}\|X_{\cdot k}^{(j)}\|_2^2\le B_1n_j,
\qquad
\max_{1\le k\le p}\|X_{\cdot k}^{(j)}\|_\infty
\le B_2\left(\frac{n_j}{\log p}\right)^{1/2}.
\label{eq:glm-column-bounds}
\end{equation}

\item\label{GLMassm: residual Bernstein}
There are deterministic $M_\varepsilon>0$ and $v_0>0$, possibly depending
on $(N_b,p)$, such that
\begin{equation}
M_\varepsilon^2\EE\left[
\exp\left\{\frac{|\varepsilon_i^{(j)}|}{M_\varepsilon}\right\}
-1-\frac{|\varepsilon_i^{(j)}|}{M_\varepsilon}
\,\middle|\,X^{(1)},\ldots,X^{(b)}
\right]
\le\frac{v_0}{2}
\label{eq:glm-residual-Bernstein}
\end{equation}
and
\begin{equation}
b''\{(X_i^{(j)})^\T\betab^*\}\le v_0
\label{eq:glm-truth-weight}
\end{equation}
for every $j\le b$ and $i\le n_j$.

\item\label{GLMassm: RSC}
There is $\underline\kappa>0$ such that
\begin{equation}
\min_{j\le b}
\inf_{\substack{u\ne0:\,
\|u_{\mc A^c}\|_1\le3\|u_{\mc A}\|_1}}
\frac{u^\T\nabla^2\widehat{\mc L}^{(j)}(\betab^*)u}{\|u\|_2^2}
\ge\underline\kappa.
\label{eq:glm-truth-rsc}
\end{equation}

\item\label{GLMassm: self concordance}
The function $b$ is three times continuously differentiable, $b''(t)>0$, and
there is $K\ge0$ such that
\$
|b'''(t)|\le Kb''(t),\qquad t\in\RR.
\$
\end{enumerate}

Condition~\ref{GLMassm: covariate norm conditions} is the batchwise form of
(A7) of \citet{jacobson2025strong}.  The first display in
\ref{GLMassm: residual Bernstein} is the conditional, batchwise form of
(A3) of \citet{shi2019linear} and (A8) of
\citet{jacobson2025strong}.  
The bound
\eqref{eq:glm-truth-weight} is additional for the ease of theory.
Conditions
\ref{GLMassm: RSC} and \ref{GLMassm: self concordance} are the batchwise
forms of (A9) and (A10), respectively, of
\citet{jacobson2025strong}. 
Indeed, on the tuning event in \maineqref{batch_b_event}, the KKT argument in the proof of Theorem~\mainref{thm:batch_b_init} gives, for the displacement
$u=\widehat\betab^{(j)}-\betab^*$ at the batch under consideration,
\begin{equation}
\|u_{\mc A^c}\|_1
\le3\|u_{\mc A}\|_1,
\qquad
\|u\|_1
\le4|\mc A|^{1/2}\|u\|_2;
\label{eq:glm-main-body-cone}
\end{equation}
see \eqref{eq:thm1-cone}--\eqref{eq:thm1-global-cone}.  
Hence every displacement generated by the proof of Theorem~\mainref{thm:batch_b_init} belongs to the cone $\CC(4|\mc A|^{1/2})$.  

\subsubsection{Score concentration}

\begin{lemma}\label{lem:glm-score-concentration}
Let $\lambda_j=C(\log p/N_j)^{1/2}$, $j=1,\ldots,b$.  Under
\ref{GLMassm: correct specification}--\ref{GLMassm: residual Bernstein}, if
\begin{equation}
C^2\ge64B_1v_0+16M_\varepsilon C
\min\left\{B_2,c_b\left(\frac{\log p}{n_1}\right)^{1/2}\right\},
\label{eq:glm-C-condition}
\end{equation}
then
\begin{equation}
\PP\left\{
\max_{1\le j\le b}\frac{\|\omegab^{(j)}\|_\infty}{\lambda_j}
\le\frac14
\,\middle|\,X^{(1)},\ldots,X^{(b)}
\right\}
\ge1-\frac{2b}{p}.
\label{eq:glm-score-event}
\end{equation}
\end{lemma}

\begin{proof}
Expanding the exponential in \eqref{eq:glm-residual-Bernstein} gives, for
every integer $m\ge2$,
\[
\EE(|\varepsilon_i^{(j)}|^m\mid X^{(1)},\ldots,X^{(b)})
\le\frac{m!}{2}v_0M_\varepsilon^{m-2}.
\]
This is the moment condition used in the argument leading to display (S5.14) of \citet{shi2019linear}.  
Bernstein's inequality then gives
\begin{equation}
\PP\left\{|\omega_k^{(j)}|>u\mid X^{(1)},\ldots,X^{(b)}\right\}
\le2\exp\left\{-\frac{u^2}{2(V_j+L_ju)}\right\},
\label{eq:glm-coordinate-Bernstein}
\end{equation}
where, by \eqref{eq:glm-column-bounds},
\[
V_j\le\frac{B_1v_0}{N_j},
\qquad
L_j\le
\frac{M_\varepsilon}{(N_j\log p)^{1/2}}
\min\left\{B_2,c_b\left(\frac{\log p}{N_j}\right)^{1/2}\right\}.
\]
With $u=\lambda_j/4$, the right-hand side of
\eqref{eq:glm-coordinate-Bernstein} is at most
\[
2\exp\left[
-\frac{C^2\log p}{32B_1v_0+8M_\varepsilon C
\min\{B_2,c_b(\log p/N_j)^{1/2}\}}
\right]
\le2p^{-2}
\]
by \eqref{eq:glm-C-condition}.  Union bounds over $k\le p$ and $j\le b$
prove \eqref{eq:glm-score-event}.
\end{proof}
\noindent Note that a sufficient condition for \eqref{eq:glm-C-condition} is
\begin{equation}
\begin{split}
C\ge8M_\varepsilon
\min\left\{B_2,c_b\left(\frac{\log p}{n_1}\right)^{1/2}\right\}
+8\left[
M_\varepsilon^2
\min\left\{B_2^2,\frac{c_b^2\log p}{n_1}\right\}
+B_1v_0
\right]^{1/2}.
\end{split}
\label{eq:glm-C-explicit}
\end{equation}

\subsubsection{Verification of the local RSC and RLH inequalities}

Define
\begin{equation}
M_{\rm GLM}
=16K\exp\{4Kc_b|\mc A|^{1/2}R_0\}
c_bv_0B_1|\mc A|.
\label{eq:glm-M-initializer}
\end{equation}

\begin{lemma}\label{lem:glm-initializer-curvature}
Suppose \ref{GLMassm: covariate norm conditions}--\ref{GLMassm: self concordance}
hold.  If
\begin{equation}
0<\kappa_{\rm GLM}\le
\underline\kappa\exp\{-8Kc_b|\mc A|^{1/2}R_0\},
\label{eq:glm-kappa-initializer}
\end{equation}
then the RSC and RLH inequalities used in the proof of
Theorem~\mainref{thm:batch_b_init} hold, respectively, with constants $\kappa_{\rm GLM}$
and $M_{\rm GLM}$, for all displacements satisfying
\eqref{eq:glm-main-body-cone} and having Euclidean norm at most $R_0$.
In particular, for every such $v,w$ and every $j\le b$,
\begin{equation}
\left\|
\{\nabla^2\widehat{\mc L}^{(j)}(\betab^*+v)
-\nabla^2\widehat{\mc L}^{(j)}(\betab^*)\}w
\right\|_\infty
\le M_{\rm GLM}\|v\|_2\|w\|_2.
\label{eq:glm-initializer-rlh}
\end{equation}
\end{lemma}

\begin{proof}
Condition~\ref{GLMassm: self concordance} gives
\begin{equation}
e^{-K|r|}b''(a)\le b''(a+r)\le e^{K|r|}b''(a),
\qquad a,r\in\RR,
\label{eq:glm-curvature-comparison}
\end{equation}
because $|(\log b'')'|\le K$.  The first relation in
\eqref{eq:glm-main-body-cone} implies
$\|u\|_1\le4|\mc A|^{1/2}\|u\|_2$.  
The mean-value theorem, displays
\eqref{eq:glm-curvature-comparison} and \eqref{eq:glm-truth-rsc} give the
RSC lower bound \eqref{eq:glm-kappa-initializer}.

For the RLH bound, the mean-value theorem,
\eqref{eq:glm-curvature-comparison}, and the Cauchy--Schwarz inequality give, for every
coordinate $k$,
\begin{align*}
&\left|e_k^\T
\{\nabla^2\widehat{\mc L}^{(j)}(\betab^*+v)
-\nabla^2\widehat{\mc L}^{(j)}(\betab^*)\}w\right|
\\
&\quad\le
Kc_b e^{4Kc_b|\mc A|^{1/2}R_0}
\{v^\T\nabla^2\widehat{\mc L}^{(j)}(\betab^*)v\}^{1/2}
\{w^\T\nabla^2\widehat{\mc L}^{(j)}(\betab^*)w\}^{1/2}.
\end{align*}
By \eqref{eq:glm-column-bounds} and \eqref{eq:glm-truth-weight},
\[
z^\T\nabla^2\widehat{\mc L}^{(j)}(\betab^*)z
\le v_0n_j^{-1}\|X^{(j)}z\|_2^2
\le16v_0B_1|\mc A|\|z\|_2^2
\]
for $z=v,w$.  This proves \eqref{eq:glm-initializer-rlh} with
\eqref{eq:glm-M-initializer}.
\end{proof}
We remark that the exponential term in \eqref{eq:glm-kappa-initializer} is not required for the linear GLM, since the Hessian does not depend on the parameter in that case.

\subsubsection{Online accumulated-error bound}

If $\lambda_j=C(\log p/N_j)^{1/2}$ and
\[
\|\widehat\betab^{(r)}-\betab^*\|_2
\le\frac{3}{2\kappa_{\rm GLM}}|\mc A|^{1/2}\lambda_r,
\qquad r<j,
\]
then Lemma~\ref{lem:batch.size} gives
\begin{equation}
\phi^{(j)}
\le
\frac{9C^2|\mc A|\log p}{4\kappa_{\rm GLM}^2N_j}
\left\{1+\log\left(\frac{N_{j-1}}{n_1}\right)\right\}.
\label{eq:glm-phi-bound}
\end{equation}
Consequently,
\begin{equation}
\left(\frac{N_j}{|\mc A|^2\log p}\right)^{1/2}
\ge
\frac{27CM_{\rm GLM}}{2\kappa_{\rm GLM}^2}
\left\{1+\log\left(\frac{N_{j-1}}{n_1}\right)\right\}
\label{eq:glm-online-sample}
\end{equation}
implies $3M_{\rm GLM}\phi^{(j)}/2\le\lambda_j/4$.

\subsubsection{Verification of the initializer theorem}

\begin{proposition}\label{prop:glm-initializer}
Suppose \ref{GLMassm: correct specification}--\ref{GLMassm: self concordance}
hold.  Let $\lambda_j=C(\log p/N_j)^{1/2}$, where
\eqref{eq:glm-C-condition} holds, and let $M_{\rm GLM}$ and $\kappa_{\rm GLM}$ satisfy
\eqref{eq:glm-M-initializer}--\eqref{eq:glm-kappa-initializer}.  Suppose
\$
R_0>
\max_{1\le j\le b}
\frac{3}{2\kappa_{\rm GLM}}|\mc A|^{1/2}\lambda_j
\$
and \eqref{eq:glm-online-sample} holds for $j=2,\ldots,b$.  Then,
conditionally on $X^{(1)},\ldots,X^{(b)}$, with probability at least
$1-2b/p$, simultaneously for $j=1,\ldots,b$,
\begin{equation}
\begin{split}
\|\widehat\betab^{(j)}-\betab^*\|_2
&\le\frac{3}{2\kappa_{\rm GLM}}|\mc A|^{1/2}\lambda_j,
\\
\|\widehat\betab^{(j)}-\betab^*\|_1
&\le\frac{6}{\kappa_{\rm GLM}}|\mc A|\lambda_j,
\end{split}
\label{eq:glm-initializer-rates}
\end{equation}
and $\widehat\betab^{(j)}-\betab^*$ satisfies
\eqref{eq:glm-main-body-cone}.
\end{proposition}

\begin{proof}
Work on \eqref{eq:glm-score-event}.  For $j=1$, $\phi^{(1)}=0$.
For $j\ge2$, assume the conclusion through batch $j-1$.
Equation \eqref{eq:glm-phi-bound} and
\eqref{eq:glm-online-sample} give
$3M_{\rm GLM}\phi^{(j)}/2\le\lambda_j/4$, while
\eqref{eq:glm-score-event} gives
$\|\omegab^{(j)}\|_\infty\le\lambda_j/4$.  Hence
\[
\lambda_j\ge
2\{\|\omegab^{(j)}\|_\infty+3M_{\rm GLM}\phi^{(j)}/2\}.
\]
Lemma~\ref{lem:decomposition} and \eqref{eq:R-beta-star-bound} therefore
give
$\|\nabla\widehat{\mc Q}_j(\betab^*)\|_\infty\le\lambda_j/2$.
By convexity,
$\widehat{\mc Q}_j(\betab)+\lambda_j\|\betab\|_1$ is bounded below by a
constant plus $(\lambda_j/2)\|\betab\|_1$.
Apply the proof of Theorem~\mainref{thm:batch_b_init}.  Its KKT step gives
\eqref{eq:glm-main-body-cone}; Lemma~\ref{lem:glm-initializer-curvature}
then supplies every RSC and RLH inequality used in the remainder and
localization steps.  The sole modification is to invoke that lemma on the
KKT displacements rather than to impose the larger event
$\mc R_j(R_0,4|\mc A|^{1/2},\kappa_{\rm GLM},M_{\rm GLM})$. 
Thus we get \eqref{eq:glm-initializer-rates}.
The probability follows from Lemma~\ref{lem:glm-score-concentration}.
\end{proof}

\subsubsection{Verification of the tightening-stage theorem}
\label{sec:glm_tightening_verification}

Suppose Assumptions~\mainref{assm:penalty_function} and
\mainref{assm:local_alternative} hold.  
Suppose that the design conditions in
Proposition~\ref{prop:glm-initializer}, with $b=B$, hold.
Let $\gamma_0,p_0,c_{\rm aug},l$, and $\delta$
be as in
Theorem~\mainref{thm:batch_b_tightening}, and use the same $\kappa_{\rm GLM}$ as in
Proposition~\ref{prop:glm-initializer}. 
Define
\begin{equation}
k_l=p\wedge2\lceil l^2\rceil,
\qquad
h_l=k_l\log(ep/k_l)+2\log\log(ep),
\qquad
\|z\|_{(k)}=\max_{|S|\le k}\|z_S\|_2.
\label{eq:glm-terminal-dimensions}
\end{equation}
Impose the following conditions.

\begin{enumerate}[label=(GLM-\arabic*),start=6,leftmargin=3\parindent]
\item\label{GLMassm: terminal random design revised}
The observations are independent across all batches and the covariate rows
are copies of a random vector $X$ satisfying
\begin{equation}
\|u^\T X\|_{\psi_2}\le K_X\|u\|_2,
\qquad u\in\RR^p.
\label{eq:glm-subgaussian-design}
\end{equation}
With $\varepsilon=Y-b'(X^\T\betab^*)$, assume
$\EE(\varepsilon\mid X)=0$ and, for constants $v_\varepsilon>0$ and
$c_\varepsilon\ge0$,
\begin{equation}
\log\EE(e^{t\varepsilon}\mid X)
\le
\frac{v_\varepsilon b''(X^\T\betab^*)t^2}
{2(1-c_\varepsilon|t|)},
\qquad c_\varepsilon|t|<1.
\label{eq:glm-score-mgf}
\end{equation}

\item\label{GLMassm: terminal hessian tails revised}
Assumption~\mainref{assm:hessian} holds, and either
\begin{equation*}
\|\{b''(X^\T\betab^*)\}^{1/2}u^\T X\|_{\psi_2}
\le K_H\|u\|_2,
\qquad u\in\RR^p,
\tag{a}\label{eq:glm-weighted-subgaussian}
\end{equation*}
or there are constants $C_H,c_H>0$ such that
\begin{equation*}
\PP\{\log b''(X^\T\betab^*)>C_H+t\}
\le2e^{-c_Ht^2},
\qquad t\ge0.
\tag{b}\label{eq:glm-log-weight-tail}
\end{equation*}
\end{enumerate}
Under \eqref{eq:glm-weighted-subgaussian}, set
\begin{equation}
\eta_\infty
=CK_H^2\left\{\sqrt{\frac{h_l}{N_B}}+\frac{h_l}{N_B}\right\},
\qquad
T_B=CK_H^2\{k_l\log(ep/k_l)+\log(2N_B)+2\log\log(ep)\}.
\label{eq:glm-terminal-scales-a}
\end{equation}
Under \eqref{eq:glm-log-weight-tail}, let
$T_B\ge Ck_l\log(ep/k_l)$ satisfy
\begin{equation}
N_B\exp\left[-c\log^2\left\{
\frac{T_B}{Ck_l\log(ep/k_l)}
\right\}\right]\le\{\log(ep)\}^{-2}
\label{eq:glm-tail-truncation}
\end{equation}
and set
\begin{equation}
\eta_\infty
=C\left\{\sqrt{\frac{h_l}{N_B}}+\frac{T_Bh_l}{N_B}\right\}.
\label{eq:glm-terminal-scales-b}
\end{equation}
Here and below, $C>0$ is sufficiently large and $c>0$ is sufficiently
small generic constants; their values depend only on the constants in
\ref{GLMassm: terminal random design revised}--\ref{GLMassm: terminal hessian tails revised}.
Increase $C$ if necessary and set
\begin{equation}
\begin{split}
M={}&4KCK_XT_B\{\log(2pN_B)+2\log\log(ep)\}^{1/2}\\
&\quad\times\exp\left[
KCK_XlR_0\{\log(2pN_B)+2\log\log(ep)\}^{1/2}
\right].
\end{split}
\label{eq:glm-terminal-M}
\end{equation}

\begin{lemma}\label{lem:glm-terminal-design}
Under \ref{GLMassm: terminal random design revised}--\ref{GLMassm: terminal hessian tails revised},
with probability at least $1-C\{\log(ep)\}^{-2}$,
\begin{align}
&\sup_{\substack{\|u\|_2=1\\\|u\|_0\le k_l}}
\left|u^\T\left\{
\frac1{N_B}\sum_{j=1}^Bn_j
\nabla^2\widehat{\mc L}^{(j)}(\betab^*)-\Jb
\right\}u\right|
\le\eta_\infty,
\label{eq:glm-sparse-hessian}
\\
&c_B\le CK_X\{\log(2pN_B)+2\log\log(ep)\}^{1/2},
\label{eq:glm-row-envelope}
\\
&\max_{j\le B}\max_{i\le n_j}
b''\{(X_i^{(j)})^\T\betab^*\}
\|X_i^{(j)}\|_{(k_l)}^2\le T_B,
\label{eq:glm-weighted-row-envelope}
\\
&\max_{j\le B}\max_{i\le n_j}
\|X_{i,\overline{\mc M}}^{(j)}\|_2
\le CK_X\{q+\log(2N_B)+2\log\log(ep)\}^{1/2},
\label{eq:glm-model-row-envelope}
\\
&\max_{j\le B}\max_{i\le n_j}
\|\operatorname{diag}(\Zb,\Ib_s)^\T
X_{i,\overline{\mc M}}^{(j)}\|_2
\le CK_X\{q-r+\log(2N_B)+2\log\log(ep)\}^{1/2}.
\label{eq:glm-tangent-row-envelope}
\end{align}
\end{lemma}

\begin{proof}
Under \eqref{eq:glm-weighted-subgaussian}, apply Bernstein's inequality to
$\{b''(X^\T\betab^*)\}^{1/2}u^\T X$ on a $1/4$-net for each support of
size at most $k_l$.  The net cardinality is at most
$9^{k_l}(ep/k_l)^{k_l}$, so a union bound gives
\eqref{eq:glm-sparse-hessian}; applying the same bound to each of the
$N_B$ weighted rows gives \eqref{eq:glm-weighted-row-envelope} with
\eqref{eq:glm-terminal-scales-a}.  This is an analogous argument used in the proofs of Lemmas S.2--S.3 of
\citet{shi2019linear}; the only changes are the replacement of the fixed
$\mc M\cup\mc S$ block by the union over $k_l$-sparse supports and the
replacement of $X$ by $\{b''(X^\T\betab^*)\}^{1/2}X$.

Under \eqref{eq:glm-log-weight-tail}, sub-Gaussianity of $X$ and a support
net give, for $z\ge Ck_l\log(ep/k_l)$,
\begin{equation}
\PP\left[
b''(X^\T\betab^*)\|X\|_{(k_l)}^2>z
\right]
\le
C\exp\left[-c\log^2\left\{
\frac{z}{Ck_l\log(ep/k_l)}
\right\}\right].
\label{eq:glm-weighted-row-tail}
\end{equation}
Moreover, the Cauchy--Schwarz inequality and the two tail assumptions give
\[
\sup_{\|u\|_2=1}
\EE[b''(X^\T\betab^*)^2(u^\T X)^4]\le C.
\]
Truncate
$b''(X^\T\betab^*)(u^\T X)^2$ on the event in
\eqref{eq:glm-weighted-row-tail} with $z=T_B$.  Bernstein's inequality and
the preceding support net give \eqref{eq:glm-sparse-hessian} with
\eqref{eq:glm-terminal-scales-b}.  Equation
\eqref{eq:glm-tail-truncation} bounds the probability of any truncation by
$C\{\log(ep)\}^{-2}$.  The preceding fourth-moment bound and
the Cauchy--Schwarz inequality bound the
single-observation truncation bias by
$C\{\log(ep)\}^{-1}N_B^{-1/2}$, which is absorbed by
$C(h_l/N_B)^{1/2}$.  This also proves
\eqref{eq:glm-weighted-row-envelope}.

Finally, coordinatewise sub-Gaussian concentration and a union bound over
$pN_B$ entries prove \eqref{eq:glm-row-envelope}.  For a fixed matrix $U$
with $d$ orthonormal columns, a $1/2$-net gives
$\PP\{\|U^\T X\|_2>CK_X(d+t)^{1/2}\}\le2e^{-t}$.
Apply this with the coordinate embedding of $\overline{\mc M}$ and with
$U=\operatorname{diag}(\Zb,\Ib_s)$, and take
$t=\log(2N_B)+2\log\log(ep)$, to obtain
\eqref{eq:glm-model-row-envelope}--\eqref{eq:glm-tangent-row-envelope}.
\end{proof}

For every symmetric matrix $A$, the block decomposition of a vector in
$\CC(l)$ into successive blocks of size $\lceil l^2\rceil$ gives
\begin{equation}
\sup_{v\in\CC(l),\,v\ne0}
\frac{|v^\T Av|}{\|v\|_2^2}
\le
4\sup_{\substack{\|u\|_2=1\\\|u\|_0\le k_l}}|u^\T Au|,
\qquad
\|A\|_{\max}
\le
\sup_{\substack{\|u\|_2=1\\\|u\|_0\le k_l}}|u^\T Au|.
\label{eq:glm-sparse-to-cone}
\end{equation}
Indeed, if the blocks are $S_1,S_2,\ldots$, then
$\sum_{a\ge2}\|v_{S_a}\|_2\le\lceil l^2\rceil^{-1/2}\|v\|_1
\le\|v\|_2$, and hence
$\sum_a\|v_{S_a}\|_2\le2\|v\|_2$.

\begin{lemma}\label{lem:glm-terminal-curvature}
On \eqref{eq:glm-row-envelope}--\eqref{eq:glm-weighted-row-envelope},
\begin{equation}
\bigcap_{j=1}^B\mc E_{\rm rlh}^{(j)}(R_0,l,M)
\label{eq:glm-terminal-rlh}
\end{equation}
holds.  Suppose also that the historical estimators satisfy the conclusion
of Proposition~\ref{prop:glm-initializer}.  On
\eqref{eq:glm-sparse-hessian},
\begin{equation}
\left\langle
\nabla\widehat{\mc Q}_B(\betab_1)
-\nabla\widehat{\mc Q}_B(\betab_2),
\betab_1-\betab_2
\right\rangle
\ge\kappa_{\rm GLM}\|\betab_1-\betab_2\|_2^2
\label{eq:glm-terminal-surrogate-rsc}
\end{equation}
for every pair in Definition~\mainref{def:local_rsc}, provided
\begin{equation}
0<\kappa_{\rm GLM}\le
\exp[-2KCK_XlR_0
\{\log(2pN_B)+2\log\log(ep)\}^{1/2}]
(J_l-4\eta_\infty).
\label{eq:glm-terminal-kappa}
\end{equation}
\end{lemma}

\begin{proof}
For $v,w\in\BB(R_0)\cap\CC(l)$, the mean-value theorem and
\eqref{eq:glm-curvature-comparison} give
\begin{align*}
&\left|e_a^\T
\{\nabla^2\widehat{\mc L}^{(j)}(\betab^*+v)
-\nabla^2\widehat{\mc L}^{(j)}(\betab^*)\}w\right|
\\
&\quad\le
Kc_Be^{Kc_BlR_0}\frac1{n_j}\sum_{i=1}^{n_j}
b''\{(X_i^{(j)})^\T\betab^*\}
|(X_i^{(j)})^\T v|\,|(X_i^{(j)})^\T w|.
\end{align*}
Equations \eqref{eq:glm-weighted-row-envelope} and
\eqref{eq:glm-sparse-to-cone}, applied to each weighted rank-one matrix,
bound the last average by $4T_B\|v\|_2\|w\|_2$.  Equations
\eqref{eq:glm-row-envelope} and \eqref{eq:glm-terminal-M} prove
\eqref{eq:glm-terminal-rlh}.

Let $v=\betab_1-\betab_2$.  The linear term in
$\widehat{\mc Q}_B$ cancels from its gradient difference.  The remaining
quadratic form is the weighted average of the current-batch Hessians along
the segment from $\betab_2$ to $\betab_1$ and the historical Hessians at
$\widehat\betab^{(j)}$.  Since
$l\ge4|\mc A|^{1/2}$, Proposition~\ref{prop:glm-initializer} places every
historical displacement in $\BB(R_0)\cap\CC(l)$.  Equations
\eqref{eq:glm-curvature-comparison} and \eqref{eq:glm-row-envelope}
therefore bound every Hessian below by
$\exp[-2KCK_XlR_0
\{\log(2pN_B)+2\log\log(ep)\}^{1/2}]$
times its value at $\betab^*$.  The resulting quadratic form is at least
this factor times
\[
v^\T\left\{N_B^{-1}\sum_{j=1}^Bn_j
\nabla^2\widehat{\mc L}^{(j)}(\betab^*)\right\}v.
\]
Assumption~\mainref{assm:hessian}, \eqref{eq:glm-sparse-hessian}, and
\eqref{eq:glm-sparse-to-cone} bound the last display below by
$(J_l-4\eta_\infty)\|v\|_2^2$.  This proves
\eqref{eq:glm-terminal-surrogate-rsc}.
\end{proof}

\begin{lemma}\label{lem:glm-terminal-score}
On the event in Lemma~\ref{lem:glm-terminal-design}, with conditional
probability at least $1-C\{\log(ep)\}^{-2}$,
\begin{align}
\|\omegab^*\|_\infty
&\le C\left[
\left\{\frac{(J_u+\eta_\infty)
\{\log(2p)+2\log\log(ep)\}}{N_B}\right\}^{1/2}
+\frac{c_B\{\log(2p)+2\log\log(ep)\}}{N_B}
\right],
\label{eq:glm-terminal-score-sup}
\\
\|\omegab^*_{\overline{\mc M}}\|_2
&\le C\left[
\sqrt{\frac{q+2\log\log(ep)}{N_B}}
+\frac{\{q+\log(2N_B)+2\log\log(ep)\}^{1/2}
\{q+2\log\log(ep)\}}{N_B}
\right].
\label{eq:glm-terminal-score-block}
\end{align}
The right-hand side of \eqref{eq:glm-terminal-score-block}, with $q$
replaced by $q-r$, also bounds
$\|\Zb^\T\omegab^*_{\mc M}\|_2+\|\omegab^*_{\mc S}\|_2$.
\end{lemma}

\begin{proof}
For deterministic $a$, conditional on all covariates,
$N_B^{-1}\sum_{j,i}\varepsilon_i^{(j)}a^\T X_i^{(j)}$ satisfies Bernstein's inequality with variance proxy
\[
\frac{v_\varepsilon}{N_B} a^\T\left\{N_B^{-1}\sum_{j=1}^Bn_j
\nabla^2\widehat{\mc L}^{(j)}(\betab^*)\right\}a
\]
and scale bounded by
$c_\varepsilon\max_{j,i}|a^\T X_i^{(j)}|/N_B$; this is the direct
application of \eqref{eq:glm-score-mgf}.  Take $a=e_k$ and a union bound
over $k\le p$ to obtain \eqref{eq:glm-terminal-score-sup}.  Apply the same
inequality on a $1/4$-net of the unit sphere in $\RR^q$, using
\eqref{eq:glm-model-row-envelope}, to obtain
\eqref{eq:glm-terminal-score-block}.  Applying it to the orthonormal
columns of $\operatorname{diag}(\Zb,\Ib_s)$ and using
\eqref{eq:glm-tangent-row-envelope} proves the last claim.
\end{proof}

On the initializer event,
\begin{equation}
\phi^*\le
\frac{9|\mc A|}{4\kappa_{\rm GLM}^2N_B}
\sum_{j=1}^{B-1}n_j\lambda_j^2.
\label{eq:glm-terminal-phi}
\end{equation}
Define
\begin{equation}
\begin{split}
r_B^{\rm G}={}&
\sqrt{\frac{q+2\log\log(ep)}{N_B}}
+\frac{\{q+\log(2N_B)+2\log\log(ep)\}^{1/2}
\{q+2\log\log(ep)\}}{N_B}
\\
&+q^{1/2}\left\{
\frac{9M|\mc A|}{4\kappa_{\rm GLM}^2N_B}
\sum_{j=1}^{B-1}n_j\lambda_j^2
+d_B+m^{1/2}\eta_\infty d_B+Md_B^2
\right\}.
\end{split}
\label{eq:glm-oracle-radius-rate}
\end{equation}
Recall $\lambda=\lambda_B$.
\begin{enumerate}[label=(GLM-\arabic*),start=8,leftmargin=3\parindent]
\item\label{GLMassm: theorem2 rates revised}
\begin{gather*}
\lambda\ge C_\lambda
\sqrt{\frac{\log(2p)+2\log\log(ep)}{N_B}},
\qquad
p_0^{-1}=O(1),
\\
K_XlR_0\{\log(2pN_B)+2\log\log(ep)\}^{1/2}=O(1),
\qquad
4\eta_\infty\le J_l/2,
\end{gather*}
where $C_\lambda p_0$ is sufficiently large and
\eqref{eq:glm-terminal-kappa} holds.  In addition,
\begin{gather}
R_0>c_{\rm aug}\gamma_0q^{1/2}\lambda\vee d_B,
\qquad
\kappa d_B\le\lambda,
\label{eq:glm-t2-localization}
\\
K_X\left[
\frac{\{\log(2pN_B)+2\log\log(ep)\}
\{\log(2p)+2\log\log(ep)\}}{N_B}
\right]^{1/2}=o(1),
\label{eq:glm-t2-score-scale}
\\
\frac{M|\mc A|}{\kappa^2N_B}
\sum_{j=1}^{B-1}n_j\lambda_j^2=o(\lambda),
\qquad
d_B+m^{1/2}\eta_\infty d_B+Md_B^2=o(\lambda).
\label{eq:glm-t2-remainders}
\end{gather}

\item\label{GLMassm: theorem3 rates revised}
There are $\gamma_1>\gamma_0$ and
$\bar\delta\in(5/(4\kappa\gamma_0),1)$ such that
\begin{equation}
p_0\in(0,1/2],
\qquad
p'_+(\gamma_1)=0,
\qquad
\min_{j\in\mc S}|\beta_j^*|
\ge(\gamma_0+\gamma_1)\lambda.
\label{eq:glm-t3-penalty}
\end{equation}
Moreover,
\begin{equation}
r_B^{\rm G}=o(\lambda),
\qquad
\left[
\eta_\infty+q^{1/2}M
\left\{\frac{9|\mc A|}{4\kappa^2N_B}
\sum_{j=1}^{B-1}n_j\lambda_j^2\right\}^{1/2}
\right]r_B^{\rm G}
+q^{1/2}M(r_B^{\rm G})^2=o(\lambda).
\label{eq:glm-t3-remainders}
\end{equation}
\end{enumerate}

\begin{proposition}\label{prop:glm-tightening}
Under Assumptions~\mainref{assm:penalty_function} and
\mainref{assm:local_alternative}, the conditions of
Proposition~\ref{prop:glm-initializer} with $b=B$, and
\ref{GLMassm: terminal random design revised}--\ref{GLMassm: theorem2 rates revised},
the full- and reduced-model conclusions of
Theorem~\mainref{thm:batch_b_tightening}, including
\eqref{eq:crude-lla-localization}, hold, for all sufficiently large $N_B$,
with probability at least
$1-2B/p-C\{\log(ep)\}^{-2}$.
\end{proposition}

\begin{proof}
Intersect the initializer event, Lemma~\ref{lem:glm-terminal-design}, and
the conditional event in Lemma~\ref{lem:glm-terminal-score}.  Conditional
expectation and a union bound give the stated probability.  Equations
\eqref{eq:R-beta-star-bound}, \eqref{eq:glm-terminal-rlh}, and
\eqref{eq:glm-terminal-phi} yield
\begin{equation}
\|\gb^*-\omegab^*\|_\infty\le\frac32M\phi^*.
\label{eq:glm-gstar-bound}
\end{equation}
Since $k_l\ge2$, \eqref{eq:glm-sparse-hessian} implies
$\mc H(\eta_\infty)$.  Lemma~\ref{lem:omega_discrepancy} therefore gives
\begin{equation}
\|\gb^\dagger-\omegab^*\|_\infty
\le2M\phi^*+2Md_B^2+(J_u+m^{1/2}\eta_\infty)d_B.
\label{eq:glm-gdagger-bound}
\end{equation}
Equations \eqref{eq:glm-terminal-score-sup},
\eqref{eq:glm-t2-score-scale}, and \eqref{eq:glm-t2-remainders}, with
$C_\lambda p_0$ sufficiently large, verify the full and reduced tuning
inequalities in Theorem~\mainref{thm:batch_b_tightening}.
Then, applying the proofs of Lemmas~\ref{lem:intermediate.rate},
\ref{lem:intermediate.rate2}, and Theorem~\mainref{thm:batch_b_tightening} yields the conclusion.
\end{proof}

\subsubsection{Verification of the strong-oracle theorem}
\label{sec:glm_strong_oracle_verification}

Write
\begin{equation}
\Ub_a=\Ib_q,
\qquad
\Ub_0=\operatorname{diag}(\Zb,\Ib_s),
\qquad
\Ab_e=\Jb_{\overline{\mc M}^c,\overline{\mc M}}\Ub_e(\Ub_e^\T\Jb_{\overline{\mc M},\overline{\mc M}}\Ub_e)^{-1},
\quad e=0,a.
\label{eq:glm-tangent-projection}
\end{equation}
Assumption~\mainref{assm:hessian} gives
\begin{equation}
\max_{j\in \overline{\mc M}^c}\|(\Ab_e)_{j,\cdot}\|_2\le J_u/J_l,
\qquad e=0,a.
\label{eq:glm-projection-row}
\end{equation}

\begin{lemma}\label{lem:glm-projected-score}
Under \ref{GLMassm: terminal random design revised}--\ref{GLMassm: terminal hessian tails revised},
with probability at least $1-C\{\log(ep)\}^{-2}$,
\begin{equation}
\begin{split}
&\max_{e\in\{0,a\}}
\|\omegab^*_{\overline{\mc M}^c}-\Ab_e\Ub_e^\T\omegab^*_{\overline{\mc M}}\|_\infty
\\
&\quad\le C\left[
\sqrt{\frac{(J_u+\eta_\infty)
\{\log(2p)+2\log\log(ep)\}}{N_B}}\right.\\
&\qquad\qquad\qquad \left.+\frac{\{\log(2pN_B)+2\log\log(ep)\}^{1/2}
\{\log(2p)+2\log\log(ep)\}}{N_B}
\right].
\end{split}
\label{eq:glm-projected-score}
\end{equation}
\end{lemma}

\begin{proof}
For $j\in \overline{\mc M}^c$, the corresponding coordinate is
\[
-N_B^{-1}\sum_{r=1}^B\sum_{i=1}^{n_r}\varepsilon_i^{(r)}
\{X_{ij}^{(r)}-(\Ab_e)_{j,\cdot}\Ub_e^\T X_{i,\overline{\mc M}}^{(r)}\}.
\]
The coefficient vector in braces has support at most $q+1\le k_l$ and
squared Euclidean norm at most $1+(J_u/J_l)^2$.  Equations
\eqref{eq:glm-subgaussian-design}, \eqref{eq:glm-sparse-hessian}, and
\eqref{eq:glm-projection-row} bound its conditional variance proxy by
$C(J_u+\eta_\infty)$; a union bound over the $pN_B$ corresponding linear
forms bounds the modulus used for the following Bernstein's inequality by
$C\{\log(2pN_B)+2\log\log(ep)\}^{1/2}$.
Bernstein's inequality and a
union bound over $j\in \overline{\mc M}^c$ prove the result.
\end{proof}

\begin{lemma}\label{lem:glm-oracle-conditions}
Under the assumptions of Proposition~\ref{prop:glm-tightening} and
\ref{GLMassm: theorem3 rates revised}, for all sufficiently large $N_B$,
with probability at least
$1-2B/p-C\{\log(ep)\}^{-2}$,
\begin{align}
\max_{e\in\{0,a\}}\|\widehat\betab_e^{\rm ora}-\betab^*\|_2
&\le Cr_B^{\rm G},
\label{eq:glm-oracle-rate}
\\
\max_{e\in\{0,a\}}
\|\nabla_{\overline{\mc M}^c}\widehat{\mc Q}_B(\widehat\betab_e^{\rm ora})\|_\infty
&\le p_0\lambda/2.
\label{eq:glm-oracle-offscore}
\end{align}
\end{lemma}

\begin{proof}
Minimize $\widehat{\mc Q}_B$ on each oracle affine set intersected with the
closed radius-$R_0$ ball centered at $\betab^*$ for the full oracle and at
$\betab^\dagger$ for the reduced oracle.  The sets are nonempty and
compact.  
The argument in the proof of Lemma~\ref{lem:oracle_rate} shows that both localized minimizers are interior and hence are global oracle minimizers.  Equation
\eqref{eq:glm-terminal-surrogate-rsc} also makes them unique on their oracle
affine sets.
The proof of Lemma~\ref{lem:oracle_rate} uses the batchwise RSC events only
through \eqref{eq:inherited-surrogate-curvature}.  Substitute
\eqref{eq:glm-terminal-surrogate-rsc} and use
\eqref{eq:glm-terminal-score-block}, \eqref{eq:glm-gstar-bound}, and
\eqref{eq:glm-gdagger-bound}.  This gives \eqref{eq:glm-oracle-rate}; for
the reduced oracle, center at $\betab^\dagger$, use the last assertion of
Lemma~\ref{lem:glm-terminal-score}, and add
$\|\betab^\dagger-\betab^*\|_2\le d_B$.  These terms are precisely those
in \eqref{eq:glm-oracle-radius-rate}.

To prove \eqref{eq:glm-oracle-offscore}, write
$\theta_a=\betab^*$, $\theta_0=\betab^\dagger$, and
$\Delta_e=\widehat\betab_e^{\rm ora}-\theta_e$.  The oracle KKT equations
and \eqref{eq:glm-tangent-projection} give
\begin{equation}
\Ub_e^\T\nabla_{\overline{\mc M}}\widehat{\mc Q}_B(\widehat\betab_e^{\rm ora})=0,
\qquad
\Jb_{\overline{\mc M}^c,\overline{\mc M}}\Ub_e-
\Ab_e\Ub_e^\T\Jb_{\overline{\mc M},\overline{\mc M}}\Ub_e=0.
\label{eq:glm-oracle-cancellation}
\end{equation}
Both $\theta_e$ and $\widehat\betab_e^{\rm ora}$ vanish on $\overline{\mc M}^c$.
Moreover, $\Cb(\Delta_0)_{\mc M}=0$ and $\Ub_a=\Ib_q$, so
\begin{equation*}
(\Delta_e)_{\overline{\mc M}^c}=0,
\qquad
(\Delta_e)_{\overline{\mc M}}=\Ub_e\Ub_e^\T(\Delta_e)_{\overline{\mc M}},
\qquad e=0,a.
\end{equation*}
The fundamental theorem of calculus, \eqref{eq:glm-oracle-cancellation},
and the preceding identities yield
\begin{align*}
\nabla_{\overline{\mc M}^c}\widehat{\mc Q}_B(\widehat\betab_e^{\rm ora})
={}&
\omegab^*_{\overline{\mc M}^c}-\Ab_e\Ub_e^\T\omegab^*_{\overline{\mc M}}
\\
&+
\left[\nabla_{\overline{\mc M}^c}\widehat{\mc Q}_B(\theta_e)-\omegab^*_{{\overline{\mc M}}^c}
-\Ab_e\Ub_e^\T
\{\nabla_{\overline{\mc M}}\widehat{\mc Q}_B(\theta_e)-\omegab^*_{\overline{\mc M}}\}\right]
\\
&+
\int_0^1
\bigg[
\left\{\nabla^2_{\overline{\mc M}^c,\overline{\mc M}}\widehat{\mc Q}_B
(\theta_e+t\Delta_e)-\Jb_{\overline{\mc M}^c,\overline{\mc M}}\right\}
\\
&\hspace{35mm}
-\Ab_e\Ub_e^\T
\left\{\nabla^2_{\overline{\mc M},\overline{\mc M}}\widehat{\mc Q}_B
(\theta_e+t\Delta_e)-\Jb_{\overline{\mc M},\overline{\mc M}}\right\}
\bigg](\Delta_e)_{\overline{\mc M}}\,{\rm d}t.
\end{align*}
By \eqref{eq:glm-gstar-bound}--\eqref{eq:glm-gdagger-bound} and
\eqref{eq:glm-projection-row}, we get
\begin{align*}
\max_{e\in\{0,a\}}
\big\|&\nabla_{\overline{\mc M}^c}\widehat{\mc Q}_B(\theta_e)-\omegab^*_{\overline{\mc M}^c}
-\Ab_e\Ub_e^\T
\{\nabla_{\overline{\mc M}}\widehat{\mc Q}_B(\theta_e)-\omegab^*_{\overline{\mc M}}\}\big\|_\infty
\\
&\le
\left(1+\frac{J_u}{J_l}q^{1/2}\right)
\{2M\phi^*+2Md_B^2+(J_u+m^{1/2}\eta_\infty)d_B\}
\\
&\le
Cq^{1/2}\{M\phi^*+d_B+m^{1/2}\eta_\infty d_B+Md_B^2\}.
\end{align*}
For every $\betab\in\RR^p$, \maineqref{def:q.function} implies
\begin{align*}
\nabla^2\widehat{\mc Q}_B(\betab)-\Jb
&=\!
\frac1{N_B}\sum_{j=1}^{B-1}n_j
\left\{
\nabla^2\widehat{\mc L}^{(j)}(\widehat\betab^{(j)})
-\nabla^2\widehat{\mc L}^{(j)}(\betab^*)
\right\}
+
\frac{n_B}{N_B}
\left\{
\nabla^2\widehat{\mc L}^{(B)}(\betab)
-\nabla^2\widehat{\mc L}^{(B)}(\betab^*)
\right\}
\\
&\qquad+
\frac1{N_B}\sum_{j=1}^Bn_j
\nabla^2\widehat{\mc L}^{(j)}(\betab^*)-\Jb.
\end{align*}
Since $q\ge1$, the definition of $r_B^{\rm G}$ gives
$d_B\le r_B^{\rm G}$.  Equation \eqref{eq:glm-oracle-rate} therefore yields
\begin{equation*}
\max_{e\in\{0,a\}}\|\Delta_e\|_2
\le Cr_B^{\rm G},
\qquad
\max_{e\in\{0,a\}}\sup_{0\le t\le1}
\|\theta_e+t\Delta_e-\betab^*\|_2
\le Cr_B^{\rm G}.
\end{equation*}
By \eqref{eq:glm-t2-localization} and
\eqref{eq:glm-t3-remainders}, $Cr_B^{\rm G}<R_0$ for all sufficiently
large $N_B$.  The Cauchy--Schwarz inequality yields
\begin{equation*}
\frac1{N_B}\sum_{j=1}^{B-1}n_j
\|\widehat\betab^{(j)}-\betab^*\|_2
\le(\phi^*)^{1/2}.
\end{equation*}
The vectors $\theta_e+t\Delta_e-\betab^*$ are supported on $\overline{\mc M}$, and
$q^{1/2}\le l$.  For $\|u\|_2=1$, the vector $\Ub_eu$ has Euclidean norm
one and is supported on $\overline{\mc M}$.  Since $q+1\le k_l$,
and \eqref{eq:glm-sparse-hessian} implies
\begin{equation*}
\max_{|S|\le k_l}
\left\|
\left\{
\frac1{N_B}\sum_{j=1}^Bn_j
\nabla^2\widehat{\mc L}^{(j)}(\betab^*)-\Jb
\right\}_{S,S}
\right\|_2
\le\eta_\infty.
\end{equation*}
Proposition~\ref{prop:glm-initializer},
\eqref{eq:glm-terminal-rlh}, the displayed identity for
$\nabla^2\widehat{\mc Q}_B(\betab)-\Jb$, and the RLH event give
\begin{align*}
&\max_{\substack{e\in\{0,a\},\ k\in \overline{\mc M}^c\\0\le t\le1}}
\left\|
\left\{\nabla^2\widehat{\mc Q}_B(\theta_e+t\Delta_e)-\Jb\right\}_{k,\overline{\mc M}}
\Ub_e
\right\|_2
\le
C\left[
\eta_\infty+M
\left\{(\phi^*)^{1/2}+r_B^{\rm G}\right\}
\right],
\\
&\max_{\substack{e\in\{0,a\}\\0\le t\le1}}
\left\|
\Ub_e^\T
\left\{\nabla^2\widehat{\mc Q}_B(\theta_e+t\Delta_e)-\Jb\right\}_{\overline{\mc M},\overline{\mc M}}
\Ub_e
\right\|_2
\le
C\left[
\eta_\infty+q^{1/2}M
\left\{(\phi^*)^{1/2}+r_B^{\rm G}\right\}
\right].
\end{align*}
For the first inequality, apply the RLH inequality to $\Ub_eu$ and take the maximum
over $k\in \overline{\mc M}^c$.  For the second inequality, use
\begin{equation*}
\|\Ub_e^\T\vb_{\overline{\mc M}}\|_2
\le\|\vb_{\overline{\mc M}}\|_2
\le q^{1/2}\|\vb\|_\infty,
\qquad \vb\in\RR^p.
\end{equation*}
It follows from \eqref{eq:glm-projection-row} that
\begin{align*}
&\max_{e\in\{0,a\}}
\bigg\|
\int_0^1
\bigg[
\left\{\nabla^2_{\overline{\mc M}^c,\overline{\mc M}}\widehat{\mc Q}_B
(\theta_e+t\Delta_e)-\Jb_{\overline{\mc M}^c,\overline{\mc M}}\right\}
\\
&\hspace{35mm}
-\Ab_e\Ub_e^\T
\left\{\nabla^2_{\overline{\mc M},\overline{\mc M}}\widehat{\mc Q}_B
(\theta_e+t\Delta_e)-\Jb_{\overline{\mc M},\overline{\mc M}}\right\}
\bigg](\Delta_e)_{\overline{\mc M}}\,{\rm d}t
\bigg\|_\infty
\\
&\qquad\le
C\left[
\eta_\infty+q^{1/2}M
\{(\phi^*)^{1/2}+r_B^{\rm G}\}
\right]r_B^{\rm G}.
\end{align*}
By \eqref{eq:glm-terminal-phi}, the definition of $r_B^{\rm G}$, and
\eqref{eq:glm-t3-remainders}, we have
\begin{equation*}
q^{1/2}\{M\phi^*+d_B+m^{1/2}\eta_\infty d_B+Md_B^2\}
\le Cr_B^{\rm G}=o(\lambda),
\end{equation*}
and
\begin{align*}
&\left[
\eta_\infty+q^{1/2}M
\{(\phi^*)^{1/2}+r_B^{\rm G}\}
\right]r_B^{\rm G}
\\
&\quad\le
\left[
\eta_\infty+q^{1/2}M
\left\{
\frac{9|\mc A|}{4\kappa^2N_B}
\sum_{j=1}^{B-1}n_j\lambda_j^2
\right\}^{1/2}
\right]r_B^{\rm G}
+q^{1/2}M(r_B^{\rm G})^2
=o(\lambda).
\end{align*}
Moreover, \ref{GLMassm: theorem2 rates revised} and
$4\eta_\infty\le J_l/2$ give
\begin{equation*}
\sqrt{\frac{(J_u+\eta_\infty)
\{\log(2p)+2\log\log(ep)\}}{N_B}}
\le\frac{C}{C_\lambda}\lambda,
\end{equation*}
whereas the lower bound for $\lambda_B$ in
\ref{GLMassm: theorem2 rates revised} and
\eqref{eq:glm-t2-score-scale} yield
\begin{equation*}
\frac{\{\log(2pN_B)+2\log\log(ep)\}^{1/2}
\{\log(2p)+2\log\log(ep)\}}{N_B}
=o(\lambda).
\end{equation*}
Substitution of these bounds and \eqref{eq:glm-projected-score} into the
displayed identity for
$\nabla_{\overline{\mc M}^c}\widehat{\mc Q}_B(\widehat\betab_e^{\rm ora})$ yields
\begin{equation*}
\max_{e\in\{0,a\}}
\|\nabla_{\overline{\mc M}^c}\widehat{\mc Q}_B
(\widehat\betab_e^{\rm ora})\|_\infty
\le
CC_\lambda^{-1}\lambda+o(\lambda)
\le\frac{p_0\lambda}{2},
\end{equation*}
where the last inequality holds for all sufficiently large $N_B$ because
$p_0^{-1}=O(1)$ and $C_\lambda p_0$ is sufficiently large.
This proves \eqref{eq:glm-oracle-offscore}.
\end{proof}

\begin{proposition}\label{prop:glm-strong-oracle}
Under the assumptions of Proposition~\ref{prop:glm-tightening} and
\ref{GLMassm: theorem3 rates revised}, for all sufficiently large $N_B$,
the conclusion of Theorem~\mainref{thm:oracle_recovery} holds with probability
at least $1-2B/p-C\{\log(ep)\}^{-2}$.
\end{proposition}

\begin{proof}
Proposition~\ref{prop:glm-tightening} gives
\eqref{eq:crude-lla-localization}.  Displays \eqref{eq:glm-t2-localization}, \eqref{eq:glm-t3-remainders}, and \eqref{eq:glm-oracle-rate} yield, for $e=0,a$,
\[
\|\widehat\betab_e^{\rm ora}-\betab^*\|_2
\le R_0\wedge\{\lambda/(5\bar\delta\kappa)\}.
\]
Equation \eqref{eq:glm-oracle-offscore} verifies the oracle-score condition,
and \eqref{eq:glm-t3-penalty} provides signal conditions.
Lemma~\ref{lem:intermediate.rate2} remains valid with
\eqref{eq:glm-terminal-surrogate-rsc} in place of
\eqref{eq:inherited-surrogate-curvature}.  Theorem
\mainref{thm:oracle_recovery} therefore applies on the common event.
\end{proof}

\subsubsection{Verification for linear, logistic, and Poisson regression}

Suppose the covariate conditions in
\ref{GLMassm: covariate norm conditions} and
\ref{GLMassm: RSC} hold.  Conditionally on the covariates, we consider the following three classes:
\begin{align*}
&Y_i^{(j)}\sim N((X_i^{(j)})^\T\betab^*,1),
&&b(t)=t^2/2,
&&\text{linear regression},
\\
&Y_i^{(j)}\sim\operatorname{Bernoulli}
\left(\frac{e^{(X_i^{(j)})^\T\betab^*}}
{1+e^{(X_i^{(j)})^\T\betab^*}}\right),
&&b(t)=\log(1+e^t),
&&\text{logistic regression},
\\
&Y_i^{(j)}\sim\operatorname{Poisson}
\{e^{(X_i^{(j)})^\T\betab^*}\},
&&b(t)=e^t,
&&\text{Poisson regression}.
\end{align*}
For Poisson regression, we further assume
\begin{equation}
4\max\left\{1,
\max_{j\le b}\max_{i\le n_j}e^{(X_i^{(j)})^\T\betab^*}
\right\}\le v_0,
\qquad v_0=O(1).
\label{eq:glm-poisson-sample-bound}
\end{equation}
Then Conditions \ref{GLMassm: correct specification},
\ref{GLMassm: residual Bernstein}, and
\ref{GLMassm: self concordance} hold with
\begin{center}
\begin{tabular}{lcccc}
\toprule
Model & $K$ & $M_\varepsilon$ & $v_0$ & $b''(t)$ \\
\midrule
Linear & $0$ & $2$ & $4$ & $1$ \\
Logistic & $1$ & $1$ & $1$ & $e^t/(1+e^t)^2$ \\
Poisson & $1$ & $2\sqrt{v_0}$ & $v_0$ & $e^t$ \\
\bottomrule
\end{tabular}
\end{center}
For the linear model, $\EE(|\varepsilon|^m\mid X)\le m!$ for $m\ge2$,
and hence
\[
4\EE\{e^{|\varepsilon|/2}-1-|\varepsilon|/2\mid X\}
\le4\sum_{m=2}^\infty2^{-m}=2.
\]
For the logistic model, $|\varepsilon|\le1$ and
$\EE(\varepsilon^2\mid X)=b''(X^\T\betab^*)\le1/4$.  Since
$e^z-1-z\le z^2$ for $0\le z\le1$,
\[
\EE\{e^{|\varepsilon|}-1-|\varepsilon|\mid X\}
\le\EE(\varepsilon^2\mid X)\le1/4.
\]
For the Poisson model, with
$M_\varepsilon=2\sqrt{v_0}$,
\[
\log\EE(e^{\pm\varepsilon/M_\varepsilon}\mid X)
=e^{X^\T\betab^*}
\{e^{\pm1/M_\varepsilon}-1\mp1/M_\varepsilon\}.
\]
The inequalities
$e^{\pm z}-1\mp z\le3z^2/5$ for $0\le z\le1/2$ and
$e^z-1\le80z/77$ for $0\le z\le3/80$ yield
\[
M_\varepsilon^2\EE\left[
e^{|\varepsilon|/M_\varepsilon}-1-|\varepsilon|/M_\varepsilon
\,\middle|\,X\right]
\le\frac{24}{77}v_0<\frac{v_0}{2}.
\]
To obtain the last display, use $e^{|z|}-1-|z|
\le(e^z-1-z)+(e^{-z}-1+z)$, $z\in\RR$.
By \eqref{eq:glm-poisson-sample-bound},
$M_\varepsilon^{-1}\le1/4$ and
$e^{X^\T\betab^*}M_\varepsilon^{-2}\le1/16$.  Hence both quantities
$e^{X^\T\betab^*}
\{e^{\pm1/M_\varepsilon}-1\mp1/M_\varepsilon\}$
are at most $3/80$, and
\begin{align*}
M_\varepsilon^2\EE\left[
e^{|\varepsilon|/M_\varepsilon}-1-|\varepsilon|/M_\varepsilon
\,\middle|\,X\right]
&\leq
\frac{80M_\varepsilon^2}{77}e^{X^\T\betab^*}
\left\{
e^{1/M_\varepsilon}-1-M_\varepsilon^{-1}
+e^{-1/M_\varepsilon}-1+M_\varepsilon^{-1}
\right\}\\
&\leq\frac{24}{77}v_0.
\end{align*}
The bounds $2=v_0/2$, $1/4<v_0/2$, and
$24v_0/77<v_0/2$ verify \eqref{eq:glm-residual-Bernstein} for the
linear, logistic, and Poisson constants in the table, respectively.  Moreover, the three choices of
$b$ satisfy, respectively,
\[
b'''(t)=0,
\qquad
|b'''(t)|=b''(t)|1-2e^t/(1+e^t)|\le b''(t),
\qquad
b'''(t)=b''(t)=e^t,
\]
so Condition~\ref{GLMassm: self concordance} holds with
$K=0,1,1$, respectively; see
\citet[Section S5]{jacobson2025strong}. 
In the linear
model, $M_{\rm GLM}=0$ and one may take
$\kappa\le\underline\kappa$; hence the term containing $M_{\rm GLM}$ in
\eqref{eq:glm-online-sample} equals zero.  In the logistic and Poisson
models, one may take
\$
M_{\rm GLM}
=16e^{4c_b|\mc A|^{1/2}R_0}c_bv_0B_1|\mc A|,
\qquad
0<\kappa\le
\underline\kappa e^{-8c_b|\mc A|^{1/2}R_0},
\$
with $v_0=1$ and the $v_0$ in
\eqref{eq:glm-poisson-sample-bound}, respectively.

In addition, \ref{GLMassm: terminal random design revised} and alternative (a) in
\ref{GLMassm: terminal hessian tails revised} hold with
\begin{center}
\begin{tabular}{lccc}
\toprule
Model & $v_\varepsilon$ & $c_\varepsilon$ & $K_H$ \\
\midrule
Linear & $1$ & $0$ & $K_X$ \\
Logistic & $1$ & $1/3$ & $K_X/2$ \\
Poisson & $1$ & $1/3$ & $\sqrt{v_0}K_X/2$ \\
\bottomrule
\end{tabular}
\end{center}
For linear regression, \eqref{eq:glm-score-mgf} is the Gaussian identity.
For logistic regression,
\[
\EE(|\varepsilon|^m\mid X)
\le b''(X^\T\betab^*)
\le\frac{m!}{2}b''(X^\T\betab^*)3^{-(m-2)},
\qquad m\ge2.
\]
For Poisson regression,
\[
\log\EE(e^{t\varepsilon}\mid X)
=b''(X^\T\betab^*)(e^t-1-t)
\le\frac{b''(X^\T\betab^*)t^2}{2(1-|t|/3)},
\qquad |t|<3.
\]
For the logistic model, the first inequality and
$m!\ge2\cdot3^{m-2}$ imply, for $|t|<3$,
\begin{align*}
\log\EE(e^{t\varepsilon}\mid X)
&\le \EE(e^{t\varepsilon}\mid X)-1
\\
&\le\sum_{m=2}^\infty
\frac{|t|^m}{m!}\EE(|\varepsilon|^m\mid X)
\le\frac{b''(X^\T\betab^*)t^2}{2(1-|t|/3)}.
\end{align*}
For the Poisson model, $m!\ge2\cdot3^{m-2}$ gives
\[
e^t-1-t
\le\sum_{m=2}^\infty\frac{|t|^m}{m!}
\le\frac{t^2}{2}\sum_{m=2}^\infty
\left(\frac{|t|}{3}\right)^{m-2}
=\frac{t^2}{2(1-|t|/3)}.
\]
Thus \eqref{eq:glm-score-mgf} holds with the values of
$v_\varepsilon$ and $c_\varepsilon$ in the table.  The values of $K_H$
follow from $b''=1$, $b''\le1/4$, and $b''\le v_0/4$, respectively.

For all three models, \eqref{eq:glm-terminal-scales-a} applies with these
values of $K_H$.  Equation \eqref{eq:glm-terminal-M} gives $M=0$ for the
linear model.  If $B_1,B_2,M_\varepsilon,v_0=O(1)$,
\eqref{eq:glm-C-explicit} permits a constant $C$ independent of $(N_j,p)$;
hence
\begin{equation}
\lambda_j\asymp\sqrt{\frac{\log p}{N_j}},
\qquad j=1,\ldots,B.
\label{eq:glm-canonical-lambda-rate}
\end{equation}
Recall that we have
\begin{equation}
\phi^*
\le
\frac{9C^2|\mc A|\log p}{4\kappa^2N_B}
\left\{1+\log\left(\frac{N_{B-1}}{n_1}\right)\right\}.
\label{eq:glm-canonical-phi-rate}
\end{equation}
Consequently, the first condition in
\eqref{eq:glm-t2-remainders} is implied by
\begin{equation}
\frac{M|\mc A|\log p}{\kappa^2N_B}
\left\{1+\log\left(\frac{N_{B-1}}{n_1}\right)\right\}
=o\left(\sqrt{\frac{\log p}{N_B}}\right).
\label{eq:glm-canonical-history-scaling}
\end{equation}
For linear regression, \eqref{eq:glm-canonical-history-scaling} holds
because $M=0$.  For logistic and Poisson regression, under
$K_XlR_0\{\log(2pN_B)+2\log\log(ep)\}^{1/2}=O(1)$,
\eqref{eq:glm-terminal-scales-a}--\eqref{eq:glm-terminal-M} give
\begin{align}
\eta_\infty
&\le CK_H^2\left\{
\sqrt{\frac{k_l\log(ep/k_l)+2\log\log(ep)}{N_B}}
+\frac{k_l\log(ep/k_l)+2\log\log(ep)}{N_B}
\right\},
\label{eq:glm-canonical-eta-rate}
\\
M
&\le CK_XK_H^2
\{k_l\log(ep/k_l)+\log(2N_B)+2\log\log(ep)\}
\{\log(2pN_B)+2\log\log(ep)\}^{1/2}.
\label{eq:glm-canonical-M-rate}
\end{align}
Substitute \eqref{eq:glm-canonical-eta-rate} for $\eta_\infty$ in
$4\eta_\infty\le J_l/2$ and the second growth condition of
\eqref{eq:glm-t2-remainders}, and substitute
\eqref{eq:glm-canonical-M-rate} for $M$ in
\eqref{eq:glm-t2-remainders}.
Under the null, $d_B=0$, and
\eqref{eq:glm-oracle-radius-rate} and
\eqref{eq:glm-canonical-phi-rate} give
\begin{equation}
\begin{split}
r_B^{\rm G}\le C\bigg[&
\sqrt{\frac{q+2\log\log(ep)}{N_B}}
+\frac{\{q+\log(2N_B)+2\log\log(ep)\}^{1/2}
\{q+2\log\log(ep)\}}{N_B}
\\
&+q^{1/2}\frac{M|\mc A|\log p}{\kappa^2N_B}
\left\{1+\log\left(\frac{N_{B-1}}{n_1}\right)\right\}
\bigg].
\end{split}
\label{eq:glm-canonical-oracle-rate}
\end{equation}
For $d_B>0$, the right-hand side of
\eqref{eq:glm-oracle-radius-rate} contains in addition
$q^{1/2}\{d_B+m^{1/2}\eta_\infty d_B+Md_B^2\}$.
For \ref{GLMassm: theorem3 rates revised}, compare
\eqref{eq:glm-canonical-oracle-rate} under the null, or
\eqref{eq:glm-oracle-radius-rate} when $d_B>0$, with
\eqref{eq:glm-canonical-lambda-rate} to evaluate
$r_B^{\rm G}=o(\lambda_B)$; substitute
\eqref{eq:glm-canonical-phi-rate}--\eqref{eq:glm-canonical-M-rate} into
\eqref{eq:glm-t3-remainders} to evaluate its remaining little-$o$ claims.

Taking the multiplicative constant in
\eqref{eq:glm-canonical-lambda-rate} sufficiently large therefore gives
\begin{equation*}
\lambda_B\ge C_\lambda
\sqrt{\frac{\log(2p)+2\log\log(ep)}{N_B}},
\end{equation*}
as required in \ref{GLMassm: theorem2 rates revised}.  By
\eqref{eq:glm-canonical-lambda-rate}, the minimal signal lower bound in
\eqref{eq:glm-t3-penalty} has order $\sqrt{\log(p)/N_B}$.

\subsubsection{Online Bahadur representation}
\label{sec:glm_bahadur_verification}

Suppose that Assumption~\mainref{assm:local_alternative}, the conditions of
Proposition~\ref{prop:glm-initializer} with $b=B$, and
\ref{GLMassm: terminal random design revised}--\ref{GLMassm: terminal hessian tails revised}
hold.  Let $k_l$, $\eta_\infty$, $M$, and $r_B^{\rm G}$ be given by
\eqref{eq:glm-terminal-dimensions},
\eqref{eq:glm-terminal-scales-a} or \eqref{eq:glm-terminal-scales-b},
\eqref{eq:glm-terminal-M}, and \eqref{eq:glm-oracle-radius-rate},
respectively.  Suppose that
$l\ge4|\mc A|^{1/2}\vee q^{1/2}$,
$4\eta_\infty\le J_l/2$, and
\eqref{eq:glm-terminal-kappa} holds.  
We further require the following condition.

\begin{enumerate}[label=(GLM-\arabic*),start=10,leftmargin=3\parindent]
\item\label{GLMassm: bahadur rates}
As $N_B\to\infty$,
\begin{gather}
\frac{B}{p}+\{\log(ep)\}^{-2}=o(1),
\quad
\|\hb_B\|_2=O\!\left(\sqrt{\frac r{N_B}}\right),
\quad
\frac{r_B^{\rm G}}{R_0}=o(1),
\quad
\frac{r^{7/2}}{N_B}=o(1),
\label{eq:glm-bahadur-probability}
\\
\begin{split}
&(1+M)\sqrt{\frac q{N_B}}
\left\{
\frac{|\mc A|}{\kappa^2}\sum_{j=1}^{B-1}n_j\lambda_j^2
+\frac{|\mc A|^{1/2}}{\kappa}r_B^{\rm G}
  \sum_{j=1}^{B-1}n_j\lambda_j
+n_B(r_B^{\rm G})^2
\right\}
\\
&\qquad\qquad
+(q+\sqrt{N_B})\eta_\infty r_B^{\rm G}=o(1).
\end{split}
\label{eq:glm-bahadur-rate}
\end{gather}
\end{enumerate}

\begin{proposition}\label{prop:glm-bahadur}
Under Assumption~\mainref{assm:local_alternative}, the conditions of Proposition
\ref{prop:glm-initializer} with $b=B$,
\ref{GLMassm: terminal random design revised}--
\ref{GLMassm: terminal hessian tails revised},
$l\ge4|\mc A|^{1/2}\vee q^{1/2}$,
$4\eta_\infty\le J_l/2$, \eqref{eq:glm-terminal-kappa}, and
\ref{GLMassm: bahadur rates},
\begin{align}
\widehat\betab^{\rm ora}_{0,\overline{\mc M}}
-\betab^*_{\overline{\mc M}}
={}&
\Jb_0^{-1/2}(\Ib-\Pb_0)\Jb_0^{-1/2}
\frac1{N_B}\sum_{j=1}^B\sum_{i=1}^{n_j}
\varepsilon_i^{(j)}X_{i,\overline{\mc M}}^{(j)}
-\Jb_0^{-1}\Cb_0^\T\Psi_0^{-1}\hb_B
+o_p(N_B^{-1/2}),
\label{eq:glm-reduced-bahadur}
\\
\widehat\betab^{\rm ora}_{a,\overline{\mc M}}
-\betab^*_{\overline{\mc M}}
={}&\Jb_0^{-1}\frac1{N_B}
\sum_{j=1}^B\sum_{i=1}^{n_j}
\varepsilon_i^{(j)}X_{i,\overline{\mc M}}^{(j)}
+o_p(N_B^{-1/2}).
\label{eq:glm-full-bahadur}
\end{align}
Moreover, Assumptions~\mainref{assm:score} and \mainref{assm:lyapunov} hold with
$K_l=J_l$ and $K_u=J_u$.
\end{proposition}

\begin{proof}
Recall that the conclusions of Proposition~\ref{prop:glm-initializer} and
Lemmas~\ref{lem:glm-terminal-design}--\ref{lem:glm-terminal-score},
hold simultaneously with probability at least
$1-2B/p-C\{\log(ep)\}^{-2}=1-o(1)$.
On this event, Proposition~\ref{prop:glm-initializer} gives
\begin{equation}
\|\widehat\betab^{(j)}-\betab^*\|_2
\le \frac{3}{2\kappa}|\mc A|^{1/2}\lambda_j,
\qquad
\widehat\betab^{(j)}-\betab^*\in\BB(R_0)\cap\CC(l),
\qquad j=1,\ldots,B-1.
\label{eq:glm-bahadur-initializer}
\end{equation}
For the quantities $r_0^{\rm ora}$ and $r_a^{\rm ora}$ in
Lemma~\ref{lem:oracle_rate}, equations
\eqref{eq:glm-terminal-score-block}, \eqref{eq:glm-terminal-phi}, and
\eqref{eq:glm-oracle-radius-rate} give
$r_0^{\rm ora}\vee r_a^{\rm ora}\le Cr_B^{\rm G}$.
Consequently, \eqref{eq:glm-bahadur-probability} gives
$2(r_0^{\rm ora}\vee r_a^{\rm ora})\le R_0$ for all sufficiently large
$N_B$.  Using \eqref{eq:glm-terminal-surrogate-rsc} in place of
\eqref{eq:inherited-surrogate-curvature} in
Lemma~\ref{lem:oracle_rate}, together with
\eqref{eq:glm-terminal-score-block}, \eqref{eq:glm-gstar-bound},
\eqref{eq:glm-gdagger-bound}, and \eqref{eq:glm-terminal-phi}, gives
\begin{equation}
\max_{e\in\{0,a\}}
\|\widehat\betab_e^{\rm ora}-\betab^*\|_2
\le Cr_B^{\rm G}=o(R_0).
\label{eq:glm-bahadur-oracle-localization}
\end{equation}
Since
$(\widehat\betab_e^{\rm ora}-\betab^*)_{\overline{\mc M}^c}=0$ and
$|\overline{\mc M}|=q\le l^2$, we get $\widehat\betab_e^{\rm ora}-\betab^*\in\CC(l)$, for 
$e=0,a$.
Equation \eqref{eq:glm-terminal-rlh} validates
$\bigcap_{j=1}^B\mc E_{\rm rlh}^{(j)}(R_0,l,M)$.
Since $k_l\ge q$, \eqref{eq:glm-sparse-hessian} yields
\begin{equation}
\left\|
\frac1{N_B}\sum_{j=1}^Bn_j
\nabla^2_{\overline{\mc M}\,\overline{\mc M}}
\widehat{\mc L}^{(j)}(\betab^*)-\Jb_0
\right\|_2
\le\eta_\infty.
\label{eq:glm-bahadur-hessian}
\end{equation}
Equations \eqref{eq:glm-sparse-hessian} and
\eqref{eq:glm-sparse-to-cone} also validate $\mc H(\eta_\infty)$.

Lemma~\ref{lem:decomposition},
\eqref{eq:glm-terminal-rlh}, \eqref{eq:glm-bahadur-hessian},
\eqref{eq:glm-score}, and the KKT arguments in
Lemma~\ref{lem:oracle_bahadur} give
\begin{equation}
\label{eq:glm-bahadur-remainder}
\begin{split}
&\max\Bigg\{
\Bigg\|
\widehat\betab^{\rm ora}_{0,\overline{\mc M}}
-\betab^*_{\overline{\mc M}}
-\Jb_0^{-1/2}(\Ib-\Pb_0)\Jb_0^{-1/2}
\frac1{N_B}\sum_{j=1}^B\sum_{i=1}^{n_j}
\varepsilon_i^{(j)}X_{i,\overline{\mc M}}^{(j)}
+\Jb_0^{-1}\Cb_0^\T\Psi_0^{-1}\hb_B
\Bigg\|_2,
\\[-1mm]
&\hspace{23mm}
\Bigg\|
\widehat\betab^{\rm ora}_{a,\overline{\mc M}}
-\betab^*_{\overline{\mc M}}
-\Jb_0^{-1}\frac1{N_B}
\sum_{j=1}^B\sum_{i=1}^{n_j}
\varepsilon_i^{(j)}X_{i,\overline{\mc M}}^{(j)}
\Bigg\|_2
\Bigg\}
\\
&\quad\le C
\frac{M\sqrt q}{N_B}
\bigg[
\sum_{j=1}^{B-1}n_j
\bigg\{
\|\widehat\betab^{(j)}-\betab^*\|_2^2
+\|\widehat\betab^{(j)}-\betab^*\|_2
\max_{e\in\{0,a\}}
\|\widehat\betab_e^{\rm ora}-\betab^*\|_2
\bigg\}
\\[-1mm]
&\hspace{38mm}
+n_B\max_{e\in\{0,a\}}
\|\widehat\betab_e^{\rm ora}-\betab^*\|_2^2
\bigg]
+C\eta_\infty\max_{e\in\{0,a\}}
\|\widehat\betab_e^{\rm ora}-\betab^*\|_2.
\end{split}
\end{equation}
After multiplication by $\sqrt{N_B}$, the right-hand side of
\eqref{eq:glm-bahadur-remainder} is, by
\eqref{eq:glm-bahadur-initializer} and
\eqref{eq:glm-bahadur-oracle-localization}, at most
\begin{align*}
&C
M\sqrt{\frac q{N_B}}
\left\{
\frac{|\mc A|}{\kappa^2}\sum_{j=1}^{B-1}n_j\lambda_j^2
+\frac{|\mc A|^{1/2}}{\kappa}r_B^{\rm G}
 \sum_{j=1}^{B-1}n_j\lambda_j
+n_B(r_B^{\rm G})^2
\right\}
+C\sqrt{N_B}\eta_\infty r_B^{\rm G}
=o(1)
\end{align*}
by \eqref{eq:glm-bahadur-rate}.  This proves
\eqref{eq:glm-reduced-bahadur}--\eqref{eq:glm-full-bahadur}.
In addition, \eqref{eq:glm-bahadur-rate} implies
\begin{equation*}
\sqrt{\frac q{N_B}}
\left\{
\frac{|\mc A|}{\kappa^2}\sum_{j=1}^{B-1}n_j\lambda_j^2
+\frac{|\mc A|^{1/2}}{\kappa}r_B^{\rm G}
 \sum_{j=1}^{B-1}n_j\lambda_j
+n_B(r_B^{\rm G})^2
\right\}
+q\eta_\infty r_B^{\rm G}=o(1),
\end{equation*}
which verifies \maineqref{eq:inference-remainder-condition} by
\eqref{eq:glm-bahadur-initializer} and
\eqref{eq:glm-bahadur-oracle-localization}.

By \eqref{eq:glm-loss} and \eqref{eq:glm-score}, we have
\begin{equation}
\nabla\ell(\betab^*;Y,X)=-\varepsilon X,
\qquad
\EE(\varepsilon^2\mid X)=b''(X^\T\betab^*),
\qquad
\Kb_0=\Jb_0.
\label{eq:glm-bartlett}
\end{equation}
Assumption~\mainref{assm:hessian} and eigenvalue calculation imply
$J_l\le\lambda_{\min}(\Kb_0)
\le\lambda_{\max}(\Kb_0)\le J_u$,
which verifies Assumption~\mainref{assm:score} with $K_l=J_l$ and $K_u=J_u$.

Equation \eqref{eq:glm-score-mgf} gives, for $t>0$,
\begin{equation*}
\PP\!\left(
|\varepsilon|>
\sqrt{2v_\varepsilon b''(X^\T\betab^*)t}+c_\varepsilon t
\,\middle|\,X
\right)\le2e^{-t},
\end{equation*}
and hence
\begin{equation}
\EE(|\varepsilon|^3\mid X)
\le C\left[
\{v_\varepsilon b''(X^\T\betab^*)\}^{3/2}
+c_\varepsilon^3
\right].
\label{eq:glm-third-error-moment}
\end{equation}
Under \eqref{eq:glm-weighted-subgaussian},
\eqref{eq:glm-third-error-moment} and
\eqref{eq:glm-subgaussian-design} imply
\begin{equation*}
\sup_{\|u\|_2=1}\EE|\varepsilon X^\T u|^3
\le C\sup_{\|u\|_2=1}
\left[
\EE|\{b''(X^\T\betab^*)\}^{1/2}X^\T u|^3
+\EE|X^\T u|^3
\right]<\infty.
\end{equation*}
Under \eqref{eq:glm-log-weight-tail},
$\EE\{b''(X^\T\betab^*)^3\}<\infty$, and
\eqref{eq:glm-third-error-moment} and
\eqref{eq:glm-subgaussian-design} yield
\begin{equation*}
\sup_{\|u\|_2=1}\EE|\varepsilon X^\T u|^3
\le C\sup_{\|u\|_2=1}
\left[
\{\EE b''(X^\T\betab^*)^3\}^{1/2}
\{\EE|X^\T u|^6\}^{1/2}
+\EE|X^\T u|^3
\right]<\infty.
\end{equation*}
By Lemma~\ref{lem:V0-conditioning}, every row of $\Omega_0$ has
Euclidean norm at most $K_l^{-1/2}$.  Therefore, 
\begin{align*}
\EE\left\|
\Omega_0\nabla_{\overline{\mc M}}
\ell(\betab^*;Y,X)
\right\|_2^3
&\le r^{1/2}\sum_{j=1}^r
\EE\left|
(\Omega_0)_{j,\cdot}
\nabla_{\overline{\mc M}}\ell(\betab^*;Y,X)
\right|^3
\le Cr^{3/2},
\\
\frac{r^{1/4}}{\sqrt{N_B}}
\EE\left\|
\Omega_0\nabla_{\overline{\mc M}}
\ell(\betab^*;Y,X)
\right\|_2^3
&\le C\sqrt{\frac{r^{7/2}}{N_B}}=o(1).
\end{align*}
By $r^{7/2}/N_B=o(1)$ in \eqref{eq:glm-bahadur-probability}, the second
line is $o(1)$, which verifies Assumption~\mainref{assm:lyapunov}.
\end{proof}

\subsubsection{Chi-square approximation of test statistics}
\label{sec:glm_test_statistic_verification}

Suppose that the conditions of Propositions~\ref{prop:glm-strong-oracle}
and \ref{prop:glm-bahadur} hold and that
\(\lambda_{\max}(\Cb\Cb^\T)=O(1)\).
By \eqref{eq:glm-full-bahadur},
\[
\sqrt{N_B}\,
\Cb_0(\widehat\betab^{\rm ora}_{a,\overline{\mc M}}
-\betab^*_{\overline{\mc M}})
=
\Cb_0\Jb_0^{-1}
\frac1{\sqrt{N_B}}
\sum_{j=1}^B\sum_{i=1}^{n_j}
\varepsilon_i^{(j)}X_{i,\overline{\mc M}}^{(j)}+o_p(1),
\]
and \eqref{eq:glm-bartlett} gives $\Cb_0\Jb_0^{-1}\Kb_0\Jb_0^{-1}\Cb_0^\T=\Vb_0$.
Equations \eqref{eq:glm-full-bahadur} and \eqref{eq:glm-bartlett} show
that the covariance part in the GLM specialization of
\maineqref{def:oracle.wald} is \(\Vb_0\).  
We assume the following condition.

\begin{enumerate}[label=(GLM-\arabic*),start=11,leftmargin=3\parindent]
\item\label{GLMassm: test statistic rates}
As \(N_B\to\infty\),
\begin{gather}
\max_{1\le j\le B}
\sup_{\substack{\|u\|_2=1\\\|u\|_0\le k_l}}
u^\T\nabla^2\widehat{\mc L}^{(j)}(\betab^*)u
=O_p(1),
\label{eq:glm-batch-upper-curvature}
\\
\sqrt r\left[
\frac q{\sqrt{N_B}}
+
\left\{
\frac{T_B|\mc A|}{\kappa^2N_B}
\sum_{j=1}^Bn_j\lambda_j^2
\right\}^{1/2}
\right]=o(1).
\label{eq:glm-test-statistic-rate}
\end{gather}
\end{enumerate}

Under \eqref{eq:glm-weighted-subgaussian}, Bernstein's inequality on a
\(1/4\)-net of the unit sphere for each support of size at most \(k_l\),
followed by a union bound over the at most
\(9^{k_l}(ep/k_l)^{k_l}B\) resulting events, yields
\begin{align*}
&\PP\left[
\max_{1\le j\le B}
\sup_{\substack{\|u\|_2=1\\\|u\|_0\le k_l}}
\left|u^\T
\{\nabla^2\widehat{\mc L}^{(j)}(\betab^*)-\Jb\}u\right|
>CK_H^2
\right]
\le C\{\log(ep)\}^{-2},
\end{align*}
provided
\$
k_l\log(ep/k_l)+\log(2B)+2\log\log(ep)
\lesssim \min_{1\le j\le B}n_j,
\$
which is a standard growth condition in high dimensions.
Since \(\lambda_{\max}(\Jb)\le J_u\), this proves
\eqref{eq:glm-batch-upper-curvature}.  Under
\eqref{eq:glm-log-weight-tail}, \eqref{eq:glm-batch-upper-curvature} is
imposed in \ref{GLMassm: test statistic rates}.

\begin{proposition}\label{prop:glm-test-statistic}
Under Propositions~\ref{prop:glm-strong-oracle} and
\ref{prop:glm-bahadur}, $\lambda_{\max}(\Cb\Cb^\T)=O(1)$, and
\ref{GLMassm: test statistic rates}, condition
\maineqref{eq:feasible-plugin} holds.  Hence, for every \(k\ge K_0+1\) and
\(T=T_W^{(k)}\) or \(T_S^{(k)}\),
\[
\sup_{z\ge0}
\left|
\PP(T\le z)-\PP\{\chi^2(r,e_B)\le z\}
\right|\longrightarrow0.
\]
The conclusions of Corollary~\mainref{cor:test} hold for both statistics.
\end{proposition}

\begin{proof}
Adding and subtracting the Hessians at \(\betab^*\) in
\maineqref{def:q.function}, and using \eqref{eq:glm-bahadur-hessian},
\eqref{eq:glm-terminal-rlh}, \eqref{eq:glm-bahadur-initializer}, and
\eqref{eq:glm-bahadur-oracle-localization}, give, for each $e \in \{e,0\}$,
\begin{equation}
\begin{split}
\left\|
\nabla^2_{\overline{\mc M}\,\overline{\mc M}}
\widehat{\mc Q}_B(\widehat\betab_e^{\rm ora})-\Jb_0
\right\|_2
\le
\eta_\infty
+\frac{CM\sqrt q}{N_B}
\left\{
\frac{|\mc A|^{1/2}}{\kappa}
\sum_{j=1}^{B-1}n_j\lambda_j
+n_Br_B^{\rm G}
\right\}.
\end{split}
\label{eq:glm-test-hessian-bound}
\end{equation}
The RLH event, $q\le l^2$,
\eqref{eq:glm-bahadur-initializer}, and
\eqref{eq:glm-bahadur-oracle-localization} yield
\begin{align*}
&\left\|
\left\{\nabla^2\widehat{\mc L}^{(j)}(\widehat\betab^{(j)})
-\nabla^2\widehat{\mc L}^{(j)}(\betab^*)\right\}_{
\overline{\mc M}\,\overline{\mc M}}
\right\|_2
\le \sqrt q M
\|\widehat\betab^{(j)}-\betab^*\|_2,
\qquad j=1,\ldots,B-1,
\\
&\left\|
\left\{\nabla^2\widehat{\mc L}^{(B)}(\widehat\betab_e^{\rm ora})
-\nabla^2\widehat{\mc L}^{(B)}(\betab^*)\right\}_{
\overline{\mc M}\,\overline{\mc M}}
\right\|_2
\le \sqrt q M
\|\widehat\betab_e^{\rm ora}-\betab^*\|_2.
\end{align*}
Together with \eqref{eq:glm-bahadur-hessian}, these inequalities imply
\eqref{eq:glm-test-hessian-bound}.

Since \(r\le q\) and \eqref{eq:glm-oracle-radius-rate} gives
\(r_B^{\rm G}\ge(q/N_B)^{1/2}\), we get
\begin{align*}
&\sqrt r\left\|
\nabla^2_{\overline{\mc M}\,\overline{\mc M}}
\widehat{\mc Q}_B(\widehat\betab_0^{\rm ora})-\Jb_0
\right\|_2
\\
&\quad\le
(q+\sqrt{N_B})\eta_\infty r_B^{\rm G}
+CM\sqrt{\frac q{N_B}}
\left\{
\frac{|\mc A|^{1/2}}{\kappa}r_B^{\rm G}
\sum_{j=1}^{B-1}n_j\lambda_j
+n_B(r_B^{\rm G})^2
\right\}
=o_p(1)
\end{align*}
by \eqref{eq:glm-bahadur-rate}.  
By the definition of \(\widehat{\Jb}^{(k)}_e\) and Proposition~\ref{prop:glm-strong-oracle}, we have
\begin{equation}
\widehat{\Jb}^{(k)}_e = \nabla^2_{\overline{\mc M}\,\overline{\mc M}}
\widehat{\mc Q}_B(\widehat\betab^{\rm ora}_e),
\qquad k\ge K_0+1.
\label{eq:glm-feasible-common-hessian}
\end{equation}
Consequently, we obtain
\begin{equation}
\sqrt r\max_{e\in\{0,a\}}
\|\widehat{\Jb}^{(k)}_e-\Jb_0\|_2=o_p(1).
\label{eq:glm-test-hessian-rate}
\end{equation}

Define
\[
\widetilde{\Kb}^{(B)}
=\frac1{N_B}\sum_{j=1}^B\sum_{i=1}^{n_j}
(\varepsilon_i^{(j)})^2
X_{i,\overline{\mc M}}^{(j)}
(X_{i,\overline{\mc M}}^{(j)})^\T.
\]
Equation \eqref{eq:glm-bartlett} gives
\(\EE\widetilde{\Kb}^{(B)}=\Kb_0\).  Integration of the conditional
tail bound obtained from \eqref{eq:glm-score-mgf} gives
\begin{equation}
\EE(|\varepsilon|^4\mid X)
\le
C\left[
\{v_\varepsilon b''(X^\T\betab^*)\}^2+c_\varepsilon^4
\right].
\label{eq:glm-fourth-error-moment}
\end{equation}
Under \eqref{eq:glm-weighted-subgaussian}, the Cauchy--Schwarz inequality and the sub-Gaussian bound yield
\[
\sup_{k,l\le p}\EE\left[
\{b''(X^\T\betab^*)\}^2X_k^2X_l^2+X_k^2X_l^2
\right]<\infty.
\]
Under \eqref{eq:glm-log-weight-tail}, the same inequality follows from
\(\EE\{b''(X^\T\betab^*)^4\}<\infty\),
\eqref{eq:glm-subgaussian-design}, and the Cauchy--Schwarz inequality.
Thus
\eqref{eq:glm-fourth-error-moment}, independence, and
\(\|A\|_2\le\|A\|_{\rm F}\) give
\begin{equation}
\begin{split}
\EE\left\|
\widetilde{\Kb}^{(B)}-\Kb_0
\right\|_{\rm F}^2
&=\frac1{N_B^2}
\sum_{j=1}^B\sum_{i=1}^{n_j}
\sum_{k,l\in\overline{\mc M}}
\operatorname{Var}\{(\varepsilon_i^{(j)})^2
X_{ik}^{(j)}X_{il}^{(j)}\}
\le\frac{Cq^2}{N_B},
\\
\left\|
\widetilde{\Kb}^{(B)}-\Kb_0
\right\|_2
&=O_p\!\left(\frac q{\sqrt{N_B}}\right).
\end{split}
\label{eq:glm-truth-score-covariance}
\end{equation}

For \(j\le B\), write $\Delta_j=\widehat\betab^{(j)}-\betab^*$, and
\begin{equation*}
\begin{split}
&d_i^{(j)}
=b'\{(X_i^{(j)})^\T\widehat\betab^{(j)}\}
-b'\{(X_i^{(j)})^\T\betab^*\},\\
&\mathbf D_B
=\frac1{N_B}\sum_{j=1}^B\sum_{i=1}^{n_j}
(d_i^{(j)})^2
X_{i,\overline{\mc M}}^{(j)}
(X_{i,\overline{\mc M}}^{(j)})^\T.
\end{split}
\end{equation*}
Proposition~\ref{prop:glm-initializer} with \(b=B\) gives
\(\Delta_j\in\BB(R_0)\cap\CC(l)\) for \(j\le B\), and \(q\le k_l\)
gives
\(\|X_{i,\overline{\mc M}}^{(j)}\|_2
\le\|X_i^{(j)}\|_{(k_l)}\).
Condition~\ref{GLMassm: self concordance}, the mean-value theorem, and
\eqref{eq:glm-weighted-row-envelope} imply
\begin{equation}
(d_i^{(j)})^2
\|X_{i,\overline{\mc M}}^{(j)}\|_2^2
\le
CT_B b''\{(X_i^{(j)})^\T\betab^*\}
\{(X_i^{(j)})^\T\Delta_j\}^2.
\label{eq:glm-fitted-score-pointwise}
\end{equation}
To obtain \eqref{eq:glm-fitted-score-pointwise}, we use
$|d_i^{(j)}|
\le b''\{(X_i^{(j)})^\T\betab^*\}
e^{K|(X_i^{(j)})^\T\Delta_j|}
|(X_i^{(j)})^\T\Delta_j|$,
and
\[
\begin{split}
\max_{j\le B}\max_{i\le n_j}
K|(X_i^{(j)})^\T\Delta_j|
&\le Kc_B\max_{j\le B}\|\Delta_j\|_1
\le Kc_BlR_0
\\
&\le CK_XlR_0
\{\log(2pN_B)+2\log\log(ep)\}^{1/2}=O(1),
\end{split}
\]
where the last inequality follows from \eqref{eq:glm-row-envelope} and
\ref{GLMassm: theorem2 rates revised}.

Since \(\mathbf D_B\) is positive semidefinite,
\eqref{eq:glm-sparse-to-cone},
\eqref{eq:glm-batch-upper-curvature}, and
Proposition~\ref{prop:glm-initializer} with \(b=B\), together with
\eqref{eq:glm-initializer-rates}, yield
\begin{equation}
\begin{split}
\|\mathbf D_B\|_2
\le\operatorname{tr}(\mathbf D_B)
\le
\frac{CT_B}{N_B}
\sum_{j=1}^Bn_j
\Delta_j^\T
\nabla^2\widehat{\mc L}^{(j)}(\betab^*)\Delta_j
=O_p\!\left(
\frac{T_B|\mc A|}{\kappa^2N_B}
\sum_{j=1}^Bn_j\lambda_j^2
\right).
\end{split}
\label{eq:glm-fitted-score-quadratic}
\end{equation}
Since
\(
\nabla\ell(\widehat\betab^{(j)};Y_i^{(j)},X_i^{(j)})
=(d_i^{(j)}-\varepsilon_i^{(j)})X_i^{(j)}
\), the Cauchy--Schwarz inequality applied to each quadratic form gives
\begin{equation}
\left\|
\widehat{\Kb}^{(B)}_{\overline{\mc M}\,\overline{\mc M}}
-\widetilde{\Kb}^{(B)}
\right\|_2
\le
2\|\widetilde{\Kb}^{(B)}\|_2^{1/2}
\|\mathbf D_B\|_2^{1/2}
+\|\mathbf D_B\|_2.
\label{eq:glm-fitted-score-covariance}
\end{equation}
Equations \eqref{eq:glm-test-statistic-rate},
\eqref{eq:glm-truth-score-covariance}, and
\eqref{eq:glm-fitted-score-quadratic}--
\eqref{eq:glm-fitted-score-covariance} imply
\[
\|\widetilde{\Kb}^{(B)}\|_2=O_p(1),
\qquad
\sqrt r\,\|\mathbf D_B\|_2^{1/2}=o_p(1),
\qquad
\sqrt r\,\|\mathbf D_B\|_2=o_p(1),
\]
where the first relation also uses \(\|\Kb_0\|_2\le K_u\), and the third
follows from the second because \(r\ge1\).  Therefore, we obtain
\begin{equation}
\sqrt r\left\|
\widehat{\Kb}^{(B)}_{\overline{\mc M}\,\overline{\mc M}}-\Kb_0
\right\|_2=o_p(1).
\label{eq:glm-test-score-covariance-rate}
\end{equation}
Equations \eqref{eq:glm-test-hessian-rate} and
\eqref{eq:glm-test-score-covariance-rate} verify
\maineqref{eq:feasible-plugin}.  Proposition~\ref{prop:glm-bahadur} verifies
the expansions and Assumptions~\mainref{assm:score} and
\mainref{assm:lyapunov} used in Theorem~\mainref{thm:oracle.test}, while
\eqref{eq:glm-bahadur-remainder} and \eqref{eq:glm-bahadur-rate} give the
\(o_p(N_B^{-1/2})\) bound required in the proof of  Theorem~\mainref{thm:oracle.test}.  
The same operator-norm Hessian bound also controls the reduced-oracle score-gradient remainder in that proof, so the sharper GLM condition applies to both statistics.
Proposition \ref{prop:glm-strong-oracle} gives the equalities required in
\eqref{eq:feasible-oracle-plugin-identification}, with its Hessian equality given by \eqref{eq:glm-feasible-common-hessian}. Therefore, the proof of Theorem~\mainref{thm:test} gives the claimed uniform approximation.

Under \(\hb_B=\mathbf0\),
\(\PP\{\chi^2(r)>\chi_\alpha^2(r)\}=\alpha\).  Under the local
alternatives in Corollary~\mainref{cor:test}, the uniform approximation implies
\[
\left|
\PP\{T>\chi_\alpha^2(r)\}
-\PP\{\chi^2(r,e_B)>\chi_\alpha^2(r)\}
\right|\longrightarrow0.
\]
\end{proof}
\noindent We note that, if \eqref{eq:glm-canonical-lambda-rate} holds, then
Lemma~\ref{lem:batch.size} gives
\[
\frac1{N_B}\sum_{j=1}^Bn_j\lambda_j^2
\lesssim
\frac{\log p}{N_B}
\left\{1+\log\left(\frac{N_B}{n_1}\right)\right\}.
\]
Consequently, \eqref{eq:glm-test-statistic-rate} is implied by
\[
\sqrt r\left[
\frac q{\sqrt{N_B}}
+\left\{
\frac{T_B|\mc A|\log p}{\kappa^2N_B}
\left\{1+\log\left(\frac{N_B}{n_1}\right)\right\}
\right\}^{1/2}
\right]=o(1).
\]

\subsubsection{Proof of the claims in Remark~\mainref{remark:ls no info lost}}\label{sec:ls no info lost proof}

Consider the linear model \(Y_i=X_i^\top\betab^*+\epsilon_i\) with least-squares loss \(\ell(\betab;Y_i,X_i)=\frac12(Y_i-X_i^\top\betab)^2\). We first examine the conventional renewable loss
\[
\widetilde{\mc Q}_b(\betab)=\frac{1}{N_b}\left\{n_b\widehat{\mathcal L}^{(b)}(\betab)+\frac12(\betab-\widehat{\betab}^{(b-1)})^\top\widehat{\Jb}^{(b-1)}(\betab-\widehat{\betab}^{(b-1)})\right\},
\]
where \(\widehat{\Jb}^{(b-1)}=\sum_{i\in\mc D_{b-1}^\star}X_iX_i^\top\). Suppose that the renewable Lasso estimator at batch \(b\) is defined by
\[
\widehat{\betab}^{(b)}\in\arg\min_{\betab}\left\{\widetilde{\mc Q}_b(\betab)+\lambda_b\|\betab\|_1\right\}.
\]

Because the least-squares gradient is affine, that is, $ \nabla\ell(\betab;Y_i,X_i)=-X_iY_i+X_iX_i^\top\betab$, we have
\#
\widehat{\Jb}^{(b-1)}(\betab-\widehat{\betab}^{(b-1)}) = \sum_{i\in\mc D_{b-1}^\star} \Big\{ \nabla\ell(\betab;Y_i,X_i) - \nabla\ell(\widehat{\betab}^{(b-1)};Y_i,X_i) \Big\}. \label{affine}
\#
Therefore,
\[
\nabla\widetilde{\mc Q}_b(\betab)
=
\frac{1}{N_b}\sum_{i\in\mc D_b^\star}\nabla\ell(\betab;Y_i,X_i)
-
\frac{1}{N_b}\sum_{i\in\mc D_{b-1}^\star}
\nabla\ell(\widehat{\betab}^{(b-1)};Y_i,X_i).
\]
Define the discrepancy
\[
\Delta_b
=
\nabla\widetilde{\mc Q}_b(\betab)
-
\frac{1}{N_b}\sum_{i\in\mc D_b^\star}\nabla\ell(\betab;Y_i,X_i).
\]
By the preceding identity, \(\Delta_b\) does not depend on \(\betab\) and satisfies
\[
\Delta_b
=
-\frac{1}{N_b}\sum_{i\in\mc D_{b-1}^\star}
\nabla\ell(\widehat{\betab}^{(b-1)};Y_i,X_i).
\]

Now, the Karush--Kuhn--Tucker (KKT) conditions for \(\widehat{\betab}^{(b-1)}\) give
\[
\nabla\widetilde{\mc Q}_{b-1}(\widehat{\betab}^{(b-1)})
+
\lambda_{b-1}\xib^{(b-1)}
=
0,
\qquad
\xib^{(b-1)}
\in
\partial\|\widehat{\betab}^{(b-1)}\|_1.
\]
By the definition of \(\Delta_{b-1}\),
\[
\frac{1}{N_{b-1}}
\sum_{i\in\mc D_{b-1}^\star}
\nabla\ell(\widehat{\betab}^{(b-1)};Y_i,X_i)
=
\nabla\widetilde{\mc Q}_{b-1}(\widehat{\betab}^{(b-1)})
-
\Delta_{b-1}
=
-\lambda_{b-1}\xib^{(b-1)}-\Delta_{b-1}.
\]
Substituting this expression into the preceding display yields
\[
\Delta_b
=
\frac{N_{b-1}}{N_b}
\left\{
\lambda_{b-1}\xib^{(b-1)}
+
\Delta_{b-1}
\right\}.
\]
Since \(\Delta_1=0\), induction gives
\[
\Delta_b
=
\frac{1}{N_b}
\sum_{k=1}^{b-1}N_k\lambda_k\xib^{(k)}.
\]
Consequently,
\[
\nabla\widetilde{\mc Q}_b(\betab)
=
\frac{1}{N_b}\sum_{i\in\mc D_b^\star}\nabla\ell(\betab;Y_i,X_i)
+
\frac{1}{N_b}\sum_{k=1}^{b-1}N_k\lambda_k\xib^{(k)}.
\]
Thus, the conventional renewable score differs exactly from the empirical score based on all observations through batch \(b\) by the accumulated first-order term \(N_b^{-1}\sum_{k=1}^{b-1}N_k\lambda_k\xib^{(k)}\).

At the true parameter, $N_b^{-1}\sum_{i\in\mc D_b^\star}\nabla\ell(\betab^*;Y_i,X_i)=-N_b^{-1}\sum_{i\in\mc D_b^\star}X_i\epsilon_i$. 
Hence,
\[
\sqrt{N_b}\nabla\widetilde{\mc Q}_b(\betab^*)=-\frac{1}{\sqrt{N_b}}\sum_{i\in\mc D_b^\star}X_i\epsilon_i+\frac{1}{\sqrt{N_b}}\sum_{k=1}^{b-1}N_k\lambda_k\xib^{(k)}.
\]
In particular, when \(b=2\),
\[
\sqrt{N_2}\nabla\widetilde{\mc Q}_2(\betab^*)
=
-\frac{1}{\sqrt{N_2}}\sum_{i\in\mc D_2^\star}X_i\epsilon_i
+
\frac{n_1\lambda_1}{\sqrt{N_2}}\xib^{(1)}.
\]
If \(\widehat{\betab}^{(1)}\neq\mathbf 0\), then \(\|\xib^{(1)}\|_\infty=1\), and hence
\[
\left\|
\sqrt{N_2}\nabla\widetilde{\mc Q}_2(\betab^*)
+
\frac{1}{\sqrt{N_2}}\sum_{i\in\mc D_2^\star}X_i\epsilon_i
\right\|_\infty
=
\frac{n_1\lambda_1}{\sqrt{N_2}}.
\]
For the usual choice \(\lambda_1\asymp\sqrt{\log p/n_1}\), this discrepancy is of order
\[
\sqrt{\frac{n_1\log p}{N_2}}.
\]
Therefore, first-order equivalence between the conventional renewable score and the full-data empirical score requires
\[
\frac{n_1\log p}{N_2}\to0.
\]
Since \(N_2=n_1+n_2\), when \(\log p\to\infty\), this requires \(n_2\gg n_1\log p\). Thus, even in the linear model, the conventional quadratic surrogate introduces a first-order distortion at the true parameter unless the incoming batch is sufficiently large.

The corrected surrogate removes this distortion exactly. We show by induction that, for every \(b\geq1\),
\[
\widehat{\mathcal Q}_b(\betab)
=
\frac{1}{N_b}\sum_{i\in\mc D_b^\star}\ell(\betab;Y_i,X_i)+C_b,
\]
where \(C_b\) does not depend on \(\betab\). The claim holds for \(b=1\) by definition. Suppose that it holds for \(b=B-1\). Then
\[
N_{B-1}\nabla\widehat{\mathcal Q}_{B-1}
(\widehat{\betab}^{(B-1)})
=
\sum_{i\in\mc D_{B-1}^\star}
\nabla\ell(\widehat{\betab}^{(B-1)};Y_i,X_i).
\]
By the definition of the corrected surrogate,
\[
N_B\nabla\widehat{\mathcal Q}_B(\betab) = \sum_{i\in\mc D_B}\nabla\ell(\betab;Y_i,X_i) + \widehat{\Jb}^{(B-1)} (\betab-\widehat{\betab}^{(B-1)}) +  N_{B-1}\nabla\widehat{\mathcal Q}_{B-1}
(\widehat{\betab}^{(B-1)}).
\]
Because the least-squares gradient is affine as in \eqref{affine}, combining this with the preceding two displays gives
\[
\nabla\widehat{\mathcal Q}_B(\betab)
=
\frac{1}{N_B}\sum_{i\in\mc D_B^\star}\nabla\ell(\betab;Y_i,X_i).
\]
The two objectives, therefore, differ only by an additive constant independent of \(\betab\), proving
\[
\widehat{\mathcal Q}_B(\betab)
=
\frac{1}{N_B}\sum_{i\in\mc D_B^\star}\ell(\betab;Y_i,X_i)+C_B.
\]
The result follows by induction.

Thus, for least squares, the contrast between the two renewable constructions is exact. That is, the conventional quadratic surrogate retains an accumulated first-order distortion generated by the penalized updates, whereas the corrected surrogate reproduces the full-data empirical loss up to an irrelevant additive constant. In particular, the corrected renewable update loses no estimation- or inference-relevant information.

% In this section we prove the claim made in Remark \mainref{remark:ls no info lost}.
% \begin{proof}
% A simple proof by induction suffices. The relation
% $
% \widehat{\mathcal{Q}}_{b}(\betab) = \frac{1}{N_b}\sum_{i \in \mc D_b^*} \ell(\betab; Y_i, X_i) + C
% $
% holds for $b = 1$ by definition. Suppose it holds for $b = B-1$. Observe that $\ell(\betab; Y_i, X_i) = \frac{1}{2}Y_i^2 - Y_iX_i^T\betab + \frac{1}{2}\betab^TX_iX_i^T\betab$, $\nabla  \ell(\betab; Y_i, X_i) = -X_iY_i + X_iX_i^T\betab$, and $\hat{\Jb}^{(b)} = \sum_{j = 1}^b \sum_{i \in D_j} X_i X_i^T = \sum_{i \in D_b^*} X_i X_i^T $. By assumption, $N_{B-1}\nabla\widehat{\mathcal{Q}}_{B-1}(\hat \betab^{(B-1)}) = -\sum_{i \in D_{B-1}^*}X_iY_i +\sum_{i \in D_{B-1}^*} X_iX_i^T\hat \betab^{(B-1)} $.
% The corrected objective in batch $B$ therefore simplifies to
% \begin{align*}
%     \widehat{\mathcal{Q}}_{B}(\betab) 
%     & = \frac{1}{N_B} \bigg\{  \frac{1}{2}\sum_{i \in \mc D_B} Y_i^2 -\sum_{i \in \mc D_B} Y_iX_i^T\betab + \frac{1}{2}\betab^T\sum_{i \in \mc D_B}X_iX_i^T\betab \\
%     & \hspace{0.8cm} + \frac{1}{2} \betab^T \hat{\Jb}^{(B-1)} \betab - (\hat{\betab}^{(B-1)})^T\hat{\Jb}^{(B-1)}\betab + \frac{1}{2} (\hat{\betab}^{(B-1)} )^T\hat{\Jb}^{(B-1)} \hat{\betab}^{(B-1)}\\
%     & \hspace{0.8cm} + (-\sum_{i \in \mc D_{B-1}^*}X_iY_i  + \sum_{i \in \mc D_{B-1}^*}X_iX_i^T \hat{\betab}^{(B-1)})^T\betab\bigg\}\\
%     & = \frac{1}{N_B} \bigg\{ -\sum_{i \in \mc D_B^*} Y_iX_i^T\betab + \frac{1}{2}\betab^T\sum_{i \in \mc D_B^*}X_iX_i^T\betab + C\bigg \},
% \end{align*}
% establishing the relation for $b = B$.
% \end{proof}

\subsubsection{Additional simulation I-1: Estimation errors}\label{sec: estimation error simulation}

We follow the data-generating process in Section~\mainref{sec:numerical}. We take $n_b=40$, $B=50$, and $p\in\{500,1000\}$, giving a total sample size of $2000$. The coefficient vector is $\boldsymbol\beta^*=(c,c,c,c,c,\mathbf 0_{p-5}^{\top})^{\top}$, with $c=1$, $0.5$, and $0.2$ for the linear, logistic, and Poisson models, respectively. We set the true intercept to zero. All comparisons use the same $200$ independent replications within each setting.

We compare four estimators: our corrected online Lasso with accumulated one-step-ahead prediction losses and the one-standard-error rule (1-SE), corrected online Lasso with 10-fold cross-validation loss on the current batch (CV), the online Lasso proposed by \citet{luo2023online} (Luo), and pooled offline Lasso benchmark with 5-fold CV. 
The tuning parameter selection for our procedures is illustrated in Section~\mainref{sec: tuning parameter selection} above.

Figure~\ref{fig:glm-estimation-mean-log} displays the average log-error measure, $\log\|\widehat{\boldsymbol\beta}^{(b)}-\boldsymbol\beta^*\|_2^2$, against the batch index. 
The pooled Lasso has the smallest mean log error in all six settings. 
The corrected online Lasso with accumulated 1-SE tuning has smaller final-batch mean log error than its current-batch CV counterpart in the linear and Poisson models, whereas the ordering is reversed in the logistic model. 
Both corrected online procedures have smaller final-batch mean log error than {\tt Luo} across the six settings.
Thus, beyond enabling valid inference under more relaxed conditions, the corrected loss can also improve estimation accuracy.

\begin{figure}[p]
\centering
\includegraphics[width=\textwidth]{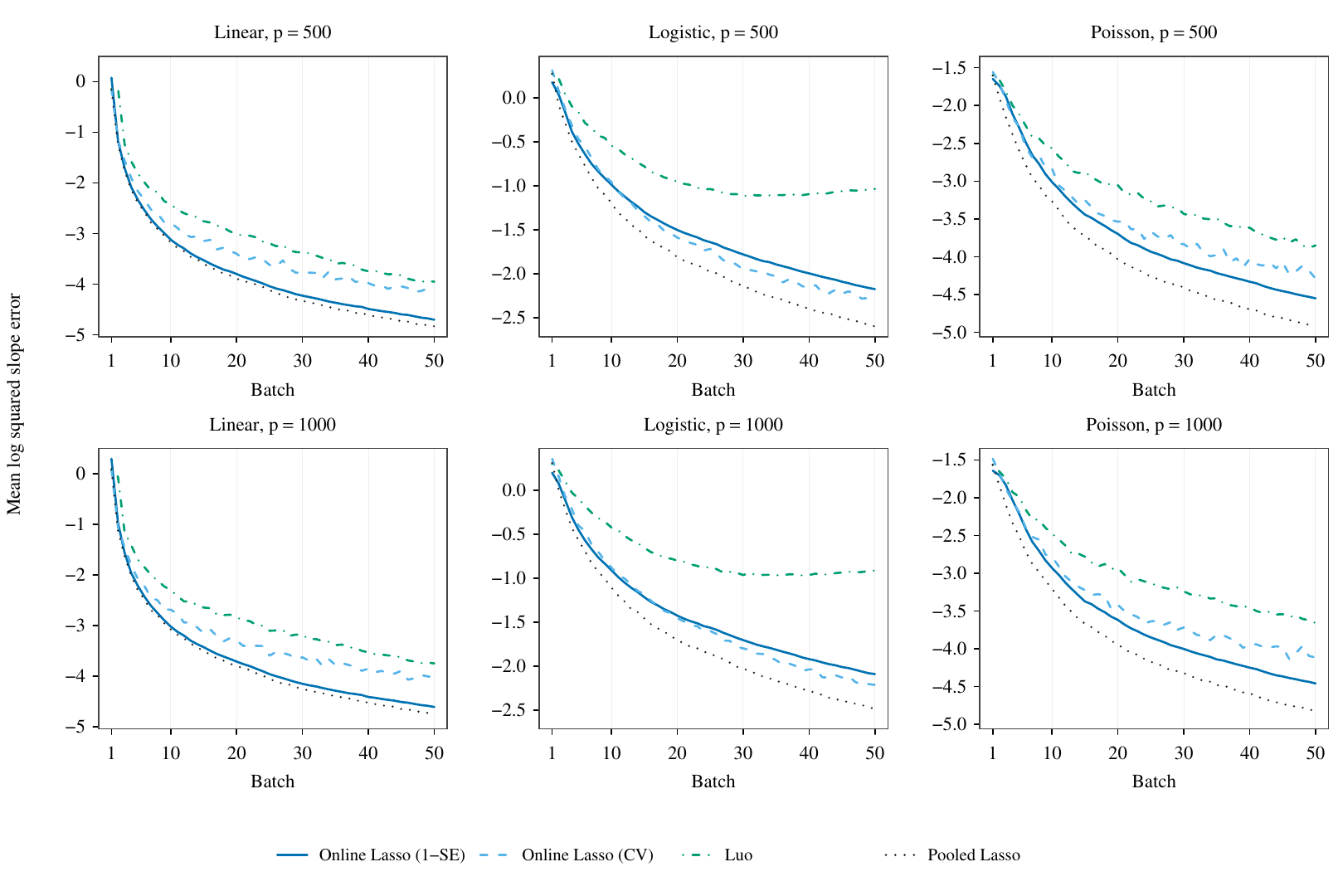}
\caption{Mean log-squared regression coefficient error versus batch index for the four estimators. Rows correspond to $p=500$ and $1000$; columns correspond to linear, logistic, and Poisson regression. Each curve averages $200$ paired replications with $n_b=40$ and $B=50$.}
\label{fig:glm-estimation-mean-log}
\end{figure}

\subsubsection{Additional simulation I-2: Inference for the linear model}
\label{sec:linear_additional_simulation_I}

In this section, we examine the finite-sample inference performance of the renewable Wald and score
tests in the linear model. 
We generate the covariate vectors as in Section~\mainref{sec: simulation}.
We set $n_b=100$,
$B=8$, and $p\in\{400,1200\}$, yielding $N_B=800$, and take
\[
\betab^*=(h,0,0,1+h,1,\mathbf{0}_{p-5}^\T)^\T,
\qquad h\in\{0,0.025,0.05,0.075,0.1,0.2,0.4\}.
\]
Here $\betab^*$ denotes the slope vector. For each setting, we use 500 independent replications and test the following five hypotheses:
\begin{align*}
H_{0,\mathrm L}^{(1)} &: \ \beta_1^*=0, &
H_{0,\mathrm L}^{(2)} &: \ \beta_4^*-\beta_5^*=0,\\
H_{0,\mathrm L}^{(3)} &: \ (\beta_1^*,\beta_8^*,\beta_9^*)^\T=\mathbf{0}_3, &
H_{0,\mathrm L}^{(4)} &: \ \beta_1^*=0,\ \beta_4^*-\beta_5^*=0,\\
H_{0,\mathrm L}^{(5)} &: \ (\beta_1^*,\beta_8^*,\beta_9^*)^\T=\mathbf{0}_3,
\ \beta_4^*-\beta_5^*=0. &&
\end{align*}
Again, the corresponding degrees of freedom are $1,1,3,2$, and $4$.
The methods, tuning procedures, LLA reweighting mechanism, competitor setups, and nominal level are those in
Section~\mainref{sec:numerical}.

Tables~\ref{tab:linear-additional-I-p400} and~\ref{tab:linear-additional-I-p1200}
report the empirical rejection probabilities. 
The results are similar across the two dimensions.
The null rejection probabilities
are close to the nominal level for both dimensions. 
The online Wald tests closely track the pooled benchmark under both estimation procedures and tuning rules, while exhibiting higher power than their score-test counterparts.
Again, our procedures showed greater power than {\tt Luo} against local alternatives in $H_{0,\mathrm L}^{(1)}$.

\begin{landscape}
\begin{table}[p]
\centering\fontsize{9}{9}\selectfont
\setlength{\tabcolsep}{5pt}
\renewcommand{\arraystretch}{1.02}
\captionsetup{font=small}
\caption{Empirical rejection probabilities for the linear model with $p=400$, $n_b=100$, $B=8$, and 500 replications per setting. The nominal level is $0.05$. $T_W$ and $T_S$ denote Wald and score tests. The {\tt Offline} and {\tt Luo} columns report Wald tests; {\tt Luo} is evaluated for the first hypothesis only.}
\label{tab:linear-additional-I-p400}
\begin{tabular*}{\linewidth}{@{\extracolsep{\fill}}llrr*{8}{r}@{}}
\toprule
 & & & & \multicolumn{4}{c}{1-SE} & \multicolumn{4}{c}{CV/GIC}\\
\cmidrule(lr){5-8}\cmidrule(lr){9-12}
 & & & & \multicolumn{2}{c}{Lasso $\to$ LLA} & \multicolumn{2}{c}{LLA $\to$ LLA} & \multicolumn{2}{c}{Lasso $\to$ LLA} & \multicolumn{2}{c}{LLA $\to$ LLA}\\
\cmidrule(lr){5-6}\cmidrule(lr){7-8}\cmidrule(lr){9-10}\cmidrule(lr){11-12}
Hypothesis & $h$ & {\tt Offline} & {\tt Luo} & $T_W$ & $T_S$ & $T_W$ & $T_S$ & $T_W$ & $T_S$ & $T_W$ & $T_S$\\
\midrule
\multirow{7}{*}{$H_{0,\mathrm L}^{(1)}$} & 0.000 & 0.054 & 0.048 & 0.054 & 0.056 & 0.054 & 0.054 & 0.054 & 0.054 & 0.054 & 0.054\\
 & 0.025 & 0.106 & 0.070 & 0.106 & 0.112 & 0.106 & 0.108 & 0.106 & 0.106 & 0.106 & 0.106\\
 & 0.050 & 0.338 & 0.218 & 0.338 & 0.332 & 0.338 & 0.332 & 0.336 & 0.336 & 0.336 & 0.336\\
 & 0.075 & 0.596 & 0.442 & 0.596 & 0.588 & 0.596 & 0.598 & 0.596 & 0.596 & 0.596 & 0.596\\
 & 0.100 & 0.822 & 0.676 & 0.822 & 0.820 & 0.822 & 0.818 & 0.822 & 0.822 & 0.822 & 0.822\\
 & 0.200 & 1.000 & 0.996 & 1.000 & 1.000 & 1.000 & 1.000 & 1.000 & 1.000 & 1.000 & 1.000\\
 & 0.400 & 1.000 & 1.000 & 1.000 & 1.000 & 1.000 & 1.000 & 1.000 & 1.000 & 1.000 & 1.000\\
\addlinespace
\multirow{7}{*}{$H_{0,\mathrm L}^{(2)}$} & 0.000 & 0.054 & --- & 0.052 & 0.054 & 0.054 & 0.054 & 0.054 & 0.054 & 0.054 & 0.054\\
 & 0.025 & 0.068 & --- & 0.068 & 0.066 & 0.068 & 0.068 & 0.068 & 0.068 & 0.068 & 0.068\\
 & 0.050 & 0.102 & --- & 0.102 & 0.106 & 0.102 & 0.102 & 0.102 & 0.102 & 0.102 & 0.102\\
 & 0.075 & 0.172 & --- & 0.172 & 0.162 & 0.172 & 0.170 & 0.172 & 0.172 & 0.172 & 0.172\\
 & 0.100 & 0.298 & --- & 0.296 & 0.276 & 0.298 & 0.280 & 0.298 & 0.298 & 0.298 & 0.298\\
 & 0.200 & 0.860 & --- & 0.858 & 0.798 & 0.862 & 0.798 & 0.852 & 0.810 & 0.852 & 0.810\\
 & 0.400 & 0.998 & --- & 0.998 & 0.998 & 0.998 & 0.998 & 0.998 & 0.998 & 0.998 & 0.998\\
\addlinespace
\multirow{7}{*}{$H_{0,\mathrm L}^{(3)}$} & 0.000 & 0.052 & --- & 0.052 & 0.052 & 0.052 & 0.052 & 0.052 & 0.052 & 0.052 & 0.052\\
 & 0.025 & 0.076 & --- & 0.076 & 0.080 & 0.076 & 0.078 & 0.076 & 0.076 & 0.076 & 0.076\\
 & 0.050 & 0.218 & --- & 0.218 & 0.214 & 0.218 & 0.218 & 0.218 & 0.218 & 0.218 & 0.218\\
 & 0.075 & 0.430 & --- & 0.430 & 0.418 & 0.428 & 0.430 & 0.430 & 0.430 & 0.430 & 0.430\\
 & 0.100 & 0.660 & --- & 0.660 & 0.652 & 0.660 & 0.658 & 0.660 & 0.660 & 0.660 & 0.660\\
 & 0.200 & 1.000 & --- & 1.000 & 0.998 & 1.000 & 0.998 & 1.000 & 1.000 & 1.000 & 1.000\\
 & 0.400 & 1.000 & --- & 1.000 & 1.000 & 1.000 & 1.000 & 1.000 & 1.000 & 1.000 & 1.000\\
\addlinespace
\multirow{7}{*}{$H_{0,\mathrm L}^{(4)}$} & 0.000 & 0.054 & --- & 0.056 & 0.054 & 0.054 & 0.056 & 0.054 & 0.054 & 0.054 & 0.054\\
 & 0.025 & 0.092 & --- & 0.092 & 0.096 & 0.092 & 0.090 & 0.092 & 0.092 & 0.092 & 0.092\\
 & 0.050 & 0.300 & --- & 0.300 & 0.296 & 0.300 & 0.298 & 0.300 & 0.300 & 0.300 & 0.300\\
 & 0.075 & 0.584 & --- & 0.586 & 0.572 & 0.584 & 0.578 & 0.584 & 0.584 & 0.584 & 0.584\\
 & 0.100 & 0.828 & --- & 0.828 & 0.822 & 0.828 & 0.824 & 0.828 & 0.828 & 0.828 & 0.828\\
 & 0.200 & 1.000 & --- & 1.000 & 1.000 & 1.000 & 1.000 & 1.000 & 1.000 & 1.000 & 1.000\\
 & 0.400 & 1.000 & --- & 1.000 & 1.000 & 1.000 & 1.000 & 1.000 & 1.000 & 1.000 & 1.000\\
\addlinespace
\multirow{7}{*}{$H_{0,\mathrm L}^{(5)}$} & 0.000 & 0.054 & --- & 0.054 & 0.054 & 0.054 & 0.052 & 0.054 & 0.054 & 0.054 & 0.054\\
 & 0.025 & 0.082 & --- & 0.082 & 0.082 & 0.082 & 0.082 & 0.082 & 0.082 & 0.082 & 0.082\\
 & 0.050 & 0.228 & --- & 0.228 & 0.228 & 0.228 & 0.230 & 0.228 & 0.228 & 0.228 & 0.228\\
 & 0.075 & 0.480 & --- & 0.478 & 0.466 & 0.478 & 0.472 & 0.480 & 0.480 & 0.480 & 0.480\\
 & 0.100 & 0.746 & --- & 0.746 & 0.728 & 0.746 & 0.738 & 0.746 & 0.746 & 0.746 & 0.746\\
 & 0.200 & 1.000 & --- & 1.000 & 1.000 & 1.000 & 1.000 & 1.000 & 1.000 & 1.000 & 1.000\\
 & 0.400 & 1.000 & --- & 1.000 & 1.000 & 1.000 & 1.000 & 1.000 & 1.000 & 1.000 & 1.000\\
\bottomrule
\end{tabular*}
\end{table}
\end{landscape}

\begin{landscape}
\begin{table}[p]
\centering\fontsize{9}{9}\selectfont
\setlength{\tabcolsep}{5pt}
\renewcommand{\arraystretch}{1.02}
\captionsetup{font=small}
\caption{Empirical rejection probabilities for the linear model with $p=1200$, $n_b=100$, $B=8$, and 500 replications per departure.}
\label{tab:linear-additional-I-p1200}
\begin{tabular*}{\linewidth}{@{\extracolsep{\fill}}llrr*{8}{r}@{}}
\toprule
 & & & & \multicolumn{4}{c}{1-SE} & \multicolumn{4}{c}{CV/GIC}\\
\cmidrule(lr){5-8}\cmidrule(lr){9-12}
 & & & & \multicolumn{2}{c}{Lasso $\to$ LLA} & \multicolumn{2}{c}{LLA $\to$ LLA} & \multicolumn{2}{c}{Lasso $\to$ LLA} & \multicolumn{2}{c}{LLA $\to$ LLA}\\
\cmidrule(lr){5-6}\cmidrule(lr){7-8}\cmidrule(lr){9-10}\cmidrule(lr){11-12}
Hypothesis & $h$ & Offline & Luo & $T_W$ & $T_S$ & $T_W$ & $T_S$ & $T_W$ & $T_S$ & $T_W$ & $T_S$\\
\midrule
\multirow{7}{*}{$H_{0,\mathrm L}^{(1)}$} & 0.000 & 0.054 & 0.050 & 0.054 & 0.056 & 0.054 & 0.056 & 0.054 & 0.054 & 0.054 & 0.054\\
 & 0.025 & 0.106 & 0.084 & 0.106 & 0.110 & 0.106 & 0.110 & 0.106 & 0.106 & 0.106 & 0.106\\
 & 0.050 & 0.338 & 0.224 & 0.336 & 0.334 & 0.336 & 0.334 & 0.336 & 0.336 & 0.336 & 0.336\\
 & 0.075 & 0.596 & 0.420 & 0.596 & 0.592 & 0.596 & 0.596 & 0.596 & 0.596 & 0.596 & 0.596\\
 & 0.100 & 0.822 & 0.628 & 0.822 & 0.818 & 0.822 & 0.822 & 0.822 & 0.822 & 0.822 & 0.822\\
 & 0.200 & 1.000 & 0.992 & 1.000 & 1.000 & 1.000 & 1.000 & 1.000 & 1.000 & 1.000 & 1.000\\
 & 0.400 & 1.000 & 1.000 & 1.000 & 1.000 & 1.000 & 1.000 & 1.000 & 1.000 & 1.000 & 1.000\\
\addlinespace
\multirow{7}{*}{$H_{0,\mathrm L}^{(2)}$} & 0.000 & 0.054 & --- & 0.054 & 0.054 & 0.054 & 0.054 & 0.054 & 0.054 & 0.054 & 0.054\\
 & 0.025 & 0.068 & --- & 0.068 & 0.068 & 0.068 & 0.066 & 0.068 & 0.068 & 0.068 & 0.068\\
 & 0.050 & 0.102 & --- & 0.102 & 0.104 & 0.102 & 0.102 & 0.102 & 0.102 & 0.102 & 0.102\\
 & 0.075 & 0.172 & --- & 0.170 & 0.168 & 0.172 & 0.170 & 0.172 & 0.172 & 0.172 & 0.172\\
 & 0.100 & 0.298 & --- & 0.298 & 0.278 & 0.300 & 0.284 & 0.298 & 0.298 & 0.298 & 0.298\\
 & 0.200 & 0.866 & --- & 0.862 & 0.790 & 0.862 & 0.798 & 0.860 & 0.834 & 0.860 & 0.834\\
 & 0.400 & 0.998 & --- & 0.998 & 0.998 & 0.998 & 0.998 & 0.998 & 0.998 & 0.998 & 0.998\\
\addlinespace
\multirow{7}{*}{$H_{0,\mathrm L}^{(3)}$} & 0.000 & 0.052 & --- & 0.052 & 0.048 & 0.052 & 0.052 & 0.052 & 0.052 & 0.052 & 0.052\\
 & 0.025 & 0.076 & --- & 0.076 & 0.074 & 0.076 & 0.074 & 0.076 & 0.076 & 0.076 & 0.076\\
 & 0.050 & 0.218 & --- & 0.218 & 0.216 & 0.218 & 0.216 & 0.218 & 0.218 & 0.218 & 0.218\\
 & 0.075 & 0.430 & --- & 0.430 & 0.424 & 0.428 & 0.426 & 0.430 & 0.430 & 0.430 & 0.430\\
 & 0.100 & 0.660 & --- & 0.660 & 0.654 & 0.660 & 0.658 & 0.660 & 0.660 & 0.660 & 0.660\\
 & 0.200 & 1.000 & --- & 1.000 & 0.998 & 1.000 & 0.998 & 1.000 & 1.000 & 1.000 & 1.000\\
 & 0.400 & 1.000 & --- & 1.000 & 1.000 & 1.000 & 1.000 & 1.000 & 1.000 & 1.000 & 1.000\\
\addlinespace
\multirow{7}{*}{$H_{0,\mathrm L}^{(4)}$} & 0.000 & 0.054 & --- & 0.054 & 0.056 & 0.054 & 0.056 & 0.054 & 0.054 & 0.054 & 0.054\\
 & 0.025 & 0.092 & --- & 0.092 & 0.096 & 0.092 & 0.092 & 0.092 & 0.092 & 0.092 & 0.092\\
 & 0.050 & 0.300 & --- & 0.300 & 0.302 & 0.298 & 0.302 & 0.300 & 0.300 & 0.300 & 0.300\\
 & 0.075 & 0.584 & --- & 0.584 & 0.578 & 0.584 & 0.578 & 0.584 & 0.584 & 0.584 & 0.584\\
 & 0.100 & 0.828 & --- & 0.828 & 0.828 & 0.828 & 0.826 & 0.828 & 0.828 & 0.828 & 0.828\\
 & 0.200 & 1.000 & --- & 1.000 & 1.000 & 1.000 & 1.000 & 1.000 & 1.000 & 1.000 & 1.000\\
 & 0.400 & 1.000 & --- & 1.000 & 1.000 & 1.000 & 1.000 & 1.000 & 1.000 & 1.000 & 1.000\\
\addlinespace
\multirow{7}{*}{$H_{0,\mathrm L}^{(5)}$} & 0.000 & 0.054 & --- & 0.054 & 0.056 & 0.054 & 0.054 & 0.054 & 0.054 & 0.054 & 0.054\\
 & 0.025 & 0.082 & --- & 0.082 & 0.080 & 0.080 & 0.082 & 0.082 & 0.082 & 0.082 & 0.082\\
 & 0.050 & 0.228 & --- & 0.228 & 0.224 & 0.228 & 0.222 & 0.228 & 0.228 & 0.228 & 0.228\\
 & 0.075 & 0.480 & --- & 0.480 & 0.474 & 0.480 & 0.476 & 0.480 & 0.480 & 0.480 & 0.480\\
 & 0.100 & 0.746 & --- & 0.746 & 0.734 & 0.746 & 0.746 & 0.746 & 0.746 & 0.746 & 0.746\\
 & 0.200 & 1.000 & --- & 1.000 & 1.000 & 1.000 & 1.000 & 1.000 & 1.000 & 1.000 & 1.000\\
 & 0.400 & 1.000 & --- & 1.000 & 1.000 & 1.000 & 1.000 & 1.000 & 1.000 & 1.000 & 1.000\\
\bottomrule
\end{tabular*}
\end{table}
\end{landscape}

\subsubsection{Additional numerical results to Section~\mainref{sec:numerical}}\label{sec: additional tables GLM}

\begin{landscape}
\begin{table}[p]
\centering\fontsize{9}{9}\selectfont
\setlength{\tabcolsep}{5pt}
\renewcommand{\arraystretch}{1.02}
\captionsetup{font=small}
\caption{Empirical rejection probabilities for logistic regression with
$p=400$, $n_b=100$, $B=8$, and 500 independent replications per value of $h$.}
\label{tab:logistic-p400}
\begin{tabular*}{\linewidth}{@{\extracolsep{\fill}}llrr*{8}{r}@{}}
\toprule
 & & & & \multicolumn{4}{c}{1-SE} & \multicolumn{4}{c}{CV/GIC}\\
\cmidrule(lr){5-8}\cmidrule(lr){9-12}
 & & & & \multicolumn{2}{c}{Lasso $\to$ LLA}
 & \multicolumn{2}{c}{LLA $\to$ LLA}
 & \multicolumn{2}{c}{Lasso $\to$ LLA}
 & \multicolumn{2}{c}{LLA $\to$ LLA}\\
\cmidrule(lr){5-6}\cmidrule(lr){7-8}
\cmidrule(lr){9-10}\cmidrule(lr){11-12}
Hypothesis & $h$ & {\tt Offline} & {\tt Luo}
 & $T_W$ & $T_S$ & $T_W$ & $T_S$
 & $T_W$ & $T_S$ & $T_W$ & $T_S$\\
\midrule
\multirow{6}{*}{$H_0^{(1)}$}
 & 0.000 & 0.050 & 0.036 & 0.068 & 0.064 & 0.056 & 0.052 & 0.044 & 0.042 & 0.044 & 0.044\\
 & 0.025 & 0.070 & 0.042 & 0.076 & 0.078 & 0.068 & 0.068 & 0.056 & 0.054 & 0.062 & 0.060\\
 & 0.050 & 0.100 & 0.066 & 0.132 & 0.126 & 0.108 & 0.102 & 0.072 & 0.070 & 0.072 & 0.072\\
 & 0.100 & 0.236 & 0.164 & 0.272 & 0.236 & 0.230 & 0.224 & 0.202 & 0.202 & 0.210 & 0.206\\
 & 0.200 & 0.618 & 0.484 & 0.606 & 0.600 & 0.616 & 0.612 & 0.602 & 0.602 & 0.604 & 0.602\\
 & 0.400 & 0.998 & 0.986 & 0.996 & 0.996 & 0.998 & 0.998 & 0.996 & 0.996 & 0.996 & 0.996\\
\addlinespace
\multirow{6}{*}{$H_0^{(2)}$}
 & 0.000 & 0.060 & --- & 0.052 & 0.052 & 0.052 & 0.052 & 0.054 & 0.054 & 0.052 & 0.052\\
 & 0.025 & 0.066 & --- & 0.062 & 0.060 & 0.060 & 0.062 & 0.062 & 0.062 & 0.056 & 0.056\\
 & 0.050 & 0.098 & --- & 0.084 & 0.082 & 0.084 & 0.082 & 0.082 & 0.082 & 0.078 & 0.078\\
 & 0.100 & 0.226 & --- & 0.210 & 0.204 & 0.206 & 0.206 & 0.212 & 0.210 & 0.196 & 0.196\\
 & 0.200 & 0.636 & --- & 0.586 & 0.560 & 0.602 & 0.578 & 0.608 & 0.600 & 0.584 & 0.580\\
 & 0.400 & 0.988 & --- & 0.970 & 0.960 & 0.978 & 0.974 & 0.980 & 0.978 & 0.972 & 0.970\\
\addlinespace
\multirow{6}{*}{$H_0^{(3)}$}
 & 0.000 & 0.058 & --- & 0.050 & 0.050 & 0.040 & 0.040 & 0.038 & 0.038 & 0.046 & 0.046\\
 & 0.025 & 0.062 & --- & 0.060 & 0.060 & 0.050 & 0.050 & 0.040 & 0.040 & 0.042 & 0.042\\
 & 0.050 & 0.072 & --- & 0.090 & 0.084 & 0.074 & 0.074 & 0.054 & 0.056 & 0.058 & 0.056\\
 & 0.100 & 0.148 & --- & 0.164 & 0.140 & 0.140 & 0.140 & 0.100 & 0.100 & 0.110 & 0.104\\
 & 0.200 & 0.472 & --- & 0.442 & 0.404 & 0.430 & 0.424 & 0.436 & 0.436 & 0.434 & 0.432\\
 & 0.400 & 0.988 & --- & 0.972 & 0.956 & 0.984 & 0.980 & 0.980 & 0.980 & 0.984 & 0.984\\
\addlinespace
\multirow{6}{*}{$H_0^{(4)}$}
 & 0.000 & 0.058 & --- & 0.046 & 0.046 & 0.046 & 0.046 & 0.046 & 0.046 & 0.048 & 0.048\\
 & 0.025 & 0.064 & --- & 0.056 & 0.054 & 0.056 & 0.056 & 0.054 & 0.054 & 0.054 & 0.054\\
 & 0.050 & 0.102 & --- & 0.088 & 0.084 & 0.088 & 0.088 & 0.092 & 0.092 & 0.092 & 0.092\\
 & 0.100 & 0.314 & --- & 0.270 & 0.258 & 0.274 & 0.274 & 0.280 & 0.276 & 0.266 & 0.260\\
 & 0.200 & 0.842 & --- & 0.818 & 0.782 & 0.824 & 0.808 & 0.822 & 0.822 & 0.816 & 0.814\\
 & 0.400 & 1.000 & --- & 1.000 & 1.000 & 1.000 & 1.000 & 1.000 & 1.000 & 1.000 & 1.000\\
\addlinespace
\multirow{6}{*}{$H_0^{(5)}$}
 & 0.000 & 0.056 & --- & 0.050 & 0.050 & 0.050 & 0.048 & 0.048 & 0.048 & 0.046 & 0.046\\
 & 0.025 & 0.058 & --- & 0.048 & 0.048 & 0.050 & 0.048 & 0.050 & 0.050 & 0.052 & 0.052\\
 & 0.050 & 0.088 & --- & 0.072 & 0.068 & 0.074 & 0.074 & 0.072 & 0.072 & 0.076 & 0.076\\
 & 0.100 & 0.222 & --- & 0.178 & 0.162 & 0.190 & 0.186 & 0.184 & 0.184 & 0.184 & 0.188\\
 & 0.200 & 0.744 & --- & 0.704 & 0.664 & 0.704 & 0.686 & 0.708 & 0.706 & 0.698 & 0.696\\
 & 0.400 & 1.000 & --- & 1.000 & 1.000 & 1.000 & 1.000 & 1.000 & 1.000 & 1.000 & 1.000\\
\bottomrule
\end{tabular*}
\end{table}
\end{landscape}

\begin{landscape}
\begin{table}[p]
\centering\fontsize{9}{9}\selectfont
\setlength{\tabcolsep}{5pt}
\renewcommand{\arraystretch}{1.02}
\captionsetup{font=small}
\caption{Empirical rejection probabilities for logistic regression with
$p=1200$, $n_b=100$, $B=8$, and 500 independent replications per value of $h$.}
\label{tab:logistic-p1200}
\begin{tabular*}{\linewidth}{@{\extracolsep{\fill}}llrr*{8}{r}@{}}
\toprule
 & & & & \multicolumn{4}{c}{1-SE} & \multicolumn{4}{c}{CV/GIC}\\
\cmidrule(lr){5-8}\cmidrule(lr){9-12}
 & & & & \multicolumn{2}{c}{Lasso $\to$ LLA}
 & \multicolumn{2}{c}{LLA $\to$ LLA}
 & \multicolumn{2}{c}{Lasso $\to$ LLA}
 & \multicolumn{2}{c}{LLA $\to$ LLA}\\
\cmidrule(lr){5-6}\cmidrule(lr){7-8}
\cmidrule(lr){9-10}\cmidrule(lr){11-12}
Hypothesis & $h$ & {\tt Offline} & {\tt Luo}
 & $T_W$ & $T_S$ & $T_W$ & $T_S$
 & $T_W$ & $T_S$ & $T_W$ & $T_S$\\
\midrule
\multirow{6}{*}{$H_0^{(1)}$}
 & 0.000 & 0.062 & 0.040 & 0.066 & 0.070 & 0.056 & 0.058 & 0.048 & 0.044 & 0.052 & 0.048\\
 & 0.025 & 0.086 & 0.062 & 0.086 & 0.082 & 0.068 & 0.068 & 0.060 & 0.060 & 0.066 & 0.062\\
 & 0.050 & 0.116 & 0.080 & 0.162 & 0.154 & 0.120 & 0.120 & 0.080 & 0.078 & 0.080 & 0.082\\
 & 0.100 & 0.260 & 0.158 & 0.338 & 0.302 & 0.244 & 0.236 & 0.200 & 0.200 & 0.224 & 0.218\\
 & 0.200 & 0.630 & 0.502 & 0.640 & 0.622 & 0.626 & 0.620 & 0.592 & 0.592 & 0.602 & 0.602\\
 & 0.400 & 0.998 & 0.980 & 0.996 & 0.996 & 0.998 & 0.998 & 0.998 & 0.998 & 0.998 & 0.998\\
\addlinespace
\multirow{6}{*}{$H_0^{(2)}$}
 & 0.000 & 0.060 & --- & 0.052 & 0.050 & 0.052 & 0.052 & 0.052 & 0.052 & 0.046 & 0.046\\
 & 0.025 & 0.066 & --- & 0.060 & 0.060 & 0.060 & 0.060 & 0.058 & 0.058 & 0.054 & 0.052\\
 & 0.050 & 0.098 & --- & 0.082 & 0.080 & 0.082 & 0.082 & 0.084 & 0.084 & 0.080 & 0.080\\
 & 0.100 & 0.226 & --- & 0.208 & 0.204 & 0.206 & 0.202 & 0.202 & 0.202 & 0.210 & 0.210\\
 & 0.200 & 0.636 & --- & 0.592 & 0.566 & 0.600 & 0.584 & 0.600 & 0.600 & 0.590 & 0.592\\
 & 0.400 & 0.990 & --- & 0.976 & 0.962 & 0.982 & 0.972 & 0.984 & 0.980 & 0.982 & 0.980\\
\addlinespace
\multirow{6}{*}{$H_0^{(3)}$}
 & 0.000 & 0.062 & --- & 0.058 & 0.058 & 0.052 & 0.050 & 0.042 & 0.040 & 0.046 & 0.044\\
 & 0.025 & 0.066 & --- & 0.078 & 0.078 & 0.060 & 0.060 & 0.042 & 0.042 & 0.046 & 0.044\\
 & 0.050 & 0.076 & --- & 0.114 & 0.118 & 0.086 & 0.090 & 0.060 & 0.060 & 0.058 & 0.060\\
 & 0.100 & 0.174 & --- & 0.228 & 0.202 & 0.156 & 0.158 & 0.108 & 0.108 & 0.118 & 0.110\\
 & 0.200 & 0.480 & --- & 0.472 & 0.430 & 0.444 & 0.428 & 0.430 & 0.428 & 0.442 & 0.440\\
 & 0.400 & 0.988 & --- & 0.974 & 0.964 & 0.980 & 0.978 & 0.978 & 0.978 & 0.982 & 0.982\\
\addlinespace
\multirow{6}{*}{$H_0^{(4)}$}
 & 0.000 & 0.058 & --- & 0.048 & 0.050 & 0.046 & 0.048 & 0.046 & 0.046 & 0.048 & 0.048\\
 & 0.025 & 0.064 & --- & 0.058 & 0.058 & 0.056 & 0.056 & 0.058 & 0.058 & 0.056 & 0.056\\
 & 0.050 & 0.102 & --- & 0.088 & 0.090 & 0.086 & 0.086 & 0.094 & 0.094 & 0.096 & 0.096\\
 & 0.100 & 0.314 & --- & 0.264 & 0.264 & 0.272 & 0.264 & 0.278 & 0.278 & 0.268 & 0.268\\
 & 0.200 & 0.842 & --- & 0.818 & 0.812 & 0.822 & 0.812 & 0.818 & 0.818 & 0.812 & 0.812\\
 & 0.400 & 1.000 & --- & 1.000 & 1.000 & 1.000 & 1.000 & 1.000 & 1.000 & 1.000 & 1.000\\
\addlinespace
\multirow{6}{*}{$H_0^{(5)}$}
 & 0.000 & 0.056 & --- & 0.050 & 0.048 & 0.050 & 0.048 & 0.050 & 0.050 & 0.046 & 0.046\\
 & 0.025 & 0.058 & --- & 0.050 & 0.048 & 0.050 & 0.048 & 0.050 & 0.052 & 0.046 & 0.046\\
 & 0.050 & 0.088 & --- & 0.072 & 0.070 & 0.072 & 0.070 & 0.076 & 0.076 & 0.072 & 0.072\\
 & 0.100 & 0.222 & --- & 0.184 & 0.180 & 0.186 & 0.184 & 0.176 & 0.176 & 0.192 & 0.194\\
 & 0.200 & 0.744 & --- & 0.700 & 0.674 & 0.706 & 0.688 & 0.708 & 0.706 & 0.704 & 0.704\\
 & 0.400 & 1.000 & --- & 1.000 & 1.000 & 1.000 & 1.000 & 1.000 & 1.000 & 1.000 & 1.000\\
\bottomrule
\end{tabular*}
\end{table}
\end{landscape}

\begin{landscape}
\begin{table}[p]
\centering\fontsize{9}{9}\selectfont
\setlength{\tabcolsep}{5pt}
\renewcommand{\arraystretch}{1.02}
\captionsetup{font=small}
\caption{Empirical rejection probabilities for Poisson regression with
$p=400$, $n_b=100$, $B=8$, and 500 independent replications per value of $h$.}
\label{tab:poisson-p400}
\begin{tabular*}{\linewidth}{@{\extracolsep{\fill}}llrr*{8}{r}@{}}
\toprule
 & & & & \multicolumn{4}{c}{1-SE} & \multicolumn{4}{c}{CV/GIC}\\
\cmidrule(lr){5-8}\cmidrule(lr){9-12}
 & & & & \multicolumn{2}{c}{Lasso $\to$ LLA}
 & \multicolumn{2}{c}{LLA $\to$ LLA}
 & \multicolumn{2}{c}{Lasso $\to$ LLA}
 & \multicolumn{2}{c}{LLA $\to$ LLA}\\
\cmidrule(lr){5-6}\cmidrule(lr){7-8}
\cmidrule(lr){9-10}\cmidrule(lr){11-12}
Hypothesis & $h$ & {\tt Offline} & {\tt Luo}
 & $T_W$ & $T_S$ & $T_W$ & $T_S$
 & $T_W$ & $T_S$ & $T_W$ & $T_S$\\
\midrule
\multirow{5}{*}{$H_0^{(1)}$}
 & 0.000 & 0.054 & 0.046 & 0.048 & 0.042 & 0.044 & 0.044 & 0.056 & 0.052 & 0.058 & 0.054\\
 & 0.040 & 0.234 & 0.148 & 0.220 & 0.192 & 0.208 & 0.186 & 0.218 & 0.206 & 0.222 & 0.212\\
 & 0.070 & 0.572 & 0.340 & 0.552 & 0.532 & 0.546 & 0.502 & 0.550 & 0.542 & 0.554 & 0.540\\
 & 0.100 & 0.880 & 0.650 & 0.872 & 0.848 & 0.870 & 0.834 & 0.866 & 0.864 & 0.862 & 0.862\\
 & 0.200 & 1.000 & 0.996 & 1.000 & 1.000 & 1.000 & 1.000 & 1.000 & 1.000 & 1.000 & 1.000\\
\addlinespace
\multirow{5}{*}{$H_0^{(2)}$}
 & 0.000 & 0.054 & --- & 0.052 & 0.044 & 0.054 & 0.040 & 0.062 & 0.062 & 0.056 & 0.056\\
 & 0.040 & 0.174 & --- & 0.164 & 0.156 & 0.170 & 0.152 & 0.168 & 0.168 & 0.174 & 0.174\\
 & 0.070 & 0.480 & --- & 0.442 & 0.390 & 0.438 & 0.372 & 0.470 & 0.470 & 0.476 & 0.476\\
 & 0.100 & 0.820 & --- & 0.804 & 0.682 & 0.802 & 0.670 & 0.832 & 0.830 & 0.818 & 0.816\\
 & 0.200 & 1.000 & --- & 1.000 & 0.998 & 1.000 & 0.998 & 1.000 & 1.000 & 1.000 & 1.000\\
\addlinespace
\multirow{5}{*}{$H_0^{(3)}$}
 & 0.000 & 0.050 & --- & 0.064 & 0.046 & 0.052 & 0.044 & 0.046 & 0.048 & 0.044 & 0.048\\
 & 0.040 & 0.168 & --- & 0.156 & 0.128 & 0.144 & 0.130 & 0.146 & 0.136 & 0.150 & 0.136\\
 & 0.070 & 0.414 & --- & 0.400 & 0.346 & 0.380 & 0.338 & 0.390 & 0.370 & 0.396 & 0.374\\
 & 0.100 & 0.742 & --- & 0.722 & 0.678 & 0.712 & 0.650 & 0.708 & 0.694 & 0.714 & 0.700\\
 & 0.200 & 1.000 & --- & 1.000 & 1.000 & 1.000 & 1.000 & 1.000 & 1.000 & 1.000 & 1.000\\
\addlinespace
\multirow{5}{*}{$H_0^{(4)}$}
 & 0.000 & 0.046 & --- & 0.048 & 0.044 & 0.048 & 0.036 & 0.048 & 0.048 & 0.046 & 0.046\\
 & 0.040 & 0.272 & --- & 0.254 & 0.236 & 0.252 & 0.230 & 0.258 & 0.258 & 0.264 & 0.264\\
 & 0.070 & 0.730 & --- & 0.734 & 0.682 & 0.716 & 0.664 & 0.730 & 0.730 & 0.728 & 0.728\\
 & 0.100 & 0.978 & --- & 0.978 & 0.950 & 0.978 & 0.944 & 0.980 & 0.980 & 0.976 & 0.976\\
 & 0.200 & 1.000 & --- & 1.000 & 1.000 & 1.000 & 1.000 & 1.000 & 1.000 & 1.000 & 1.000\\
\addlinespace
\multirow{5}{*}{$H_0^{(5)}$}
 & 0.000 & 0.046 & --- & 0.048 & 0.036 & 0.048 & 0.038 & 0.048 & 0.048 & 0.046 & 0.046\\
 & 0.040 & 0.196 & --- & 0.186 & 0.156 & 0.186 & 0.148 & 0.194 & 0.194 & 0.190 & 0.190\\
 & 0.070 & 0.610 & --- & 0.604 & 0.550 & 0.596 & 0.530 & 0.604 & 0.604 & 0.610 & 0.610\\
 & 0.100 & 0.944 & --- & 0.938 & 0.892 & 0.936 & 0.884 & 0.940 & 0.940 & 0.940 & 0.940\\
 & 0.200 & 1.000 & --- & 1.000 & 1.000 & 1.000 & 1.000 & 1.000 & 1.000 & 1.000 & 1.000\\
\bottomrule
\end{tabular*}
\end{table}
\end{landscape}

\begin{landscape}
\begin{table}[p]
\centering\fontsize{9}{9}\selectfont
\setlength{\tabcolsep}{5pt}
\renewcommand{\arraystretch}{1.02}
\captionsetup{font=small}
\caption{Empirical rejection probabilities for Poisson regression with
$p=1200$, $n_b=100$, $B=8$, and 500 independent replications per value of $h$.}
\label{tab:poisson-p1200}
\begin{tabular*}{\linewidth}{@{\extracolsep{\fill}}llrr*{8}{r}@{}}
\toprule
 & & & & \multicolumn{4}{c}{1-SE} & \multicolumn{4}{c}{CV/GIC}\\
\cmidrule(lr){5-8}\cmidrule(lr){9-12}
 & & & & \multicolumn{2}{c}{Lasso $\to$ LLA}
 & \multicolumn{2}{c}{LLA $\to$ LLA}
 & \multicolumn{2}{c}{Lasso $\to$ LLA}
 & \multicolumn{2}{c}{LLA $\to$ LLA}\\
\cmidrule(lr){5-6}\cmidrule(lr){7-8}
\cmidrule(lr){9-10}\cmidrule(lr){11-12}
Hypothesis & $h$ & {\tt Offline} & {\tt Luo}
 & $T_W$ & $T_S$ & $T_W$ & $T_S$
 & $T_W$ & $T_S$ & $T_W$ & $T_S$\\
\midrule
\multirow{5}{*}{$H_0^{(1)}$}
 & 0.000 & 0.062 & 0.042 & 0.068 & 0.046 & 0.056 & 0.044 & 0.056 & 0.044 & 0.062 & 0.044\\
 & 0.040 & 0.264 & 0.118 & 0.246 & 0.198 & 0.218 & 0.192 & 0.246 & 0.200 & 0.254 & 0.208\\
 & 0.070 & 0.594 & 0.314 & 0.586 & 0.534 & 0.546 & 0.524 & 0.572 & 0.542 & 0.560 & 0.536\\
 & 0.100 & 0.892 & 0.590 & 0.878 & 0.854 & 0.866 & 0.838 & 0.870 & 0.866 & 0.878 & 0.874\\
 & 0.200 & 1.000 & 0.984 & 1.000 & 1.000 & 1.000 & 1.000 & 1.000 & 1.000 & 1.000 & 1.000\\
\addlinespace
\multirow{5}{*}{$H_0^{(2)}$}
 & 0.000 & 0.054 & --- & 0.050 & 0.048 & 0.052 & 0.042 & 0.054 & 0.054 & 0.052 & 0.052\\
 & 0.040 & 0.174 & --- & 0.158 & 0.146 & 0.166 & 0.140 & 0.176 & 0.176 & 0.176 & 0.176\\
 & 0.070 & 0.480 & --- & 0.436 & 0.384 & 0.444 & 0.356 & 0.458 & 0.458 & 0.470 & 0.470\\
 & 0.100 & 0.820 & --- & 0.798 & 0.706 & 0.804 & 0.674 & 0.822 & 0.820 & 0.824 & 0.820\\
 & 0.200 & 1.000 & --- & 1.000 & 1.000 & 1.000 & 0.996 & 1.000 & 1.000 & 1.000 & 1.000\\
\addlinespace
\multirow{5}{*}{$H_0^{(3)}$}
 & 0.000 & 0.076 & --- & 0.082 & 0.046 & 0.060 & 0.042 & 0.052 & 0.044 & 0.056 & 0.042\\
 & 0.040 & 0.192 & --- & 0.186 & 0.124 & 0.146 & 0.112 & 0.160 & 0.136 & 0.166 & 0.138\\
 & 0.070 & 0.452 & --- & 0.424 & 0.356 & 0.388 & 0.350 & 0.414 & 0.360 & 0.414 & 0.372\\
 & 0.100 & 0.748 & --- & 0.734 & 0.694 & 0.724 & 0.660 & 0.722 & 0.700 & 0.720 & 0.710\\
 & 0.200 & 1.000 & --- & 1.000 & 1.000 & 1.000 & 1.000 & 1.000 & 0.998 & 1.000 & 1.000\\
\addlinespace
\multirow{5}{*}{$H_0^{(4)}$}
 & 0.000 & 0.046 & --- & 0.046 & 0.042 & 0.048 & 0.040 & 0.048 & 0.048 & 0.046 & 0.046\\
 & 0.040 & 0.272 & --- & 0.250 & 0.228 & 0.248 & 0.218 & 0.266 & 0.266 & 0.264 & 0.264\\
 & 0.070 & 0.730 & --- & 0.736 & 0.690 & 0.724 & 0.640 & 0.740 & 0.740 & 0.738 & 0.738\\
 & 0.100 & 0.978 & --- & 0.980 & 0.964 & 0.978 & 0.944 & 0.976 & 0.976 & 0.978 & 0.978\\
 & 0.200 & 1.000 & --- & 1.000 & 1.000 & 1.000 & 1.000 & 1.000 & 1.000 & 1.000 & 1.000\\
\addlinespace
\multirow{5}{*}{$H_0^{(5)}$}
 & 0.000 & 0.046 & --- & 0.048 & 0.044 & 0.048 & 0.042 & 0.048 & 0.048 & 0.048 & 0.048\\
 & 0.040 & 0.196 & --- & 0.194 & 0.170 & 0.186 & 0.148 & 0.202 & 0.202 & 0.210 & 0.210\\
 & 0.070 & 0.610 & --- & 0.596 & 0.538 & 0.590 & 0.510 & 0.608 & 0.608 & 0.620 & 0.620\\
 & 0.100 & 0.944 & --- & 0.942 & 0.910 & 0.936 & 0.872 & 0.938 & 0.938 & 0.940 & 0.940\\
 & 0.200 & 1.000 & --- & 1.000 & 1.000 & 1.000 & 1.000 & 1.000 & 1.000 & 1.000 & 1.000\\
\bottomrule
\end{tabular*}
\end{table}
\end{landscape}

\subsubsection{Additional real data analysis}\label{sec:real data poisson}
We return to the Capital Bikeshare data analyzed in Section~\mainref{sec:real-data} for a second analysis using Poisson regression.
We model the number of arrivals at each ending rental station at each hour in the previous dataset. Our training data consist of 41,354 rental station-hour pairings from January 2024 and our test data consist of 38,686 such pairs from February 2024.
The data contain 73 predictors including the daily characteristics, end time (categorized as before), and end station indicators from the previous analysis.

We use the same dates to define batches and the same tuning approach as in the logistic regression analysis. We likewise make predictions on the test data after fitting to each batch and compute the loss for those predictions. Figure \ref{fig:bikeshare_pois_pred} plots the batch number against the loss. We observe a generally decreasing trend, though there is a notable dip in the loss with the second batch followed by an increase with the third batch. This could be due to covariate drift, with the second batch better resembling conditions in the test batch than later batches.

For our hypothesis test, we again examine the effect of rush hour, testing whether the mean number of arrivals changes during rush hour -- as before, the null hypothesis is H$_0: (\beta_{\rm MR}, \beta_{\rm AR})^T = (0,0)^T$.
Using 5-step LLA estimators with GIC tuning, the online partial penalized Wald and score statistics are $T_W = 7378.25$ and $T_S = 7377.91$, respectively, giving us $p$-values $< 2.2\times10^{-16}$ and providing overwhelming evidence that rush hour is a significant predictor of the mean number of arrivals.

\begin{figure}[t]
    \centering
    \includegraphics[width=0.5\linewidth]{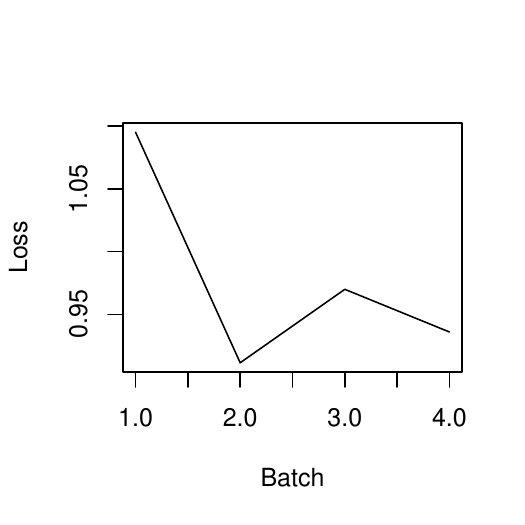}
    \caption{Loss on test data for predictions made after training on each batch in Capital Bikeshare Poisson regression analysis.}
    \label{fig:bikeshare_pois_pred}
\end{figure}

\subsection{Tobit model}

For $j=1,\ldots,B$, let
$x_i^{(j)}=(1,x_{i1}^{(j)},\ldots,x_{ip}^{(j)})^\top\in\mathbb R^{p+1}$ and
suppose that
\begin{equation}
\begin{split}
Y_i^{*(j)}
&=(x_i^{(j)})^\top\beta_{\mathrm T}^*
+\sigma^*\varepsilon_i^{(j)},
\\
Y_i^{(j)}&=\max\{Y_i^{*(j)},0\},
\qquad
D_i^{(j)}=\mathbf 1\{Y_i^{(j)}>0\}.
\end{split}
\label{eq:tobit-latent-model}
\end{equation}
The pairs $(Y_i^{(j)},x_i^{(j)})$ are independent copies of $(Y,X)$, and the errors are independent $N(0,1)$
variables independent of the covariates. All expectations and probabilities
below are with respect to this joint distribution.  Set
\$
\alpha^*=\frac{\beta_{\mathrm T}^*}{\sigma^*},
\qquad
\tau^*=\frac1{\sigma^*},
\qquad
\theta^*=((\alpha^*)^\top,\tau^*)^\top,
\qquad
\Theta_{\mathrm T}=\mathbb R^{p+1}\times(0,\infty),
\$
and $d_{\mathrm T}=p+2$.  For $y\ge0$, write $D(y)=\mathbf 1(y>0)$ and define
\$
\ell_{\mathrm T}(\theta;y,x)
= D(y)\left\{-\log\tau+\frac12(\tau y-x^\top\alpha)^2\right\} -\{1-D(y)\}\log\Phi(-x^\top\alpha),
\$
where $\theta=(\alpha^\top,\tau)^\top\in\Theta_{\mathrm T}$.
This is the Olsen parametrization \citep{olsen1978note} in Section~3 of
\citet{jacobsonzou2024tobit} and Section~3.1 of
\citet{jacobsonzou2025testing}.  In this subsection, $\theta$ replaces the
generic parameter in Sections~\mainref{sec:method}--\mainref{sec:theory},
$d_{\mathrm T}$ is its ambient dimension, and
\$
\widehat{\mc L}^{(j)}(\theta)
=\frac1{n_j}\sum_{i=1}^{n_j}
\ell_{\mathrm T}(\theta;Y_i^{(j)},x_i^{(j)}),
\$
and with optimization subject to $\Theta_{\mathrm T}$.  In particular,
$\widehat{\mc Q}_j$, $\omegab^{(j)}$, $\phi^{(j)}$, $\BB(R)$, and $\CC(l)$
have the same definitions in Sections~\mainref{sec:method}--\mainref{sec:theory}.

Let $\varphi$ denote the standard normal density and put
\begin{equation}
\mathfrak g(t)=\frac{\varphi(t)}{\Phi(t)},
\qquad
\mathfrak h(t)=\mathfrak g(t)\{t+\mathfrak g(t)\}.
\label{eq:tobit-mills-definitions}
\end{equation}
A direct differentiation and the conditional variance formula for a truncated
standard normal variable give
\begin{equation}
\mathfrak g'(t)=-\mathfrak h(t),
\qquad
0<\mathfrak h(t)<1,
\qquad t\in\mathbb R,
\label{eq:tobit-mills-properties-a}
\end{equation}
and Lemma~S.6 of \citet{jacobsonzou2024tobit} gives
\begin{equation}
\sup_{t\in\mathbb R}|\mathfrak h'(t)|
=\sup_{t\in\mathbb R}|\mathfrak g''(t)|<4.3.
\label{eq:tobit-mills-properties}
\end{equation}
The score and Hessian are
\begin{equation}
\begin{split}
\nabla_\alpha\ell_{\mathrm T}(\theta;y,x)
&=D(y)x(x^\top\alpha-\tau y)
+\{1-D(y)\}\mathfrak g(-x^\top\alpha)x,
\\
\nabla_\tau\ell_{\mathrm T}(\theta;y,x)
&=D(y)\{-\tau^{-1}+y(\tau y-x^\top\alpha)\},
\end{split}
\label{eq:tobit-score-formula}
\end{equation}
and
\begin{equation}
\begin{split}
\nabla^2\ell_{\mathrm T}(\theta;y,x)
={}&D(y)
\begin{pmatrix}
xx^\top&-yx\\[-1mm]
-yx^\top&y^2+\tau^{-2}
\end{pmatrix}
+\{1-D(y)\}
\begin{pmatrix}
\mathfrak h(-x^\top\alpha)xx^\top&0\\
0&0
\end{pmatrix}.
\end{split}
\label{eq:tobit-hessian-formula}
\end{equation}
We remark that both matrices on the right-hand side are positive semidefinite \citep{jacobsonzou2024tobit,jacobsonzou2025testing}.  Hence
$\ell_{\mathrm T}(\cdot;y,x)$ and $\widehat{\mc L}^{(j)}$ are convex on
$\Theta_{\mathrm T}$.

Let
\$
\mc P_{\mathrm T}=\{1,\ldots,p\},
\qquad
\mc U_{\mathrm T}=\{0,p+1\},
\qquad
\mc A_{\mathrm T}
=\{k\in\mc P_{\mathrm T}:\alpha_k^*\ne0\}\cup\mc U_{\mathrm T},
\qquad
a_{\mathrm T}=|\mc A_{\mathrm T}|.
\$
The initializer is obtained from \maineqref{def:bth.estimator} by replacing its
penalty by $\lambda_j\|\alpha_{\mc P_{\mathrm T}}\|_1$.  Thus the intercept
and scale coordinates are unpenalized.  
Recall that the proof of
Theorem~\mainref{thm:batch_b_init} uses the support cone only through the KKT
inequality.  In the present parametrization that inequality is
\begin{equation}
\|u_{\mc A_{\mathrm T}^c}\|_1
\le3\|u_{\mc A_{\mathrm T}}\|_1,
\qquad\text{which implies}\qquad
u\in\CC(4a_{\mathrm T}^{1/2}).
\label{eq:tobit-main-cone-reduction}
\end{equation}
Indeed, the subgradient is zero on $\mc U_{\mathrm T}$ and
$\theta^*_{\mc A_{\mathrm T}^c}=0$; the calculation in
\eqref{eq:thm1-subgradient}--\eqref{eq:thm1-global-cone} is otherwise
unchanged.  Accordingly, all localization statements below use the
conic constraint $\CC(4a_{\mathrm T}^{1/2})$.

The linear hypothesis of interest transforms as
\$
\Cb\beta_{\mathrm T,\mc M}^*=\tb
\quad\Longleftrightarrow\quad
\widetilde{\Cb}_{\mathrm T}\theta_{\mc M_{\mathrm T}}^*=0,
\qquad
\mc M_{\mathrm T}=\{0\}\cup\mc M\cup\{p+1\},
\quad
\widetilde{\Cb}_{\mathrm T}
=\begin{pmatrix}\mathbf0_{r\times1}&\Cb&-\tb\end{pmatrix},
\$
where the coordinates in $\theta_{\mc M_{\mathrm T}}$ are ordered as
$(\alpha_0,\alpha_{\mc M}^\top,\tau)^\top$ (see Section~3.1 of \citet{jacobsonzou2025testing}).

\subsubsection{Assumptions for the Tobit model}

For $j\le B$ and $i\le n_j$, put
$\eta_i^{(j)}=(x_i^{(j)})^\top\alpha^*$ and
\begin{equation}
\zeta_i^{(j)}
=D_i^{(j)}\begin{pmatrix}-x_i^{(j)}\\Y_i^{(j)}\end{pmatrix},
\qquad
\widehat\Sigma_{\mathrm T}^{(j)}
=\frac1{n_j}\sum_{i=1}^{n_j}\zeta_i^{(j)}(\zeta_i^{(j)})^\top,
\qquad
\Sigma_{\mathrm T}^{(j)}
=\mathbb E\widehat\Sigma_{\mathrm T}^{(j)}.
\label{eq:tobit-augmented-gram}
\end{equation}
We impose the following conditions.
\begin{enumerate}[label=(Tobit-\arabic*),leftmargin=3\parindent]
\item\label{Tobitassm:model}
The model in \eqref{eq:tobit-latent-model} holds and, for constants
$0<\underline\tau\le\overline\tau<\infty$,
\begin{equation}
\underline\tau\le\tau^*\le\overline\tau.
\label{eq:tobit-scale-bounds}
\end{equation}

\item\label{Tobitassm:design}
For constants $B_1\ge1$ and $B_\eta<\infty$, almost surely,
\begin{equation}
\max_{\substack{j\le B,\ i\le n_j\\0\le k\le p}}
|x_{ik}^{(j)}|\le B_1^{1/2},
\qquad
\max_{j\le B,\ i\le n_j}|\eta_i^{(j)}|\le B_\eta^{1/2}.
\label{eq:tobit-bahadur-design-bounds}
\end{equation}

\item\label{Tobitassm:expected-re}
For a constant $\underline\kappa_{\mathrm T}>0$,
\begin{equation}
\min_{j\le B}
\inf_{\substack{u\ne0:\ 
\|u_{\mc A_{\mathrm T}^c}\|_1
\le3\|u_{\mc A_{\mathrm T}}\|_1}}
\frac{u^\top\Sigma_{\mathrm T}^{(j)}u}{\|u\|_2^2}
\ge\underline\kappa_{\mathrm T}.
\label{eq:tobit-expected-re}
\end{equation}
\end{enumerate}
Condition \eqref{eq:tobit-expected-re} is the batchwise population version
of (A0) in \citet{jacobsonzou2024tobit}. The two bounds in
\eqref{eq:tobit-bahadur-design-bounds} are the uniform online versions of
the last two bounds in (A4) of \citet{jacobsonzou2025testing};
\eqref{eq:tobit-design-envelopes} follows from them. The constants do not
depend on $p$, $n_j$, or $B$.  
From \ref{Tobitassm:design}, for every $j\le B$, one has
\begin{equation} 
\begin{split}
\max_{0\le k\le p}\frac1{n_j}\sum_{i=1}^{n_j}(x_{ik}^{(j)})^2
&\le B_1,
\qquad
\frac1{n_j}\sum_{i=1}^{n_j}(\eta_i^{(j)})^2\le B_\eta,
\\
\max_{0\le k,l\le p}\frac1{n_j}\sum_{i=1}^{n_j}
(x_{ik}^{(j)})^2(x_{il}^{(j)})^2&\le B_{1}^2.
\end{split}
\label{eq:tobit-design-envelopes}
\end{equation}

\subsubsection{Verification of the initializer theorem}

Fix $b\in\{1,\ldots,B\}$. Integrating the single observation bounds in Lemmas~S.1--S.2 of \citet{jacobsonzou2024tobit} over the covariate row and
using \eqref{eq:tobit-bahadur-design-bounds} yields
\begin{equation}
\begin{split}
\mathbb E\nabla\ell_{\mathrm T}(\theta^*;Y_i^{(j)},x_i^{(j)})&=0,
\\
\mathbb E\exp\{t[\nabla_\alpha\ell_{\mathrm T}
(\theta^*;Y_i^{(j)},x_i^{(j)})]_k\}
&\le\exp(t^2B_1/2),\qquad t\in\mathbb R,
\\
\mathbb E\exp\{t\tau^*\nabla_\tau\ell_{\mathrm T}
(\theta^*;Y_i^{(j)},x_i^{(j)})\}
&\le\exp\{(16+4B_\eta)t^2/2\},\qquad |t|<1/4.
\end{split}
\label{eq:tobit-score-mgf}
\end{equation}
Let
\begin{equation}
V_{\mathrm T}=16+4B_\eta,
\qquad
\lambda_j=C\left(\frac{\log d_{\mathrm T}}{N_j}\right)^{1/2}.
\label{eq:tobit-lambda}
\end{equation}
If
\begin{equation}
C^2\ge64\max\left\{B_1,\frac{V_{\mathrm T}}{\underline\tau^2}\right\},
\qquad
C\ge\frac{64}{\underline\tau}
\left(\frac{\log d_{\mathrm T}}{n_1}\right)^{1/2},
\label{eq:tobit-C-conditions}
\end{equation}
then multiplication of the moment generating functions in
\eqref{eq:tobit-score-mgf}, followed by coordinatewise Chernoff bounds, gives
\begin{equation}
\mathbb P\left\{
\max_{1\le j\le b}\frac{\|\omegab^{(j)}\|_\infty}{\lambda_j}
\le\frac14\right\}
\ge1-\frac{2b}{d_{\mathrm T}},
\qquad b\le B.
\label{eq:tobit-simultaneous-score-event}
\end{equation}
Let $c>0$ be sufficiently small, depending only on the constants in
\ref{Tobitassm:model}--\ref{Tobitassm:design}, and define
\begin{equation*}
\begin{split}
\pi_{\mathrm T,j}(t)
&=2d_{\mathrm T}^2\exp\{-c n_j\min(t^2,t)\},\qquad t>0,
\\
\Pi_{\mathrm T,b}(t)&=\sum_{j=1}^b\pi_{\mathrm T,j}(t),
\qquad t_{\mathrm T}=\frac{\underline\kappa_{\mathrm T}}{32a_{\mathrm T}}.
\end{split}
\end{equation*}
The entries of $\zeta_i^{(j)}(\zeta_i^{(j)})^\top$ are
$D x_kx_l$, $-x_kY$, and $Y^2$. Since
$0\le\tau^*Y\le B_\eta^{1/2}+|\varepsilon|$,
\eqref{eq:tobit-bahadur-design-bounds} and the Gaussian moments bound the
absolute $m$th moments of the centered entries of both $\zeta\zeta^\top$
and $(1+Y^2)XX^\top$ by $m!C^m$, for every $m\ge2$.
Here and below, $C$ denotes a sufficiently large constant depending only on
the stated model and population constants. The two symmetric matrices have
$d_{\mathrm T}(d_{\mathrm T}+1)/2+(d_{\mathrm T}-1)d_{\mathrm T}/2
=d_{\mathrm T}^2$ upper-triangular entries in total. Bernstein's inequality
and a union bound over these entries and batches therefore give
\begin{equation}
\begin{split}
&\mathbb P\Bigg\{\max_{1\le j\le b}\Bigg[
\|\widehat\Sigma_{\mathrm T}^{(j)}-\Sigma_{\mathrm T}^{(j)}\|_{\max}
\\
&\qquad\qquad\qquad\vee\left\|\frac1{n_j}\sum_{i=1}^{n_j}
\{1+(Y_i^{(j)})^2\}x_i^{(j)}(x_i^{(j)})^\top
-\mathbb E\{(1+Y^2)XX^\top\}\right\|_{\max}
\Bigg]\le t_{\mathrm T}\Bigg\}\\
&\quad\ge1-\Pi_{\mathrm T,b}(t_{\mathrm T}).
\end{split}
\label{eq:tobit-gram-event-probability}
\end{equation}
Let $l_{\mathrm T}=4a_{\mathrm T}^{1/2}$.  On the event in
\eqref{eq:tobit-gram-event-probability}, condition
\ref{Tobitassm:expected-re} and \eqref{eq:tobit-main-cone-reduction} imply
\begin{equation}
u^\top\widehat\Sigma_{\mathrm T}^{(j)}u
\ge\frac{\underline\kappa_{\mathrm T}}2\|u\|_2^2
\label{eq:tobit-empirical-re}
\end{equation}
for every $j\le b$ and every $u$ satisfying the first inequality in
\eqref{eq:tobit-main-cone-reduction}.  Equation
\eqref{eq:tobit-hessian-formula} also gives, for every
$\theta\in\Theta_{\mathrm T}$,
\begin{equation}
\nabla^2\widehat{\mc L}^{(j)}(\theta)
\succeq\widehat\Sigma_{\mathrm T}^{(j)}.
\label{eq:tobit-global-rsc-lower-bound}
\end{equation}
Thus \eqref{eq:tobit-empirical-re} supplies the restricted curvature used in
the proof of Theorem~\mainref{thm:batch_b_init}.

Under Assumption~\mainref{assm:hessian},
\eqref{eq:tobit-hessian-formula} and
\eqref{eq:tobit-bahadur-design-bounds} imply
\begin{align*}
\Phi(-B_\eta^{1/2})\mathbb E(XX^\top)
&\preceq\mathbb E\nabla^2_{\alpha,\alpha}
\ell_{\mathrm T}(\theta^*;Y,X)\preceq J_u\Ib_{p+1},
\\
\mathbb E\{(1+Y^2)XX^\top\}
&\preceq
\frac{1+(B_\eta+1)/\underline\tau^2}{\Phi(-B_\eta^{1/2})}
J_u\Ib_{p+1},
\end{align*}
where the second line uses
$\mathbb E(Y^2\mid X)\le(B_\eta+1)/\underline\tau^2$.
On the event in \eqref{eq:tobit-gram-event-probability}, every
$v\in\CC(l_{\mathrm T})$ therefore satisfies
\[
\frac1{n_j}\sum_{i=1}^{n_j}\{1+(Y_i^{(j)})^2\}
\{(x_i^{(j)})^\top v_\alpha\}^2
\le\left\{
\frac{1+(B_\eta+1)/\underline\tau^2}{\Phi(-B_\eta^{1/2})}J_u
+\frac{\underline\kappa_{\mathrm T}}2\right\}\|v\|_2^2,
\]
since $\|v_\alpha\|_1^2\le16a_{\mathrm T}\|v\|_2^2$.
For $0<R_0\le\underline\tau/2$, the Cauchy--Schwarz inequality and \eqref{eq:tobit-mills-properties} consequently imply, for
$v,w\in\BB(R_0)\cap\CC(l_{\mathrm T})$ and $j\le b$,
\begin{equation}
\left\|\{\nabla^2\widehat{\mc L}^{(j)}(\theta^*+v)
-\nabla^2\widehat{\mc L}^{(j)}(\theta^*)\}w\right\|_\infty
\le M_{\mathrm T}\|v\|_2\|w\|_2,
\label{eq:tobit-rlh-bound}
\end{equation}
where
\$
M_{\mathrm T}
=\max\left\{
4.3B_1^{1/2}\left[
\frac{1+(B_\eta+1)/\underline\tau^2}{\Phi(-B_\eta^{1/2})}J_u
+\frac{\underline\kappa_{\mathrm T}}2\right],
2(\underline\tau-R_0)^{-3}\right\}.
\$
For the $\alpha$ coordinates, use $|x_k|\le B_1^{1/2}$ in the empirical
mean of $4.3|x_k(x^\top v_\alpha)(x^\top w_\alpha)|$.
For the $\tau$ coordinate, use
$|\{(\tau^*+v_\tau)^{-2}-(\tau^*)^{-2}\}w_\tau|
\le2(\underline\tau-R_0)^{-3}|v_\tau w_\tau|$.

\begin{proposition}
\label{prop:tobit-initializer-verification}
Suppose that \ref{Tobitassm:model}--\ref{Tobitassm:expected-re},
Assumption~\mainref{assm:hessian}, \eqref{eq:tobit-lambda}, and
\eqref{eq:tobit-C-conditions} hold.  Let
$0<\kappa\le\underline\kappa_{\mathrm T}/2$ and suppose that
\$
\max_{1\le j\le b}\frac{3a_{\mathrm T}^{1/2}\lambda_j}{2\kappa}
<R_0\le\frac{\underline\tau}{2},
\$
and, for $j=2,\ldots,b$,
\begin{equation}
\left(\frac{N_j}{a_{\mathrm T}^2\log d_{\mathrm T}}\right)^{1/2}
\ge\frac{27CM_{\mathrm T}}{2\kappa^2}
\left\{1+\log\left(\frac{N_{j-1}}{n_1}\right)\right\}.
\label{eq:tobit-online-sample-condition}
\end{equation}
Then, with probability at least
\$
1-\frac{2b}{d_{\mathrm T}}-\Pi_{\mathrm T,b}(t_{\mathrm T}),
\$
every initializer exists and, simultaneously for $j=1,\ldots,b$,
\begin{equation}
\|\widehat\theta^{(j)}-\theta^*\|_2
\le\frac{3a_{\mathrm T}^{1/2}\lambda_j}{2\kappa},
\qquad
\|\widehat\theta^{(j)}-\theta^*\|_1
\le\frac{6a_{\mathrm T}\lambda_j}{\kappa},
\qquad
\widehat\theta^{(j)}-\theta^*\in\BB(R_0)\cap\CC(l_{\mathrm T}).
\label{eq:tobit-final-rates}
\end{equation}
\end{proposition}

\begin{proof}
Intersect the events in \eqref{eq:tobit-simultaneous-score-event} and
\eqref{eq:tobit-gram-event-probability}.  Inductively, the bounds through
batch $j-1$ and Lemma~\ref{lem:batch.size} give
\$
\phi^{(j)}
\le\frac{9C^2a_{\mathrm T}\log d_{\mathrm T}}{4\kappa^2N_j}
\left\{1+\log\left(\frac{N_{j-1}}{n_1}\right)\right\}.
\$
Consequently, \eqref{eq:tobit-online-sample-condition} gives
$(3/2)M_{\mathrm T}\phi^{(j)}\le\lambda_j/4$, while
\eqref{eq:tobit-simultaneous-score-event} gives
$\|\omegab^{(j)}\|_\infty\le\lambda_j/4$.
Equation~\eqref{eq:tobit-main-cone-reduction} gives the required cone,
and \eqref{eq:tobit-empirical-re}--\eqref{eq:tobit-global-rsc-lower-bound}
give curvature at least $\kappa$ on that cone for the batch and surrogate
losses. The argument in the proof of
Proposition~\ref{prop:glm-initializer} therefore gives the bounds in
\eqref{eq:tobit-final-rates}.
\end{proof}

\subsubsection{Verification of the tightening-stage theorem}
\label{sec:tobit_tightening_verification}

Let
\begin{equation*}
\begin{split}
\mc S_{\mathrm T}
&=\{j\in\mc P_{\mathrm T}\setminus\mc M:\alpha_j^*\ne0\},
\qquad
s_{\mathrm T}=|\mc S_{\mathrm T}|,
\\
\overline{\mc M}_{\mathrm T}
&=\mc M_{\mathrm T}\cup\mc S_{\mathrm T},
\qquad
m_{\mathrm T}=m+2,
\qquad
q_{\mathrm T}=m_{\mathrm T}+s_{\mathrm T}.
\end{split}
\end{equation*}
Then $\mc A_{\mathrm T}\subseteq\overline{\mc M}_{\mathrm T}$ and
$a_{\mathrm T}\le q_{\mathrm T}$.  The reduced and full LLA iterates and
their oracle targets are those in \maineqref{def:LLA.0}, \maineqref{def:LLA.a},
\maineqref{def:constrained_oracle_estimator}, and
\maineqref{def:unconstrained_oracle_estimator}, after the substitutions
\begin{equation}
\begin{split}
(\betab,\mc M,\overline{\mc M})
=(\theta,\mc M_{\mathrm T},\overline{\mc M}_{\mathrm T}),\qquad
(\Cb,\tb)=(\widetilde{\Cb}_{\mathrm T},\mathbf0).
\end{split}
\label{eq:tobit-terminal-main-substitution}
\end{equation}
We denote them by
$\widehat\theta_{\mathrm T,e}^{({\rm LLA},k)}$ and
$\widehat\theta_{\mathrm T,e}^{\rm ora}$, $e\in\{0,a\}$.
Suppose that the local alternative in Assumption~\mainref{assm:local_alternative}
holds for $\beta_{\mathrm T}^*$.  Define
\$
\hb_{\mathrm T,B}=\widetilde{\Cb}_{\mathrm T}
\theta_{\mc M_{\mathrm T}}^*=\tau^*\hb_B,
\qquad
\theta_{\mathrm T,\mc M_{\mathrm T}}^\dagger
=\theta_{\mc M_{\mathrm T}}^*
-\widetilde{\Cb}_{\mathrm T}^\top
(\widetilde{\Cb}_{\mathrm T}\widetilde{\Cb}_{\mathrm T}^\top)^{-1}
\hb_{\mathrm T,B},
\qquad
\theta_{\mathrm T,\mc M_{\mathrm T}^c}^\dagger
=\theta_{\mc M_{\mathrm T}^c}^*,
\$
and $d_{\mathrm T,B}=C_{\max}\|\hb_{\mathrm T,B}\|_2$.  Since
$\widetilde{\Cb}_{\mathrm T}\widetilde{\Cb}_{\mathrm T}^\top
=\Cb\Cb^\top+\tb\tb^\top\succeq\Cb\Cb^\top$,
\begin{equation}
\widetilde{\Cb}_{\mathrm T}\theta_{\mathrm T,\mc M_{\mathrm T}}^\dagger=0,
\qquad
\|\theta_{\mathrm T}^\dagger-\theta^*\|_2
\le d_{\mathrm T,B}.
\label{eq:tobit-terminal-projection-bound}
\end{equation}
Let $\Zb_{\mathrm T}$ have orthonormal columns spanning
$\ker(\widetilde{\Cb}_{\mathrm T})$.  Use the constants
$\gamma_0,p_0,c_{\rm aug},l$, and $\delta$ in
Theorem~\mainref{thm:batch_b_tightening}, with $m+s=q_{\mathrm T}$.

Denote
\begin{equation}
\begin{split}
\overline\Sigma_{\mathrm T,B}
&=\frac1{N_B}\sum_{j=1}^Bn_j\widehat\Sigma_{\mathrm T}^{(j)},
\qquad
\Sigma_{\mathrm T,B}
 =\mathbb E\overline\Sigma_{\mathrm T,B},
\\
\overline H_{\mathrm T,B}^*
&=\frac1{N_B}\sum_{j=1}^Bn_j
\nabla^2\widehat{\mc L}^{(j)}(\theta^*),
\qquad
J_{\mathrm T,B}
 =\mathbb E\overline H_{\mathrm T,B}^*.
\end{split}
\label{eq:tobit-terminal-pooled-matrices}
\end{equation}
Let $k_{\mathrm T,l}=(p+1)\wedge\lceil l^2\rceil$ and
\$
Q_{\mathrm T,l}
=\max_{\substack{1\le j\le B,\ 0\le a\le p\\
S\subseteq\{0,\ldots,p\},\ 1\le|S|\le k_{\mathrm T,l}}}
\lambda_{\max}\left\{
\frac1{n_j}\sum_{i=1}^{n_j}|x_{ia}^{(j)}|
x_{iS}^{(j)}(x_{iS}^{(j)})^\top\right\}.
\$
On the event in \eqref{eq:tobit-gram-event-probability} with $b=B$,
$|x_{ia}^{(j)}|\le B_1^{1/2}$ and the population bound preceding
\eqref{eq:tobit-rlh-bound} give
\begin{equation*}
\begin{split}
Q_{\mathrm T,l}&\le\overline Q_{\mathrm T},
\\
\overline Q_{\mathrm T}
&=B_1^{1/2}\left\{
\frac{1+(B_\eta+1)/\underline\tau^2}{\Phi(-B_\eta^{1/2})}J_u
+\frac{\underline\kappa_{\mathrm T}k_{\mathrm T,l}}{32a_{\mathrm T}}
\right\}.
\end{split}
\end{equation*}
Under Assumption~\mainref{assm:hessian},
$J_{\mathrm T,B}=\nabla^2\mathcal L(\theta^*)$ is the population matrix
$\Jb$ in the main text. Set
\$
J_{\mathrm T,u}=J_u.
\$
For every row and coordinate subset, the operator norm bound on $\Jb$ then yields
\begin{equation}
\max_{0\le a\le p+1}
\|(J_{\mathrm T,B})_{a,\overline{\mc M}_{\mathrm T}}\|_2
\le J_{\mathrm T,u}.
\label{eq:tobit-terminal-J-row-bound}
\end{equation}
Since $d_{\mathrm T}=p+2\ge2$,
$2\log\log(e d_{\mathrm T})\ge1$ and
$2\log\log(e d_{\mathrm T})\le2\log d_{\mathrm T}$.  Define
\begin{equation}
\begin{split}
u_{\mathrm T,B}
&=2\log(2d_{\mathrm T})+2\log\log(e d_{\mathrm T}),
\\
\eta_{\mathrm T,\infty}
&=C\left\{\sqrt{\frac{u_{\mathrm T,B}}{N_B}}
+\frac{u_{\mathrm T,B}}{N_B}\right\},
\\
\nu_{\mathrm T}^2&=C(J_u+1),
\qquad c_{\mathrm T}=Ck_{\mathrm T,l}^{1/2},
\\
\chi_{\mathrm T,\infty}
&=\nu_{\mathrm T}\left\{
\frac{2\{\log(2d_{\mathrm T})+2\log\log(e d_{\mathrm T})\}}
{N_B}\right\}^{1/2}
+\frac{2c_{\mathrm T}\{\log(2d_{\mathrm T})
+2\log\log(e d_{\mathrm T})\}}{N_B},
\\
\chi_{\mathrm T}(v)
&=2\nu_{\mathrm T}\left\{
\frac{2\{v\log5+2\log\log(e d_{\mathrm T})\}}{N_B}\right\}^{1/2}
+\frac{4c_{\mathrm T}\{v\log5+2\log\log(e d_{\mathrm T})\}}{N_B}.
\end{split}
\label{eq:tobit-terminal-concentration-rates}
\end{equation}

\begin{lemma}
\label{lem:tobit-terminal-concentration}
Under \ref{Tobitassm:model}, \ref{Tobitassm:design}, and
Assumption~\mainref{assm:hessian}, with probability at least
$1-8\{\log(e d_{\mathrm T})\}^{-2}$,
\begin{equation}
\begin{split}
\|\overline\Sigma_{\mathrm T,B}-\Sigma_{\mathrm T,B}\|_{\max}
\vee
\|\overline H_{\mathrm T,B}^*-J_{\mathrm T,B}\|_{\max}
&\le\eta_{\mathrm T,\infty},
\\
\|\omegab^{(B)}\|_\infty
&\le\chi_{\mathrm T,\infty},
\\
\|\omegab_{\overline{\mc M}_{\mathrm T}}^{(B)}\|_2
&\le\chi_{\mathrm T}(q_{\mathrm T}),
\\
\|\Zb_{\mathrm T}^\top\omegab_{\mc M_{\mathrm T}}^{(B)}\|_2
+\|\omegab_{\mc S_{\mathrm T}}^{(B)}\|_2
&\le\sqrt2\,\chi_{\mathrm T}(q_{\mathrm T}-r).
\end{split}
\label{eq:tobit-terminal-concentration-event}
\end{equation}
\end{lemma}

\begin{proof}
The entries in \eqref{eq:tobit-hessian-formula} and
\eqref{eq:tobit-augmented-gram} satisfy the same moment bound used for
\eqref{eq:tobit-gram-event-probability}. Apply Bernstein's inequality to
each centered entry, with $u_{\mathrm T,B}$ in
\eqref{eq:tobit-terminal-concentration-rates}. The union over both
matrices has complement probability at most
$\{\log(e d_{\mathrm T})\}^{-2}$ and validates the first line of
\eqref{eq:tobit-terminal-concentration-event}.

For the score bounds, let $a=(a_\alpha^\top,a_\tau)^\top$ have at most
$k_{\mathrm T,l}$ nonzero $\alpha$ coordinates. The observed score is the
conditional expectation, given $(Y,X)$, of the complete Gaussian score.
Consequently, for every integer $m\ge2$, Jensen's inequality gives
\begin{align*}
\mathbb E|a^\top\nabla\ell_{\mathrm T}(\theta^*;Y,X)|^m
\le
\mathbb E\left|
\{-a_\alpha^\top X+(a_\tau/\tau^*)X^\top\alpha^*\}\varepsilon
+(a_\tau/\tau^*)(\varepsilon^2-1)\right|^m.
\end{align*}
By \eqref{eq:tobit-bahadur-design-bounds} and
\eqref{eq:tobit-hessian-formula}, we get
\[
\left\{\inf_{|t|\le B_\eta^{1/2}}\mathfrak h(t)\right\}
\mathbb E(XX^\top)
\preceq(J_{\mathrm T,B})_{\alpha,\alpha}\preceq J_u\Ib_{p+1}.
\]
The infimum is positive by \eqref{eq:tobit-mills-properties-a}.
Thus the coefficient of $\varepsilon$ in the preceding display has
absolute value at most $Ck_{\mathrm T,l}^{1/2}\|a\|_2$ and second moment
at most $C(J_u+1)\|a\|_2^2$. Independence of $X$ and $\varepsilon$,
together with the Gaussian moments, gives
\[
\mathbb E|a^\top\nabla\ell_{\mathrm T}(\theta^*;Y,X)|^m
\le\frac{m!}{2}\,C(J_u+1)\|a\|_2^2
\{Ck_{\mathrm T,l}^{1/2}\|a\|_2\}^{m-2}.
\]
Expanding the moment generating function and enlarging the constants in
\eqref{eq:tobit-terminal-concentration-rates} therefore yields
\begin{equation}
\mathbb E\exp\{tN_Ba^\top\omegab^{(B)}\}
\le\exp\{N_B\nu_{\mathrm T}^2t^2\|a\|_2^2/2\},
\qquad |t|<\{c_{\mathrm T}\|a\|_2\}^{-1}.
\label{eq:tobit-terminal-pooled-score-mgf}
\end{equation}
Apply Chernoff's inequality at
$\log(2d_{\mathrm T})+2\log\log(e d_{\mathrm T})$ to the coordinate
directions. For $1/2$-nets of the unit spheres in
$\mathbb R^{q_{\mathrm T}}$ and
$\operatorname{col}\{\operatorname{diag}(\Zb_{\mathrm T},
\Ib_{s_{\mathrm T}})\}$, use, respectively,
$q_{\mathrm T}\log5+2\log\log(e d_{\mathrm T})$ and
$(q_{\mathrm T}-r)\log5+2\log\log(e d_{\mathrm T})$.
The nets have at most $5^{q_{\mathrm T}}$ and $5^{q_{\mathrm T}-r}$
elements. Moreover,
\[
\|\Zb_{\mathrm T}^\top\omegab_{\mc M_{\mathrm T}}^{(B)}\|_2
+\|\omegab_{\mc S_{\mathrm T}}^{(B)}\|_2
\le\sqrt2\,
\|\operatorname{diag}(\Zb_{\mathrm T},\Ib_{s_{\mathrm T}})^\top
\omegab_{\overline{\mc M}_{\mathrm T}}^{(B)}\|_2.
\]
The union of these bounds proves
\eqref{eq:tobit-terminal-concentration-event} with the stated complement
probability.
\end{proof}

Define
\$
M_{\mathrm T,l}
=\max\left\{4.3\,\overline Q_{\mathrm T},
2(\underline\tau-R_0)^{-3}\right\}.
\$
On the event in \eqref{eq:tobit-gram-event-probability} with $b=B$, every $v\in\CC(l)$
satisfies $\|v_\alpha\|_1^2\le\min(l^2,p+1)\|v\|_2^2
\le k_{\mathrm T,l}\|v\|_2^2$, and hence
\[
\frac1{n_j}\sum_{i=1}^{n_j}\{1+(Y_i^{(j)})^2\}
\{(x_i^{(j)})^\top v_\alpha\}^2
\le B_1^{-1/2}\overline Q_{\mathrm T}\|v\|_2^2.
\]
The Cauchy--Schwarz inequality, $|x_{ia}^{(j)}|\le B_1^{1/2}$, and
\eqref{eq:tobit-mills-properties} bound each $\alpha$ coordinate of the
Hessian difference applied to $w\in\CC(l)$ by
$4.3\overline Q_{\mathrm T}\|v\|_2\|w\|_2$.
For $R_0\le\underline\tau/2$, the $\tau$-coordinate bound in
\eqref{eq:tobit-rlh-bound} directly applies.
Consequently, for $j\le B$,
\begin{equation}
\left\|\{\nabla^2\widehat{\mc L}^{(j)}(\theta^*+v)
-\nabla^2\widehat{\mc L}^{(j)}(\theta^*)\}w\right\|_\infty
\le M_{\mathrm T,l}\|v\|_2\|w\|_2,
\qquad v,w\in\BB(R_0)\cap\CC(l).
\label{eq:tobit-terminal-rlh}
\end{equation}
By homogeneity in $w$, \eqref{eq:tobit-terminal-rlh} also holds for every
$w\in\CC(l)$.

To obtain the terminal restricted curvature, use the same population
matrix $\Jb=J_{\mathrm T,B}$ as in
Assumption~\mainref{assm:hessian}. For $a,b\in\mathbb R$ and
$|\eta|\le B_\eta^{1/2}$, truncated-normal moments give
\begin{align*}
&\mathbb E\!\left[\mathbf1\{\eta+\varepsilon>0\}
\{-a+b(\eta+\varepsilon)\}^2\right]
\\
&\quad=\Phi(\eta)\left[\{-a+b\{\eta+\mathfrak g(\eta)\}\}^2
+\{1-\mathfrak h(\eta)\}b^2\right]
\ge C^{-1}(a^2+b^2).
\end{align*}
Indeed, $\Phi(\eta)$ and $1-\mathfrak h(\eta)$ have positive minima on
this compact interval, and $\eta+\mathfrak g(\eta)$ is bounded there.
Since
$\Phi(-\eta)\mathfrak h(-\eta)a^2+\Phi(\eta)b^2\le a^2+b^2$,
apply the preceding inequality with
$a=X^\top u_\alpha$ and $b=u_\tau/\tau^*$, and then integrate over $X$.
Equations \eqref{eq:tobit-hessian-formula} and
\eqref{eq:tobit-terminal-pooled-matrices} give
$J_{\mathrm T,B}\preceq C\Sigma_{\mathrm T,B}$.
Decreasing the constant in \eqref{eq:tobit-expected-re}, if necessary,
Assumption~\mainref{assm:hessian} therefore yields
\begin{equation}
u^\top\Sigma_{\mathrm T,B}u
\ge\underline\kappa_{\mathrm T}\|u\|_2^2,
\qquad u\in\mathbb R^{d_{\mathrm T}}.
\label{eq:tobit-terminal-population-bounds}
\end{equation}
For every $\theta\in\Theta_{\mathrm T}$,
\eqref{eq:tobit-global-rsc-lower-bound} implies
\begin{equation*}
\nabla^2\widehat{\mc Q}_B(\theta)
=\frac1{N_B}\left\{
\sum_{j<B}n_j\nabla^2\widehat{\mc L}^{(j)}
(\widehat\theta^{(j)})
+n_B\nabla^2\widehat{\mc L}^{(B)}(\theta)\right\}
\succeq\overline\Sigma_{\mathrm T,B}.
\end{equation*}
If
\$
l^2\eta_{\mathrm T,\infty}
\le\frac{\underline\kappa_{\mathrm T}}2,
\$
then, by \eqref{eq:tobit-terminal-population-bounds}, on the first event in
\eqref{eq:tobit-terminal-concentration-event}, every $v\in\CC(l)$ obeys
\begin{equation*}
v^\top\overline\Sigma_{\mathrm T,B}v
\ge v^\top\Sigma_{\mathrm T,B}v
-\eta_{\mathrm T,\infty}\|v\|_1^2
\ge\{\underline\kappa_{\mathrm T}
-l^2\eta_{\mathrm T,\infty}\}\|v\|_2^2
\ge\frac{\underline\kappa_{\mathrm T}}2\|v\|_2^2.
\end{equation*}
Consequently, for every $\theta_1,\theta_2\in\Theta_{\mathrm T}$ such that
$\theta_1-\theta_2\in\CC(l)$, we obtain
\begin{equation}
\begin{split}
\left\langle
\nabla\widehat{\mc Q}_B(\theta_1)
-\nabla\widehat{\mc Q}_B(\theta_2),
\theta_1-\theta_2\right\rangle
\ge
(\theta_1-\theta_2)^\top\overline\Sigma_{\mathrm T,B}
(\theta_1-\theta_2)
\ge\frac{\underline\kappa_{\mathrm T}}2
\|\theta_1-\theta_2\|_2^2,
\end{split}
\label{eq:tobit-terminal-surrogate-rsc}
\end{equation}
In particular, \eqref{eq:tobit-terminal-surrogate-rsc} implies
\eqref{eq:inherited-surrogate-curvature} for every pair used below.

Set $\lambda=\lambda_B$ and
\begin{equation*}
\begin{split}
\Phi_{\mathrm T,B}
=\frac{9a_{\mathrm T}}{4\kappa^2N_B}
\sum_{j=1}^{B-1}n_j\lambda_j^2,
\qquad
\Gamma_{\mathrm T,B}
=2M_{\mathrm T,l}d_{\mathrm T,B}^2
+\{J_{\mathrm T,u}+m_{\mathrm T}^{1/2}
\eta_{\mathrm T,\infty}\}d_{\mathrm T,B}.
\end{split}
\end{equation*}
Assume, in addition, that
\begin{enumerate}[label=(Tobit-\arabic*),start=4,leftmargin=3\parindent]
\item\label{Tobitassm:terminal-rates}
$0<\kappa\le\underline\kappa_{\mathrm T}/2$ and
\begin{equation}
\begin{gathered}
l^2\eta_{\mathrm T,\infty}
\le\frac{\underline\kappa_{\mathrm T}}2,
\qquad
\chi_{\mathrm T,\infty}\le\frac{p_0\lambda}{8},
\qquad
2M_{\mathrm T,l}\Phi_{\mathrm T,B}\le\frac{p_0\lambda}{8},
\\
\Gamma_{\mathrm T,B}\le\frac{p_0\lambda}{8},
\qquad
\kappa d_{\mathrm T,B}\le\frac\lambda4,
\qquad
c_{\rm aug}\gamma_0q_{\mathrm T}^{1/2}\lambda
\vee d_{\mathrm T,B}<R_0\le\frac{\underline\tau}{2}.
\end{gathered}
\label{eq:tobit-terminal-rate-conditions}
\end{equation}
\end{enumerate}
Define
\begin{equation}
\begin{split}
r_{\mathrm T,a}^{\rm ora}
&=\kappa^{-1}\left\{
\|\omegab_{\overline{\mc M}_{\mathrm T}}^{(B)}\|_2
+2M_{\mathrm T,l}q_{\mathrm T}^{1/2}\phi^{(B)}\right\},
\\
r_{\mathrm T,0}^{\rm ora}
&=\kappa^{-1}\bigg[
\|\Zb_{\mathrm T}^\top\omegab_{\mc M_{\mathrm T}}^{(B)}\|_2
+\|\omegab_{\mc S_{\mathrm T}}^{(B)}\|_2
+2q_{\mathrm T}^{1/2}
\{2M_{\mathrm T,l}\phi^{(B)}+\Gamma_{\mathrm T,B}\}\bigg]
+d_{\mathrm T,B}.
\end{split}
\label{eq:tobit-terminal-oracle-radii}
\end{equation}

\begin{proposition}
\label{prop:tobit-tightening-verification}
Suppose that Assumptions~\mainref{assm:penalty_function},
\mainref{assm:local_alternative}, and \mainref{assm:hessian}, conditions
\ref{Tobitassm:model}--\ref{Tobitassm:terminal-rates}, and those of
Proposition~\ref{prop:tobit-initializer-verification} with $b=B$ hold.
Then, with probability at least
\begin{equation}
1-\frac{2B}{d_{\mathrm T}}-\Pi_{\mathrm T,B}(t_{\mathrm T})
-8\{\log(e d_{\mathrm T})\}^{-2},
\label{eq:tobit-terminal-probability}
\end{equation}
for every $k\ge1$,
\begin{equation}
\begin{split}
\|\widehat\theta_{\mathrm T,a}^{({\rm LLA},k)}-\theta^*\|_2
\le{}&\delta
\|\widehat\theta_{\mathrm T,a}^{({\rm LLA},k-1)}-\theta^*\|_2
+\kappa^{-1}\left\|p'_{\lambda,+}
\{(|\alpha_{\mc S_{\mathrm T}}^*|-\gamma_0\lambda)_+\}\right\|_2
+r_{\mathrm T,a}^{\rm ora},
\\
\|\widehat\theta_{\mathrm T,0}^{({\rm LLA},k)}-\theta^*\|_2
\le{}&\delta
\|\widehat\theta_{\mathrm T,0}^{({\rm LLA},k-1)}-\theta^*\|_2
+\kappa^{-1}\left\|p'_{\lambda,+}
\{(|\alpha_{\mc S_{\mathrm T}}^*|-\gamma_0\lambda)_+\}\right\|_2
+r_{\mathrm T,0}^{\rm ora}.
\end{split}
\label{eq:tobit-terminal-contractions}
\end{equation}
The solution sets of all weighted convex programs are nonempty, and the crude
localization conclusion \eqref{eq:crude-lla-localization} holds with
$s+m=q_{\mathrm T}$.
\end{proposition}

\begin{proof}
Intersect the events in
Proposition~\ref{prop:tobit-initializer-verification} and
Lemma~\ref{lem:tobit-terminal-concentration}.  By the definition of
$\phi^{(B)}$ and \eqref{eq:tobit-final-rates}, we have
\begin{equation*}
\phi^{(B)}
=\frac1{N_B}\sum_{j<B}n_j
\|\widehat\theta^{(j)}-\theta^*\|_2^2
\le\frac{9a_{\mathrm T}}{4\kappa^2N_B}
\sum_{j<B}n_j\lambda_j^2
=\Phi_{\mathrm T,B}.
\end{equation*}
Moreover, $l\ge4q_{\mathrm T}^{1/2}\ge l_{\mathrm T}$, so
\eqref{eq:tobit-final-rates} places every historical displacement in
$\BB(R_0)\cap\CC(l)$.

Lemma~\ref{lem:decomposition} and \eqref{eq:tobit-terminal-rlh} yield
\begin{equation}
\|\nabla\widehat{\mc Q}_B(\theta^*)-\omegab^{(B)}\|_\infty
\le\frac32M_{\mathrm T,l}\phi^{(B)}.
\label{eq:tobit-terminal-gstar-bound}
\end{equation}
By \eqref{eq:tobit-terminal-projection-bound} and
\eqref{eq:tobit-terminal-rate-conditions},
\(\theta_{\mathrm T}^\dagger-\theta^*\in\BB(R_0)\cap\CC(l)\).
Equations \eqref{eq:R-general-center-bound},
\eqref{eq:tobit-terminal-J-row-bound},
\eqref{eq:tobit-terminal-concentration-event}, and
\eqref{eq:tobit-terminal-projection-bound} give
\begin{equation*}
\begin{split}
\frac1{N_B}\sum_{j<B}n_j\|R_j(\theta_{\mathrm T}^\dagger)\|_\infty
&\le2M_{\mathrm T,l}\phi^{(B)}+M_{\mathrm T,l}d_{\mathrm T,B}^2,
\\
\left\|\frac1{N_B}\sum_{j=1}^Bn_j
\{\nabla\widehat{\mc L}^{(j)}(\theta_{\mathrm T}^\dagger)
-\nabla\widehat{\mc L}^{(j)}(\theta^*)\}\right\|_\infty
&\le\{J_{\mathrm T,u}+m_{\mathrm T}^{1/2}
\eta_{\mathrm T,\infty}\}d_{\mathrm T,B}
+\frac12M_{\mathrm T,l}d_{\mathrm T,B}^2.
\end{split}
\end{equation*}
Substitution in \eqref{eq:surrogate-gradient-remainder} yields
\begin{equation}
\|\nabla\widehat{\mc Q}_B(\theta_{\mathrm T}^\dagger)
-\omegab^{(B)}\|_\infty
\le2M_{\mathrm T,l}\phi^{(B)}+\Gamma_{\mathrm T,B}.
\label{eq:tobit-terminal-gdagger-bound}
\end{equation}
Equations \eqref{eq:tobit-terminal-concentration-event} and
\eqref{eq:tobit-terminal-rate-conditions} imply
\begin{equation*}
\begin{split}
&\|\omegab^{(B)}\|_\infty+2M_{\mathrm T,l}\phi^{(B)}
\le\frac{p_0\lambda}{4},\\
&\frac2{p_0}\{\|\omegab^{(B)}\|_\infty
+2M_{\mathrm T,l}\phi^{(B)}+\Gamma_{\mathrm T,B}\}
+\kappa d_{\mathrm T,B}
\le\frac{3\lambda}{4}+\frac\lambda4=\lambda.
\end{split}
\end{equation*}
These are \maineqref{event:theorem2-full} and
\maineqref{event:theorem2-reduced} under
\eqref{eq:tobit-terminal-main-substitution}.
Equations \eqref{eq:tobit-terminal-gstar-bound}--
\eqref{eq:tobit-terminal-gdagger-bound} also give
\begin{equation*}
\max\left\{
\|\nabla\widehat{\mc Q}_B(\theta^*)\|_\infty,
\|\nabla\widehat{\mc Q}_B(\theta_{\mathrm T}^\dagger)\|_\infty
\right\}
\le\frac{p_0\lambda}{2}.
\end{equation*}
The induction argument in the proof of
Theorem~\mainref{thm:batch_b_tightening} therefore yields
\begin{equation*}
\begin{split}
\|\widehat\theta_{\mathrm T,a}^{({\rm LLA},k)}-\theta^*\|_2
&\le\kappa^{-1}\left\{1+\frac{p_0}{2}
(c_{\rm aug}^2+1)^{1/2}\right\}q_{\mathrm T}^{1/2}\lambda,
\\
\|\widehat\theta_{\mathrm T,0}^{({\rm LLA},k)}
-\theta_{\mathrm T}^\dagger\|_2+d_{\mathrm T,B}
&\le\kappa^{-1}\left\{\frac54+\frac{p_0}{\sqrt2}
(c_{\rm aug}^2+1)^{1/2}\right\}q_{\mathrm T}^{1/2}\lambda.
\end{split}
\end{equation*}
Each right-hand side is strictly smaller than
$\kappa^{-1}\{2+p_0(c_{\rm aug}^2+1)^{1/2}\}
q_{\mathrm T}^{1/2}\lambda
=c_{\rm aug}\gamma_0q_{\mathrm T}^{1/2}\lambda<R_0$.
Replace
\eqref{eq:inherited-surrogate-curvature} by
\eqref{eq:tobit-terminal-surrogate-rsc} in the proofs of
Lemma~\ref{lem:intermediate.rate} and
Theorem~\mainref{thm:batch_b_tightening}.  Equations
\eqref{eq:penalty-weight-vector},
\eqref{eq:tobit-terminal-gstar-bound}, and
\eqref{eq:tobit-terminal-gdagger-bound} then give
\eqref{eq:tobit-terminal-contractions} and
\eqref{eq:crude-lla-localization}.
\end{proof}
\subsubsection{Verification of the strong-oracle theorem}
\label{sec:tobit_strong_oracle_verification}

Let $T=\overline{\mc M}_{\mathrm T}$ and
\$
U_{\mathrm T,0}=\operatorname{diag}(\Zb_{\mathrm T},\Ib_{s_{\mathrm T}}),
\qquad
U_{\mathrm T,a}=\Ib_{q_{\mathrm T}},
\qquad
\varrho_{\mathrm T}=\frac{J_{\mathrm T,u}}{J_l},
\$
and, for $e\in\{0,a\}$,
\begin{equation}
A_{\mathrm T,e}
=(J_{\mathrm T,B})_{T^c,T}U_{\mathrm T,e}
\{U_{\mathrm T,e}^\top(J_{\mathrm T,B})_{T,T}
U_{\mathrm T,e}\}^{-1}.
\label{eq:tobit-terminal-tangent-projection}
\end{equation}
Every vector in the column space of $U_{\mathrm T,e}$ belongs to
$\CC(q_{\mathrm T}^{1/2})\subseteq\CC(l)$.  Assumption~\mainref{assm:hessian}
and \eqref{eq:tobit-terminal-J-row-bound} give
\begin{equation}
\lambda_{\min}\{U_{\mathrm T,e}^\top(J_{\mathrm T,B})_{T,T}
U_{\mathrm T,e}\}\ge J_l,
\qquad
\max_{j\in T^c}\|(A_{\mathrm T,e})_{j,\cdot}\|_2
\le\varrho_{\mathrm T}.
\label{eq:tobit-terminal-projection-bounds}
\end{equation}
We note that $q_{\mathrm T}\le p+1$ and
$k_{\mathrm T,l}\ge q_{\mathrm T}$.  For $j\in T^c$, the vector equal to
$1$ at $j$, to
$-U_{\mathrm T,e}(A_{\mathrm T,e})_{j,\cdot}^\top$ on $T$, and to zero
elsewhere has at most $q_{\mathrm T}$ nonzero $\alpha$ coordinates and
squared norm at most $1+\varrho_{\mathrm T}^2$.  Apply
\eqref{eq:tobit-terminal-pooled-score-mgf} to these vectors and take a
union over $j$ and $e$. With probability at least $1-\{\log(e d_{\mathrm T})\}^{-2}$, one has
\begin{equation}
\begin{split}
\max_{e\in\{0,a\}}
\|\omegab_{T^c}^{(B)}
-A_{\mathrm T,e}U_{\mathrm T,e}^\top\omegab_T^{(B)}\|_\infty
\le{}&(1+\varrho_{\mathrm T}^2)^{1/2}
\bigg[
\nu_{\mathrm T}\left\{
\frac{2\{\log(4d_{\mathrm T})+2\log\log(e d_{\mathrm T})\}}
{N_B}\right\}^{1/2}
\\
&\hspace{31mm}
+\frac{2c_{\mathrm T}\{\log(4d_{\mathrm T})
+2\log\log(e d_{\mathrm T})\}}{N_B}
\bigg].
\end{split}
\label{eq:tobit-terminal-projected-score}
\end{equation}
Define
\begin{equation}
\begin{split}
\overline r_{\mathrm T,B}
={}&\kappa^{-1}\{\chi_{\mathrm T}(q_{\mathrm T})
+2M_{\mathrm T,l}q_{\mathrm T}^{1/2}\Phi_{\mathrm T,B}\}
\\
&\vee\left[
\kappa^{-1}\left\{\sqrt2\,\chi_{\mathrm T}(q_{\mathrm T}-r)
+2q_{\mathrm T}^{1/2}
(2M_{\mathrm T,l}\Phi_{\mathrm T,B}+\Gamma_{\mathrm T,B})\right\}
+d_{\mathrm T,B}\right],
\\
\Psi_{\mathrm T,B}
={}&\mathbf1\{T^c\ne\varnothing\}\bigg[
(1+\varrho_{\mathrm T}^2)^{1/2}
\bigg\{
\nu_{\mathrm T}\left[
\frac{2\{\log(4d_{\mathrm T})+2\log\log(e d_{\mathrm T})\}}
{N_B}\right]^{1/2}
\\
&\hspace{24mm}
+\frac{2c_{\mathrm T}\{\log(4d_{\mathrm T})
+2\log\log(e d_{\mathrm T})\}}{N_B}
\bigg\}
\\
&+(1+\varrho_{\mathrm T}q_{\mathrm T}^{1/2})
\bigg\{2M_{\mathrm T,l}\Phi_{\mathrm T,B}+\Gamma_{\mathrm T,B}
\\
&\hspace{31mm}
+\big[q_{\mathrm T}^{1/2}\eta_{\mathrm T,\infty}
+M_{\mathrm T,l}\{\Phi_{\mathrm T,B}^{1/2}
+\overline r_{\mathrm T,B}\}\big]\overline r_{\mathrm T,B}\bigg\}
\bigg].
\end{split}
\label{eq:tobit-terminal-strong-oracle-rates}
\end{equation}
The last condition is as follows.
\begin{enumerate}[label=(Tobit-\arabic*),start=5,leftmargin=3\parindent]
\item\label{Tobitassm:strong-oracle}
There are $\gamma_1>\gamma_0$ and
$\bar\delta\in(5/(4\kappa\gamma_0),1)$ such that
$p_0\in(0,1/2]$, $p'_+(\gamma_1)=0$, $\kappa\gamma_0>5/4$, and
\begin{equation}
\begin{gathered}
s_{\mathrm T}\ge1
\quad\Longrightarrow\quad
\min_{j\in\mc S_{\mathrm T}}|\alpha_j^*|
\ge(\gamma_0+\gamma_1)\lambda,
\\
\overline r_{\mathrm T,B}
<R_0\wedge\frac{\lambda}{5\bar\delta\kappa},
\qquad
\Psi_{\mathrm T,B}\le\frac{p_0\lambda}{2}.
\end{gathered}
\label{eq:tobit-terminal-strong-oracle-conditions}
\end{equation}
\end{enumerate}

On the event in Proposition~\ref{prop:tobit-tightening-verification},
minimize $\widehat{\mc Q}_B$ on the full oracle affine set intersected with
the closed radius-$R_0$ ball centered at $\theta^*$ and on the reduced
oracle affine set intersected with the closed
radius-$(R_0-d_{\mathrm T,B})$ ball centered at
$\theta_{\mathrm T}^\dagger$.  The calculation in the proof of
Lemma~\ref{lem:oracle_rate}, with
\eqref{eq:tobit-terminal-surrogate-rsc} in place of
\eqref{eq:inherited-surrogate-curvature}, gives
\begin{equation}
\begin{gathered}
\|\widehat\theta_{\mathrm T,a}^{\rm ora}-\theta^*\|_2
\le\kappa^{-1}\{\chi_{\mathrm T}(q_{\mathrm T})
+2M_{\mathrm T,l}q_{\mathrm T}^{1/2}\Phi_{\mathrm T,B}\}
\le\overline r_{\mathrm T,B},
\\
\|\widehat\theta_{\mathrm T,0}^{\rm ora}
-\theta_{\mathrm T}^\dagger\|_2
\le\kappa^{-1}\{\sqrt2\,\chi_{\mathrm T}(q_{\mathrm T}-r)
+2q_{\mathrm T}^{1/2}(2M_{\mathrm T,l}\Phi_{\mathrm T,B}
+\Gamma_{\mathrm T,B})\}
\le\overline r_{\mathrm T,B}-d_{\mathrm T,B},
\\
\max_{e\in\{0,a\}}
\|\widehat\theta_{\mathrm T,e}^{\rm ora}-\theta^*\|_2
\le\overline r_{\mathrm T,B}.
\end{gathered}
\label{eq:tobit-terminal-strong-oracle-rate}
\end{equation}
The first two bounds are strictly below the corresponding ball radii by
\eqref{eq:tobit-terminal-strong-oracle-conditions}, and the convexity makes the two minimizers global.  Any two
oracle minimizers differ by a vector supported on $T$, which belongs to
$\CC(q_{\mathrm T}^{1/2})$; therefore
\eqref{eq:tobit-terminal-surrogate-rsc} makes each minimizer unique.

Let
$\mathfrak w_{\mathrm T,e}^{\rm ora}
=\nabla\widehat{\mc Q}_B(\widehat\theta_{\mathrm T,e}^{\rm ora})$.
On the intersection with the event in
\eqref{eq:tobit-terminal-projected-score}, put
$\theta_a=\theta^*$, $\theta_0=\theta_{\mathrm T}^\dagger$, and
$\Delta_e=\widehat\theta_{\mathrm T,e}^{\rm ora}-\theta_e$.
The oracle KKT equations and
\eqref{eq:tobit-terminal-tangent-projection} give
\begin{equation}
U_{\mathrm T,e}^\top\mathfrak w_{\mathrm T,e,T}^{\rm ora}=0,
\qquad
(J_{\mathrm T,B})_{T^c,T}U_{\mathrm T,e}
-A_{\mathrm T,e}U_{\mathrm T,e}^\top
(J_{\mathrm T,B})_{T,T}U_{\mathrm T,e}=0.
\label{eq:tobit-terminal-oracle-cancellation}
\end{equation}
Moreover,
$(\Delta_e)_{T^c}=0$ and
$(\Delta_e)_T=U_{\mathrm T,e}U_{\mathrm T,e}^\top(\Delta_e)_T$.
Subtracting $A_{\mathrm T,e}$ times the first equality in
\eqref{eq:tobit-terminal-oracle-cancellation} and applying the fundamental
theorem of calculus yields
\begin{align}
\mathfrak w_{\mathrm T,e,T^c}^{\rm ora}
={}&\omegab_{T^c}^{(B)}
-A_{\mathrm T,e}U_{\mathrm T,e}^\top\omegab_T^{(B)}
\notag\\
&+\nabla_{T^c}\widehat{\mc Q}_B(\theta_e)-\omegab_{T^c}^{(B)}
-A_{\mathrm T,e}U_{\mathrm T,e}^\top
\{\nabla_T\widehat{\mc Q}_B(\theta_e)-\omegab_T^{(B)}\}
\notag\\
&+\int_0^1\bigg[
\{\nabla^2\widehat{\mc Q}_B(\theta_e+t\Delta_e)
-J_{\mathrm T,B}\}_{T^c,T}U_{\mathrm T,e}
\notag\\
&\hspace{22mm}
-A_{\mathrm T,e}U_{\mathrm T,e}^\top
\{\nabla^2\widehat{\mc Q}_B(\theta_e+t\Delta_e)
-J_{\mathrm T,B}\}_{T,T}U_{\mathrm T,e}
\bigg]U_{\mathrm T,e}^\top(\Delta_e)_T\,dt,
\label{eq:tobit-terminal-oracle-score-identity}
\end{align}
where we used the expansion in Part~(ii) of the proof of Theorem~1 of
\citet{jacobsonzou2025testing} (see (S1.4)--(S1.6) therein).

For every $\theta\in\Theta_{\mathrm T}$, we have
\begin{equation*}
\begin{split}
\nabla^2\widehat{\mc Q}_B(\theta)-J_{\mathrm T,B}
={}&\overline H_{\mathrm T,B}^*-J_{\mathrm T,B}
+\frac1{N_B}\sum_{j<B}n_j
\{\nabla^2\widehat{\mc L}^{(j)}(\widehat\theta^{(j)})
-\nabla^2\widehat{\mc L}^{(j)}(\theta^*)\}
\\
&+\frac{n_B}{N_B}
\{\nabla^2\widehat{\mc L}^{(B)}(\theta)
-\nabla^2\widehat{\mc L}^{(B)}(\theta^*)\}.
\end{split}
\end{equation*}
By \eqref{eq:tobit-terminal-strong-oracle-rate}, for $0\le t\le1$,
$\theta_e+t\Delta_e-\theta^*$ is supported on $T$, belongs to
$\BB(R_0)\cap\CC(l)$, and has Euclidean norm at most
$\overline r_{\mathrm T,B}$.
The definition of $\phi^{(B)}$ gives
$N_B^{-1}\sum_{j<B}n_j
\|\widehat\theta^{(j)}-\theta^*\|_2
\le\Phi_{\mathrm T,B}^{1/2}$.
Equations \eqref{eq:tobit-terminal-concentration-event},
\eqref{eq:tobit-terminal-rlh}, and
\eqref{eq:tobit-terminal-strong-oracle-rate} therefore give
\begin{equation}\label{eq:tobit-terminal-oracle-hessian-bounds}
\begin{split}
&\max_{\substack{e\in\{0,a\},\ 0\le t\le1\\j\in T^c}}
\left\|
\{\nabla^2\widehat{\mc Q}_B(\theta_e+t\Delta_e)
-J_{\mathrm T,B}\}_{j,T}U_{\mathrm T,e}
\right\|_2
\le q_{\mathrm T}^{1/2}\eta_{\mathrm T,\infty}
+M_{\mathrm T,l}\{\Phi_{\mathrm T,B}^{1/2}
+\overline r_{\mathrm T,B}\},\\
&\max_{\substack{e\in\{0,a\}\\0\le t\le1}}
\left\|U_{\mathrm T,e}^\top
\{\nabla^2\widehat{\mc Q}_B(\theta_e+t\Delta_e)
-J_{\mathrm T,B}\}_{T,T}U_{\mathrm T,e}\right\|_2
\le q_{\mathrm T}^{1/2}\left[
q_{\mathrm T}^{1/2}\eta_{\mathrm T,\infty}
+M_{\mathrm T,l}\{\Phi_{\mathrm T,B}^{1/2}
+\overline r_{\mathrm T,B}\}\right].
\end{split}
\end{equation}
Equations \eqref{eq:tobit-terminal-gstar-bound}--
\eqref{eq:tobit-terminal-gdagger-bound} and
\eqref{eq:tobit-terminal-projection-bounds} also yield
\begin{equation}
\begin{split}
\max_{e\in\{0,a\}}
\big\|&\nabla_{T^c}\widehat{\mc Q}_B(\theta_e)-\omegab_{T^c}^{(B)}
-A_{\mathrm T,e}U_{\mathrm T,e}^\top
\{\nabla_T\widehat{\mc Q}_B(\theta_e)-\omegab_T^{(B)}\}
\big\|_\infty
\\
&\le(1+\varrho_{\mathrm T}q_{\mathrm T}^{1/2})
\{2M_{\mathrm T,l}\Phi_{\mathrm T,B}+\Gamma_{\mathrm T,B}\}.
\end{split}
\label{eq:tobit-terminal-oracle-center-bound}
\end{equation}
Since $\|U_{\mathrm T,e}^\top(\Delta_e)_T\|_2
\le\overline r_{\mathrm T,B}$,
\eqref{eq:tobit-terminal-projected-score} and
\eqref{eq:tobit-terminal-oracle-score-identity},
\eqref{eq:tobit-terminal-oracle-hessian-bounds}, and
\eqref{eq:tobit-terminal-oracle-center-bound} imply
\begin{equation}
\max_{e\in\{0,a\}}
\|\mathfrak w_{\mathrm T,e,T^c}^{\rm ora}\|_\infty
\le\Psi_{\mathrm T,B}
\le\frac{p_0\lambda}{2}.
\label{eq:tobit-terminal-oracle-offscore}
\end{equation}

Set
\begin{equation*}
\widehat s_{\mathrm T}
=|\operatorname{supp}(\widehat\alpha_{\mc M_{\mathrm T}^c}^{(B)})|
+s_{\mathrm T},
\qquad
K_{\mathrm T}^{\rm LLA}
=\begin{cases}
0,&\widehat s_{\mathrm T}=0,\\[1mm]
\displaystyle
\left\lceil
\frac{\log(\widehat s_{\mathrm T}^{1/2}/\bar\delta)}
{\log(1/\bar\delta)}
\right\rceil,&\widehat s_{\mathrm T}\ge1.
\end{cases}
\end{equation*}

\begin{proposition}
\label{prop:tobit-strong-oracle-verification}
Suppose that the assumptions and conditions of
Proposition~\ref{prop:tobit-tightening-verification} and
\ref{Tobitassm:strong-oracle} hold. Then, with probability at least
\begin{equation}
1-\frac{2B}{d_{\mathrm T}}-\Pi_{\mathrm T,B}(t_{\mathrm T})
-9\{\log(e d_{\mathrm T})\}^{-2},
\label{eq:tobit-terminal-strong-oracle-probability}
\end{equation}
every sequence of LLA solutions satisfies
\begin{equation}
\widehat\theta_{\mathrm T,e}^{({\rm LLA},k)}
=\widehat\theta_{\mathrm T,e}^{\rm ora},
\qquad e\in\{0,a\},
\qquad k\ge K_{\mathrm T}^{\rm LLA}+1.
\label{eq:tobit-terminal-strong-oracle-recovery}
\end{equation}
\end{proposition}

\begin{proof}
The union bound applied to
\eqref{eq:tobit-terminal-probability} and
\eqref{eq:tobit-terminal-projected-score} gives
\eqref{eq:tobit-terminal-strong-oracle-probability}.  On this event,
Proposition~\ref{prop:tobit-tightening-verification} gives
\eqref{eq:crude-lla-localization}, and
\eqref{eq:tobit-terminal-strong-oracle-rate} and
\eqref{eq:tobit-terminal-oracle-offscore} verify
\maineqref{eq:oracle-recovery-score-radius} after the substitutions in
\eqref{eq:tobit-terminal-main-substitution}.  The remaining inequalities
in Theorem~\mainref{thm:oracle_recovery} are exactly
\eqref{eq:tobit-terminal-strong-oracle-conditions}. 
Every null coordinate in
$\mc B_{e,0}$ belongs to
$\operatorname{supp}(\widehat\alpha_{\mc M_{\mathrm T}^c}^{(B)})$, and at
most $s_{\mathrm T}$ coordinates in $\mc S_{\mathrm T}$ can belong to
$\mc B_{e,0}$.  Hence $|\mc B_{e,0}|\le\widehat s_{\mathrm T}$.
Equation \eqref{eq:bad-set-contraction} gives
\begin{equation*}
|\mc B_{e,k}|^{1/2}
\le\bar\delta^k\widehat s_{\mathrm T}^{1/2}<1,
\qquad k\ge K_{\mathrm T}^{\rm LLA}.
\end{equation*}
Thus $\mc B_{e,k}=\varnothing$ for
$k\ge K_{\mathrm T}^{\rm LLA}$, and
\eqref{eq:strong-oracle-upper} at the next update gives
\eqref{eq:tobit-terminal-strong-oracle-recovery}.
\end{proof}

\subsubsection{Online Bahadur representation}
\label{sec:tobit_bahadur_verification}

Set $T=\overline{\mc M}_{\mathrm T}$.  Under the substitutions in
\eqref{eq:tobit-terminal-main-substitution}, take
$\Jb=J_{\mathrm T,B}$, so that
\begin{equation}
\Jb_0=(J_{\mathrm T,B})_{T,T},
\qquad
\Cb_0=(\widetilde{\Cb}_{\mathrm T},\mathbf0_{r\times s_{\mathrm T}}).
\label{eq:tobit-bahadur-main-objects}
\end{equation}
The definitions of $\Psi_0$ and $\Pb_0$ use the $\Jb_0$ and $\Cb_0$ in
\eqref{eq:tobit-bahadur-main-objects}.

For the remainders in the Bahadur representation, we impose the following rates.
\begin{enumerate}[label=(Tobit-\arabic*),start=6,leftmargin=3\parindent]
\item\label{Tobitassm:bahadur}
As $N_B\to\infty$,
\begin{gather}
\frac{2B}{d_{\mathrm T}}+\Pi_{\mathrm T,B}(t_{\mathrm T})
+\{\log(e d_{\mathrm T})\}^{-2}=o(1),
\qquad
\frac{\overline r_{\mathrm T,B}}{R_0}=o(1),
\qquad
\frac{r^{5/2}q_{\mathrm T}}{N_B}=o(1),
\label{eq:tobit-bahadur-probability-and-lyapunov}
\\
\begin{split}
&(1+M_{\mathrm T,l})\sqrt{\frac{q_{\mathrm T}}{N_B}}
\left\{
\frac{a_{\mathrm T}}{\kappa^2}
\sum_{j=1}^{B}n_j\lambda_j^2
+\frac{a_{\mathrm T}^{1/2}}{\kappa}\overline r_{\mathrm T,B}
\sum_{j=1}^{B}n_j\lambda_j
+n_B\overline r_{\mathrm T,B}^{\,2}
\right\}
\\
&\qquad
+\sqrt{N_B}\,q_{\mathrm T}\eta_{\mathrm T,\infty}
\overline r_{\mathrm T,B}=o(1).
\end{split}
\label{eq:tobit-bahadur-remainder-rate}
\end{gather}
\end{enumerate}

\begin{proposition}\label{prop:tobit-bahadur-verification}
Suppose that Assumptions~\mainref{assm:local_alternative} and
\mainref{assm:hessian}, conditions
\ref{Tobitassm:model}--\ref{Tobitassm:terminal-rates}, those of
Proposition~\ref{prop:tobit-initializer-verification} with $b=B$, and
\ref{Tobitassm:bahadur} hold.  Then
\begin{align}
\widehat\theta_{\mathrm T,0,T}^{\rm ora}-\theta_T^*
={}&
\Jb_0^{-1/2}(\Pb_0-\Ib_{q_{\mathrm T}})\Jb_0^{-1/2}
\frac1{N_B}\sum_{j=1}^B\sum_{i=1}^{n_j}
\nabla_T\ell_{\mathrm T}
(\theta^*;Y_i^{(j)},x_i^{(j)})
\notag\\
&-\Jb_0^{-1}\Cb_0^\top\Psi_0^{-1}\hb_{\mathrm T,B}
+o_p(N_B^{-1/2}),
\label{eq:tobit-reduced-bahadur}
\\
\widehat\theta_{\mathrm T,a,T}^{\rm ora}-\theta_T^*
={}&
-\Jb_0^{-1}
\frac1{N_B}\sum_{j=1}^B\sum_{i=1}^{n_j}
\nabla_T\ell_{\mathrm T}
(\theta^*;Y_i^{(j)},x_i^{(j)})
+o_p(N_B^{-1/2}).
\label{eq:tobit-full-bahadur}
\end{align}
\end{proposition}

\begin{proof}
On the intersection of the events in
Proposition~\ref{prop:tobit-initializer-verification} and
Lemma~\ref{lem:tobit-terminal-concentration}, the complement
probability is at most
$2B/d_{\mathrm T}+\Pi_{\mathrm T,B}(t_{\mathrm T})
+8\{\log(e d_{\mathrm T})\}^{-2}=o(1)$.
On this event, \eqref{eq:tobit-final-rates} and the definition of
$\phi^{(B)}$ yield
\begin{equation}
\begin{gathered}
\|\widehat\theta^{(j)}-\theta^*\|_2
\le\frac{3a_{\mathrm T}^{1/2}\lambda_j}{2\kappa},
\qquad j=1,\ldots,B,
\\
\phi^{(B)}
=\frac1{N_B}\sum_{j=1}^{B-1}n_j
\|\widehat\theta^{(j)}-\theta^*\|_2^2
\le\Phi_{\mathrm T,B}.
\end{gathered}
\label{eq:tobit-bahadur-initializer-rate}
\end{equation}
Equations
\eqref{eq:tobit-terminal-concentration-event},
\eqref{eq:tobit-terminal-gstar-bound},
\eqref{eq:tobit-terminal-gdagger-bound}, and
\eqref{eq:tobit-terminal-oracle-radii} give
$r_{\mathrm T,0}^{\rm ora}\vee r_{\mathrm T,a}^{\rm ora}
\le\overline r_{\mathrm T,B}$.
Hence \eqref{eq:tobit-bahadur-probability-and-lyapunov} gives
$2(r_{\mathrm T,0}^{\rm ora}\vee r_{\mathrm T,a}^{\rm ora})\le R_0$
for all sufficiently large $N_B$.  Applying
Lemma~\ref{lem:oracle_rate}, with
\eqref{eq:tobit-terminal-surrogate-rsc} in place of
\eqref{eq:inherited-surrogate-curvature}, yields oracle minimizers satisfying
\begin{equation}
\max_{e\in\{0,a\}}
\|\widehat\theta_{\mathrm T,e}^{\rm ora}-\theta^*\|_2
\le\overline r_{\mathrm T,B}=o(R_0).
\label{eq:tobit-bahadur-oracle-localization}
\end{equation}
Any two minimizers of the same oracle problem differ only on $T$; hence
\eqref{eq:tobit-terminal-surrogate-rsc} makes each minimizer unique.
Every difference in \eqref{eq:tobit-bahadur-oracle-localization} is
supported on $T$ and therefore belongs to
$\CC(q_{\mathrm T}^{1/2})\subseteq\CC(l)$.
Equation~\eqref{eq:tobit-terminal-rlh} gives
$\bigcap_{j=1}^B\mc E_{\rm rlh}^{(j)}(R_0,l,M_{\mathrm T,l})$, and the
first line of \eqref{eq:tobit-terminal-concentration-event} gives
$\mc H(\eta_{\mathrm T,\infty})$ under
\eqref{eq:tobit-bahadur-main-objects}.

Let $\mc R_{\mathrm T,0}$ and $\mc R_{\mathrm T,a}$ denote the
remainders in \eqref{eq:reduced-oracle-bahadur} and
\eqref{eq:full-oracle-bahadur}, respectively, after
\eqref{eq:tobit-terminal-main-substitution}.  For $e\in\{0,a\}$, the
proof of Lemma~\ref{lem:oracle_bahadur},
Lemma~\ref{lem:decomposition},
\eqref{eq:tobit-terminal-rlh}, and the first line of
\eqref{eq:tobit-terminal-concentration-event} give
\begin{equation*}
\begin{split}
\|\mc R_{\mathrm T,e}\|_2
\lesssim{}&
\frac{M_{\mathrm T,l}q_{\mathrm T}^{1/2}}{N_B}
\bigg[
\sum_{j=1}^{B-1}n_j
\bigg\{
\|\widehat\theta^{(j)}-\theta^*\|_2^2
+\|\widehat\theta^{(j)}-\theta^*\|_2
 \|\widehat\theta_{\mathrm T,e}^{\rm ora}-\theta^*\|_2
\bigg\}
+n_B\|\widehat\theta_{\mathrm T,e}^{\rm ora}-\theta^*\|_2^2
\bigg]\\
&+q_{\mathrm T}\eta_{\mathrm T,\infty}
\|\widehat\theta_{\mathrm T,e}^{\rm ora}-\theta^*\|_2.
\end{split}
\end{equation*}
Equations \eqref{eq:tobit-bahadur-initializer-rate} and
\eqref{eq:tobit-bahadur-oracle-localization} show that both
$\sqrt{N_B}\max_{e\in\{0,a\}}\|\mc R_{\mathrm T,e}\|_2$ and the
left-hand side of \maineqref{eq:inference-remainder-condition}, uniformly in
$e\in\{0,a\}$, are bounded by
\begin{equation*}
\begin{split}
C(1+M_{\mathrm T,l})\sqrt{\frac{q_{\mathrm T}}{N_B}}
\left\{
\frac{a_{\mathrm T}}{\kappa^2}
\sum_{j=1}^{B}n_j\lambda_j^2
+\frac{a_{\mathrm T}^{1/2}}{\kappa}\overline r_{\mathrm T,B}
\sum_{j=1}^{B}n_j\lambda_j
+n_B\overline r_{\mathrm T,B}^{\,2}
\right\}
+C\sqrt{N_B}\,q_{\mathrm T}\eta_{\mathrm T,\infty}
\overline r_{\mathrm T,B},
\end{split}
\end{equation*}
which is $o(1)$ by \eqref{eq:tobit-bahadur-remainder-rate}.  Applying
Lemma~\ref{lem:oracle_bahadur} under
\eqref{eq:tobit-terminal-main-substitution} proves
\eqref{eq:tobit-reduced-bahadur}--\eqref{eq:tobit-full-bahadur}.

For the score covariance, take
\begin{equation}
\Kb_0=
\frac1{N_B}\sum_{j=1}^B\sum_{i=1}^{n_j}
\mathbb E\!\left[
\nabla_T\ell_{\mathrm T}(\theta^*;Y_i^{(j)},x_i^{(j)})
\nabla_T\ell_{\mathrm T}(\theta^*;Y_i^{(j)},x_i^{(j)})^\top
\right],
\label{eq:tobit-bahadur-score-covariance}
\end{equation}
and evaluate $\Vb_0$ and $\Omega_0$ with
\eqref{eq:tobit-bahadur-main-objects} and
\eqref{eq:tobit-bahadur-score-covariance}.
For every row of the design matrix, the Tobit model satisfies
\begin{equation}
\Phi\{-(x_i^{(j)})^\top\alpha\}
+\int_0^\infty
\tau\varphi\{\tau y-(x_i^{(j)})^\top\alpha\}\,\mathrm dy=1,
\qquad \theta\in\Theta_{\mathrm T}.
\label{eq:tobit-observed-mass-identity}
\end{equation}
We differentiate \eqref{eq:tobit-observed-mass-identity} twice at $\theta^*$
and integrate over the covariate row. The Gaussian moments and
\eqref{eq:tobit-bahadur-design-bounds} justify differentiation under the
integral and yield
\begin{equation}
\mathbb E\!\left[
\nabla\ell_{\mathrm T}(\theta^*;Y_i^{(j)},x_i^{(j)})
\nabla\ell_{\mathrm T}(\theta^*;Y_i^{(j)},x_i^{(j)})^\top
\right]
=
\mathbb E\!\left[
\nabla^2\ell_{\mathrm T}(\theta^*;Y_i^{(j)},x_i^{(j)})
\right].
\label{eq:tobit-bartlett-identity}
\end{equation}
Consequently, averaging over $(j,i)$ yields
\begin{equation}
\Kb_0=\Jb_0=(J_{\mathrm T,B})_{T,T}.
\label{eq:tobit-pooled-information-identity}
\end{equation}
By Assumption~\mainref{assm:hessian}, we get
\begin{equation*}
u^\top\Jb_0u\ge J_l\|u\|_2^2,
\qquad \lambda_{\min}(\Jb_0)\ge J_l.
\end{equation*}
The upper population bound in Assumption~\mainref{assm:hessian} gives
\begin{equation}
u^\top\Jb_0u\le J_u\|u\|_2^2,
\qquad\lambda_{\max}(\Jb_0)\le J_u.
\label{eq:tobit-bahadur-information-upper}
\end{equation}
Thus \eqref{eq:tobit-pooled-information-identity} gives
$J_l\le\lambda_{\min}(\Kb_0)\le\lambda_{\max}(\Kb_0)\le J_u$.
These are the bounds used in Lemma~\ref{lem:V0-conditioning} and
Theorem~\mainref{thm:oracle.test}.

Partition $\Omega_0=(\Omega_{0,\alpha},\Omega_{0,\tau})$ according to the
$\alpha$ and $\tau$ coordinates in $T$.  From
\eqref{eq:tobit-pooled-information-identity} and the definition of
$\Omega_0$, $\Omega_0\Jb_0\Omega_0^\top=\Ib_r$ and
$\|\Omega_0\|_2\le J_l^{-1/2}$.
Equations \eqref{eq:tobit-score-formula} and
\eqref{eq:tobit-bahadur-design-bounds}, together with Gaussian moments, imply
\begin{equation*}
\mathbb E\left\|\Omega_0\nabla_T
\ell_{\mathrm T}(\theta^*;Y,X)\right\|_2^3
\le C\left\{\mathbb E\|\Omega_{0,\alpha}X_{T\setminus\{p+1\}}\|_2^3
+\|\Omega_{0,\tau}\|_2^3\right\}.
\end{equation*}
The covariance bound proved before
\eqref{eq:tobit-terminal-pooled-score-mgf} gives
\begin{equation*}
\mathbb E\|\Omega_{0,\alpha}X_{T\setminus\{p+1\}}\|_2^2
\le C\|\Omega_{0,\alpha}\|_{\rm F}^2\le Cr.
\end{equation*}
Also, \eqref{eq:tobit-bahadur-design-bounds} implies
$\|\Omega_{0,\alpha}X_{T\setminus\{p+1\}}\|_2\le Cq_{\mathrm T}^{1/2}$
almost surely. Therefore
\begin{equation}
\begin{aligned}
\mathbb E\|\Omega_{0,\alpha}X_{T\setminus\{p+1\}}\|_2^3
&\le Cq_{\mathrm T}^{1/2}
\mathbb E\|\Omega_{0,\alpha}X_{T\setminus\{p+1\}}\|_2^2
\le Crq_{\mathrm T}^{1/2},
\\
\frac{r^{1/4}}{N_B^{3/2}}
\sum_{j=1}^B\sum_{i=1}^{n_j}
\mathbb E\left\|\Omega_0\nabla_T\ell_{\mathrm T}
(\theta^*;Y_i^{(j)},x_i^{(j)})\right\|_2^3
&\le C\left(\frac{r^{5/2}q_{\mathrm T}}{N_B}\right)^{1/2}=o(1).
\end{aligned}
\label{eq:tobit-bahadur-summed-third-moment}
\end{equation}
The summands are independent and centered under the joint law of $(Y,X)$,
and \eqref{eq:tobit-pooled-information-identity} gives their standardized
total covariance $\Ib_r$. Thus Theorem~1.1 of
\citet{bentkus2005lyapunov} yields, for
$\Zb_r\sim N(\mathbf0_r,\Ib_r)$,
\begin{equation*}
\begin{split}
\sup_{\mathcal A\ {\rm convex}}
\left|\mathbb P\left\{
\frac1{\sqrt{N_B}}\sum_{j=1}^B\sum_{i=1}^{n_j}
\Omega_0\nabla_T\ell_{\mathrm T}
(\theta^*;Y_i^{(j)},x_i^{(j)})\in\mathcal A\right\}
-\mathbb P(\Zb_r\in\mathcal A)\right|
\leq C\left(\frac{r^{5/2}q_{\mathrm T}}{N_B}\right)^{1/2},
\end{split}
\end{equation*}
which is $o(1)$.
Moreover,
$\widetilde{\Cb}_{\mathrm T}\widetilde{\Cb}_{\mathrm T}^\top
\succeq\Cb\Cb^\top$ yields
$\lambda_{\max}\{(
\widetilde{\Cb}_{\mathrm T}\widetilde{\Cb}_{\mathrm T}^\top)^{-1}\}
\le C_{\max}^2$.
\end{proof}

\subsubsection{Chi-square approximation of test statistics}
\label{sec:tobit_test_statistic_verification}

Set $T=\overline{\mc M}_{\mathrm T}$ and apply
\eqref{eq:tobit-terminal-main-substitution}.  In the Tobit specialization
of \maineqref{def:oracle.wald}, the covariance is $\Vb_0$:
the leading score term in \eqref{eq:tobit-full-bahadur}, after multiplication
by $\sqrt{N_B}\Cb_0$, has covariance $\Cb_0\Jb_0^{-1}\Kb_0\Jb_0^{-1}\Cb_0^\top$, which equals $\Vb_0$.
Moreover, since
$\widetilde{\Cb}_{\mathrm T}\widetilde{\Cb}_{\mathrm T}^{\top}
=\Cb\Cb^\top+\tb\tb^\top\succeq\Cb\Cb^\top$, define
\begin{equation}
\begin{gathered}
R_{\mathrm T}
=(\widetilde{\Cb}_{\mathrm T}
\widetilde{\Cb}_{\mathrm T}^{\top})^{-1/2},
\qquad
R_{\mathrm T}^2
\preceq(\Cb\Cb^\top)^{-1},
\qquad
\|R_{\mathrm T}\|_2\le C_{\max},
\\
R_{\mathrm T}\widetilde{\Cb}_{\mathrm T}
\theta_{\mc M_{\mathrm T}}=0
\quad\Longleftrightarrow\quad
\widetilde{\Cb}_{\mathrm T}\theta_{\mc M_{\mathrm T}}=0,
\qquad
(R_{\mathrm T}\widetilde{\Cb}_{\mathrm T})
(R_{\mathrm T}\widetilde{\Cb}_{\mathrm T})^\top=\Ib_r.
\end{gathered}
\label{eq:tobit-row-normalized-constraint}
\end{equation}
For every nonsingular $R$, every positive definite $V$ and $J$, every
full-row-rank $C$, and every conformable $v$, one can easily check that
\begin{equation*}
\begin{split}
(Rv)^\top(RVR^\top)^{-1}(Rv)
&=v^\top V^{-1}v,
\\
(RC)^\top(RVR^\top)^{-1}(RC)
&=C^\top V^{-1}C,
\\
J^{-1}(RC)^\top
\{RCJ^{-1}C^\top R^\top\}^{-1}R
&=J^{-1}C^\top(CJ^{-1}C^\top)^{-1}.
\end{split}
\end{equation*}
With $V=CJ^{-1}C^\top$, the second identity, pre- and post-multiplied by $J^{-1/2}$, gives the same matrix $\Pb_0$.  Right
multiplication of the third identity by $\hb_{\mathrm T,B}$ gives the
same constrained term in \eqref{eq:tobit-reduced-bahadur}.  Therefore
\eqref{eq:tobit-row-normalized-constraint} does not change the feasible
sets, the oracle estimators, the two statistics, or $e_B$.
For every $u\in\mathbb R^{q_{\mathrm T}}$, we obtain
\begin{equation}
\left\|
(R_{\mathrm T}\Vb_0R_{\mathrm T}^\top)^{-1/2}
R_{\mathrm T}\Cb_0\Jb_0^{-1}u
\right\|_2
=\|\Vb_0^{-1/2}\Cb_0\Jb_0^{-1}u\|_2
=\|\Omega_0u\|_2.
\label{eq:tobit-row-normalized-lyapunov}
\end{equation}
Thus \eqref{eq:tobit-bahadur-summed-third-moment} is unchanged by
\eqref{eq:tobit-row-normalized-constraint}.

Suppose that the conditions of
Propositions~\ref{prop:tobit-strong-oracle-verification} and
\ref{prop:tobit-bahadur-verification} hold. Additionally, we impose
the following conditions.

\begin{enumerate}[label=(Tobit-\arabic*),start=7,leftmargin=3\parindent]
\item\label{Tobitassm:test-statistic}
As $N_B\to\infty$,
\begin{gather}
\|R_{\mathrm T}\hb_{\mathrm T,B}\|_2
=O\!\left(\sqrt{\frac r{N_B}}\right),
\label{eq:tobit-test-local-alternative-rate}
\\
\sqrt r\left[
\frac{q_{\mathrm T}}{\sqrt{N_B}}
+
\left\{
\frac{q_{\mathrm T}a_{\mathrm T}}{\kappa^2N_B}
\sum_{j=1}^Bn_j\lambda_j^2
\right\}^{1/2}
\right]=o(1).
\label{eq:tobit-test-statistic-rate}
\end{gather}
\end{enumerate}

\begin{proposition}\label{prop:tobit-test-statistic-verification}
Suppose that the conditions of
Propositions~\ref{prop:tobit-strong-oracle-verification} and
\ref{prop:tobit-bahadur-verification} and
\ref{Tobitassm:test-statistic} hold.  Under
\eqref{eq:tobit-terminal-main-substitution} and
\eqref{eq:tobit-row-normalized-constraint}, condition
\maineqref{eq:feasible-plugin} holds.
Consequently, for every
$k\ge K_{\mathrm T}^{\rm LLA}+1$ and
$T'=T_W^{(k)}$ or $T_S^{(k)}$,
\begin{equation*}
\sup_{z\ge0}
\left|
\PP(T'\le z)-\PP\{\chi^2(r,e_B)\le z\}
\right|\longrightarrow0.
\end{equation*}
Under the null, both tests have asymptotic size $\alpha$. The local-power
conclusion of Corollary~\mainref{cor:test} holds under its local-alternative rate.
\end{proposition}

\begin{proof}
Denote
\begin{equation*}
\begin{split}
\tau_J
=\max_{e\in\{0,a\}}
\left\|\nabla^2_{T,T}\widehat{\mc Q}_B
(\widehat\theta_{\mathrm T,e}^{\rm ora})-\Jb_0\right\|_2,
\qquad
\tau_K
=\left\|\widehat{\Kb}^{(B)}_{T,T}-\Kb_0\right\|_2.
\end{split}
\end{equation*}
For $e\in\{0,a\}$, the definition of $\widehat{\mc Q}_B$ implies
\begin{equation*}
\begin{split}
\nabla^2_{T,T}\widehat{\mc Q}_B
(\widehat\theta_{\mathrm T,e}^{\rm ora})-\Jb_0
&=(\overline H_{\mathrm T,B}^*-J_{\mathrm T,B})_{T,T}
+\frac1{N_B}\sum_{j=1}^{B-1}n_j
\left\{\nabla^2\widehat{\mc L}^{(j)}(\widehat\theta^{(j)})
-\nabla^2\widehat{\mc L}^{(j)}(\theta^*)\right\}_{T,T}
\\
&\qquad
+\frac{n_B}{N_B}
\left\{\nabla^2\widehat{\mc L}^{(B)}
(\widehat\theta_{\mathrm T,e}^{\rm ora})
-\nabla^2\widehat{\mc L}^{(B)}(\theta^*)\right\}_{T,T}.
\end{split}
\end{equation*}
The first line of \eqref{eq:tobit-terminal-concentration-event}, \eqref{eq:tobit-terminal-rlh}, and $q_{\mathrm T}^{1/2}\le l$ yield
\begin{equation*}
\begin{split}
\left\|\nabla^2_{T,T}\widehat{\mc Q}_B
(\widehat\theta_{\mathrm T,e}^{\rm ora})-\Jb_0\right\|_2
\le q_{\mathrm T}\eta_{\mathrm T,\infty}
+\frac{M_{\mathrm T,l}q_{\mathrm T}^{1/2}}{N_B}
\left\{
\sum_{j=1}^{B-1}n_j
\|\widehat\theta^{(j)}-\theta^*\|_2
+n_B\|\widehat\theta_{\mathrm T,e}^{\rm ora}-\theta^*\|_2
\right\}.
\end{split}
\end{equation*}
Equations \eqref{eq:tobit-bahadur-initializer-rate} and
\eqref{eq:tobit-bahadur-oracle-localization} therefore provide
\begin{equation}
\tau_J
\lesssim
q_{\mathrm T}\eta_{\mathrm T,\infty}
+\frac{M_{\mathrm T,l}q_{\mathrm T}^{1/2}}{N_B}
\left\{
\frac{a_{\mathrm T}^{1/2}}{\kappa}
\sum_{j=1}^{B-1}n_j\lambda_j
+n_B\overline r_{\mathrm T,B}
\right\}.
\label{eq:tobit-test-hessian-bound}
\end{equation}
Equations \eqref{eq:tobit-terminal-concentration-rates} and
\eqref{eq:tobit-terminal-strong-oracle-rates} give
$\sqrt{N_B}\,\overline r_{\mathrm T,B}\gtrsim q_{\mathrm T}^{1/2}$.
Since $r\le q_{\mathrm T}$, \eqref{eq:tobit-test-hessian-bound} implies
\begin{equation*}
\begin{split}
\sqrt r\,\tau_J
\lesssim &
\sqrt{N_B}\,q_{\mathrm T}\eta_{\mathrm T,\infty}
\overline r_{\mathrm T,B}
+M_{\mathrm T,l}\sqrt{\frac{q_{\mathrm T}}{N_B}}
\left\{
\frac{a_{\mathrm T}^{1/2}}{\kappa}\overline r_{\mathrm T,B}
\sum_{j=1}^{B-1}n_j\lambda_j
+n_B\overline r_{\mathrm T,B}^{\,2}
\right\},
\end{split}
\end{equation*}
which is $o_p(1)$ by \eqref{eq:tobit-bahadur-remainder-rate}.  Thus
\begin{equation}
\sqrt r\,\tau_J=o_p(1).
\label{eq:tobit-test-hessian-rate}
\end{equation}
Equations \eqref{eq:tobit-terminal-strong-oracle-recovery} and
\eqref{eq:feasible-support-identification}, under
\eqref{eq:tobit-terminal-main-substitution}, give, for
$k\ge K_{\mathrm T}^{\rm LLA}+1$,
\begin{equation}
\widehat{\Jb}^{(k)}_e
=\nabla^2_{T,T}\widehat{\mc Q}_B
(\widehat\theta_{\mathrm T,e}^{\rm ora}),
\qquad
\widehat{\Kb}^{(k)}_e=\widehat{\Kb}^{(B)}_{T,T},
\qquad e\in\{0,a\}.
\label{eq:tobit-feasible-common-hessian}
\end{equation}
Thus \eqref{eq:tobit-test-hessian-rate} applies to each feasible Hessian
at its corresponding oracle estimator.

Define
\begin{equation*}
\widetilde{\Kb}_{\mathrm T,B}
=\frac1{N_B}\sum_{j=1}^B\sum_{i=1}^{n_j}
s_i^{(j)}(\theta^*)s_i^{(j)}(\theta^*)^\top,
\qquad
s_i^{(j)}(\theta)=
\nabla_T\ell_{\mathrm T}
(\theta;Y_i^{(j)},x_i^{(j)}).
\end{equation*}
Equations \eqref{eq:tobit-bartlett-identity} and
\eqref{eq:tobit-pooled-information-identity} give
$\mathbb E\widetilde{\Kb}_{\mathrm T,B}=\Kb_0$.
The score formula \eqref{eq:tobit-score-formula}, the scale bounds
\eqref{eq:tobit-scale-bounds},
\eqref{eq:tobit-bahadur-design-bounds}, and the Gaussian moment bounds imply
\begin{equation*}
\max_{\substack{j\le B,\ i\le n_j\\ k,l\in T}}
\mathbb E\left(
[s_i^{(j)}(\theta^*)]_k^2
[s_i^{(j)}(\theta^*)]_l^2
\right)\lesssim1.
\end{equation*}
The independence in
\eqref{eq:tobit-latent-model}, together with
$\|A\|_2\le\|A\|_{\rm F}$, therefore gives
\begin{equation}
\begin{split}
\mathbb E\left(
\left\|\widetilde{\Kb}_{\mathrm T,B}-\Kb_0\right\|_{\rm F}^2
\right)
\lesssim\frac{q_{\mathrm T}^2}{N_B},
\qquad
\left\|\widetilde{\Kb}_{\mathrm T,B}-\Kb_0\right\|_2
=O_p\left(\frac{q_{\mathrm T}}{\sqrt{N_B}}\right).
\end{split}
\label{eq:tobit-truth-score-covariance-rate}
\end{equation}

For $j\le B$, let
$\Delta_j=\widehat\theta^{(j)}-\theta^*$ and
\begin{equation*}
d_i^{(j)}
=s_i^{(j)}(\widehat\theta^{(j)})-s_i^{(j)}(\theta^*),
\qquad
\mathbf D_{\mathrm T,B}
=\frac1{N_B}\sum_{j=1}^B\sum_{i=1}^{n_j}
d_i^{(j)}(d_i^{(j)})^\top.
\end{equation*}
By \eqref{eq:tobit-final-rates},
$\|\Delta_j\|_1\le4a_{\mathrm T}^{1/2}\|\Delta_j\|_2$ and
$\tau^*+t\Delta_{j,\tau}\ge\underline\tau/2$ for $0\le t\le1$.
The score formula \eqref{eq:tobit-score-formula},
$|\mathfrak g'|\le1$, and \eqref{eq:tobit-bahadur-design-bounds} yield
\begin{equation}
\|d_i^{(j)}\|_2^2
\le Cq_{\mathrm T}\left[
\{1+(Y_i^{(j)})^2\}
\big\{(x_i^{(j)})^\top\Delta_{j,\alpha}\big\}^2
+\{1+(Y_i^{(j)})^4\}\Delta_{j,\tau}^2\right].
\label{eq:tobit-fitted-score-pointwise}
\end{equation}
On the event in \eqref{eq:tobit-gram-event-probability} with $b=B$,
the bound on
$\mathbb E\{(1+Y^2)XX^\top\}$ obtained from
Assumption~\mainref{assm:hessian} and \eqref{eq:tobit-final-rates} give,
simultaneously for $j\le B$,
\[
\frac1{n_j}\sum_{i=1}^{n_j}\{1+(Y_i^{(j)})^2\}
\big\{(x_i^{(j)})^\top\Delta_{j,\alpha}\big\}^2
\le C\|\Delta_j\|_2^2+t_{\mathrm T}\|\Delta_j\|_1^2
\le\frac{Ca_{\mathrm T}\lambda_j^2}{\kappa^2}.
\]
Equations \eqref{eq:tobit-scale-bounds} and
\eqref{eq:tobit-bahadur-design-bounds} yield $\mathbb E(Y^4)\le C$.
Then, Markov's inequality implies
\$
\frac1{N_B}\sum_{j=1}^B\sum_{i=1}^{n_j}
\lambda_j^2\{1+(Y_i^{(j)})^4\}
=O_p\left(\frac1{N_B}\sum_{j=1}^Bn_j\lambda_j^2\right).
\$
Together with
$\Delta_{j,\tau}^2\le\|\Delta_j\|_2^2
\le9a_{\mathrm T}\lambda_j^2/(4\kappa^2)$,
these bounds and \eqref{eq:tobit-fitted-score-pointwise} yield
\$
\|\mathbf D_{\mathrm T,B}\|_2
\le\operatorname{tr}(\mathbf D_{\mathrm T,B})
=O_p\left\{
\frac{q_{\mathrm T}a_{\mathrm T}}{\kappa^2N_B}
\sum_{j=1}^Bn_j\lambda_j^2\right\}.
\$
Thus \eqref{eq:tobit-test-statistic-rate} gives
$\sqrt r\|\mathbf D_{\mathrm T,B}\|_2^{1/2}=o_p(1)$.

For every unit vector $v\in\mathbb R^{q_{\mathrm T}}$, the Cauchy--Schwarz inequality applied to $v^\top s_i^{(j)}(\theta^*)$ and $v^\top d_i^{(j)}$ yields
\begin{equation}
\left\|
\widehat{\Kb}^{(B)}_{T,T}-\widetilde{\Kb}_{\mathrm T,B}
\right\|_2
\le
2\|\widetilde{\Kb}_{\mathrm T,B}\|_2^{1/2}
\|\mathbf D_{\mathrm T,B}\|_2^{1/2}
+\|\mathbf D_{\mathrm T,B}\|_2.
\label{eq:tobit-fitted-score-covariance}
\end{equation}
Condition \eqref{eq:tobit-test-statistic-rate},
\eqref{eq:tobit-truth-score-covariance-rate}, and the preceding bound give
$\|\widetilde{\Kb}_{\mathrm T,B}\|_2=O_p(1)$ and
$\sqrt r\{\|\widetilde{\Kb}_{\mathrm T,B}-\Kb_0\|_2
+\|\mathbf D_{\mathrm T,B}\|_2^{1/2}\}=o_p(1)$.
Since $r\ge1$, this also yields
$\sqrt r\|\mathbf D_{\mathrm T,B}\|_2=o_p(1)$.  Thus
\eqref{eq:tobit-fitted-score-covariance} and the triangle inequality give
\begin{equation}
\begin{split}
\sqrt r\,\tau_K
\leq\sqrt r\,
\|\widetilde{\Kb}_{\mathrm T,B}-\Kb_0\|_2
+2\sqrt r\,\|\widetilde{\Kb}_{\mathrm T,B}\|_2^{1/2}
\|\mathbf D_{\mathrm T,B}\|_2^{1/2}
+\sqrt r\,\|\mathbf D_{\mathrm T,B}\|_2
=o_p(1).
\end{split}
\label{eq:tobit-test-score-covariance-rate}
\end{equation}

Equations \eqref{eq:tobit-test-hessian-rate} and
\eqref{eq:tobit-test-score-covariance-rate} verify
\maineqref{eq:feasible-plugin}.  Equations
\eqref{eq:tobit-pooled-information-identity}--
\eqref{eq:tobit-bahadur-information-upper} provide the bounds on $\Jb_0$ and
$\Kb_0$ used in Lemma~\ref{lem:V0-conditioning} and the proof of
Theorem~\mainref{thm:oracle.test}.  Equations
\eqref{eq:tobit-bahadur-summed-third-moment},
\eqref{eq:tobit-bahadur-probability-and-lyapunov}, and
\eqref{eq:tobit-row-normalized-lyapunov} supply the
triangular-array bound used in Step~1 in the proof of
Theorem~\mainref{thm:oracle.test}.  
Equations
\eqref{eq:tobit-bahadur-oracle-localization} and
\eqref{eq:tobit-bahadur-remainder-rate} give the remaining radius and
remainder bounds, and \eqref{eq:tobit-test-local-alternative-rate} yields
the required rate after \eqref{eq:tobit-row-normalized-constraint}.  The failure
probabilities in \eqref{eq:tobit-bahadur-probability-and-lyapunov} and
\eqref{eq:tobit-terminal-strong-oracle-probability} tend to zero.
Finally, \eqref{eq:tobit-terminal-strong-oracle-recovery} and
\eqref{eq:feasible-support-identification} give the estimator, support,
constraint-matrix, and score-covariance identities in
\eqref{eq:feasible-oracle-plugin-identification}.  For the Hessians, the
definition of $\widehat{\Jb}^{(k)}_e$ gives
\eqref{eq:tobit-feasible-common-hessian}; hence
\eqref{eq:tobit-test-hessian-rate} applies to each feasible Hessian.
The proofs of Theorems~\mainref{thm:oracle.test} and \mainref{thm:test} now give
the claimed uniform approximation.

Under \eqref{eq:tobit-row-normalized-constraint}, we get
\begin{equation*}
\begin{split}
&N_B(R_{\mathrm T}\hb_{\mathrm T,B})^\top
(R_{\mathrm T}\Vb_0R_{\mathrm T}^\top)^{-1}
(R_{\mathrm T}\hb_{\mathrm T,B})
=N_B\hb_{\mathrm T,B}^\top\Vb_0^{-1}\hb_{\mathrm T,B}=e_B,
\\
&\|R_{\mathrm T}\hb_{\mathrm T,B}\|_2
\le C_{\max}\overline\tau\|\hb_B\|_2.
\end{split}
\end{equation*}
Thus $\hb_B=\mathbf0$ gives
$\PP\{\chi^2(r)>\chi_\alpha^2(r)\}=\alpha$, while
$\|\hb_B\|_2=O[\{\min(q_{\mathrm T}-r,r)/N_B\}^{1/2}]$ gives
$
\left|
\PP\{T'>\chi_\alpha^2(r)\}
-\PP\{\chi^2(r,e_B)>\chi_\alpha^2(r)\}
\right|\to0
$.
These are the two conclusions of Corollary~\mainref{cor:test}.
\end{proof}

Since \eqref{eq:tobit-lambda} and Lemma~\ref{lem:batch.size} give
$N_B^{-1}\sum_{j=1}^Bn_j\lambda_j^2
\lesssim N_B^{-1}\log d_{\mathrm T}\{1+\log(N_B/n_1)\}$,
condition \eqref{eq:tobit-test-statistic-rate} is implied by
\begin{equation*}
\sqrt r\left[
\frac{q_{\mathrm T}}{\sqrt{N_B}}
+
\left\{
\frac{q_{\mathrm T}a_{\mathrm T}\log d_{\mathrm T}}
{\kappa^2N_B}
\left\{1+\log\left(\frac{N_B}{n_1}\right)\right\}
\right\}^{1/2}
\right]=o(1).
\end{equation*}
For bounded $B$, the bound
$\sum_{j=1}^Bn_j/N_j\le B$ removes the logarithmic batch factor.
Together with $a_{\mathrm T}\le q_{\mathrm T}$ and $r\le q_{\mathrm T}$,
this shows that $q_{\mathrm T}^3\log d_{\mathrm T}=o(N_B)$ is sufficient
for \eqref{eq:tobit-test-statistic-rate}.  If
$\log d_{\mathrm T}=O(\log q_{\mathrm T})$, this has the same order as
$(s+m)^3\log(s+m)=o(n)$ in Theorem~2 of
\citet{jacobsonzou2025testing}; otherwise the logarithmic factor is
$\log d_{\mathrm T}$ rather than $\log q_{\mathrm T}$.

\subsubsection{Additional simulation II}
\label{sec:tobit_additional_simulation_II}

In this section, we explore the finite-sample inference performance of the renewable Wald and score
tests in the Tobit model. We generate covariate vectors
$X_i\sim N(\mathbf{0}_p,\Sigma)$ with $\Sigma_{jk}=0.5^{|j-k|}$ and
responses satisfying
\[
Y_i^*=1+X_i^\T\betab^*+\epsilon_i,\qquad
Y_i=\max\{0,Y_i^*\},\qquad \epsilon_i\sim N(0,1),
\]
where the errors are independent of the covariates. We set $p=400$,
$n_b=100$, and $B=8$, yielding $N_B=800$, and choose
\[
\betab^*=(h,0,0,1+h,1,\mathbf{0}_{p-5}^\T)^\T,
\qquad h\in\{0,0.05,0.1,0.2,0.4\}.
\]
For each value of $h$, we use 500 independent replications and test the following five hypotheses:
\begin{align*}
H_{0,\mathrm T}^{(1)} &: \ \beta_1^*=0, &
H_{0,\mathrm T}^{(2)} &: \ \beta_4^*-\beta_5^*=0,\\
H_{0,\mathrm T}^{(3)} &: \ (\beta_1^*,\beta_2^*,\beta_3^*)^\T=\mathbf{0}_3, &
H_{0,\mathrm T}^{(4)} &: \ \beta_1^*=0,\ \beta_4^*-\beta_5^*=0,\\
H_{0,\mathrm T}^{(5)} &: \ (\beta_1^*,\beta_2^*,\beta_3^*)^\T=\mathbf{0}_3,
\ \beta_4^*-\beta_5^*=0. &&
\end{align*}
Note that the corresponding degrees of freedom are $1,1,3,2$, and $4$.
The methods, tuning procedures, LLA reweighting mechanism, competitor setups, and nominal level are those in
Section~\mainref{sec:numerical}.

Table~\ref{tab:tobit-additional-II} reports the empirical rejection
probabilities. The null rejection probabilities are broadly comparable
across the online procedures and the pooled benchmark. The first hypothesis, $H_{0,\mathrm T}^{(1)}$, exhibits mildly conservative rejection behavior, whereas $H_{0,\mathrm T}^{(2)}$ tends to reject more often than the nominal level for both the pooled benchmark and most online procedures.

With LLA $\to$ LLA estimation, the online results are generally close to those of the pooled benchmark under both tuning rules.  With Lasso $\to$ LLA estimation and 1-SE tuning, the score test has lower power than the corresponding Wald test at several intermediate departures from the null.
This difference is smaller with LLA $\to$ LLA estimation. 
Under CV/GIC tuning, the Wald and score rejection probabilities
are very similar for both estimators.

\begin{table}[p]
\centering\footnotesize
\setlength{\tabcolsep}{6pt}
\renewcommand{\arraystretch}{1.02}
\caption{Empirical rejection probabilities for the Tobit model with Gaussian covariates, $p=400$, $n_b=100$, $B=8$, and 500 independent replications per value of $h$.}
\label{tab:tobit-additional-II}
\begin{tabular}{llr*{8}{r}}
\toprule
 & & & \multicolumn{4}{c}{1-SE} & \multicolumn{4}{c}{CV/GIC}\\
\cmidrule(lr){4-7}\cmidrule(lr){8-11}
 & & & \multicolumn{2}{c}{Lasso $\to$ LLA} & \multicolumn{2}{c}{LLA $\to$ LLA} & \multicolumn{2}{c}{Lasso $\to$ LLA} & \multicolumn{2}{c}{LLA $\to$ LLA}\\
\cmidrule(lr){4-5}\cmidrule(lr){6-7}\cmidrule(lr){8-9}\cmidrule(lr){10-11}
Hypothesis & $h$ & {\tt Offline} & $T_W$ & $T_S$ & $T_W$ & $T_S$ & $T_W$ & $T_S$ & $T_W$ & $T_S$\\
\midrule
\multirow{5}{*}{$H_{0,\mathrm T}^{(1)}$} & 0.000 & 0.040 & 0.036 & 0.036 & 0.040 & 0.038 & 0.040 & 0.040 & 0.040 & 0.040\\
 & 0.050 & 0.264 & 0.242 & 0.216 & 0.256 & 0.258 & 0.252 & 0.252 & 0.264 & 0.264\\
 & 0.100 & 0.710 & 0.692 & 0.644 & 0.708 & 0.700 & 0.700 & 0.700 & 0.710 & 0.710\\
 & 0.200 & 1.000 & 1.000 & 0.996 & 1.000 & 1.000 & 1.000 & 1.000 & 1.000 & 1.000\\
 & 0.400 & 1.000 & 1.000 & 1.000 & 1.000 & 1.000 & 1.000 & 1.000 & 1.000 & 1.000\\
\addlinespace
\multirow{5}{*}{$H_{0,\mathrm T}^{(2)}$} & 0.000 & 0.066 & 0.066 & 0.058 & 0.066 & 0.066 & 0.066 & 0.066 & 0.066 & 0.066\\
 & 0.050 & 0.128 & 0.126 & 0.128 & 0.126 & 0.120 & 0.126 & 0.126 & 0.126 & 0.126\\
 & 0.100 & 0.338 & 0.318 & 0.278 & 0.338 & 0.314 & 0.334 & 0.334 & 0.340 & 0.340\\
 & 0.200 & 0.784 & 0.732 & 0.666 & 0.770 & 0.712 & 0.770 & 0.756 & 0.768 & 0.756\\
 & 0.400 & 1.000 & 1.000 & 0.998 & 1.000 & 1.000 & 1.000 & 1.000 & 1.000 & 1.000\\
\addlinespace
\multirow{5}{*}{$H_{0,\mathrm T}^{(3)}$} & 0.000 & 0.062 & 0.050 & 0.050 & 0.054 & 0.056 & 0.052 & 0.052 & 0.058 & 0.058\\
 & 0.050 & 0.174 & 0.144 & 0.152 & 0.162 & 0.168 & 0.158 & 0.158 & 0.170 & 0.170\\
 & 0.100 & 0.552 & 0.508 & 0.472 & 0.548 & 0.548 & 0.528 & 0.528 & 0.550 & 0.550\\
 & 0.200 & 0.996 & 0.990 & 0.986 & 0.992 & 0.992 & 0.994 & 0.994 & 0.994 & 0.994\\
 & 0.400 & 1.000 & 1.000 & 1.000 & 1.000 & 1.000 & 1.000 & 1.000 & 1.000 & 1.000\\
\addlinespace
\multirow{5}{*}{$H_{0,\mathrm T}^{(4)}$} & 0.000 & 0.056 & 0.048 & 0.046 & 0.052 & 0.054 & 0.054 & 0.054 & 0.054 & 0.054\\
 & 0.050 & 0.290 & 0.268 & 0.226 & 0.290 & 0.272 & 0.282 & 0.282 & 0.290 & 0.290\\
 & 0.100 & 0.750 & 0.734 & 0.658 & 0.750 & 0.740 & 0.734 & 0.734 & 0.748 & 0.748\\
 & 0.200 & 1.000 & 1.000 & 0.996 & 1.000 & 1.000 & 1.000 & 1.000 & 1.000 & 1.000\\
 & 0.400 & 1.000 & 1.000 & 1.000 & 1.000 & 1.000 & 1.000 & 1.000 & 1.000 & 1.000\\
\addlinespace
\multirow{5}{*}{$H_{0,\mathrm T}^{(5)}$} & 0.000 & 0.062 & 0.054 & 0.054 & 0.060 & 0.066 & 0.056 & 0.056 & 0.060 & 0.060\\
 & 0.050 & 0.212 & 0.194 & 0.200 & 0.200 & 0.204 & 0.200 & 0.200 & 0.210 & 0.210\\
 & 0.100 & 0.658 & 0.634 & 0.592 & 0.656 & 0.650 & 0.650 & 0.650 & 0.660 & 0.660\\
 & 0.200 & 0.998 & 0.998 & 0.996 & 0.998 & 0.998 & 0.998 & 0.998 & 0.998 & 0.998\\
 & 0.400 & 1.000 & 1.000 & 1.000 & 1.000 & 1.000 & 1.000 & 1.000 & 1.000 & 1.000\\
\bottomrule
\end{tabular}
\end{table}

\subsection{Nonparametric Box--Cox model}

Fix an integer $K\ge1$, nonrandom thresholds
$y_{01}<\cdots<y_{0K}$, and nonrandom weights $w_1,\ldots,w_K>0$
satisfying $\sum_{k=1}^K w_k=1$.  For $j=1,\ldots,B$ and
$i=1,\ldots,n_j$, suppose
\begin{equation}
g(Y_i^{(j)})=(x_i^{(j)})^\top\beta^*+\varepsilon_i^{(j)},
\label{eq:boxcox-model}
\end{equation}
where $g:\mathbb R\to\mathbb R$ is strictly increasing.  The observations are independent draws from the
joint law of $(Y,X)$, and $\varepsilon$ is independent of $X$ with law
$N(0,1)$.  All probabilities and expectations below refer to this joint
law.  Let
\[
b_k^*=g(y_{0k}),\qquad
D_{ik}^{(j)}=\mathbf 1\{Y_i^{(j)}\ge y_{0k}\},\qquad
S_{ik}^{(j)}=2D_{ik}^{(j)}-1,
\]
and define
\$
\theta^*=\begin{pmatrix}\beta^*\\ \boldsymbol b^*\end{pmatrix},
\qquad
\boldsymbol b^*=(b_1^*,\ldots,b_K^*)^\top,
\qquad
d_{\mathrm{BC}}=p+K,
\$
and
\$
z_{ik}^{(j)}
=
\begin{pmatrix}x_i^{(j)}\\-e_k\end{pmatrix},
\qquad
\eta_{ik}^{*(j)}=(z_{ik}^{(j)})^\top\theta^*,
\$
where $e_k$ is the $k$th canonical basis vector in $\mathbb R^K$.
Then
\begin{equation}
\mathbb P\{D_{ik}^{(j)}=1\mid x_i^{(j)}\}
=\Phi(\eta_{ik}^{*(j)}).
\label{eq:boxcox-probit-law}
\end{equation}
Equation~\eqref{eq:boxcox-probit-law} is the composite-probit parametrization developed in
\citet[Section.~2.1, display~(4)]{zhouzou2024boxcox} and
\citet[Section~2.1, displays~(3)--(4)]{zhouzou2024testing}.
We note that the thresholds are fixed before observing the data.

Let $L_{\mathrm P}(t)=-\log\Phi(t)$ and use $\mathfrak g$ and $\mathfrak h$
from \eqref{eq:tobit-mills-definitions}.  By
\citet[Lemma~2]{zhouzou2024boxcox} and
\eqref{eq:tobit-mills-properties}, one has
\begin{equation}
\begin{gathered}
L_{\mathrm P}'(t)=-\mathfrak g(t),\qquad
L_{\mathrm P}''(t)=\mathfrak h(t)\in(0,1),\qquad
|L_{\mathrm P}'''(t)|\le4.3,\\
|\log\mathfrak h(t+r)-\log\mathfrak h(t)|
\le\max\{|t|,1\}|r|+r^2/2,\qquad t,r\in\mathbb R.
\end{gathered}
\label{eq:boxcox-probit-derivatives}
\end{equation}
For the last inequality, the bound between displays (53) and (54)
in the proof of \citet[Theorem~2]{zhouzou2024boxcox} gives
$|\mathfrak h'(t)|\le\max\{|t|,1\}\mathfrak h(t)$.
Integrate along $t+sr$, $0\le s\le1$, using
$\max\{|t+sr|,1\}\le\max\{|t|,1\}+s|r|$.
For $s_k(y)=2\mathbf 1\{y\ge y_{0k}\}-1$ and
$z_k(x)=(x^\top,-e_k^\top)^\top$, define
\$
\ell_{\mathrm{BC}}(\theta;y,x)
=
\sum_{k=1}^K w_k
L_{\mathrm P}\{s_k(y)z_k(x)^\top\theta\},
\qquad
\widehat{\mc L}^{(j)}(\theta)
=
\frac1{n_j}\sum_{i=1}^{n_j}
\ell_{\mathrm{BC}}(\theta;Y_i^{(j)},x_i^{(j)}).
\$
This is the negative composite-probit objective function introduced in
\citet[Section~3.3, display~(24)]{zhouzou2024boxcox} and the negative of the
composite log-likelihood in
\citet[Section~3.2, display~(9)]{zhouzou2024testing}.  Taking a direct differentiation yields
\begin{equation}
\begin{split}
\nabla\ell_{\mathrm{BC}}(\theta;y,x)
&=
\sum_{k=1}^K w_kL_{\mathrm P}'\{s_k(y)z_k(x)^\top\theta\}
s_k(y)z_k(x),
\\
\nabla^2\ell_{\mathrm{BC}}(\theta;y,x)
&=
\sum_{k=1}^K w_k\mathfrak h\{s_k(y)z_k(x)^\top\theta\}
z_k(x)z_k(x)^\top.
\end{split}
\label{eq:boxcox-score-hessian}
\end{equation}
Thus $\ell_{\mathrm{BC}}(\cdot;y,x)$ is twice continuously differentiable and convex \citep{zhouzou2024boxcox,zhouzou2024testing}.

Let
\begin{equation*}
\begin{split}
\mc P_{\mathrm{BC}}&=\{1,\ldots,p\},
\qquad
\mc U_{\mathrm{BC}}=\{p+1,\ldots,p+K\},
\\
\mc M_{\mathrm{BC}}&=\mc M\cup\mc U_{\mathrm{BC}},
\qquad
\widetilde{\Cb}_{\mathrm{BC}}
=\begin{pmatrix}\Cb&\mathbf0_{r\times K}\end{pmatrix}.
\end{split}
\end{equation*}
Under the ordering
$\theta_{\mc M_{\mathrm{BC}}}=(\beta_{\mc M}^\top,\boldsymbol b^\top)^\top$,
display (6) of \citet{zhouzou2024testing} becomes
\[
\Cb\beta_{\mc M}^*=\tb
\quad\Longleftrightarrow\quad
\widetilde{\Cb}_{\mathrm{BC}}
\theta_{\mc M_{\mathrm{BC}}}^*=\tb.
\]

In the remainder of this subsection, the definitions in
Section~\mainref{sec:finite_sample_theory} are applied with $(\betab^*,p,\ell,\widehat\betab^{(j)},\mc M,\Cb)$ replaced by $(\theta^*,d_{\mathrm{BC}},\ell_{\mathrm{BC}},
\widehat\theta^{(j)},\mc M_{\mathrm{BC}},
\widetilde{\Cb}_{\mathrm{BC}})$.
The corrected surrogate $\widehat{\mc Q}_j$ is
\maineqref{def:q.function} under this substitution, and the initializer is
\begin{equation}
\widehat\theta^{(j)}
\in
\argmin_{\theta=(\beta^\top,\boldsymbol b^\top)^\top}
\{\widehat{\mc Q}_j(\theta)+\lambda_j\|\beta\|_1\}.
\label{eq:boxcox-online-initializer}
\end{equation}
Note that only the coordinates in $\mc P_{\mathrm{BC}}$ are penalized.

Define
\begin{equation*}
\mc A_{\mathrm{BC}}
=\operatorname{supp}(\beta^*)\cup\mc U_{\mathrm{BC}},
\qquad
a_{\mathrm{BC}}=|\mc A_{\mathrm{BC}}|,
\end{equation*}
and
\[
c_{\mathrm{BC},B}
=
\max_{j\le B}\max_{i\le n_j}\max_{k\le K}
\|z_{ik}^{(j)}\|_\infty,
\qquad
H_{\mathrm{BC},B}
=
\max_{j\le B}\max_{i\le n_j}\max_{k\le K}
|\eta_{ik}^{*(j)}|.
\]

\subsubsection{Assumptions for the nonparametric Box--Cox model}

\begin{enumerate}[label=(Box--Cox-\arabic*),leftmargin=3\parindent]
\item\label{BoxCoxassm:model}
Equation~\eqref{eq:boxcox-model} holds with the independent observations
and normal errors specified above.  The integer $K$ and the positive
weights are fixed.

\item\label{BoxCoxassm:design}
For a constant $B_1\ge1$, almost surely,
\begin{equation}
\max_{j\le B}\max_{1\le a\le d_{\mathrm{BC}}}
\frac1{n_j}\sum_{i=1}^{n_j}\sum_{k=1}^K
w_k(z_{ik,a}^{(j)})^2\le B_1,
\qquad H_{\mathrm{BC},B}\le B_1.
\label{eq:boxcox-design-conditions}
\end{equation}

\item\label{BoxCoxassm:re}
For a constant $\underline\kappa_{\mathrm{BC}}>0$,
\begin{equation}
\min_{j\le B}
\inf_{\substack{u\ne0:\\
\|u_{\mc A_{\mathrm{BC}}^c}\|_1
\le3\|u_{\mc A_{\mathrm{BC}}}\|_1}}
\frac{u^\top\mathbb E\{\nabla^2\widehat{\mc L}^{(j)}(\theta^*)\}u}
{\|u\|_2^2}
\ge\underline\kappa_{\mathrm{BC}}.
\label{eq:boxcox-design-re}
\end{equation}
\end{enumerate}
The inactive-column bound and the active Gram-matrix upper bound in condition
(A1) of \citet[Section~3.2.1]{zhouzou2024testing} imply the fixed-design
counterpart of the first inequality in
\eqref{eq:boxcox-design-conditions}.  The second inequality in
\eqref{eq:boxcox-design-conditions} is the uniform batchwise version
of the third condition in (A1) therein.
Under the stated i.i.d.\ model, by \eqref{eq:boxcox-design-conditions}, one has $|X_a|\le B_1^{1/2}$ and $c_{\mathrm{BC},B}\le B_1^{1/2}$, almost surely.
By \eqref{eq:boxcox-probit-law} and
\eqref{eq:tobit-mills-definitions}, we get
\begin{align*}
\mathbb E\{\mathfrak h(S_{ik}^{(j)}\eta_{ik}^{*(j)})
\mid x_i^{(j)}\}
&=\frac{\varphi(\eta_{ik}^{*(j)})^2}
{\Phi(\eta_{ik}^{*(j)})\Phi(-\eta_{ik}^{*(j)})},\\
u^\top\mathbb E\{\nabla^2\widehat{\mc L}^{(j)}(\theta^*)\}u
&=\mathbb E\left[
\frac1{n_j}\sum_{i=1}^{n_j}\sum_{k=1}^K
\frac{w_k\varphi(\eta_{ik}^{*(j)})^2
\{(z_{ik}^{(j)})^\top u\}^2}
{\Phi(\eta_{ik}^{*(j)})\Phi(-\eta_{ik}^{*(j)})}
\right].
\end{align*}
Condition~\ref{BoxCoxassm:re} is the batchwise population
restricted-eigenvalue formulation of the curvature condition (C),
display (27), in \citet[Theorem~2]{zhouzou2024boxcox}.  
The display \eqref{eq:boxcox-design-re} supplies the curvature used
in Theorem~\mainref{thm:batch_b_init}, with the probit weights included
in the expectation.

The support inequality
$\|u_{\mc A_{\mathrm{BC}}^c}\|_1
\le3\|u_{\mc A_{\mathrm{BC}}}\|_1$ implies
\begin{equation}
\|u\|_1\le4\|u_{\mc A_{\mathrm{BC}}}\|_1
\le4a_{\mathrm{BC}}^{1/2}\|u\|_2.
\label{eq:boxcox-main-cone-implication}
\end{equation}
Under Assumption~\mainref{assm:hessian}, let
\begin{equation}
B_3=
\frac{B_1^{1/2}(J_u+\underline\kappa_{\mathrm{BC}}/2)}
{\inf_{|t|\le B_1}\mathfrak h(t)},
\qquad M_{\mathrm{BC}}=4.3B_3,
\qquad G_{\mathrm{BC}}=\mathfrak g(-B_1).
\label{eq:boxcox-initializer-constants}
\end{equation}
The denominator is positive by continuity and
\eqref{eq:boxcox-probit-derivatives}, so these constants do not depend
on $a_{\mathrm{BC}}$ or $d_{\mathrm{BC}}$.  Choose
\begin{equation}
\begin{gathered}
0<\kappa\le
\frac{\underline\kappa_{\mathrm{BC}}}{2}
\exp\{-8B_1^{3/2}a_{\mathrm{BC}}^{1/2}R_0
      -32B_1a_{\mathrm{BC}}R_0^2\},\\
\lambda_j=C\left\{\frac{\log d_{\mathrm{BC}}}{N_j}\right\}^{1/2},
\qquad C\ge8(2B_1)^{1/2}.
\end{gathered}
\label{eq:boxcox-initializer-tuning}
\end{equation}

\begin{proposition}
\label{prop:boxcox-initializer-verification}
Suppose Assumption~\mainref{assm:hessian},
conditions~\ref{BoxCoxassm:model}--\ref{BoxCoxassm:re}, and
\eqref{eq:boxcox-initializer-tuning} hold,
\begin{equation*}
R_0>
\max_{1\le j\le B}
\frac{3}{2\kappa}a_{\mathrm{BC}}^{1/2}\lambda_j,
\end{equation*}
and, for $j=2,\ldots,B$,
\begin{equation}
\left\{\frac{N_j}{a_{\mathrm{BC}}^2\log d_{\mathrm{BC}}}\right\}^{1/2}
\ge
\frac{27CM_{\mathrm{BC}}}{2\kappa^2}
\left\{1+\log\left(\frac{N_{j-1}}{n_1}\right)\right\}.
\label{eq:boxcox-online-sample-condition}
\end{equation}
Then, with probability at least
\[
1-\frac{2B}{d_{\mathrm{BC}}}
-2d_{\mathrm{BC}}^2\sum_{j=1}^B
\exp\left\{-\frac{n_j\underline\kappa_{\mathrm{BC}}^2}
{2048B_1^2a_{\mathrm{BC}}^2}\right\},
\]
every initializer exists and, simultaneously for $j=1,\ldots,B$,
\begin{equation}
\begin{split}
\|\widehat\theta^{(j)}-\theta^*\|_2
&\le
\frac{3}{2\kappa}a_{\mathrm{BC}}^{1/2}\lambda_j,
\\
\|\widehat\theta^{(j)}-\theta^*\|_1
&\le
\frac6\kappa a_{\mathrm{BC}}\lambda_j,
\\
\|(\widehat\theta^{(j)}-\theta^*)_{\mc A_{\mathrm{BC}}^c}\|_1
&\le
3\|(\widehat\theta^{(j)}-\theta^*)_{\mc A_{\mathrm{BC}}}\|_1,
\qquad
\widehat\theta^{(j)}-\theta^*
\in\BB(R_0)\cap\CC(4a_{\mathrm{BC}}^{1/2}).
\end{split}
\label{eq:boxcox-initializer-rates}
\end{equation}
\end{proposition}

\begin{proof}
Let $\xi_i^{(j)}=\nabla\ell_{\mathrm{BC}}
(\theta^*;Y_i^{(j)},x_i^{(j)})$.  For each threshold, the normality gives
\[
-S_{ik}^{(j)}\mathfrak g(S_{ik}^{(j)}\eta_{ik}^{*(j)})
=-\mathbb E(\varepsilon_i^{(j)}
 \mid D_{ik}^{(j)},x_i^{(j)}).
\]
Jensen's inequality, followed by expectation over $X$, therefore yields
\begin{equation}
\mathbb E\xi_i^{(j)}=0,\qquad
\mathbb E\exp(t\xi_{i,a}^{(j)})\le\exp(B_1t^2/2),
\qquad t\in\mathbb R.
\label{eq:boxcox-score-centering}
\end{equation}
For the second inequality, apply H\"older's inequality to the $K$ terms before integrating over $X$; their squared coefficient average is at
most $B_1$ by \eqref{eq:boxcox-design-conditions}.

By \eqref{eq:boxcox-score-hessian} and
\eqref{eq:boxcox-design-conditions}, every Hessian entry
belongs to $[-B_1,B_1]$.  Hoeffding's inequality and a union over entries yield, for each $j$,
\[
\mathbb P\left\{
\left\|\nabla^2\widehat{\mc L}^{(j)}(\theta^*)
-\mathbb E\{\nabla^2\widehat{\mc L}^{(j)}(\theta^*)\}\right\|_{\max}
>\frac{\underline\kappa_{\mathrm{BC}}}{32a_{\mathrm{BC}}}\right\}
\le2d_{\mathrm{BC}}^2
\exp\left\{-\frac{n_j\underline\kappa_{\mathrm{BC}}^2}
{2048B_1^2a_{\mathrm{BC}}^2}\right\}.
\]
Equations \eqref{eq:boxcox-score-centering} and
\eqref{eq:boxcox-initializer-tuning} also give
$\mathbb P\{\|\omegab^{(j)}\|_\infty>\lambda_j/4\}
\le2/d_{\mathrm{BC}}$.  A union bound over batches yields
\begin{equation}
\begin{split}
\mathbb P\bigg[
&\max_{j\le B}\frac{\|\omegab^{(j)}\|_\infty}{\lambda_j}\le\frac14,
\quad \max_{j\le B}\left\|\nabla^2\widehat{\mc L}^{(j)}(\theta^*)
-\mathbb E\{\nabla^2\widehat{\mc L}^{(j)}(\theta^*)\}\right\|_{\max}
\le\frac{\underline\kappa_{\mathrm{BC}}}{32a_{\mathrm{BC}}}\bigg]
\\
&\quad\ge1-\frac{2B}{d_{\mathrm{BC}}}
-2d_{\mathrm{BC}}^2\sum_{j=1}^B
\exp\left\{-\frac{n_j\underline\kappa_{\mathrm{BC}}^2}
{2048B_1^2a_{\mathrm{BC}}^2}\right\}.
\end{split}
\label{eq:boxcox-score-event}
\end{equation}
Here $\omegab^{(j)}$ is the pooled score in
Section~\mainref{sec:finite_sample_theory}.

Since the observations have the same joint law,
$\mathbb E\{\nabla^2\widehat{\mc L}^{(j)}(\theta^*)\}=\Jb$.
On \eqref{eq:boxcox-score-event},
Assumption~\mainref{assm:hessian} and the lower truth weights in
\eqref{eq:boxcox-design-conditions} give, for every $u$,
\[
\begin{split}
\frac1{n_j}\sum_{i,k}w_k\{(z_{ik}^{(j)})^\top u\}^2
\leq\frac{u^\top\nabla^2\widehat{\mc L}^{(j)}(\theta^*)u}
{\inf_{|t|\le B_1}\mathfrak h(t)}\leq\frac{J_u\|u\|_2^2+
\underline\kappa_{\mathrm{BC}}\|u\|_1^2/(32a_{\mathrm{BC}})}
{\inf_{|t|\le B_1}\mathfrak h(t)}.
\end{split}
\]
For $u\in\CC(4a_{\mathrm{BC}}^{1/2})$, the last expression is at most
$(B_3/B_1^{1/2})\|u\|_2^2$.
Consequently, the coordinate bound and the Cauchy--Schwarz inequality give, for
$u,v\in\CC(4a_{\mathrm{BC}}^{1/2})$,
\[
\max_a\frac1{n_j}\sum_{i,k}w_k|z_{ik,a}^{(j)}|
 |(z_{ik}^{(j)})^\top u|\,|(z_{ik}^{(j)})^\top v|
\le B_3\|u\|_2\|v\|_2.
\]
The mean-value theorem and \eqref{eq:boxcox-probit-derivatives}
therefore yield
\begin{equation}
\left\|\{\nabla^2\widehat{\mc L}^{(j)}(\theta^*+v)
-\nabla^2\widehat{\mc L}^{(j)}(\theta^*)\}u\right\|_\infty
\le M_{\mathrm{BC}}\|v\|_2\|u\|_2.
\label{eq:boxcox-initializer-rlh}
\end{equation}

Suppose $\theta_1-\theta_2$ and $\theta_2-\theta^*$ have norm at most
$R_0$ and satisfy the support inequality in
\eqref{eq:boxcox-design-re}.  Along their segment,
\eqref{eq:boxcox-main-cone-implication} gives
$\|\theta_2+t(\theta_1-\theta_2)-\theta^*\|_1
\le8a_{\mathrm{BC}}^{1/2}R_0$ for $0\le t\le1$.
On \eqref{eq:boxcox-score-event},
\eqref{eq:boxcox-design-re} and
\eqref{eq:boxcox-main-cone-implication} yield
\[
u^\top\nabla^2\widehat{\mc L}^{(j)}(\theta^*)u
\ge\left\{\underline\kappa_{\mathrm{BC}}
-16a_{\mathrm{BC}}\frac{\underline\kappa_{\mathrm{BC}}}
{32a_{\mathrm{BC}}}\right\}\|u\|_2^2
=\frac{\underline\kappa_{\mathrm{BC}}}{2}\|u\|_2^2
\]
for every $u$ satisfying the support inequality in
\eqref{eq:boxcox-design-re}.  Since $|S_{ik}^{(j)}\eta_{ik}^{*(j)}|\le B_1$
and $c_{\mathrm{BC},B}\le B_1^{1/2}$, apply the last inequality in
\eqref{eq:boxcox-probit-derivatives} to each Hessian along
this segment.  Integration and \eqref{eq:boxcox-initializer-tuning} yield
\begin{equation}
\begin{split}
\langle\nabla\widehat{\mc L}^{(j)}(\theta_1)
-\nabla\widehat{\mc L}^{(j)}(\theta_2),\theta_1-\theta_2\rangle
\ge
\frac{\underline\kappa_{\mathrm{BC}}}{2}
 e^{-8B_1^{3/2}a_{\mathrm{BC}}^{1/2}R_0
       -32B_1a_{\mathrm{BC}}R_0^2}
\|\theta_1-\theta_2\|_2^2
\ge\kappa\|\theta_1-\theta_2\|_2^2.
\end{split}
\label{eq:boxcox-initializer-rsc}
\end{equation}
The same lower bound holds for the historical Hessians evaluated at
$\widehat\theta^{(h)}$ whenever those estimators satisfy
\eqref{eq:boxcox-initializer-rates}.  Averaging these quadratic forms as
in the proof of Lemma~\ref{lem:surrogate-curvature-inheritance} verifies
\eqref{eq:inherited-surrogate-curvature} on the initializer support cone.

Inductively, the initializer bounds through batch $j-1$ and
Lemma~\ref{lem:batch.size} yield
\$
\phi^{(j)}
\le
\frac{9C^2a_{\mathrm{BC}}\log d_{\mathrm{BC}}}
{4\kappa^2N_j}
\left\{1+\log\left(\frac{N_{j-1}}{n_1}\right)\right\}.
\$
For $j=1$, $\phi^{(1)}=0$; for $j\ge2$,
\eqref{eq:boxcox-online-sample-condition} implies
$3M_{\mathrm{BC}}\phi^{(j)}/2\le\lambda_j/4$.
Together with \eqref{eq:boxcox-score-event}, this verifies the score
inequality in \maineqref{batch_b_event}.
Minimize the objective over the closed radius-$R_0$ ball
centered at $\theta^*$.  The basic inequality in the proof of
Theorem~\mainref{thm:batch_b_init}, with the penalty subgradient zero on
$\mc U_{\mathrm{BC}}\subset\mc A_{\mathrm{BC}}$, yields
$\|u_{\mc A_{\mathrm{BC}}^c}\|_1
\le3\|u_{\mc A_{\mathrm{BC}}}\|_1$.
The unpenalized threshold coordinates contribute as in display (49) of
\citet{zhouzou2024boxcox}.
Equations \eqref{eq:boxcox-initializer-rlh}--
\eqref{eq:boxcox-initializer-rsc} provide every loss inequality used in the proof of
Theorem~\mainref{thm:batch_b_init}, restricted to the KKT cone as in Proposition~\ref{prop:glm-initializer}. 
By an analogous argument to the other examples above, we get the conclusion.
\end{proof}

\subsubsection{Verification of the tightening-stage theorem}
\label{sec:BC_tightening_verification}

Let
\begin{equation*}
\begin{split}
\mc S_{\mathrm{BC}}
&=
\{j\in\mc P_{\mathrm{BC}}\setminus\mc M:\beta_j^*\ne0\},
\qquad
s_{\mathrm{BC}}=|\mc S_{\mathrm{BC}}|,
\\
\overline{\mc M}_{\mathrm{BC}}
&=\mc M_{\mathrm{BC}}\cup\mc S_{\mathrm{BC}},
\qquad
q_{\mathrm{BC}}
=|\overline{\mc M}_{\mathrm{BC}}|=m+K+s_{\mathrm{BC}}.
\end{split}
\end{equation*}
All complements below are relative to $\{1,\ldots,d_{\mathrm{BC}}\}$.
Then
$\mc M_{\mathrm{BC}}^c=\mc P_{\mathrm{BC}}\setminus\mc M$,
$\mc A_{\mathrm{BC}}\subseteq\overline{\mc M}_{\mathrm{BC}}$, and
$a_{\mathrm{BC}}\le q_{\mathrm{BC}}$.  Under
Assumption~\mainref{assm:local_alternative}, define
\begin{equation*}
\theta^\dagger_{\mc M_{\mathrm{BC}}}
=
\theta^*_{\mc M_{\mathrm{BC}}}
-
\widetilde{\Cb}_{\mathrm{BC}}^\top
(\Cb\Cb^\top)^{-1}\hb_B,
\qquad
\theta^\dagger_{\mc M_{\mathrm{BC}}^c}
=\theta^*_{\mc M_{\mathrm{BC}}^c}.
\end{equation*}
Since
$\widetilde{\Cb}_{\mathrm{BC}}
\widetilde{\Cb}_{\mathrm{BC}}^\top=\Cb\Cb^\top$, one has
\begin{equation} 
\widetilde{\Cb}_{\mathrm{BC}}
\theta^\dagger_{\mc M_{\mathrm{BC}}}=\tb,
\qquad
\|\theta^\dagger-\theta^*\|_2
\le d_B,
\qquad
\operatorname{supp}(\theta^\dagger-\theta^*)\subseteq\mc M.
\label{eq:boxcox-feasible-center-properties}
\end{equation}
Let
$\Zb_{\mathrm{BC}}\in
\mathbb R^{(m+K)\times(m+K-r)}$
have orthonormal columns spanning
$\ker(\widetilde{\Cb}_{\mathrm{BC}})$.

Define the reduced and full LLA iterates by
\maineqref{def:LLA.0}--\maineqref{def:LLA.a}, under the substitution preceding
\eqref{eq:boxcox-online-initializer}, with starting value
$\widehat\theta^{(B)}$ and with the coordinates in
$\mc M_{\mathrm{BC}}$ unpenalized.  Denote them by
$\widehat\theta_0^{({\rm LLA},k)}$ and
$\widehat\theta_a^{({\rm LLA},k)}$.  Let
$\widehat\theta_0^{\rm ora}$ and $\widehat\theta_a^{\rm ora}$ be the
estimators in \maineqref{def:constrained_oracle_estimator}--
\maineqref{def:unconstrained_oracle_estimator} under the same substitution.

Let $\gamma_0,p_0,c_{\rm aug},\delta$, and $l$ be the quantities in
Theorem~\mainref{thm:batch_b_tightening}, with $q=q_{\mathrm{BC}}$.  Thus
\begin{equation}
p_0=p'_+(\gamma_0)>0,
\qquad
\delta=\frac{\sqrt5}{2\kappa\gamma_0}<1,
\qquad
l=(2+2/p_0)(c_{\rm aug}^2+1)^{1/2}q_{\mathrm{BC}}^{1/2}.
\label{eq:boxcox-terminal-l}
\end{equation}

On the event in \eqref{eq:boxcox-score-event},
Assumption~\mainref{assm:hessian} and
\eqref{eq:boxcox-design-conditions} yield
\begin{equation}
\begin{gathered}
\overline B_{2,\mathrm{BC}}
=\frac{J_u+q_{\mathrm{BC}}\underline\kappa_{\mathrm{BC}}/(32a_{\mathrm{BC}})}
{\inf_{|t|\le B_1}\mathfrak h(t)},\qquad
\overline B_{3,\mathrm{BC}}
=\frac{B_1^{1/2}\{J_u+l^2\underline\kappa_{\mathrm{BC}}/(32a_{\mathrm{BC}})\}}
{\inf_{|t|\le B_1}\mathfrak h(t)},
\\
\max_{j\le B}\sup_{\substack{\operatorname{supp}(u)
\subseteq\overline{\mc M}_{\mathrm{BC}}\\\|u\|_2=1}}
\frac1{n_j}\sum_{i,k}w_k\{(z_{ik}^{(j)})^\top u\}^2
\le\overline B_{2,\mathrm{BC}},
\\
\max_{j\le B,a}\sup_{\substack{u,v\in\CC(l)\\
\|u\|_2=\|v\|_2=1}}
\frac1{n_j}\sum_{i,k}w_k|z_{ik,a}^{(j)}|
|(z_{ik}^{(j)})^\top u|\,|(z_{ik}^{(j)})^\top v|
\le\overline B_{3,\mathrm{BC}}.
\end{gathered}
\label{eq:boxcox-terminal-design-bounds}
\end{equation}
The quadratic bound preceding \eqref{eq:boxcox-initializer-rlh} implies the first inequality with $\|u\|_1\le q_{\mathrm{BC}}^{1/2}\|u\|_2$.
For the second, apply the Cauchy--Schwarz inequality to the two linear forms, use
$|z_{ik,a}^{(j)}|\le B_1^{1/2}$, and then
$\|u\|_1\le l\|u\|_2$, $\|v\|_1\le l\|v\|_2$.

Define
\begin{equation}
\overline H_{\mathrm{BC},B}^*
=
\frac1{N_B}\sum_{j=1}^Bn_j
\nabla^2\widehat{\mc L}^{(j)}(\theta^*),
\qquad
J_{\mathrm{BC},B}
=
\mathbb E\overline H_{\mathrm{BC},B}^*
=\nabla^2\mc L(\theta^*)=\Jb.
\label{eq:boxcox-terminal-hessians}
\end{equation}
By \eqref{eq:boxcox-probit-law}, we get
\begin{equation*}
\mathbb E\{\mathfrak h(S_{ik}^{(j)}\eta_{ik}^{*(j)})
\mid x_i^{(j)}\}
=
\frac{\varphi^2(\eta_{ik}^{*(j)})}
{\Phi(\eta_{ik}^{*(j)})\Phi(-\eta_{ik}^{*(j)})}
\in(0,1).
\end{equation*}
Assumption~\mainref{assm:hessian} implies
\begin{equation}
\max_a\|(J_{\mathrm{BC},B})_{a,\overline{\mc M}_{\mathrm{BC}}}\|_2
\le J_{\mathrm{BC},u}:=J_u.
\label{eq:boxcox-J-row-bound}
\end{equation}
The expectation in \eqref{eq:boxcox-terminal-hessians} is with respect
to the joint law of $(Y,X)$.

Since $d_{\mathrm{BC}}=p+K\ge2$,
$2\log\log(e d_{\mathrm{BC}})\ge1$ and
$2\log\log(e d_{\mathrm{BC}})\le2\log d_{\mathrm{BC}}$.  
Define
\begin{equation}
\begin{split}
\eta_{\mathrm{BC},\infty}
&=B_1\left[
\frac{2\{\log(2d_{\mathrm{BC}}^2)+2\log\log(e d_{\mathrm{BC}})\}}
{N_B}\right]^{1/2},\\
\chi_{\mathrm{BC},\infty}
&=\left[
\frac{2B_1\{\log(2d_{\mathrm{BC}})+2\log\log(e d_{\mathrm{BC}})\}}
{N_B}\right]^{1/2},\\
\chi_{\mathrm{BC}}(v)
&=2\left[
\left\{\frac{2J_u\{v\log5+\log2+
2\log\log(e d_{\mathrm{BC}})\}}{N_B}\right\}^{1/2}\right.\\
&\hspace{10mm}\left.+
\frac{2G_{\mathrm{BC}}(B_1q_{\mathrm{BC}})^{1/2}
\{v\log5+\log2+2\log\log(e d_{\mathrm{BC}})\}}{3N_B}
\right],\qquad v\ge1,
\end{split}
\label{eq:boxcox-terminal-concentration-rates}
\end{equation}
with $\chi_{\mathrm{BC}}(0)=0$.  Set
\begin{equation*}
M_{\mathrm{BC},l}=4.3\overline B_{3,\mathrm{BC}}.
\end{equation*}
Again write $\lambda=\lambda_B$ and define
\begin{equation}
\begin{split}
\Phi_{\mathrm{BC},B}
&=
\frac{9a_{\mathrm{BC}}}{4\kappa^2N_B}
\sum_{j=1}^{B-1}n_j\lambda_j^2,
\\
\Gamma_{\mathrm{BC},B}
&=
2M_{\mathrm{BC},l}d_B^2
+\{J_{\mathrm{BC},u}
+m^{1/2}\eta_{\mathrm{BC},\infty}\}d_B.
\end{split}
\label{eq:boxcox-terminal-history-alternative}
\end{equation}

\begin{enumerate}[label=(Box--Cox-\arabic*),start=4,leftmargin=3\parindent]
\item\label{BoxCoxassm:terminal-rates}
The same $\kappa$ as in
Proposition~\ref{prop:boxcox-initializer-verification} satisfies
\begin{equation}
\begin{split}
0<\kappa\le{}&
\exp\{-2B_1^{3/2}lR_0-2B_1l^2R_0^2\}
(J_l-l^2\eta_{\mathrm{BC},\infty}),
\end{split}
\label{eq:boxcox-terminal-kappa}
\end{equation}
and
\begin{equation}
\begin{gathered}
\chi_{\mathrm{BC},\infty}\le\frac{p_0\lambda}{8},
\qquad
2M_{\mathrm{BC},l}\Phi_{\mathrm{BC},B}
\le\frac{p_0\lambda}{8},
\qquad
\Gamma_{\mathrm{BC},B}\le\frac{p_0\lambda}{8},
\\
\kappa d_B\le\frac\lambda4,
\qquad
c_{\rm aug}\gamma_0q_{\mathrm{BC}}^{1/2}\lambda
\vee d_B<R_0.
\end{gathered}
\label{eq:boxcox-terminal-rate-conditions}
\end{equation}
\end{enumerate}
Under the notation of Theorem~\mainref{thm:batch_b_tightening}, define
\begin{equation*}
\begin{split}
r_{\mathrm{BC},a}^{\rm ora}
&=
\kappa^{-1}
\{\|\omegab^{(B)}_{\overline{\mc M}_{\mathrm{BC}}}\|_2
+2M_{\mathrm{BC},l}q_{\mathrm{BC}}^{1/2}\phi^{(B)}\},
\\
r_{\mathrm{BC},0}^{\rm ora}
&=
\kappa^{-1}\big[
\|\Zb_{\mathrm{BC}}^\top
\omegab^{(B)}_{\mc M_{\mathrm{BC}}}\|_2
+\|\omegab^{(B)}_{\mc S_{\mathrm{BC}}}\|_2
+2q_{\mathrm{BC}}^{1/2}
\{2M_{\mathrm{BC},l}\phi^{(B)}+\Gamma_{\mathrm{BC},B}\}
\big]+d_B.
\end{split}
\end{equation*}

\begin{proposition}
\label{prop:boxcox-tightening-verification}
Suppose Assumptions~\mainref{assm:penalty_function},
\mainref{assm:local_alternative}, and \mainref{assm:hessian}, conditions
\ref{BoxCoxassm:model}--\ref{BoxCoxassm:re},
\ref{BoxCoxassm:terminal-rates}, and the conditions
of Proposition~\ref{prop:boxcox-initializer-verification} hold.  Then,
with probability at least
\begin{equation}
\begin{split}
1-\frac{2B}{d_{\mathrm{BC}}}
-2d_{\mathrm{BC}}^2\sum_{j=1}^B
\exp\left\{-\frac{n_j\underline\kappa_{\mathrm{BC}}^2}
{2048B_1^2a_{\mathrm{BC}}^2}\right\}
-4\{\log(e d_{\mathrm{BC}})\}^{-2},
\end{split}
\label{eq:boxcox-terminal-probability}
\end{equation}
the conclusions of
Theorem~\mainref{thm:batch_b_tightening} hold for the two Box--Cox LLA
sequences.  In particular, for $e\in\{0,a\}$ and $k\ge1$,
\begin{equation}
\begin{split}
\|\widehat\theta_e^{({\rm LLA},k)}-\theta^*\|_2
\leq
\delta
\|\widehat\theta_e^{({\rm LLA},k-1)}-\theta^*\|_2
+
\kappa^{-1}
\left\|
p'_{\lambda,+}\{(|\beta^*_{\mc S_{\mathrm{BC}}}|
-\gamma_0\lambda)_+\}
\right\|_2
+r_{\mathrm{BC},e}^{\rm ora}.
\end{split}
\label{eq:boxcox-terminal-contraction}
\end{equation}
Moreover,
\begin{equation}
\begin{split}
r_{\mathrm{BC},a}^{\rm ora}
&\le
\kappa^{-1}\{
\chi_{\mathrm{BC}}(q_{\mathrm{BC}})
+2M_{\mathrm{BC},l}q_{\mathrm{BC}}^{1/2}\Phi_{\mathrm{BC},B}\},
\\
r_{\mathrm{BC},0}^{\rm ora}
&\le
\kappa^{-1}\big[
\sqrt2\,\chi_{\mathrm{BC}}(q_{\mathrm{BC}}-r)
+2q_{\mathrm{BC}}^{1/2}
\{2M_{\mathrm{BC},l}\Phi_{\mathrm{BC},B}
+\Gamma_{\mathrm{BC},B}\}\big]+d_B.
\end{split}
\label{eq:boxcox-terminal-oracle-rate-bounds}
\end{equation}
\end{proposition}

\begin{proof}
The following event has probability at least
$1-4\{\log(e d_{\mathrm{BC}})\}^{-2}$:
\begin{equation}
\begin{gathered}
\|\overline H_{\mathrm{BC},B}^*-J_{\mathrm{BC},B}\|_{\max}
\le\eta_{\mathrm{BC},\infty},\qquad
\|\omegab^{(B)}\|_\infty\le\chi_{\mathrm{BC},\infty},\\
\|\omegab^{(B)}_{\overline{\mc M}_{\mathrm{BC}}}\|_2
\le\chi_{\mathrm{BC}}(q_{\mathrm{BC}}),\\
\|\Zb_{\mathrm{BC}}^\top
\omegab^{(B)}_{\mc M_{\mathrm{BC}}}\|_2+
\|\omegab^{(B)}_{\mc S_{\mathrm{BC}}}\|_2
\le\sqrt2\,\chi_{\mathrm{BC}}(q_{\mathrm{BC}}-r).
\end{gathered}
\label{eq:boxcox-terminal-concentration-event}
\end{equation}
Intersect this event with \eqref{eq:boxcox-score-event}.
The first two inequalities follow from the Hessian and score bounds used in \eqref{eq:boxcox-score-event}, with
the fixed levels in \eqref{eq:boxcox-terminal-concentration-rates}.
Put $\xi_i^{(j)}=\nabla\ell_{\mathrm{BC}}
(\theta^*;Y_i^{(j)},x_i^{(j)})$.
For a unit $u$ supported on $\overline{\mc M}_{\mathrm{BC}}$, we have
\[
|u^\top\xi_i^{(j)}|\le
G_{\mathrm{BC}}(B_1q_{\mathrm{BC}})^{1/2},\qquad
\mathbb E(u^\top\xi_i^{(j)})^2
\le u^\top\Jb u\le J_u.
\]
Applying the same argument used in the proof of Proposition~\ref{prop:tobit-tightening-verification} to
$\operatorname{diag}(\Zb_{\mathrm{BC}},\Ib_{s_{\mathrm{BC}}})^\top
\omegab_{\overline{\mc M}_{\mathrm{BC}}}^{(B)}$ gives the fourth with
$v=q_{\mathrm{BC}}-r$, since
$\|a\|_2+\|b\|_2\le\sqrt2(\|a\|_2^2+\|b\|_2^2)^{1/2}$.
Each complement probability is at most
$\{\log(e d_{\mathrm{BC}})\}^{-2}$.

For a pair in Definition~\mainref{def:local_rsc}, put
$\Delta=\theta_1-\theta_2$.  The initializer event gives
$\phi^{(B)}\le\Phi_{\mathrm{BC},B}$ and places every historical
displacement in $\BB(R_0)\cap\CC(l)$, since
$l\ge4q_{\mathrm{BC}}^{1/2}\ge4a_{\mathrm{BC}}^{1/2}$.
For $0\le t\le1$,
$\|\theta_2+t\Delta-\theta^*\|_1\le2lR_0$; every historical
displacement has $\ell_1$ norm at most $lR_0$.
Apply the last inequality in \eqref{eq:boxcox-probit-derivatives}
to each summand of the historical and current Hessians.
Assumption~\mainref{assm:hessian},
$|\Delta^\top(\overline H_{\mathrm{BC},B}^*-\Jb)\Delta|
\le l^2\eta_{\mathrm{BC},\infty}\|\Delta\|_2^2$,
and \eqref{eq:boxcox-terminal-kappa} then give
\begin{equation}
\begin{split}
&\Delta^\top\nabla^2\widehat{\mc Q}_B(\theta_2+t\Delta)\Delta
\ge e^{-2B_1^{3/2}lR_0-2B_1l^2R_0^2}
\Delta^\top\overline H_{\mathrm{BC},B}^*\Delta,\\
&\langle\nabla\widehat{\mc Q}_B(\theta_1)
-\nabla\widehat{\mc Q}_B(\theta_2),\Delta\rangle\ge
 e^{-2B_1^{3/2}lR_0-2B_1l^2R_0^2}
(J_l-l^2\eta_{\mathrm{BC},\infty})\|\Delta\|_2^2
\ge\kappa\|\Delta\|_2^2,\\
&\|\{\nabla^2\widehat{\mc L}^{(j)}(\theta^*+u)
-\nabla^2\widehat{\mc L}^{(j)}(\theta^*)\}v\|_\infty
\le M_{\mathrm{BC},l}\|u\|_2\|v\|_2 .
\end{split}
\label{eq:boxcox-terminal-loss-events}
\end{equation}
The last inequality follows from $|L_{\mathrm P}'''|\le4.3$ and
\eqref{eq:boxcox-terminal-design-bounds} for $u,v\in\CC(l)$.
The second inequality verifies
\eqref{eq:inherited-surrogate-curvature} for the pairs in
\eqref{eq:weighted-full-bound}--\eqref{eq:weighted-reduced-bound}.

By \eqref{eq:boxcox-feasible-center-properties} and
\eqref{eq:boxcox-terminal-rate-conditions},
$\theta^\dagger-\theta^*\in\BB(R_0)\cap\CC(l)$ and
$\|\theta^\dagger-\theta^*\|_1\le m^{1/2}d_B$.
Equations \eqref{eq:R-general-center-bound},
\eqref{eq:boxcox-terminal-loss-events},
\eqref{eq:boxcox-J-row-bound}, and
\eqref{eq:boxcox-terminal-concentration-event} give
\begin{equation*}
\begin{aligned}
\frac1{N_B}\sum_{j<B}n_j\|R_j(\theta^\dagger)\|_\infty
&\le2M_{\mathrm{BC},l}\phi^{(B)}+M_{\mathrm{BC},l}d_B^2,\\
\left\|\frac1{N_B}\sum_{j=1}^Bn_j
\{\nabla\widehat{\mc L}^{(j)}(\theta^\dagger)
-\nabla\widehat{\mc L}^{(j)}(\theta^*)\}\right\|_\infty
&\le\{J_{\mathrm{BC},u}+m^{1/2}\eta_{\mathrm{BC},\infty}\}d_B
+\tfrac12M_{\mathrm{BC},l}d_B^2.
\end{aligned}
\end{equation*}
Lemma~\ref{lem:decomposition} and substitution in
\eqref{eq:surrogate-gradient-remainder} yield
\begin{equation}
\begin{split}
\|\nabla\widehat{\mc Q}_B(\theta^*)
-\omegab^{(B)}\|_\infty
&\le\frac32M_{\mathrm{BC},l}\phi^{(B)},
\\
\|\nabla\widehat{\mc Q}_B(\theta^\dagger)
-\omegab^{(B)}\|_\infty
&\le2M_{\mathrm{BC},l}\phi^{(B)}
+\Gamma_{\mathrm{BC},B}.
\end{split}
\label{eq:boxcox-terminal-center-score-bounds}
\end{equation}
Indeed, \eqref{eq:boxcox-J-row-bound} and
$\operatorname{supp}(\theta^\dagger-\theta^*)\subseteq\mc M$ give
$\|J_{\mathrm{BC},B}(\theta^\dagger-\theta^*)\|_\infty
\le J_{\mathrm{BC},u}d_B$ directly.
Equations \eqref{eq:boxcox-terminal-concentration-event},
\eqref{eq:boxcox-terminal-rate-conditions}, and
$\phi^{(B)}\le\Phi_{\mathrm{BC},B}$ yield
\begin{align*}
\|\omegab^{(B)}\|_\infty+2M_{\mathrm{BC},l}\phi^{(B)}
&\le p_0\lambda/4,\\
2p_0^{-1}\{\|\omegab^{(B)}\|_\infty
+2M_{\mathrm{BC},l}\phi^{(B)}+\Gamma_{\mathrm{BC},B}\}
+\kappa d_B&\le3\lambda/4+\lambda/4=\lambda.
\end{align*}
These are \maineqref{event:theorem2-full} and
\maineqref{event:theorem2-reduced}, respectively.

As in the proof of Proposition~\ref{prop:tobit-tightening-verification}, the induction argument in the proof of
Theorem~\mainref{thm:batch_b_tightening} yields
\begin{equation*}
\begin{aligned}
\max\left\{
\|\widehat\theta_a^{({\rm LLA},k)}-\theta^*\|_2,
\|\widehat\theta_0^{({\rm LLA},k)}-\theta^\dagger\|_2+d_B
\right\}
&<\kappa^{-1}\{2+p_0(c_{\rm aug}^2+1)^{1/2}\}
q_{\mathrm{BC}}^{1/2}\lambda\\
&=c_{\rm aug}\gamma_0q_{\mathrm{BC}}^{1/2}\lambda<R_0.
\end{aligned}
\end{equation*}

The two lines of \eqref{eq:boxcox-terminal-center-score-bounds} replace
\eqref{eq:gstar-score} and \eqref{eq:gdagger-score}.  Since
$\|v_A\|_2\le|A|^{1/2}\|v\|_\infty$ and
$\|\Zb_{\mathrm{BC}}\|_2=1$, they also yield the bounds used in
\eqref{eq:gdagger-block-score}; the second inequality preceding
\eqref{eq:boxcox-terminal-center-score-bounds} replaces
\eqref{eq:omega-transport}.  Lemma~\ref{lem:intermediate.rate} and the
proof of Theorem~\mainref{thm:batch_b_tightening} therefore give
\eqref{eq:boxcox-terminal-contraction} and
\eqref{eq:crude-lla-localization}.  The last two inequalities in
\eqref{eq:boxcox-terminal-concentration-event} and
$\phi^{(B)}\le\Phi_{\mathrm{BC},B}$ give
\eqref{eq:boxcox-terminal-oracle-rate-bounds}.  A union bound with
Proposition~\ref{prop:boxcox-initializer-verification} yields
\eqref{eq:boxcox-terminal-probability}.
\end{proof}

By \eqref{eq:boxcox-initializer-tuning} and
Lemma~\ref{lem:batch.size}, we have
\begin{equation*}
\Phi_{\mathrm{BC},B}
\le
\frac{9C^2a_{\mathrm{BC}}\log d_{\mathrm{BC}}}
{4\kappa^2N_B}
\left\{1+\log\left(\frac{N_{B-1}}{n_1}\right)\right\}.
\end{equation*}
Iterating \eqref{eq:boxcox-terminal-contraction} gives
\begin{equation*}
\begin{split}
\|\widehat\theta_e^{({\rm LLA},k)}-\theta^*\|_2
\le{}&
\delta^k\|\widehat\theta^{(B)}-\theta^*\|_2
\\
&+\frac1{1-\delta}\left[
\kappa^{-1}
\left\|p'_{\lambda,+}\{(|\beta^*_{\mc S_{\mathrm{BC}}}|
-\gamma_0\lambda)_+\}\right\|_2
+r_{\mathrm{BC},e}^{\rm ora}\right].
\end{split}
\end{equation*}
If $\mc S_{\mathrm{BC}}=\varnothing$, the penalty-derivative term is zero.
If $\mc S_{\mathrm{BC}}\ne\varnothing$ and
$\min_{j\in\mc S_{\mathrm{BC}}}|\beta_j^*|
\ge(\gamma_0+c_1)\lambda$, then
$(|\beta_j^*|-\gamma_0\lambda)_+\ge c_1\lambda$ for every
$j\in\mc S_{\mathrm{BC}}$, so Assumption~\mainref{assm:penalty_function}
makes the same term zero.

\subsubsection{Verification of the strong-oracle theorem}
\label{sec:BC_strong_oracle_verification}

Write $T=\overline{\mc M}_{\mathrm{BC}}$ and
\begin{equation*}
U_{\mathrm{BC},a}=\Ib_{q_{\mathrm{BC}}},
\qquad
U_{\mathrm{BC},0}
=\operatorname{diag}(\Zb_{\mathrm{BC}},\Ib_{s_{\mathrm{BC}}}),
\qquad
\varrho_{\mathrm{BC}}=J_{\mathrm{BC},u}/\kappa,
\end{equation*}
and, for $e\in\{0,a\}$,
\begin{equation}
A_{\mathrm{BC},e}
=(J_{\mathrm{BC},B})_{T^c,T}U_{\mathrm{BC},e}
\{U_{\mathrm{BC},e}^\top(J_{\mathrm{BC},B})_{T,T}
U_{\mathrm{BC},e}\}^{-1}.
\label{eq:boxcox-terminal-projection}
\end{equation}
Assumption~\mainref{assm:hessian},
\eqref{eq:boxcox-terminal-kappa}, and
\eqref{eq:boxcox-J-row-bound} yield
\begin{equation}
\lambda_{\min}\{U_{\mathrm{BC},e}^\top
(J_{\mathrm{BC},B})_{T,T}U_{\mathrm{BC},e}\}\ge\kappa,
\qquad
T^c\ne\varnothing\ \Longrightarrow\ 
\max_{j\in T^c}\|(A_{\mathrm{BC},e})_{j,\cdot}\|_2
\le\varrho_{\mathrm{BC}}.
\label{eq:boxcox-terminal-projection-bounds}
\end{equation}
Indeed, the columns of $U_{\mathrm{BC},e}$ are orthonormal, so
Assumption~\mainref{assm:hessian}, with $\Jb=J_{\mathrm{BC},B}$, gives the
first bound because $J_l\ge\kappa$ by
\eqref{eq:boxcox-probit-derivatives} and
\eqref{eq:boxcox-terminal-kappa}.  Its inverse bound and
\eqref{eq:boxcox-J-row-bound} give the row bound when $T^c\ne\varnothing$.

Define
\begin{equation*}
\begin{split}
\chi_{\mathrm{BC},\mathrm{proj}}
&=\left[
\frac{2J_u(1+\varrho_{\mathrm{BC}}^2)
\{\log(4d_{\mathrm{BC}})+2\log\log(e d_{\mathrm{BC}})\}}{N_B}
\right]^{1/2}\\
&\quad+
\frac{2G_{\mathrm{BC}}B_1^{1/2}
(1+\varrho_{\mathrm{BC}}q_{\mathrm{BC}}^{1/2})
\{\log(4d_{\mathrm{BC}})+2\log\log(e d_{\mathrm{BC}})\}}
{3N_B}.
\end{split}
\end{equation*}
\begin{equation}
\begin{split}
\overline r_{\mathrm{BC},B}
={}&\kappa^{-1}\{
\chi_{\mathrm{BC}}(q_{\mathrm{BC}})
+2M_{\mathrm{BC},l}q_{\mathrm{BC}}^{1/2}\Phi_{\mathrm{BC},B}\}
\\
&\vee\bigg[
\kappa^{-1}\big[
\sqrt2\,\chi_{\mathrm{BC}}(q_{\mathrm{BC}}-r)
+2q_{\mathrm{BC}}^{1/2}
\{2M_{\mathrm{BC},l}\Phi_{\mathrm{BC},B}
+\Gamma_{\mathrm{BC},B}\}\big]+d_B\bigg],
\\
\Psi_{\mathrm{BC},B}
={}&\mathbf1\{T^c\ne\varnothing\}\bigg[
\chi_{\mathrm{BC},\mathrm{proj}}
+(1+\varrho_{\mathrm{BC}}q_{\mathrm{BC}}^{1/2})
\bigg\{2M_{\mathrm{BC},l}\Phi_{\mathrm{BC},B}
+\Gamma_{\mathrm{BC},B}
\\
&\hspace{25mm}
+\big[q_{\mathrm{BC}}^{1/2}\eta_{\mathrm{BC},\infty}
+M_{\mathrm{BC},l}
\{\Phi_{\mathrm{BC},B}^{1/2}+\overline r_{\mathrm{BC},B}\}
\big]\overline r_{\mathrm{BC},B}\bigg\}\bigg].
\end{split}
\label{eq:boxcox-terminal-strong-oracle-rates}
\end{equation}

\begin{enumerate}[label=(Box--Cox-\arabic*),start=5,leftmargin=3\parindent]
\item\label{BoxCoxassm:strong-oracle}
There are $\gamma_1>\gamma_0$ and
$\bar\delta\in(5/(4\kappa\gamma_0),1)$ such that
$p_0\in(0,1/2]$, $p'_+(\gamma_1)=0$, $\kappa\gamma_0>5/4$, and
\begin{equation}
\begin{gathered}
s_{\mathrm{BC}}>0
\quad\Longrightarrow\quad
\min_{j\in\mc S_{\mathrm{BC}}}|\beta_j^*|
\ge(\gamma_0+\gamma_1)\lambda,
\\
\overline r_{\mathrm{BC},B}
<R_0\wedge\frac{\lambda}{5\bar\delta\kappa},
\qquad
\Psi_{\mathrm{BC},B}\le\frac{p_0\lambda}{2}.
\end{gathered}
\label{eq:boxcox-terminal-strong-oracle-conditions}
\end{equation}
\end{enumerate}

Let
\begin{equation*}
\widehat s_{\mathrm{BC}}
=|\operatorname{supp}(\widehat\beta^{(B)}_{\mc M_{\mathrm{BC}}^c})|
+s_{\mathrm{BC}},
\qquad
K_{\mathrm{BC}}^{\rm LLA}
=
\begin{cases}
0,&\widehat s_{\mathrm{BC}}=0,\\[1mm]
\displaystyle
\left\lceil
\frac{\log(\widehat s_{\mathrm{BC}}^{1/2}/\bar\delta)}
{\log(1/\bar\delta)}
\right\rceil,&\widehat s_{\mathrm{BC}}>0.
\end{cases}
\end{equation*}

\begin{proposition}
\label{prop:boxcox-strong-oracle-verification}
Suppose that the conditions of
Proposition~\ref{prop:boxcox-tightening-verification} and
\ref{BoxCoxassm:strong-oracle} hold.  Then, with probability at least
\begin{equation}
\begin{split}
1-\frac{2B}{d_{\mathrm{BC}}}
-2d_{\mathrm{BC}}^2\sum_{j=1}^B
\exp\left\{-\frac{n_j\underline\kappa_{\mathrm{BC}}^2}
{2048B_1^2a_{\mathrm{BC}}^2}\right\}
-5\{\log(e d_{\mathrm{BC}})\}^{-2},
\end{split}
\label{eq:boxcox-terminal-strong-oracle-probability}
\end{equation}
every sequence of LLA solutions satisfies
\begin{equation}
\widehat\theta_e^{({\rm LLA},k)}
=\widehat\theta_e^{\rm ora},
\qquad e\in\{0,a\},
\qquad k\ge K_{\mathrm{BC}}^{\rm LLA}+1.
\label{eq:boxcox-terminal-strong-oracle-recovery}
\end{equation}
\end{proposition}

\begin{proof}
If $T^c\ne\varnothing$, Bernstein inequalities and a union
bound over $e\in\{0,a\}$ and $h\in T^c$ give, with probability at least
$1-\{\log(e d_{\mathrm{BC}})\}^{-2}$,
\begin{equation}
\max_{e\in\{0,a\}}
\|\omegab^{(B)}_{T^c}
-A_{\mathrm{BC},e}U_{\mathrm{BC},e}^\top
\omegab^{(B)}_T\|_\infty
\le\chi_{\mathrm{BC},\mathrm{proj}}.
\label{eq:boxcox-terminal-projected-score}
\end{equation}
For fixed $e$ and $h\in T^c$, its $(j,i)$ summand is
$\nabla_h\ell_{\mathrm{BC}}$ minus
$(A_{\mathrm{BC},e})_{h,\cdot}U_{\mathrm{BC},e}^\top
\nabla_T\ell_{\mathrm{BC}}$, both evaluated at
$(\theta^*;Y_i^{(j)},x_i^{(j)})$.  By
\eqref{eq:boxcox-terminal-projection-bounds}, its absolute value is at
most $G_{\mathrm{BC}}B_1^{1/2}
(1+\varrho_{\mathrm{BC}}q_{\mathrm{BC}}^{1/2})$.
Its mean is zero, and its variance is at most
$J_u(1+\varrho_{\mathrm{BC}}^2)$ by the Cauchy--Schwarz inequality, as in the proof of
\eqref{eq:boxcox-terminal-concentration-event}.
The deviation level
$\log(4d_{\mathrm{BC}})+2\log\log(e d_{\mathrm{BC}})$ yields complement
probability $\{\log(e d_{\mathrm{BC}})\}^{-2}$, since
$4d_{\mathrm{BC}}e^{-\log(4d_{\mathrm{BC}})}=1$ and
$e^{-2\log\log(e d_{\mathrm{BC}})}=\{\log(e d_{\mathrm{BC}})\}^{-2}$.

On the intersection of \eqref{eq:boxcox-score-event} and
\eqref{eq:boxcox-terminal-concentration-event},
Proposition~\ref{prop:boxcox-initializer-verification},
\eqref{eq:boxcox-terminal-center-score-bounds},
$\phi^{(B)}\le\Phi_{\mathrm{BC},B}$ imply
\begin{equation}
\begin{split}
\kappa^{-1}\|\nabla_T\widehat{\mc Q}_B(\theta^*)\|_2
&\le
\kappa^{-1}\{
\chi_{\mathrm{BC}}(q_{\mathrm{BC}})
+2M_{\mathrm{BC},l}q_{\mathrm{BC}}^{1/2}\Phi_{\mathrm{BC},B}\},
\\
\kappa^{-1}\|U_{\mathrm{BC},0}^\top
\nabla_T\widehat{\mc Q}_B(\theta^\dagger)\|_2+d_B
&\le
\kappa^{-1}\big[
\sqrt2\,\chi_{\mathrm{BC}}(q_{\mathrm{BC}}-r)
+2q_{\mathrm{BC}}^{1/2}
\{2M_{\mathrm{BC},l}\Phi_{\mathrm{BC},B}
+\Gamma_{\mathrm{BC},B}\}\big]+d_B.
\end{split}
\label{eq:boxcox-terminal-oracle-center-rates}
\end{equation}
We intersect the sets in
\maineqref{def:unconstrained_oracle_estimator} and
\maineqref{def:constrained_oracle_estimator} with
$\{\|\theta-\theta^*\|_2\le R_0\}$ and
$\{\|\theta-\theta^\dagger\|_2\le R_0-d_B\}$, respectively, and minimize
$\widehat{\mc Q}_B$ on them.
The KKT equations, \eqref{eq:inherited-surrogate-curvature}, and the Cauchy--Schwarz inequality imply
\begin{equation}
\|\widehat\theta_a^{\rm ora}-\theta^*\|_2
\le\kappa^{-1}\|\nabla_T\widehat{\mc Q}_B(\theta^*)\|_2,
\qquad
\|\widehat\theta_0^{\rm ora}-\theta^\dagger\|_2
\le\kappa^{-1}\|U_{\mathrm{BC},0}^\top
\nabla_T\widehat{\mc Q}_B(\theta^\dagger)\|_2.
\label{eq:boxcox-terminal-oracle-radial-rates}
\end{equation}
For the reduced problem, if a minimizer is on the boundary and
$\Delta=\widehat\theta_0^{\rm ora}-\theta^\dagger$, its normal multiplier
$\mu\ge0$ satisfies
\begin{equation*}
\kappa\|\Delta\|_2^2
\le-\langle\nabla_T\widehat{\mc Q}_B(\theta^\dagger),
\Delta_T\rangle-\mu\|\Delta\|_2
\le\|U_{\mathrm{BC},0}^\top
\nabla_T\widehat{\mc Q}_B(\theta^\dagger)\|_2\|\Delta\|_2.
\end{equation*}
For the full problem, $U_{\mathrm{BC},a}=\Ib_{q_{\mathrm{BC}}}$.
The reduced right-hand side in
\eqref{eq:boxcox-terminal-oracle-radial-rates} is at most
$\overline r_{\mathrm{BC},B}-d_B<R_0-d_B$, and the full right-hand side is
at most $\overline r_{\mathrm{BC},B}<R_0$, by
\eqref{eq:boxcox-terminal-oracle-center-rates} and
\eqref{eq:boxcox-terminal-strong-oracle-conditions}.
The
triangle inequality yields
\begin{equation}
\max_{e\in\{0,a\}}
\|\widehat\theta_e^{\rm ora}-\theta^*\|_2
\le\overline r_{\mathrm{BC},B}.
\label{eq:boxcox-terminal-oracle-rate}
\end{equation}
Again note that both minimizers are unique.

We work with $T^c\ne\varnothing$.  On the
intersection with \eqref{eq:boxcox-terminal-projected-score}, put
$\theta_a=\theta^*$, $\theta_0=\theta^\dagger$, and
$\Delta_e=\widehat\theta_e^{\rm ora}-\theta_e$.  The oracle KKT equations
and \eqref{eq:boxcox-terminal-projection} yield
\begin{equation}
U_{\mathrm{BC},e}^\top
\nabla_T\widehat{\mc Q}_B(\widehat\theta_e^{\rm ora})=0,
\qquad
(J_{\mathrm{BC},B})_{T^c,T}U_{\mathrm{BC},e}
-A_{\mathrm{BC},e}U_{\mathrm{BC},e}^\top
(J_{\mathrm{BC},B})_{T,T}U_{\mathrm{BC},e}=0.
\label{eq:boxcox-terminal-oracle-cancellation}
\end{equation}
Moreover,
$(\Delta_e)_{T^c}=0$ and
$(\Delta_e)_T=U_{\mathrm{BC},e}U_{\mathrm{BC},e}^\top(\Delta_e)_T$.
Subtracting $A_{\mathrm{BC},e}$ times the first equality in
\eqref{eq:boxcox-terminal-oracle-cancellation} and applying the fundamental
theorem of calculus gives
\begin{align}
\nabla_{T^c}\widehat{\mc Q}_B(\widehat\theta_e^{\rm ora})
={}&\omegab^{(B)}_{T^c}
-A_{\mathrm{BC},e}U_{\mathrm{BC},e}^\top\omegab^{(B)}_T
\notag\\
&+\nabla_{T^c}\widehat{\mc Q}_B(\theta_e)-\omegab^{(B)}_{T^c}
-A_{\mathrm{BC},e}U_{\mathrm{BC},e}^\top
\{\nabla_T\widehat{\mc Q}_B(\theta_e)-\omegab^{(B)}_T\}
\notag\\
&+\int_0^1\bigg[
\{\nabla^2\widehat{\mc Q}_B(\theta_e+t\Delta_e)
-J_{\mathrm{BC},B}\}_{T^c,T}U_{\mathrm{BC},e}
\notag\\
&\hspace{19mm}
-A_{\mathrm{BC},e}U_{\mathrm{BC},e}^\top
\{\nabla^2\widehat{\mc Q}_B(\theta_e+t\Delta_e)
-J_{\mathrm{BC},B}\}_{T,T}U_{\mathrm{BC},e}
\bigg]U_{\mathrm{BC},e}^\top(\Delta_e)_T\,dt.
\label{eq:boxcox-terminal-oracle-score-identity}
\end{align}
For every $\theta\in\mathbb R^{d_{\mathrm{BC}}}$, we have
\begin{equation*}
\begin{split}
\nabla^2\widehat{\mc Q}_B(\theta)-J_{\mathrm{BC},B}
&=\overline H_{\mathrm{BC},B}^*-J_{\mathrm{BC},B}
+\frac1{N_B}\sum_{j<B}n_j
\{\nabla^2\widehat{\mc L}^{(j)}(\widehat\theta^{(j)})
-\nabla^2\widehat{\mc L}^{(j)}(\theta^*)\}
\\
&\qquad+\frac{n_B}{N_B}
\{\nabla^2\widehat{\mc L}^{(B)}(\theta)
-\nabla^2\widehat{\mc L}^{(B)}(\theta^*)\}.
\end{split}
\end{equation*}
For $e\in\{0,a\}$ and $0\le t\le1$,
$\theta_e+t\Delta_e-\theta^*$ is supported on $T$, belongs to
$\BB(R_0)\cap\CC(l)$, and has Euclidean norm at most
$\overline r_{\mathrm{BC},B}$.  Also,
$N_B^{-1}\sum_{j<B}n_j\|\widehat\theta^{(j)}-\theta^*\|_2
\le\{\phi^{(B)}\}^{1/2}\le\Phi_{\mathrm{BC},B}^{1/2}$.
Equations \eqref{eq:boxcox-terminal-concentration-event},
\eqref{eq:boxcox-terminal-loss-events}, and
\eqref{eq:boxcox-terminal-oracle-rate} therefore imply
\begin{align}
&\max_{\substack{e\in\{0,a\},\ 0\le t\le1\\j\in T^c}}
\left\|
\{\nabla^2\widehat{\mc Q}_B(\theta_e+t\Delta_e)
-J_{\mathrm{BC},B}\}_{j,T}U_{\mathrm{BC},e}
\right\|_2
\notag\\
&\qquad\le q_{\mathrm{BC}}^{1/2}\eta_{\mathrm{BC},\infty}
+M_{\mathrm{BC},l}
\{\Phi_{\mathrm{BC},B}^{1/2}+\overline r_{\mathrm{BC},B}\},
\notag\\
&\max_{\substack{e\in\{0,a\}\\0\le t\le1}}
\left\|U_{\mathrm{BC},e}^\top
\{\nabla^2\widehat{\mc Q}_B(\theta_e+t\Delta_e)
-J_{\mathrm{BC},B}\}_{T,T}U_{\mathrm{BC},e}\right\|_2
\notag\\
&\qquad\le q_{\mathrm{BC}}^{1/2}
\left[q_{\mathrm{BC}}^{1/2}\eta_{\mathrm{BC},\infty}
+M_{\mathrm{BC},l}
\{\Phi_{\mathrm{BC},B}^{1/2}+\overline r_{\mathrm{BC},B}\}\right].
\label{eq:boxcox-terminal-oracle-hessian-bounds}
\end{align}
For the second inequality, take $j\in T$ in the first bound and use
$\|A\|_2\le\|A\|_{\mathrm F}$ and
$\|A\|_{\mathrm F}\le
q_{\mathrm{BC}}^{1/2}\max_{j\in T}\|A_{j,\cdot}\|_2$.
Equations \eqref{eq:boxcox-terminal-center-score-bounds} and
\eqref{eq:boxcox-terminal-projection-bounds} give
\begin{equation}
\begin{split}
\max_{e\in\{0,a\}}\big\|&
\nabla_{T^c}\widehat{\mc Q}_B(\theta_e)-\omegab_{T^c}^{(B)}
-A_{\mathrm{BC},e}U_{\mathrm{BC},e}^\top
\{\nabla_T\widehat{\mc Q}_B(\theta_e)-\omegab_T^{(B)}\}
\big\|_\infty
\\
&\le(1+\varrho_{\mathrm{BC}}q_{\mathrm{BC}}^{1/2})
\{2M_{\mathrm{BC},l}\Phi_{\mathrm{BC},B}
+\Gamma_{\mathrm{BC},B}\}.
\end{split}
\label{eq:boxcox-terminal-oracle-center-bound}
\end{equation}
Since
$\|U_{\mathrm{BC},e}^\top(\Delta_e)_T\|_2
\le\overline r_{\mathrm{BC},B}$,
\eqref{eq:boxcox-terminal-projected-score},
\eqref{eq:boxcox-terminal-oracle-score-identity},
\eqref{eq:boxcox-terminal-oracle-hessian-bounds}, and
\eqref{eq:boxcox-terminal-oracle-center-bound} imply
\begin{equation}
\max_{e\in\{0,a\}}
\|\nabla_{T^c}\widehat{\mc Q}_B
(\widehat\theta_e^{\rm ora})\|_\infty
\le\Psi_{\mathrm{BC},B}\le\frac{p_0\lambda}{2}.
\label{eq:boxcox-terminal-oracle-offscore}
\end{equation}
Equations \eqref{eq:boxcox-terminal-oracle-rate} and
\eqref{eq:boxcox-terminal-oracle-offscore} verify
\maineqref{eq:oracle-recovery-score-radius}; the remaining inequalities in
Theorem~\mainref{thm:oracle_recovery} are
\eqref{eq:boxcox-terminal-strong-oracle-conditions}.
Proposition~\ref{prop:boxcox-tightening-verification} gives
\eqref{eq:crude-lla-localization}.  In the proof of
Theorem~\mainref{thm:oracle_recovery},
every null coordinate in $\mc B_{e,0}$ belongs to
$\operatorname{supp}(\widehat\beta_{\mc M_{\mathrm{BC}}^c}^{(B)})$, and at
most $s_{\mathrm{BC}}$ coordinates in $\mc S_{\mathrm{BC}}$ can belong to
$\mc B_{e,0}$.  Hence $|\mc B_{e,0}|\le\widehat s_{\mathrm{BC}}$, and
\eqref{eq:bad-set-contraction} yields
\begin{equation*}
|\mc B_{e,k}|^{1/2}
\le\bar\delta^k\widehat s_{\mathrm{BC}}^{1/2}
\le\bar\delta<1,
\qquad k\ge K_{\mathrm{BC}}^{\rm LLA}.
\end{equation*}
Thus $\mc B_{e,k}=\varnothing$ for
$k\ge K_{\mathrm{BC}}^{\rm LLA}$, and
\eqref{eq:strong-oracle-upper} at the next update yields
\eqref{eq:boxcox-terminal-strong-oracle-recovery}.
Finally, the union bound applied to
\eqref{eq:boxcox-terminal-probability} and
\eqref{eq:boxcox-terminal-projected-score} gives
\eqref{eq:boxcox-terminal-strong-oracle-probability}.
\end{proof}

\subsubsection{Online Bahadur representation}
\label{sec:BC_bahadur_verification}

Set
$T=\overline{\mc M}_{\mathrm{BC}}$ and order its coordinates as
$(\mc M_{\mathrm{BC}},\mc S_{\mathrm{BC}})$.  Under the substitution preceding
\eqref{eq:boxcox-online-initializer}, take $\Jb=J_{\mathrm{BC},B}$ and
\begin{equation}
\begin{gathered}
\Jb_0=(J_{\mathrm{BC},B})_{T,T},
\qquad
\Cb_0=
\begin{pmatrix}
\widetilde{\Cb}_{\mathrm{BC}}&
\mathbf0_{r\times s_{\mathrm{BC}}}
\end{pmatrix},
\\
\Psi_0=\Cb_0\Jb_0^{-1}\Cb_0^\top,
\qquad
\Pb_0=\Jb_0^{-1/2}\Cb_0^\top
\Psi_0^{-1}\Cb_0\Jb_0^{-1/2}.
\end{gathered}
\label{eq:boxcox-bahadur-main-objects}
\end{equation}
Define
\begin{equation}
\Kb_0
=
\frac1{N_B}\sum_{j=1}^B\sum_{i=1}^{n_j}
\mathbb E\!\left[
\nabla_T\ell_{\mathrm{BC}}
(\theta^*;Y_i^{(j)},x_i^{(j)})
\nabla_T\ell_{\mathrm{BC}}
(\theta^*;Y_i^{(j)},x_i^{(j)})^\top
\right].
\label{eq:boxcox-bahadur-score-covariance}
\end{equation}
Set
\begin{equation*}
\Vb_0=\Cb_0\Jb_0^{-1}\Kb_0\Jb_0^{-1}\Cb_0^\top,
\qquad
\Omega_0=\Vb_0^{-1/2}\Cb_0\Jb_0^{-1}.
\end{equation*}

\begin{enumerate}[label=(Box--Cox-\arabic*),start=6,leftmargin=3\parindent]
\item\label{BoxCoxassm:bahadur}
As $N_B\to\infty$,
\begin{gather}
\begin{gathered}
\frac{2B}{d_{\mathrm{BC}}}
+2d_{\mathrm{BC}}^2\sum_{j=1}^B
\exp\left\{-\frac{n_j\underline\kappa_{\mathrm{BC}}^2}
{2048B_1^2a_{\mathrm{BC}}^2}\right\}
+\{\log(e d_{\mathrm{BC}})\}^{-2}=o(1),
\\
\|\hb_B\|_2
=O\!\left(\sqrt{\frac r{N_B}}\right),
\end{gathered}
\label{eq:boxcox-bahadur-probability}
\\
\frac{\overline r_{\mathrm{BC},B}}{R_0}=o(1),
\qquad
\lambda_{\min}(\Kb_0)\ge K_l>0,
\label{eq:boxcox-bahadur-localization-and-score}
\\
\frac{c_{\mathrm{BC},B}^2r^{5/2}q_{\mathrm{BC}}}{N_B}=o(1),
\label{eq:boxcox-bahadur-lyapunov}
\\
\begin{split}
&(1+M_{\mathrm{BC},l})
\sqrt{\frac{q_{\mathrm{BC}}}{N_B}}
\left\{
\frac{a_{\mathrm{BC}}}{\kappa^2}
\sum_{j=1}^{B}n_j\lambda_j^2
+\frac{a_{\mathrm{BC}}^{1/2}}{\kappa}
\overline r_{\mathrm{BC},B}
\sum_{j=1}^{B}n_j\lambda_j
+n_B\overline r_{\mathrm{BC},B}^{\,2}
\right\}
\\
&\hspace{38mm}
+\sqrt{N_B}\,q_{\mathrm{BC}}
\eta_{\mathrm{BC},\infty}
\overline r_{\mathrm{BC},B}=o(1).
\end{split}
\label{eq:boxcox-bahadur-remainder-rate}
\end{gather}
\end{enumerate}

\begin{proposition}
\label{prop:boxcox-bahadur-verification}
Suppose Assumptions~\mainref{assm:penalty_function},
\mainref{assm:local_alternative}, and \mainref{assm:hessian}, the conditions of
Proposition~\ref{prop:boxcox-initializer-verification},
conditions~\ref{BoxCoxassm:design}--
\ref{BoxCoxassm:terminal-rates}, and
condition~\ref{BoxCoxassm:bahadur} hold.  Then
\begin{align}
\widehat\theta_{0,T}^{\rm ora}-\theta_T^*
={}&
\Jb_0^{-1/2}(\Pb_0-\Ib_{q_{\mathrm{BC}}})\Jb_0^{-1/2}
\frac1{N_B}\sum_{j=1}^B\sum_{i=1}^{n_j}
\nabla_T\ell_{\mathrm{BC}}
(\theta^*;Y_i^{(j)},x_i^{(j)})
\notag\\
&-\Jb_0^{-1}\Cb_0^\top\Psi_0^{-1}\hb_B
+o_p(N_B^{-1/2}),
\label{eq:boxcox-reduced-bahadur}
\\
\widehat\theta_{a,T}^{\rm ora}-\theta_T^*
={}&
-\Jb_0^{-1}
\frac1{N_B}\sum_{j=1}^B\sum_{i=1}^{n_j}
\nabla_T\ell_{\mathrm{BC}}
(\theta^*;Y_i^{(j)},x_i^{(j)})
+o_p(N_B^{-1/2}).
\label{eq:boxcox-full-bahadur}
\end{align}
\end{proposition}

\begin{proof}
We work on the intersection of the events in
\eqref{eq:boxcox-score-event} and
\eqref{eq:boxcox-terminal-concentration-event}.  Proposition
\ref{prop:boxcox-initializer-verification} holds on the event in
\eqref{eq:boxcox-score-event}, and the complement probability
of the intersection is at most
\begin{equation*}
\frac{2B}{d_{\mathrm{BC}}}
+2d_{\mathrm{BC}}^2\sum_{j=1}^B
\exp\left\{-\frac{n_j\underline\kappa_{\mathrm{BC}}^2}
{2048B_1^2a_{\mathrm{BC}}^2}\right\}
+4\{\log(e d_{\mathrm{BC}})\}^{-2}=o(1)
\end{equation*}
by \eqref{eq:boxcox-bahadur-probability}.  On this intersection,
\eqref{eq:boxcox-initializer-rates} and
\eqref{eq:boxcox-terminal-history-alternative} yield
\begin{equation}
\begin{gathered}
\|\widehat\theta^{(j)}-\theta^*\|_2
\le
\frac{3a_{\mathrm{BC}}^{1/2}\lambda_j}{2\kappa},
\qquad j=1,\ldots,B,
\\
\phi^{(B)}
=\frac1{N_B}\sum_{j=1}^{B-1}n_j
\|\widehat\theta^{(j)}-\theta^*\|_2^2
\le\Phi_{\mathrm{BC},B}.
\end{gathered}
\label{eq:boxcox-bahadur-initializer-rate}
\end{equation}
Since $p_0\le1$ and $a_{\mathrm{BC}}\le q_{\mathrm{BC}}$,
\eqref{eq:boxcox-terminal-l} gives
$l\ge4a_{\mathrm{BC}}^{1/2}$.  Hence, for $j<B$, one has
\begin{equation*}
\widehat\theta^{(j)}-\theta^*
\in\BB(R_0)\cap\CC(l).
\end{equation*}
Equations \eqref{eq:boxcox-feasible-center-properties} and
\eqref{eq:boxcox-terminal-rate-conditions} imply
\begin{equation*}
\theta^\dagger-\theta^*
\in\BB(R_0)\cap\CC(l),
\end{equation*}
because $\operatorname{supp}(\theta^\dagger-\theta^*)\subseteq\mc M$
and $m^{1/2}\le q_{\mathrm{BC}}^{1/2}\le l$.

Equations \eqref{eq:boxcox-terminal-loss-events},
\eqref{eq:boxcox-terminal-center-score-bounds}, and
\eqref{eq:boxcox-terminal-oracle-center-rates} verify the curvature and the score requirements of Lemma~\ref{lem:oracle_rate}.  Here
$l\ge q_{\mathrm{BC}}^{1/2}$ and
\eqref{eq:boxcox-bahadur-localization-and-score} gives
$2\overline r_{\mathrm{BC},B}<R_0$ for all sufficiently large $N_B$.
In particular, we get
\[
\kappa^{-1}\|\nabla_T\widehat{\mc Q}_B(\theta^*)\|_2
\le\overline r_{\mathrm{BC},B}<R_0,
\qquad
\kappa^{-1}\|U_{\mathrm{BC},0}^\top
\nabla_T\widehat{\mc Q}_B(\theta^\dagger)\|_2
\le\overline r_{\mathrm{BC},B}-d_B<R_0-d_B.
\]
Hence, we have
\begin{equation}
\max_{e\in\{0,a\}}
\|\widehat\theta_e^{\rm ora}-\theta^*\|_2
\le\overline r_{\mathrm{BC},B}=o(R_0).
\label{eq:boxcox-bahadur-oracle-localization}
\end{equation}
Both differences are supported on $T$, so they belong to $\CC(l)$.
Equations \eqref{eq:boxcox-terminal-oracle-rate-bounds} and
\eqref{eq:boxcox-bahadur-localization-and-score} also give
$\max_{e\in\{0,a\}}r_{\mathrm{BC},e}^{\rm ora}\le R_0/2$
for all sufficiently large $N_B$.

Since $\Jb_0=(J_{\mathrm{BC},B})_{T,T}$ is a principal submatrix of
$J_{\mathrm{BC},B}$, Assumption~\mainref{assm:hessian} and eigenvalue calculus give
\begin{equation}
0<J_l
\le\lambda_{\min}(\Jb_0)
\le\lambda_{\max}(\Jb_0)
\le J_u<\infty.
\label{eq:boxcox-bahadur-hessian-bounds}
\end{equation}
The Hessian inequality in
\eqref{eq:boxcox-terminal-loss-events} yields
$\bigcap_{j=1}^B\mc E_{\rm rlh}^{(j)}
(R_0,l,M_{\mathrm{BC},l})$, and the first inequality in
\eqref{eq:boxcox-terminal-concentration-event} implies
$\mc H(\eta_{\mathrm{BC},\infty})$ after
\eqref{eq:boxcox-bahadur-main-objects}.

We then apply Lemma~\ref{lem:oracle_bahadur} with $m+s=q_{\mathrm{BC}}$, $\Cb_a=\Cb_0$, $M=M_{\mathrm{BC},l}$, and $\eta_\infty=\eta_{\mathrm{BC},\infty}$.
Lemma~\ref{lem:decomposition},
\eqref{eq:boxcox-terminal-loss-events}, and
\eqref{eq:boxcox-terminal-concentration-event} provide, for
$e\in\{0,a\}$,
\begin{equation}
\begin{split}
\|\mc R_{\mathrm{BC},e}\|_2
\lesssim{}&
\frac{M_{\mathrm{BC},l}q_{\mathrm{BC}}^{1/2}}{N_B}
\bigg[
\sum_{j=1}^{B-1}n_j
\bigg\{
\|\widehat\theta^{(j)}-\theta^*\|_2^2
+\|\widehat\theta^{(j)}-\theta^*\|_2
\|\widehat\theta_e^{\rm ora}-\theta^*\|_2
\bigg\}
+n_B\|\widehat\theta_e^{\rm ora}-\theta^*\|_2^2
\bigg]
\\
&+q_{\mathrm{BC}}\eta_{\mathrm{BC},\infty}
\|\widehat\theta_e^{\rm ora}-\theta^*\|_2,
\end{split}
\label{eq:boxcox-bahadur-remainder}
\end{equation}
where $\mc R_{\mathrm{BC},0}$ and
$\mc R_{\mathrm{BC},a}$ are the remainders in
\eqref{eq:reduced-oracle-bahadur} and
\eqref{eq:full-oracle-bahadur} after the preceding substitutions.
Substitution of \eqref{eq:boxcox-bahadur-initializer-rate} and
\eqref{eq:boxcox-bahadur-oracle-localization} into
\eqref{eq:boxcox-bahadur-remainder} yields
\begin{equation*}
\begin{split}
\sqrt{N_B}\max_{e\in\{0,a\}}
\|\mc R_{\mathrm{BC},e}\|_2
&\lesssim
M_{\mathrm{BC},l}\sqrt{\frac{q_{\mathrm{BC}}}{N_B}}
\left\{
\frac{a_{\mathrm{BC}}}{\kappa^2}
\sum_{j=1}^{B-1}n_j\lambda_j^2
+\frac{a_{\mathrm{BC}}^{1/2}}{\kappa}
\overline r_{\mathrm{BC},B}
\sum_{j=1}^{B-1}n_j\lambda_j
+n_B\overline r_{\mathrm{BC},B}^{\,2}
\right\}
\\
&\qquad
+\sqrt{N_B}\,q_{\mathrm{BC}}
\eta_{\mathrm{BC},\infty}
\overline r_{\mathrm{BC},B}=o(1)
\end{split}
\end{equation*}
by \eqref{eq:boxcox-bahadur-remainder-rate}.  Since
\eqref{eq:boxcox-bahadur-initializer-rate} also holds for $j=B$,
substitution in the left-hand side of
\maineqref{eq:inference-remainder-condition} replaces the two sums above by
$\sum_{j=1}^B$, replaces $M_{\mathrm{BC},l}$ by
$1+M_{\mathrm{BC},l}$, and leaves the $n_B$ term unchanged.  Its last
term is bounded by the one above because
$q_{\mathrm{BC}}\eta_{\mathrm{BC},\infty}\overline r_{\mathrm{BC},B}
\le\sqrt{N_B}q_{\mathrm{BC}}\eta_{\mathrm{BC},\infty}
\overline r_{\mathrm{BC},B}$.  Thus
\eqref{eq:boxcox-bahadur-remainder-rate} also proves
\maineqref{eq:inference-remainder-condition}.  Lemma
\ref{lem:oracle_bahadur} now proves
\eqref{eq:boxcox-reduced-bahadur}--
\eqref{eq:boxcox-full-bahadur}.

Taking a direct differentiation in \eqref{eq:boxcox-score-hessian} gives
\begin{equation}
\nabla_T\ell_{\mathrm{BC}}
(\theta^*;Y_i^{(j)},x_i^{(j)})
=-
\sum_{k=1}^K w_k
\frac{\varphi(\eta_{ik}^{*(j)})}
{\Phi(\eta_{ik}^{*(j)})\Phi(-\eta_{ik}^{*(j)})}
\{D_{ik}^{(j)}-\Phi(\eta_{ik}^{*(j)})\}
z_{ik,T}^{(j)}.
\label{eq:boxcox-bahadur-score-form}
\end{equation}
Since the thresholds are ordered, substitution in
\eqref{eq:boxcox-bahadur-score-covariance} yields
\begin{equation}
\begin{split}
\Kb_0
={}&\frac1{N_B}\sum_{j=1}^B\sum_{i=1}^{n_j}
\mathbb E\left[
\sum_{k=1}^K\sum_{k'=1}^K w_kw_{k'}
\frac{\varphi(\eta_{ik}^{*(j)})
\varphi(\eta_{ik'}^{*(j)})}
{\Phi(\eta_{i,\min(k,k')}^{*(j)})
 \Phi(-\eta_{i,\max(k,k')}^{*(j)})}
z_{ik,T}^{(j)}(z_{ik',T}^{(j)})^\top
\right].
\end{split}
\label{eq:boxcox-bahadur-score-covariance-expanded}
\end{equation}
The expectation includes $X$, and every threshold cross term is present.
For a unit $u$, the Cauchy--Schwarz inequality applied to
\eqref{eq:boxcox-bahadur-score-form} and the information identity for
each threshold yield
\[
u^\top\Kb_0u
\le\mathbb E\sum_{k=1}^K w_k
\frac{\varphi^2(z_k(X)^\top\theta^*)}
{\Phi(z_k(X)^\top\theta^*)\Phi(-z_k(X)^\top\theta^*)}
\{z_{k,T}(X)^\top u\}^2
=u^\top\Jb_0u.
\]
Thus, we get
\begin{equation}
K_l\le\lambda_{\min}(\Kb_0)
\le\lambda_{\max}(\Kb_0)\le J_u=:K_u
\label{eq:boxcox-bahadur-score-upper}
\end{equation}
by \eqref{eq:boxcox-bahadur-localization-and-score}.
We remark that, if $K>1$, \eqref{eq:boxcox-bahadur-score-covariance-expanded}
requires the sandwich covariance $\Kb_0$ rather than its
single-threshold counterpart $\Jb_0$.

Lemma~\ref{lem:V0-conditioning},
\eqref{eq:boxcox-bahadur-hessian-bounds},
\eqref{eq:boxcox-bahadur-localization-and-score}, and
\eqref{eq:boxcox-bahadur-score-upper} yield $\|\Omega_0\|_2\le K_l^{-1/2}$ and $\Omega_0\Kb_0\Omega_0^\top=\Ib_r$.
Equations \eqref{eq:boxcox-design-conditions} and
\eqref{eq:boxcox-bahadur-score-form} give
\begin{equation}
\max_{j,i}\|\Omega_0\nabla_T\ell_{\mathrm{BC}}
(\theta^*;Y_i^{(j)},x_i^{(j)})\|_2
\le K_l^{-1/2}G_{\mathrm{BC}}(B_1q_{\mathrm{BC}})^{1/2}.
\label{eq:boxcox-bahadur-score-envelope}
\end{equation}
Moreover, we get
\begin{equation}
\frac1{N_B}\sum_{j=1}^B\sum_{i=1}^{n_j}
\mathbb E\|\Omega_0\nabla_T\ell_{\mathrm{BC}}
(\theta^*;Y_i^{(j)},x_i^{(j)})\|_2^2
=\operatorname{tr}(\Omega_0\Kb_0\Omega_0^\top)=r.
\label{eq:boxcox-bahadur-score-second-moment}
\end{equation}
By \eqref{eq:boxcox-design-conditions} and
\eqref{eq:boxcox-initializer-constants}, we have
$G_{\mathrm{BC}}=\mathfrak g(-B_1)=O(1)$.
The threshold coordinates yield $c_{\mathrm{BC},B}\ge1$.
Multiplying \eqref{eq:boxcox-bahadur-score-second-moment} by the envelope in \eqref{eq:boxcox-bahadur-score-envelope}, and using
\eqref{eq:boxcox-bahadur-lyapunov}, verifies
Assumption~\mainref{assm:lyapunov}.
The scores indexed by $(j,i)$ are independent and centered.
Thus \citet[Theorem~1.1]{bentkus2005lyapunov} implies, for
$\Zb_r\sim N(\mathbf0_r,\Ib_r)$,
\begin{equation*}
\sup_{\mc B\ {\rm convex}}\left|
\PP\left\{\frac1{\sqrt{N_B}}\sum_{j=1}^B\sum_{i=1}^{n_j}
\Omega_0\nabla_T\ell_{\mathrm{BC}}
(\theta^*;Y_i^{(j)},x_i^{(j)})\in\mc B\right\}
-\PP(\Zb_r\in\mc B)\right|\lesssim
\frac{G_{\mathrm{BC}}q_{\mathrm{BC}}^{1/2}r^{5/4}}
{\sqrt{N_B}},
\end{equation*}
which is $o(1)$.
Finally, we obtain
\begin{equation*}
\Cb_0\theta_T^*-\tb=\hb_B,
\quad
\lambda_{\max}\{(\Cb_0\Cb_0^\top)^{-1}\}
=\lambda_{\max}\{(\Cb\Cb^\top)^{-1}\}
\le C_{\max}^2,
\quad
\|\hb_B\|_2=O\!\left(\sqrt{\frac r{N_B}}\right).
\end{equation*}
\end{proof}

\subsubsection{Chi-square approximation of test statistics}
\label{sec:BC_test_statistic_verification}

Let
$T=\overline{\mc M}_{\mathrm{BC}}$ and, for $e\in\{0,a\}$,
$\widehat{\overline{\mc M}}{}^{(k)}_e
=\mc M_{\mathrm{BC}}\cup\widehat{\mc S}^{(k)}_e$.  Use
$\Jb_0$, $\Cb_0$, $\Kb_0$, and $\Vb_0$ in
\eqref{eq:boxcox-bahadur-main-objects}--
\eqref{eq:boxcox-bahadur-score-covariance}.    
Indeed, \eqref{eq:boxcox-full-bahadur} and
\eqref{eq:boxcox-bahadur-score-covariance} identify the covariance of
the leading score term as
$\Cb_0\Jb_0^{-1}\Kb_0\Jb_0^{-1}\Cb_0^\top=\Vb_0$.

Define
$R_{\mathrm{BC}}=(\widetilde{\Cb}_{\mathrm{BC}}
\widetilde{\Cb}_{\mathrm{BC}}^\top)^{-1/2}
=(\Cb\Cb^\top)^{-1/2}$.  Then
\begin{equation}
R_{\mathrm{BC}}\widetilde{\Cb}_{\mathrm{BC}}
\theta_{\mc M_{\mathrm{BC}}}=R_{\mathrm{BC}}\tb
\quad\Longleftrightarrow\quad
\widetilde{\Cb}_{\mathrm{BC}}
\theta_{\mc M_{\mathrm{BC}}}=\tb,
\qquad
(R_{\mathrm{BC}}\widetilde{\Cb}_{\mathrm{BC}})
(R_{\mathrm{BC}}\widetilde{\Cb}_{\mathrm{BC}})^\top=\Ib_r.
\label{eq:boxcox-row-normalized-constraint}
\end{equation}
Use $(RVR^\top)^{-1}=R^{-\top}V^{-1}R^{-1}$ first with
$V=\Vb_0$ and then with $V=\Cb_0\Jb_0^{-1}\Cb_0^\top$.
Multiplication by $R_{\mathrm{BC}}$ leaves the feasible set, $\Pb_0$,
and both quadratic forms unchanged, and yields
\begin{equation}
\begin{gathered}
\left\|
(R_{\mathrm{BC}}\Vb_0R_{\mathrm{BC}}^\top)^{-1/2}
R_{\mathrm{BC}}\Cb_0\Jb_0^{-1}u
\right\|_2
=\|\Vb_0^{-1/2}\Cb_0\Jb_0^{-1}u\|_2
=\|\Omega_0u\|_2,
\\
N_B(R_{\mathrm{BC}}\hb_B)^\top
(R_{\mathrm{BC}}\Vb_0R_{\mathrm{BC}}^\top)^{-1}
(R_{\mathrm{BC}}\hb_B)
=N_B\hb_B^\top\Vb_0^{-1}\hb_B=e_B.
\end{gathered}
\label{eq:boxcox-row-normalized-identities}
\end{equation}
Assumption~\mainref{assm:local_alternative} implies
$\|R_{\mathrm{BC}}\|_2\le C_{\max}$.
Recall that $\chi^2(r,e_B)$ denotes the noncentral chi-square law with $r$ degrees of freedom and $e_B=N_B\hb_B^\top\Vb_0^{-1}\hb_B$.

\begin{proposition}
\label{prop:boxcox-test-statistic-verification}
Suppose that the conditions of
Propositions~\ref{prop:boxcox-strong-oracle-verification} and
\ref{prop:boxcox-bahadur-verification} hold.  Then
\begin{equation}
\begin{split}
\sqrt r\bigg[
\max_{e\in\{0,a\}}
\left\|\nabla^2_{T,T}\widehat{\mc Q}_B
(\widehat\theta_e^{\rm ora})-\Jb_0\right\|_2
+\left\|\widehat{\Kb}^{(B)}_{T,T}-\Kb_0\right\|_2
\bigg]=o_p(1).
\end{split}
\label{eq:boxcox-feasible-plugin-bound}
\end{equation}
Consequently, for every $k\ge K_{\mathrm{BC}}^{\rm LLA}+1$ and
$T'=T_W^{(k)}$ or $T_S^{(k)}$,
\begin{equation}
\sup_{z\ge0}
\left|
\PP(T'\le z)-\PP\{\chi^2(r,e_B)\le z\}
\right|\longrightarrow0.
\label{eq:boxcox-test-conditional-limit}
\end{equation}
Under $\hb_B=0$, the rejection probability at
$\chi_\alpha^2(r)$ converges to $\alpha$.
If, in addition,
\begin{equation*}
\|\hb_B\|_2
=O\!\left[
\left\{\frac{\min(q_{\mathrm{BC}}-r,r)}{N_B}\right\}^{1/2}
\right],
\end{equation*}
then the local-power conclusion of Corollary~\mainref{cor:test} holds.
\end{proposition}

\begin{proof}
For each $k\le K$, the threshold coordinate $p+k$ satisfies
$|z_{ik,p+k}^{(j)}|=1$ in the $k$th term.  Since
$\sum_{k=1}^Kw_k=1$, $w_k>0$, and $K$ is fixed,
\eqref{eq:boxcox-design-conditions} and
\eqref{eq:boxcox-terminal-design-bounds} yield
\begin{equation*}
c_{\mathrm{BC},B}\ge1,\qquad
B_1\ge\max_{k\le K}w_k\ge K^{-1},\qquad
\overline B_{2,\mathrm{BC}}\ge\max_{k\le K}w_k\ge K^{-1}.
\end{equation*}
Consequently, \eqref{eq:boxcox-terminal-concentration-rates} and
\eqref{eq:boxcox-terminal-strong-oracle-rates} imply
\begin{equation}
\eta_{\mathrm{BC},\infty}
\gtrsim\frac{c_{\mathrm{BC},B}}{\sqrt{N_B}},
\qquad
\overline r_{\mathrm{BC},B}
\ge\kappa^{-1}\chi_{\mathrm{BC}}(q_{\mathrm{BC}})
\gtrsim\sqrt{\frac{q_{\mathrm{BC}}}{N_B}}.
\label{eq:boxcox-test-lower-rates}
\end{equation}

For $e\in\{0,a\}$, the definition of $\widehat{\mc Q}_B$ gives
\begin{equation*}
\begin{split}
\nabla^2_{T,T}\widehat{\mc Q}_B
(\widehat\theta_e^{\rm ora})-\Jb_0
={}&(\overline H_{\mathrm{BC},B}^*-J_{\mathrm{BC},B})_{T,T}
+\frac1{N_B}\bigg[
\sum_{j=1}^{B-1}n_j\left\{
\nabla^2\widehat{\mc L}^{(j)}(\widehat\theta^{(j)})
-\nabla^2\widehat{\mc L}^{(j)}(\theta^*)
\right\}_{T,T}
\\
&\hspace{27mm}+n_B\left\{
\nabla^2\widehat{\mc L}^{(B)}(\widehat\theta_e^{\rm ora})
-\nabla^2\widehat{\mc L}^{(B)}(\theta^*)
\right\}_{T,T}\bigg].
\end{split}
\end{equation*}
The first inequality in
\eqref{eq:boxcox-terminal-concentration-event}, the Hessian
inequality in \eqref{eq:boxcox-terminal-loss-events},
\eqref{eq:boxcox-bahadur-initializer-rate}, and
\eqref{eq:boxcox-bahadur-oracle-localization} yield
\begin{equation}
\begin{split}
&\sqrt r\max_{e\in\{0,a\}}
\left\|\nabla^2_{T,T}\widehat{\mc Q}_B
(\widehat\theta_e^{\rm ora})-\Jb_0\right\|_2
\\
&\quad\lesssim\sqrt r\left[
q_{\mathrm{BC}}\eta_{\mathrm{BC},\infty}
+\frac{M_{\mathrm{BC},l}q_{\mathrm{BC}}^{1/2}}{N_B}
\left\{
\frac{a_{\mathrm{BC}}^{1/2}}{\kappa}
\sum_{j=1}^{B-1}n_j\lambda_j
+n_B\overline r_{\mathrm{BC},B}
\right\}\right]
\\
&\quad\lesssim
\sqrt{N_B}\,q_{\mathrm{BC}}
\eta_{\mathrm{BC},\infty}\overline r_{\mathrm{BC},B}
+M_{\mathrm{BC},l}\sqrt{\frac{q_{\mathrm{BC}}}{N_B}}
\left\{
\frac{a_{\mathrm{BC}}^{1/2}}{\kappa}
\overline r_{\mathrm{BC},B}
\sum_{j=1}^{B-1}n_j\lambda_j
+n_B\overline r_{\mathrm{BC},B}^{\,2}
\right\}
=o_p(1).
\end{split}
\label{eq:boxcox-test-hessian-rate}
\end{equation}
The second inequality uses
$\sqrt r\le q_{\mathrm{BC}}^{1/2}\lesssim
\sqrt{N_B}\,\overline r_{\mathrm{BC},B}$ from
\eqref{eq:boxcox-test-lower-rates} and
\eqref{eq:boxcox-terminal-kappa}; the last equality follows from
\eqref{eq:boxcox-bahadur-remainder-rate}.

On the event in
Proposition~\ref{prop:boxcox-strong-oracle-verification}, equations
\eqref{eq:boxcox-terminal-strong-oracle-recovery} and
\eqref{eq:feasible-support-identification} give, for
$k\ge K_{\mathrm{BC}}^{\rm LLA}+1$ and $e\in\{0,a\}$,
\begin{equation}
\begin{gathered}
\widehat{\overline{\mc M}}{}^{(k)}_e=T,
\qquad
\widehat{\Cb}^{(k)}_e=\Cb_0,
\\
\widehat{\Jb}^{(k)}_e
=\nabla^2_{T,T}\widehat{\mc Q}_B
(\widehat\theta_e^{\rm ora}),
\qquad
\widehat{\Kb}^{(k)}_0=\widehat{\Kb}^{(k)}_a
=\widehat{\Kb}^{(B)}_{T,T}.
\end{gathered}
\label{eq:boxcox-feasible-common-hessian}
\end{equation}

For the second norm in
\eqref{eq:boxcox-feasible-plugin-bound}, let
\begin{equation*}
s_i^{(j)}(\theta)
=\nabla_T\ell_{\mathrm{BC}}
(\theta;Y_i^{(j)},x_i^{(j)}),
\qquad
\widetilde{\Kb}_{\mathrm{BC},B}
=\frac1{N_B}\sum_{j=1}^B\sum_{i=1}^{n_j}
s_i^{(j)}(\theta^*)s_i^{(j)}(\theta^*)^\top.
\end{equation*}
Equations \eqref{eq:boxcox-score-centering} and
\eqref{eq:boxcox-bahadur-score-covariance} give
$\mathbb E\widetilde{\Kb}_{\mathrm{BC},B}=\Kb_0$.
The normal conditional-expectation identity used in
\eqref{eq:boxcox-score-centering}, Jensen's inequality, and
\eqref{eq:boxcox-design-conditions} imply
\[
\mathbb E\|s_i^{(j)}(\theta^*)\|_2^4
\le3\mathbb E\sum_{k=1}^Kw_k\|z_{ik,T}^{(j)}\|_2^4
\le3B_1^2q_{\mathrm{BC}}^2.
\]
Independence and
$\|ss^\top\|_{\rm F}^2=\|s\|_2^4$ therefore yield
\begin{equation}
\begin{split}
\mathbb E\|\widetilde{\Kb}_{\mathrm{BC},B}-\Kb_0\|_{\rm F}^2
\le\frac{3B_1^2q_{\mathrm{BC}}^2}{N_B},\qquad
\|\widetilde{\Kb}_{\mathrm{BC},B}-\Kb_0\|_2
=O_p\left(\frac{q_{\mathrm{BC}}}{\sqrt{N_B}}\right).
\end{split}
\label{eq:boxcox-truth-score-covariance-rate}
\end{equation}
Since $r\le q_{\mathrm{BC}}$, the last term in
\eqref{eq:boxcox-bahadur-remainder-rate} yields
\begin{equation}
\frac{c_{\mathrm{BC},B}q_{\mathrm{BC}}\sqrt r}{\sqrt{N_B}}
\lesssim
\sqrt{N_B}\,q_{\mathrm{BC}}
\eta_{\mathrm{BC},\infty}\overline r_{\mathrm{BC},B}
=o(1).
\label{eq:boxcox-test-truth-score-rate}
\end{equation}

For $j\le B$, set
$\Delta_j=\widehat\theta^{(j)}-\theta^*$ and
\begin{equation*}
\begin{split}
d_i^{(j)}
=s_i^{(j)}(\widehat\theta^{(j)})-s_i^{(j)}(\theta^*),
\qquad
\mathbf D_{\mathrm{BC},B}
=\frac1{N_B}\sum_{j=1}^B\sum_{i=1}^{n_j}
d_i^{(j)}(d_i^{(j)})^\top.
\end{split}
\end{equation*}
By \eqref{eq:boxcox-score-hessian}, the Cauchy--Schwarz inequality in the sum over $k$
and the integral over $t$, $0<\mathfrak h<1$, and
$\|z_{ik,T}^{(j)}\|_2^2\le B_1q_{\mathrm{BC}}$ give
\begin{equation}
\|d_i^{(j)}\|_2^2
\le B_1q_{\mathrm{BC}}\int_0^1
\Delta_j^\top\nabla^2\ell_{\mathrm{BC}}
(\theta^*+t\Delta_j;Y_i^{(j)},x_i^{(j)})\Delta_j\,dt.
\label{eq:boxcox-fitted-score-pointwise}
\end{equation}
On the initializer event in \eqref{eq:boxcox-score-event},
$\|\Delta_j\|_1\le4a_{\mathrm{BC}}^{1/2}\|\Delta_j\|_2$ for every
$j\le B$.  The quadratic-form bound proved before
\eqref{eq:boxcox-initializer-rlh},
\eqref{eq:boxcox-initializer-constants}, and $0<\mathfrak h<1$ yield,
for $0\le t\le1$,
\begin{align*}
\Delta_j^\top\nabla^2\widehat{\mc L}^{(j)}
(\theta^*+t\Delta_j)\Delta_j
\le\frac1{n_j}\sum_{i=1}^{n_j}\sum_{k=1}^K
w_k\{(z_{ik}^{(j)})^\top\Delta_j\}^2\leq\frac{B_3}{B_1^{1/2}}\|\Delta_j\|_2^2.
\end{align*}
Substitution in \eqref{eq:boxcox-fitted-score-pointwise} and
\eqref{eq:boxcox-initializer-rates} yields
\begin{equation}
\begin{split}
\|\mathbf D_{\mathrm{BC},B}\|_2
\le\operatorname{tr}(\mathbf D_{\mathrm{BC},B})
\le\frac{B_1^{1/2}B_3q_{\mathrm{BC}}}{N_B}
\sum_{j=1}^Bn_j\|\Delta_j\|_2^2
=O_p\!\left\{
\frac{q_{\mathrm{BC}}a_{\mathrm{BC}}}{\kappa^2N_B}
\sum_{j=1}^Bn_j\lambda_j^2
\right\}.
\end{split}
\label{eq:boxcox-fitted-score-quadratic}
\end{equation}
Consequently, we obtain
\begin{equation}
\begin{split}
r\|\mathbf D_{\mathrm{BC},B}\|_2
&=O_p\!\left[
\left\{r\sqrt{\frac{q_{\mathrm{BC}}}{N_B}}\right\}
\left\{\sqrt{\frac{q_{\mathrm{BC}}}{N_B}}
\frac{a_{\mathrm{BC}}}{\kappa^2}
\sum_{j=1}^Bn_j\lambda_j^2\right\}
\right]=o_p(1).
\end{split}
\label{eq:boxcox-fitted-score-quadratic-rate}
\end{equation}
The first factor is $o(1)$ by
\eqref{eq:boxcox-bahadur-lyapunov} and $c_{\mathrm{BC},B}\ge1$;
the second is $o(1)$ by \eqref{eq:boxcox-bahadur-remainder-rate}.

For every unit $v\in\mathbb R^{q_{\mathrm{BC}}}$,
the Cauchy--Schwarz inequality applied to
$v^\top s_i^{(j)}(\theta^*)$ and $v^\top d_i^{(j)}$ yields
\begin{equation}
\left\|\widehat{\Kb}^{(B)}_{T,T}
-\widetilde{\Kb}_{\mathrm{BC},B}\right\|_2
\le2\|\widetilde{\Kb}_{\mathrm{BC},B}\|_2^{1/2}
\|\mathbf D_{\mathrm{BC},B}\|_2^{1/2}
+\|\mathbf D_{\mathrm{BC},B}\|_2.
\label{eq:boxcox-fitted-score-covariance}
\end{equation}
Equations \eqref{eq:boxcox-bahadur-score-upper},
\eqref{eq:boxcox-truth-score-covariance-rate}, and
\eqref{eq:boxcox-test-truth-score-rate} give
$\|\widetilde{\Kb}_{\mathrm{BC},B}\|_2=O_p(1)$ and
$\sqrt r\|\widetilde{\Kb}_{\mathrm{BC},B}-\Kb_0\|_2=o_p(1)$.
Together with
\eqref{eq:boxcox-fitted-score-quadratic-rate}--
\eqref{eq:boxcox-fitted-score-covariance}, these inequalities yield
\begin{equation}
\begin{split}
\sqrt r\|\widehat{\Kb}^{(B)}_{T,T}-\Kb_0\|_2
\le{}&
\sqrt r\|\widetilde{\Kb}_{\mathrm{BC},B}-\Kb_0\|_2
\\
&+2\|\widetilde{\Kb}_{\mathrm{BC},B}\|_2^{1/2}
\{r\|\mathbf D_{\mathrm{BC},B}\|_2\}^{1/2}
+\sqrt r\|\mathbf D_{\mathrm{BC},B}\|_2
=o_p(1).
\end{split}
\label{eq:boxcox-test-score-covariance-rate}
\end{equation}

Equations \eqref{eq:boxcox-test-hessian-rate} and
\eqref{eq:boxcox-test-score-covariance-rate} prove
\eqref{eq:boxcox-feasible-plugin-bound}, hence
\maineqref{eq:feasible-plugin}.  Apply the proofs of
Theorems~\mainref{thm:oracle.test}--\mainref{thm:test} with the surrogate
curvature in \eqref{eq:boxcox-terminal-loss-events}, the localization and
Bahadur representation argument above, and
\eqref{eq:boxcox-row-normalized-constraint}--
\eqref{eq:boxcox-row-normalized-identities}.
The identities in \eqref{eq:boxcox-feasible-common-hessian} and
\eqref{eq:boxcox-feasible-plugin-bound} then provide
\eqref{eq:boxcox-test-conditional-limit} and conclude its null-size and local-power claims.
\end{proof}

By \eqref{eq:boxcox-initializer-tuning} and
Lemma~\ref{lem:batch.size}, note that
\[
\frac1{N_B}\sum_{j=1}^Bn_j\lambda_j^2
\lesssim\frac{\log d_{\mathrm{BC}}}{N_B}
\left\{1+\log\left(\frac{N_B}{n_1}\right)\right\}.
\]
We remark that equations
\eqref{eq:boxcox-truth-score-covariance-rate},
\eqref{eq:boxcox-fitted-score-quadratic}, and
\eqref{eq:boxcox-fitted-score-covariance} show that the covariance
bound \eqref{eq:boxcox-test-score-covariance-rate} is implied by
\begin{equation*}
\sqrt r\left[
\frac{q_{\mathrm{BC}}}{\sqrt{N_B}}
+
\left\{
\frac{q_{\mathrm{BC}}a_{\mathrm{BC}}\log d_{\mathrm{BC}}}
{\kappa^2N_B}
\left\{1+\log\left(\frac{N_B}{n_1}\right)\right\}
\right\}^{1/2}
\right]=o(1).
\end{equation*}

\subsubsection{Additional simulation III}
\label{sec:BC_additional_simulation_III}

In this section, we examine the finite-sample inference performance of the renewable Wald and score
tests in the nonparametric Box--Cox model. 
We generate covariate vectors
$X_i\sim N(\mathbf{0}_p,\Sigma)$ with
$\Sigma_{jk}=0.5^{|j-k|}$ and responses satisfying
\[
Z_i=1+X_i^\T\betab^*+\epsilon_i,\qquad
Y_i=\frac12\sqrt[3]{2Z_i-1}+\frac12,\qquad
\epsilon_i\sim N(0,1),
\]
where the errors are independent of the covariates. This corresponds to the increasing transformation
$g_1(y)=\{(2y-1)^3+1\}/2$; the intercept in $Z_i$ is absorbed into the
unpenalized threshold parameters in the composite-probit formulation of \citet{zhouzou2024boxcox,zhouzou2024testing}.
We set $p=900$, $n_b=150$, and $B=8$, yielding $N_B=1200$, and take
\[
\betab^*=(h,0,0,1+h,1,\mathbf{0}_{p-5}^\T)^\T,
\qquad h\in\{0,0.025,0.05,0.1,0.2\}.
\]
For each value of $h$, we use 500 independent replications and test the
following five hypotheses:
\begin{align*}
H_{0,\mathrm{BC}}^{(1)} &: \ \beta_1^*=0, &
H_{0,\mathrm{BC}}^{(2)} &: \ \beta_4^*-\beta_5^*=0,\\
H_{0,\mathrm{BC}}^{(3)} &: \ (\beta_1^*,\beta_8^*,\beta_9^*)^\T=\mathbf{0}_3, &
H_{0,\mathrm{BC}}^{(4)} &: \ \beta_1^*=0,\ \beta_4^*-\beta_5^*=0,\\
H_{0,\mathrm{BC}}^{(5)} &: \ (\beta_1^*,\beta_8^*,\beta_9^*)^\T=\mathbf{0}_3,
\ \beta_4^*-\beta_5^*=0. &&
\end{align*}
Again, the corresponding degrees of freedom are $1,1,3,2$, and $4$.
The methods, tuning procedures, LLA reweighting mechanism, competitor setups, and nominal level are those in Section~\mainref{sec:numerical}.

We use $K=20$ thresholds at probability levels equally spaced between
$0.05$ and $0.95$, with equal composite-likelihood weights.
The online procedures estimate these thresholds from the empirical
quantiles of the first batch and keep them throughout the stream.
The offline benchmark instead uses empirical quantiles of the pooled
sample. 

The four renewable procedures and the pooled benchmark use the historical
estimation strategies, tuning procedures, and LLA reweighting counts in
Section~\mainref{sec:numerical}, applied to the
composite-probit loss.
All tests use the nominal level $\alpha=0.05$.

Table~\ref{tab:BC-additional-III} reports the empirical rejection
probabilities. With LLA $\to$ LLA estimation, the renewable Wald and score tests have null rejection probabilities close to the nominal level
under both tuning rules. Their powers are generally close to those of the
pooled benchmark and increase with $h$.

With Lasso $\to$ LLA estimation, CV/GIC tuning yields elevated null rejection probabilities for the equality and joint hypotheses. The 1-SE
rule reduces this discrepancy in several cases, although $H_{0,\mathrm{BC}}^{(2)}$, $H_{0,\mathrm{BC}}^{(4)}$, and $H_{0,\mathrm{BC}}^{(5)}$ still exhibit overrejection. In this setting, using LLA throughout the stream improves null calibration while preserving power broadly comparable to the offline benchmark procedure.

\begin{table}[p]
\centering\footnotesize
\setlength{\tabcolsep}{6pt}
\renewcommand{\arraystretch}{1.02}
\caption{Empirical rejection probabilities for the nonparametric Box--Cox model with Gaussian covariates, $p=900$, $n_b=150$, $B=8$, and 500 independent replications per value of $h$.}
\label{tab:BC-additional-III}
\begin{tabular}{llr*{8}{r}}
\toprule
 & & & \multicolumn{4}{c}{1-SE} & \multicolumn{4}{c}{CV/GIC}\\
\cmidrule(lr){4-7}\cmidrule(lr){8-11}
 & & & \multicolumn{2}{c}{Lasso $\to$ LLA} & \multicolumn{2}{c}{LLA $\to$ LLA} & \multicolumn{2}{c}{Lasso $\to$ LLA} & \multicolumn{2}{c}{LLA $\to$ LLA}\\
\cmidrule(lr){4-5}\cmidrule(lr){6-7}\cmidrule(lr){8-9}\cmidrule(lr){10-11}
Hypothesis & $h$ & {\tt Offline} & $T_W$ & $T_S$ & $T_W$ & $T_S$ & $T_W$ & $T_S$ & $T_W$ & $T_S$\\
\midrule
\multirow{6}{*}{$H_{0,\mathrm{BC}}^{(1)}$} & 0.000 & 0.034 & 0.042 & 0.036 & 0.044 & 0.044 & 0.052 & 0.052 & 0.044 & 0.044\\
 & 0.025 & 0.134 & 0.128 & 0.112 & 0.126 & 0.124 & 0.172 & 0.172 & 0.134 & 0.134\\
 & 0.050 & 0.416 & 0.430 & 0.404 & 0.422 & 0.422 & 0.496 & 0.496 & 0.428 & 0.428\\
 & 0.100 & 0.900 & 0.898 & 0.878 & 0.898 & 0.898 & 0.928 & 0.928 & 0.902 & 0.902\\
 & 0.200 & 1.000 & 1.000 & 1.000 & 1.000 & 1.000 & 1.000 & 1.000 & 1.000 & 1.000\\
\addlinespace
\multirow{6}{*}{$H_{0,\mathrm{BC}}^{(2)}$} & 0.000 & 0.048 & 0.092 & 0.086 & 0.058 & 0.058 & 0.094 & 0.094 & 0.056 & 0.056\\
 & 0.025 & 0.082 & 0.126 & 0.122 & 0.080 & 0.080 & 0.130 & 0.132 & 0.086 & 0.086\\
 & 0.050 & 0.146 & 0.188 & 0.188 & 0.142 & 0.142 & 0.204 & 0.202 & 0.144 & 0.144\\
 & 0.100 & 0.400 & 0.504 & 0.446 & 0.414 & 0.414 & 0.542 & 0.542 & 0.426 & 0.426\\
 & 0.200 & 0.938 & 0.924 & 0.904 & 0.932 & 0.914 & 0.958 & 0.958 & 0.936 & 0.936\\
\addlinespace
\multirow{6}{*}{$H_{0,\mathrm{BC}}^{(3)}$} & 0.000 & 0.050 & 0.048 & 0.048 & 0.046 & 0.046 & 0.094 & 0.094 & 0.056 & 0.056\\
 & 0.025 & 0.104 & 0.092 & 0.086 & 0.088 & 0.088 & 0.160 & 0.160 & 0.094 & 0.094\\
 & 0.050 & 0.292 & 0.260 & 0.244 & 0.258 & 0.260 & 0.380 & 0.380 & 0.276 & 0.276\\
 & 0.100 & 0.798 & 0.794 & 0.752 & 0.780 & 0.780 & 0.860 & 0.860 & 0.790 & 0.790\\
 & 0.200 & 1.000 & 1.000 & 1.000 & 1.000 & 1.000 & 1.000 & 1.000 & 1.000 & 1.000\\
\addlinespace
\multirow{6}{*}{$H_{0,\mathrm{BC}}^{(4)}$} & 0.000 & 0.054 & 0.080 & 0.070 & 0.050 & 0.050 & 0.102 & 0.102 & 0.048 & 0.048\\
 & 0.025 & 0.120 & 0.158 & 0.142 & 0.112 & 0.112 & 0.204 & 0.204 & 0.116 & 0.116\\
 & 0.050 & 0.402 & 0.426 & 0.400 & 0.366 & 0.366 & 0.514 & 0.514 & 0.386 & 0.386\\
 & 0.100 & 0.926 & 0.938 & 0.908 & 0.934 & 0.934 & 0.958 & 0.958 & 0.934 & 0.934\\
 & 0.200 & 1.000 & 1.000 & 1.000 & 1.000 & 1.000 & 1.000 & 1.000 & 1.000 & 1.000\\
\addlinespace
\multirow{6}{*}{$H_{0,\mathrm{BC}}^{(5)}$} & 0.000 & 0.062 & 0.072 & 0.056 & 0.058 & 0.056 & 0.106 & 0.106 & 0.064 & 0.064\\
 & 0.025 & 0.110 & 0.134 & 0.126 & 0.104 & 0.104 & 0.188 & 0.188 & 0.116 & 0.116\\
 & 0.050 & 0.326 & 0.358 & 0.324 & 0.308 & 0.308 & 0.416 & 0.416 & 0.326 & 0.328\\
 & 0.100 & 0.870 & 0.898 & 0.848 & 0.868 & 0.868 & 0.936 & 0.936 & 0.872 & 0.872\\
 & 0.200 & 1.000 & 1.000 & 1.000 & 1.000 & 1.000 & 1.000 & 1.000 & 1.000 & 1.000\\
\bottomrule
\end{tabular}
\end{table}

  \clearpage
\fi
\bibliography{allref}

@article{yang2026communication,
  title={Communication-efficient and distributed-oracle estimation for high-dimensional quantile regression},
  author={Yang, Songshan and Gu, Yifan and Yang, Hanfang and He, Xuming},
  journal={The Annals of Statistics},
  volume={54},
  number={1},
  pages={442--465},
  year={2026},
  publisher={Institute of Mathematical Statistics}
}

@article{vdGeer2014,
author = {Sara van de Geer and Peter B{\"u}hlmann and Ya’acov Ritov and Ruben Dezeure},
title = {{On asymptotically optimal confidence regions and tests for high-dimensional models}},
volume = {42},
journal = {The Annals of Statistics},
number = {3},
publisher = {Institute of Mathematical Statistics},
pages = {1166 -- 1202},
year = {2014},
doi = {10.1214/14-AOS1221},
URL = {https://doi.org/10.1214/14-AOS1221}
}

@article{zhang2014confidence,
  title={Confidence intervals for low dimensional parameters in high dimensional linear models},
  author={Zhang, Cun-Hui and Zhang, Stephanie S},
  journal={Journal of the Royal Statistical Society Series B: Statistical Methodology},
  volume={76},
  number={1},
  pages={217--242},
  year={2014},
  publisher={Oxford University Press}
}

@article{han2026adaptive,
  title={Adaptive debiased lasso in high-dimensional generalized linear models with streaming data},
  author={Han, Ruijian and Luo, Lan and Luo, Yuanhang and Lin, Yuanyuan and Huang, Jian},
  journal={Journal of the American Statistical Association},
  pages={1--11},
  year={2026},
  publisher={Taylor \& Francis}
}

@article{zhang2026renewable,
  title={Renewable $\ell_1$-regularized linear support vector machine with high-dimensional streaming data},
  author={Zhang, Na and Xie, Jinhan and Yan, Xiaodong and Jiang, Bei and Li, Ting and Kong, Linglong},
  journal={Journal of Computational and Graphical Statistics},
  volume={35},
  pages={1177--1189},
  year={2026},
  publisher={Taylor \& Francis}
}

@article{jiang2022renewable,
  title={Renewable quantile regression for streaming data sets},
  author={Jiang, Rong and Yu, Keming},
  journal={Neurocomputing},
  volume={508},
  pages={208--224},
  year={2022},
  publisher={Elsevier}
}

@article{schifano2016online,
  title={Online updating of statistical inference in the big data setting},
  author={Schifano, Elizabeth D and Wu, Jing and Wang, Chun and Yan, Jun and Chen, Ming-Hui},
  journal={Technometrics},
  volume={58},
  number={3},
  pages={393--403},
  year={2016},
  publisher={Taylor \& Francis}
}

@article{gao2026renewable,
  title={Renewable Lasso without Batch-Number Constraints: A Gradient-Enhanced Approach},
  author={Gao, Junzhuo and Peng, Ling and Guo, Xu and Lian, Heng},
  journal={arXiv preprint arXiv:2606.11738},
  year={2026}
}

@article{olsen1978note,
  title={Note on the uniqueness of the maximum likelihood estimator for the Tobit model},
  author={Olsen, Randall J},
  journal={Econometrica: Journal of the Econometric Society},
  pages={1211--1215},
  year={1978},
  publisher={JSTOR}
}

@article{liu2026high,
  title={High-dimensional online learning via asynchronous decomposition: Non-divergent results, dynamic regularization, and beyond},
  author={Liu, Shixiang and Li, Zhifan and Yang, Hanming and Yin, Jianxin},
  journal={arXiv preprint arXiv:2603.20696},
  year={2026}
}

@article{rissanen1986stochastic,
  title={Stochastic complexity and modeling},
  author={Rissanen, Jorma},
  journal={The annals of statistics},
  pages={1080--1100},
  year={1986},
  publisher={JSTOR}
}

@article{tashman2000out,
  title={Out-of-sample tests of forecasting accuracy: An analysis and review},
  author={Tashman, Leonard J.},
  journal={International Journal of Forecasting},
  volume={16},
  number={4},
  pages={437--450},
  year={2000},
  publisher={Elsevier}
}

@article{bentkus2005lyapunov,
  title={A Lyapunov-type bound in $\mathbb{R}^d$},
  author={Bentkus, Vidmantas},
  journal={Theory of Probability \& Its Applications},
  volume={49},
  number={2},
  pages={311--323},
  year={2005},
  publisher={SIAM}
}

@article{loh2015regularized,
  title={Regularized {M}-estimators with nonconvexity: Statistical and algorithmic theory for local optima},
  author={Loh, Po-Ling and Wainwright, Martin J.},
  journal={Journal of Machine Learning Research},
  volume={16},
  number={19},
  pages={559--616},
  year={2015}
}

@article{jiang2024renewable,
  title={Renewable Huber estimation method for streaming datasets},
  author={Jiang, Rong and Liang, Lei and Yu, Keming},
  journal={Electronic Journal of Statistics},
  volume={18},
  number={1},
  pages={674--705},
  year={2024},
  publisher={The Institute of Mathematical Statistics and the Bernoulli Society}
}

@article{luo2023online,
  title={Online inference in high-dimensional generalized linear models with streaming data},
  author={Luo, Lan and Han, Ruijian and Lin, Yuanyuan and Huang, Jian},
  journal={Electronic journal of statistics},
  volume={17},
  number={2},
  pages={3443},
  year={2023}
}

@article{negahban2012unified,
author = {Sahand N. Negahban and Pradeep Ravikumar and Martin J. Wainwright and Bin Yu},
title = {{A Unified Framework for High-Dimensional Analysis of $M$-Estimators with Decomposable Regularizers}},
volume = {27},
journal = {Statistical Science},
number = {4},
publisher = {Institute of Mathematical Statistics},
pages = {538 -- 557},
year = {2012},
doi = {10.1214/12-STS400},
URL = {https://doi.org/10.1214/12-STS400}
}

@article{fan2014strong,
  title={Strong oracle optimality of folded concave penalized estimation},
  author={Fan, Jianqing and Xue, Lingzhou and Zou, Hui},
  journal={Annals of Statistics},
  volume={42},
  number={3},
  pages={819},
  year={2014}
}

@article{jacobson2025strong,
  title={Strong oracle guarantees for partial penalized tests of high-dimensional generalized linear models},
  author={Jacobson, Tate},
  journal={Journal of the Royal Statistical Society Series B: Statistical Methodology},
  volume={87},
  number={4},
  pages={1150--1170},
  year={2025},
  publisher={Oxford University Press UK},
  doi={10.1093/jrsssb/qkaf010}
}

@article{jacobsonzou2024tobit,
  title={High-Dimensional Censored Regression via the Penalized {Tobit} Likelihood},
  author={Jacobson, Tate and Zou, Hui},
  journal={Journal of Business \& Economic Statistics},
  volume={42},
  number={1},
  pages={286--297},
  year={2024},
  doi={10.1080/07350015.2023.2182309}
}

@article{jacobsonzou2025testing,
  title={Linear Hypothesis Testing for High Dimensional {Tobit} Models},
  author={Jacobson, Tate and Zou, Hui},
  journal={Statistica Sinica},
  volume={35},
  number={4},
  pages={2185--2204},
  year={2025},
  doi={10.5705/ss.202023.0118}
}

@article{zhouzou2024boxcox,
  title={The nonparametric {Box--Cox} model for high-dimensional regression analysis},
  author={Zhou, He and Zou, Hui},
  journal={Journal of Econometrics},
  volume={239},
  number={2},
  pages={105419},
  year={2024},
  doi={10.1016/j.jeconom.2023.01.025}
}

@article{zhouzou2024testing,
  title   = {A Non-Parametric Box-Cox Approach to Robustifying High-Dimensional Linear Hypothesis Testing},
  author  = {Zhou, He and Zou, Hui},
  journal = {arXiv preprint arXiv:2405.12816},
  year    = {2024}
}

@article{shi2019linear,
  title={Linear hypothesis testing for high dimensional generalized linear models},
  author={Shi, Chengchun and Song, Rui and Chen, Zhao and Li, Runze},
  journal={Annals of Statistics},
  volume={47},
  number={5},
  pages={2671},
  year={2019}
}

@book{mccullagh1989generalized,
  title={Generalized linear models},
  author={McCullagh, Peter and Nelder, John A and McCullagh, P},
  volume={2},
  year={1989},
  publisher={Chapman and hall London}
}

@article{zou2008one,
  title={One-step sparse estimates in nonconcave penalized likelihood models},
  author={Zou, Hui and Li, Runze},
  journal={Annals of statistics},
  volume={36},
  number={4},
  pages={1509},
  year={2008}
}

@misc{han2021online,
  title        = {Online Debiased Lasso for Streaming Data},
  author       = {Han, Ruijian and Luo, Lan and Lin, Yuanyuan and Huang, Jian},
  year         = {2021},
  howpublished = {arXiv preprint arXiv:2106.05925},
  note         = {Available at \url{https://arxiv.org/abs/2106.05925}}
}

@article{fan2001variable,
  title={Variable selection via nonconcave penalized likelihood and its oracle properties},
  author={Fan, Jianqing and Li, Runze},
  journal={Journal of the American statistical Association},
  volume={96},
  number={456},
  pages={1348--1360},
  year={2001},
  publisher={Taylor \& Francis}
}

@article{Zhang2010MCP,
author = {Cun-Hui Zhang},
title = {{Nearly unbiased variable selection under minimax concave penalty}},
volume = {38},
journal = {The Annals of Statistics},
number = {2},
publisher = {Institute of Mathematical Statistics},
pages = {894 -- 942},
year = {2010},
doi = {10.1214/09-AOS729},
URL = {https://doi.org/10.1214/09-AOS729}
}

@article{fan2018lamm,
  title={I-LAMM for sparse learning: Simultaneous control of algorithmic complexity and statistical error},
  author={Fan, Jianqing and Liu, Han and Sun, Qiang and Zhang, Tong},
  journal={Annals of statistics},
  volume={46},
  number={2},
  pages={814},
  year={2018}
}

@article{han2024inference,
  title={Inference on high-dimensional single-index models with streaming data},
  author={Han, Dongxiao and Xie, Jinhan and Liu, Jin and Sun, Liuquan and Huang, Jian and Jiang, Bei and Kong, Linglong},
  journal={Journal of Machine Learning Research},
  volume={25},
  number={337},
  pages={1--68},
  year={2024}
}

@article{gu2025collaborative,
  title={Collaborative Inference for Sparse High-Dimensional Models with Non-Shared Data},
  author={Gu, Yifan and Yang, Hanfang and Yang, Songshan and Zou, Hui},
  journal={Journal of the American Statistical Association},
  volume={just-accepted},
  pages={1--30},
  year={2026},
  publisher={Taylor \& Francis}
}

@article{luo2020renewable,
  title={Renewable estimation and incremental inference in generalized linear models with streaming data sets},
  author={Luo, Lan and Song, Peter X-K},
  journal={Journal of the Royal Statistical Society Series B: Statistical Methodology},
  volume={82},
  number={1},
  pages={69--97},
  year={2020},
  publisher={Oxford University Press}
}

@article{xie2024statistical,
  title={Statistical inference for smoothed quantile regression with streaming data},
  author={Xie, Jinhan and Yan, Xiaodong and Jiang, Bei and Kong, Linglong},
  journal={Journal of Econometrics},
  volume={249},
  pages={105924},
  year={2025},
  publisher={Elsevier}
}

@article{han2024online,
  title={Online inference with debiased stochastic gradient descent},
  author={Han, Ruijian and Luo, Lan and Lin, Yuanyuan and Huang, Jian},
  journal={Biometrika},
  volume={111},
  number={1},
  pages={93--108},
  year={2024},
  publisher={Oxford University Press}
}

@misc{capitalbikeshare,
  author       = {{Capital Bikeshare}},
  title        = {Capital Bikeshare Trip History Data},
  howpublished = {Data set},
  url          = {https://capitalbikeshare.com/system-data},
  note         = {Retrieved September 24, 2026},
  year = {2024}
}

@misc{kantor2023lcdv2,
  author       = {Kantor, Diana and Casey, Nancy W. and Menne, Matthew J. and Buddenberg, Andrew},
  title        = {Local Climatological Data ({LCD}), Version 2},
  year         = {2023},
  doi          = {10.25921/96dw-mb77},
  url          = {https://doi.org/10.25921/96dw-mb77},
  note         = {Subset used: Ronald Reagan Washington National Airport, 2024. {NOAA National Centers for Environmental Information}. Accessed 2026-09-24}
}
\end{document}